\documentclass[11pt]{article}

\usepackage{color, amsmath, amssymb, amsthm, graphicx, bbm, footnote, mathtools, bm, enumitem, subcaption}

\usepackage[top=1in, bottom=1in, left=1in, right=1in]{geometry}

\usepackage[toc,page,header]{appendix}
\usepackage{titletoc}
\usepackage{titling}
\allowdisplaybreaks

\usepackage{setspace}
\usepackage{adjustbox}
\usepackage[colorlinks,allcolors=blue]{hyperref}
\usepackage[round]{natbib}  
\usepackage{authblk}
\usepackage[para]{footmisc}  

\usepackage{apptools}
\AtAppendix{\counterwithin{theorem}{section}}
\AtAppendix{\counterwithin{remark}{section}}

\newtheorem{theorem}{Theorem}

\newtheorem{proposition}[theorem]{Proposition}
\newtheorem{corollary}[theorem]{Corollary}

\newtheorem{remark}{Remark}
\newtheorem{assump}{Assumption}

\newtheorem{innercustomthm}{Theorem}
\newenvironment{customthm}[1]
{\renewcommand\theinnercustomthm{#1}\innercustomthm}
{\endinnercustomthm}

\newtheorem{innercustomprop}{Proposition}
\newenvironment{customprop}[1]
{\renewcommand\theinnercustomprop{#1}\innercustomprop}
{\endinnercustomprop}

\newcommand{\ind}{{\perp\!\!\!\perp}}

\newcommand{\norm}[1]{\left|\left| #1 \right|\right|}

\renewcommand{\hat}{\widehat}
\renewcommand{\tilde}{\widetilde}

\renewcommand{\P}{\mbox{$\mathrm{P}$}}
\newcommand{\E}{\mbox{$\mathbb{E}$}}

\definecolor{emerald}{rgb}{0.31, 0.78, 0.47}

\title{\textbf{Doubly Robust Inference on Causal Derivative Effects for Continuous Treatments}}
\author{Yikun Zhang$^{\ast}$ \;\text{and}\;  Yen-Chi Chen$^{\dagger}$}
\date{\normalsize	
	\textit{Department of Statistics, University of Washington}\\
		$^{\ast}$\href{mailto:yikun@uw.edu}{yikun@uw.edu}\;\; $^{\dagger}$\href{mailto:yenchic@uw.edu}{yenchic@uw.edu} \\~\\	
	\today}
\begin{document}
	\maketitle

\begin{abstract}
Statistical methods for causal inference with continuous treatments mainly focus on estimating the mean potential outcome function, commonly known as the dose-response curve. However, it is often not the dose-response curve but its derivative function that signals the treatment effect. 
In this paper, we investigate nonparametric inference on the derivative of the dose-response curve with and without the positivity condition.
Under the positivity and other regularity conditions, we propose a doubly robust (DR) inference method for estimating the derivative of the dose-response curve using kernel smoothing.
When the positivity condition is violated, we demonstrate the inconsistency of conventional inverse probability weighting (IPW) and DR estimators, and introduce novel bias-corrected IPW and DR estimators. 
In all settings, our DR estimator achieves asymptotic normality at the standard nonparametric rate of convergence with nonparametric efficiency guarantees.
Additionally, our approach reveals an interesting connection to nonparametric support and level set estimation problems. 
Finally, we demonstrate the applicability of our proposed estimators through simulations and a case study of evaluating a job training program.
\end{abstract}
	
\section{Introduction}
\label{sec:Intro}

This paper investigates the construction of a doubly robust estimator for the derivative of the continuous treatment effect using kernel smoothing. The analysis considers scenarios both with and without the positivity condition. Specifically, positivity (Assumption~\ref{assump:positivity}) requires that every individual has a nonzero chance (measured by a conditional density function) of being exposed to any treatment level $T=t$ across all possible values of the covariate vector $\bm{S}\in \mathcal{S}\subset \mathbb{R}^d$. Let $Y(t)$ be the potential outcome \citep{rubin1974estimating} that would have been observed under treatment level $T=t$. The focus of this work is the (causal) derivative effect curve $t\mapsto \theta(t) := \frac{d}{dt}\mathbb{E}\left[Y(t)\right]$, where $t\mapsto m(t):= \mathbb{E}\left[Y(t)\right]$ represents the (causal) dose-response curve.

Valid inference on $\theta(t)$ is essential for understanding how the outcome of interest $Y\in \mathcal{Y}$ changes with treatment $t$, offering insights beyond the expected value $\E\left[Y(t)\right]=m(t)$ of the potential outcome across the population. Indeed, the derivative effect curve $\theta(t)$ serves as the most natural continuous-treatment counterpart to the average treatment effect $\E\left[Y(1)\right] - \E\left[Y(0)\right]$ in binary treatment settings. Despite the importance of estimating $\theta(t)$, the current research on continuous treatments has primarily focused on inferring $m(t)$ \citep{diaz2013targeted,kennedy2017non,bonvini2022fast,takatsu2022debiased} or the average derivative effect $\E\left[\theta(T)\right] = \E\left[\frac{\partial}{\partial t}\E\left(Y|T,\bm{S}\right)\right]$ under some regularity conditions \citep{hardle1989investigating,powell1989semiparametric,newey1993efficiency,cattaneo2010robust,hirshberg2020debiased,hines2023optimally}, with limited attention to $\theta(t)$ itself.
As demonstrated by \autoref{fig:deriv_illust}, a dose-response curve $m(t)$ attains the same value at two distinct treatment levels $t_1,t_2$, and the associated average derivative effect $\mathbb{E}\left[\theta(T)\right]$ is identical to zero. Nevertheless, the actual causal effects differ and remain nonzero at $t_1,t_2$, which can be effectively captured by our target estimand $\theta(t)$. In reality, $\mathbb{E}\left[\theta(T)\right]$ only quantifies the overall causal effects, while $\theta(t)$ provides more precise treatment effects at a personalized level of interest.

\begin{figure}[t]
	\centering
	\includegraphics[width=1\linewidth]{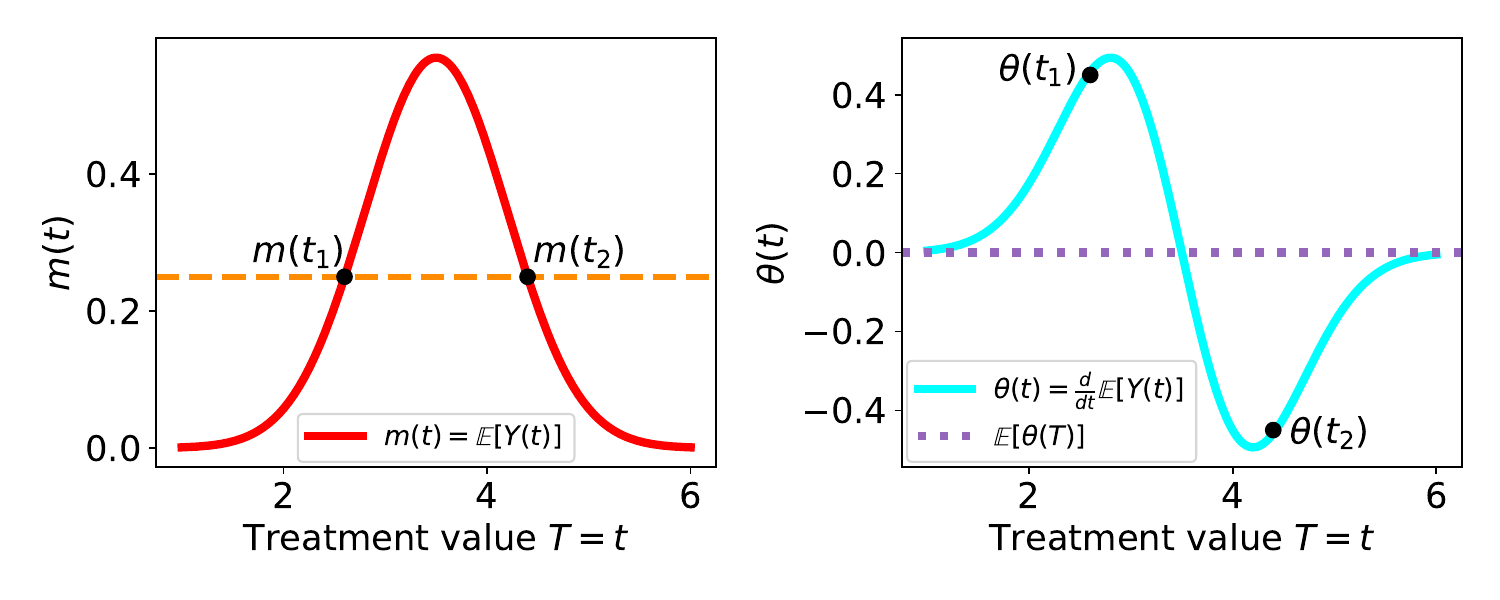}
	\caption{Illustrations of a dose-response curve $m(t)$, its corresponding derivative effect curve $\theta(t)$, and average derivative effects $\mathbb{E}\left[\theta(T)\right]$. Here, $m(t)$ is symmetric with respect to $t=3.5$ so that $m(t_1)=m(t_2)$ and $\mathbb{E}\left[\theta(T)\right]=0$ but $\theta(t_1) > 0 > \theta(t_2)$.}
	\label{fig:deriv_illust}
\end{figure}

To achieve precise inference on $\theta(t)$ without numerical approximation, a straightforward approach is to impose structural assumptions on the conditional mean outcome function $\E\left(Y|T=t,\bm{S}=\bm{s}\right)$ or directly on the dose-response curve $m(t)$, known as the marginal structural modeling \citep{robins2000marginal,neugebauer2007nonparametric}. Although this approach can easily construct an estimator of $\theta(t)$ via a standard differentiation on the estimated dose-response curve, those structural assumptions are difficult to verify in practice. Alternatively, existing methods for derivative estimation \citep{gasser1984estimating,mack1989derivative,zhou2000derivative}, combined with the inverse probability weighting (IPW) technique \citep{hirano2004propensity,imai2004causal}, can define an estimator of $\theta(t)$; see, \emph{e.g.}, our proposed IPW estimator in \autoref{sec:theta_pos}. Yet, this approach requires correct specification and accurate estimation of the conditional density of $T$ given $\bm{S}$. The sensitivity of these approaches to model misspecification and estimation motivates us to propose a doubly robust (DR) inference procedures for $\theta(t)$ that accommodates misspecification in either the outcome regression or the conditional density models \citep{robins1986new,van2003unified,bang2005doubly} while imposing less stringent requirements on the estimation rates of convergence.

The existing inference methods for $m(t)$ and the above discussion on $\theta(t)$ relies on the positivity condition (Assumption~\ref{assump:positivity}), which may be violated in observational studies with continuous treatments \citep{cole2008constructing,westreich2010invited}. When positivity fails, the identifications of both $m(t)$ and $\theta(t)$ become infeasible without structural assumptions; see \autoref{sec:inconsistency} for details. \cite{zhang2024nonparametric} address this problem without positivity by imposing an assumption on the potential outcome model that can be satisfied by additive confounding models and proposing a regression adjustment (RA) estimator of $\theta(t)$. We extend their identification and estimation strategies to propose IPW and DR estimators of $\theta(t)$ under additive confounding models. This extension not only advances the field but also reveals novel connections between the derivative effect curve inference and classical support estimation problems \citep{cuevas1997plug,cuevas2009set}.

\subsection{Contributions and Outline of the Paper}

\hspace{4pt} {\bf 1. Identification and Estimation:} Under the positivity and other regularity conditions that are stated in \autoref{sec:background}, we propose our IPW and DR estimators of $\theta(t)$ using kernel smoothing in \autoref{sec:theta_pos}. 
In particular, our proposed DR estimator leverages a local polynomial approximation to the outcome variable and is robust to the misspecification of either the outcome regression or the conditional density models. 

{\bf 2. Challenges and Remedies Under Violations of Positivity:}
When the positivity condition is violated, we demonstrate that the conventional approaches exhibit two types of bias due to lack of identification and support discrepancy in \autoref{sec:inconsistency}.
To resolve these issues, we adopt the additive structural assumption to maintain identification and utilize nonparametric set estimation techniques to develop our bias-corrected IPW and DR estimators of $\theta(t)$ in \autoref{subsec:bias_corrected_IPW_DR}.

{\bf 3. Asymptotic Theory:} We establish the consistency and asymptotic properties of RA, IPW, and DR estimators of $\theta(t)$ when the nuisance functions are nonparametrically estimated under cross-fitting; see \autoref{subsec:asymp_theta_pos} with positivity and \autoref{subsec:asymp_theta_nopos} without positivity. Specifically, our proposed DR estimators are asymptotically normal and can be used to conduct valid (uniform) inference on $\theta(t)$ with nonparametric efficiency guarantees.

{\bf 4. Numerical Experiments:} We showcase the finite-sample performances of our proposed estimators of $\theta(t)$ with and without the positivity condition through simulations and a case study of the Job Corps program in the United States in \autoref{sec:experiments} and \autoref{app:add_sim}. All the codes for our experiments are available at \url{https://github.com/zhangyk8/npDRDeriv}, and we provide some practical considerations for implementing our proposed estimators in \autoref{app:practical}.

\subsection{Other Related Works}
\label{subsec:related_works}

The dose-response curve $m(t)$ and its derivative $\theta(t)$ are non-regular target parameters, as they lack unique G\^ateaux derivatives and Riesz representers, depending on how the treatment distribution is localized at $t$ \citep{van1991differentiable,carone2019toward,ichimura2022influence}. As one of the key ingredients in this paper, kernel-based localization is a common approach in the literature, which has been used to construct IPW or DR estimators of $m(t)$ \citep{kallus2018policy,su2019non,huber2020direct,colangelo2020double,klosin2021automatic}. An alternative localization method is through the basis approach or series estimator \citep{chen2014sieve,chen2014sieveM,luedtke2024one}. Additionally, a general form of the IPW estimator of $m(t)$ was studied by \cite{galvao2015uniformly}. Under the positivity condition, the RA or G-computation \citep{robins1986new} estimators of $\theta(t)$ have been explored by \cite{gill2001causal,flores2007estimation,lee2018partial}. Recently, \cite{colangelo2020double} and \cite{bong2023local} also considered approximating $\theta(t)$ via finite differences of the estimated dose-response curve or a closely related matching method. Shortly after the first version of this paper, \cite{zeng2025nonparametric} proposed their DR estimator of $\theta(t)$ by regressing the pseudo-outcome in \cite{kennedy2017non} via local polynomial regression, but they did not address the positivity violation.


Growing interest in relaxing the positivity condition has led to new developments in causal inference. For continuous treatments, \cite{branson2023causal} studied a smoothed causal effect with trimmed conditional densities, while \cite{schindl2024incremental} examined stochastic interventions via exponentially tilted treatment distributions. Notably, dynamic stochastic interventions with continuous treatments can be robust to the violation of positivity \citep{diaz2020causal,bonvini2023incremental,mcclean2024fair}. To our knowledge, no existing works directly consider nonparametric inference on $\theta(t)$ without positivity, and our work takes an initial step to fill in this gap.

\subsection{Notations}

Throughout this paper, we consider an outcome variable $Y\in \mathcal{Y}\subset \mathbb{R}$, an univariate continuous treatment $T\in \mathcal{T}\subset \mathbb{R}$, and a vector of continuous confounding variables or covariates $\bm{S}=(S_1,...,S_d)\in \mathcal{S}\subset \mathbb{R}^d$ with a fixed dimension $d$. We write $Y \ind \bm{X}$ when the random variables $Y,\bm{X}$ are independent. The common distribution and expectation of $\bm{U}=(Y,T,\bm{S})$ are denoted by $\P$ and $\E$ respectively, whose Lebesgue density is $p(y,t,\bm{s}) = p_{Y|T,\bm{S}}(y|t,\bm{s})\cdot p_{T|\bm{S}}(t|\bm{s}) \cdot p_S(\bm{s})$. Here, $p_T(t)$ and $p_S(\bm{s})$ are the marginal densities of $T$ and $\bm{S}$, respectively, and $p_{T|\bm{S}}(t|\bm{s}) = \frac{\partial}{\partial t} \P\left(T\leq t|\bm{S} = \bm{s}\right)$ is the conditional density of $T$ given covariates $\bm{S}=\bm{s}$. We also denote the joint density of $(T,\bm{S})$ by $p(t,\bm{s}) = p_{T|\bm{S}}(t|\bm{s})\cdot p_S(\bm{s}) = p_{\bm{S}|T}(\bm{s}|t)\cdot p_T(t)$ and the support of $p_{\bm{S}|T}(\bm{s}|t)$ by $\mathcal{S}(t)$ for $t\in \mathcal{T}$. Our methodology and theory in this paper also apply when $\bm{S}$ consists of discrete covariates, where it suffices to, \emph{e.g.}, replace the density $p_{\bm{S}}(\bm{s})$, with the probability mass function $\P(\bm{S}=\bm{s})$ and substitute integration with summation. For any real-valued $\P$-integrable function $f$, we write $\P f = \int f(\bm{u})\, d\P(\bm{u})$ and denote the $L_p(\P)$-norm of $f$ by $\norm{f}_{L_p} := \left(\int |f(\bm{u})|^p d\P(\bm{u})\right)^{\frac{1}{p}}$. If $\hat{f}$ is estimated on an independent data sample, then $\norm{\hat{f}}_{L_p} := \left(\int \left|\hat{f}(\bm{u})\right|^p d\P(\bm{u})\right)^{\frac{1}{p}}$.
Additionally, we let $\mathbb{P}_n$ denote the empirical measure so that $\mathbb{P}_n f = \frac{1}{n} \sum_{i=1}^n f(\bm{U}_i) = \int f(\bm{u}) d\mathbb{P}_n(\bm{u})$ and $\mathbb{G}_n(f) = \sqrt{n}\left(\mathbb{P}_n-\P\right)f$. Finally, we use $\mathbbm{1}_A$ to denote the indicator function of a set $A$. The big-$O$ notation $h_n=O(g_n)$ means that $|h_n|$ is upper bounded by a positive constant multiple of $g_n >0$ when $n$ is sufficiently large. In contrast, $h_n=o(g_n)$ when $\lim_{n\to\infty} \frac{|h_n|}{g_n}=0$. For random variables, $o_P(1)$ is short for a sequence of random variables converging to zero in probability, while $O_P(1)$ denotes the sequence that is bounded in probability.

\section{Basic Framework and Identification With Positivity}
\label{sec:background}

Suppose that the data sample consists of independent and identically distributed (i.i.d.) observations $\left\{(Y_i,T_i,\bm{S}_i)\right\}_{i=1}^n \subset \mathcal{Y}\times \mathcal{T} \times \mathcal{S}$.
Since the main estimands of interest $\theta(t)=\frac{d}{dt}\E\left[Y(t)\right]$ and $m(t)=\E\left[Y(t)\right]$ are defined by potential outcomes that are not directly observable, we introduce some identification conditions for identifying $m(t)$ and $\theta(t)$ with observed data $\left\{(Y_i,T_i,\bm{S}_i)\right\}_{i=1}^n$.

\begin{assump}[Basic identification conditions]
	\label{assump:id_cond}
	\noindent %
	\begin{enumerate}[label=(\alph*)]
		\item (Consistency) $T=t$ implies that $Y(t) = Y$ for any $t\in \mathcal{T}$.
		\item (Ignorability or unconfoundedness) $Y(t) \ind T \,\big|\, \bm{S}$ for all $t\in \mathcal{T}$.
		\item (Treatment variation) The conditional variance of $T$ given $\bm{S}=\bm{s}$ is strictly positive for all $\bm{s}\in \mathcal{S}$, \emph{i.e.}, $\mathrm{Var}(T|\bm{S}=\bm{s})>0$.
		\item (Interchangeability) The equality $\frac{d}{dt}\mathbb{E}\left[\mu(t,\bm{S})\right] =\E\left[\frac{\partial}{\partial t} \mu(t,\bm{S})\right]$ holds true when $\mu(t,\bm{s}) = \E(Y|T=t,\bm{S}=\bm{s})$ is well-defined on $\mathcal{T}\times \mathcal{S}$.
	\end{enumerate}
\end{assump}

Assumption~\ref{assump:id_cond}(a,b) are standard identification conditions for causal dose-response curves \citep{gill2001causal,kennedy2017non}, while Example 1 in \cite{zhang2024nonparametric} demonstrates the necessity of imposing Assumption~\ref{assump:id_cond}(c) for identifiability. In particular, Assumption~\ref{assump:id_cond}(c) ensures that the distribution of $(T,\bm{S})$ has a nontrivial support in $\mathcal{T}\times \mathcal{S}$. 

Finally, Assumption~\ref{assump:id_cond}(d) only requires the interchangeability of the expectation and (partial) differentiation when $\mu(t,\bm{s})$ is well-defined on $\mathcal{T}\times \mathcal{S}$. It is a mild condition and can be satisfied when $\left|\frac{\partial}{\partial t} \mu(t,\bm{S})\right|$ is upper bounded by an integrable function with respect to the distribution of $\bm{S}$; see Theorem 1.1 and Example 1.8 in \cite{shao2003mathematical}. We also emphasize that the requirement of $\mu(t,\bm{s})$ being well-defined on $\mathcal{T}\times \mathcal{S}$ in Assumption~\ref{assump:id_cond}(d) is entailed by the following positivity condition. In general, when the positivity condition is violated, the conditional mean outcome (or regression) function $\mu(t,\bm{s})$ is no longer well-defined outside the support of the joint density $p(t,\bm{s})$.

\begin{assump}[Positivity]
	\label{assump:positivity}
	The conditional density $p_{T|\bm{S}}(t|\bm{s})$ is bounded away from 0 for all $(t,\bm{s})\in \mathcal{T} \times \mathcal{S}$, \emph{i.e.}, there exists a constant $p_{\min}>0$ such that $p_{T|\bm{S}}(t|\bm{s}) \geq p_{\min}$. 
\end{assump}

Under Assumptions \autoref{assump:id_cond}
and \autoref{assump:positivity}, the dose-response curve $m(t)$ and its derivative $\theta(t)$ are identifiable as:
$$
m(t)=\E\left[\mu(t,\bm{S})\right] \quad \text{ and } \quad \theta(t) = \E\left[\frac{\partial}{\partial t} \mu(t,\bm{S})\right],
$$
respectively. 
In \autoref{subsec:m_pos} and \autoref{sec:theta_pos}, we first study nonparametric inference on $m(t)$ and $\theta(t)$ this positivity condition. Later, in \autoref{sec:inconsistency} and \autoref{sec:theta_nopos}, we examine the identification and estimation issues as well as our proposed inference method without positivity.

\subsection{Nonparametric Estimation on $m(t)$ With Positivity}
\label{subsec:m_pos}

Before discussing our estimation strategy on the derivative effect curve $\theta(t)$, we first review the existing approaches for estimating the dose-response curve $m(t)$ via kernel smoothing when the positivity condition is valid. Specifically, under Assumptions~\ref{assump:id_cond} and \ref{assump:positivity}, there are three major estimation strategies for $t\mapsto m(t)=\E\left[Y(t)\right]$ with observed data $\left\{(Y_i,T_i,\bm{S}_i)\right\}_{i=1}^n$ listed as follows.

$\bullet$ {\bf Regression Adjustment (RA) Estimator:} Since $m(t)$ coincides with the form $\E\left[\mu(t,\bm{S})\right]$ under Assumptions~\ref{assump:id_cond} and \ref{assump:positivity}, it leads to a plug-in estimator as:
	\begin{equation}
		\label{m_RA}
		\hat{m}_{\mathrm{RA}}(t)  = \frac{1}{n}\sum_{i=1}^n \hat{\mu}(t,\bm{S}_i),
	\end{equation}
	where $\hat{\mu}(t,\bm{s})$ is a (consistent) estimator of the conditional mean outcome function $\mu(t,\bm{s})$.
	
$\bullet$ {\bf Inverse Probability Weighting (IPW) Estimator:} The IPW estimator follows from the rationale that $m(t) = \E\left[\frac{Y \cdot \mathbbm{1}_{\{T=t\}}}{p_{T|\bm{S}}(t|\bm{S})} \right]$ under Assumptions~\ref{assump:id_cond} and \ref{assump:positivity} in the discrete treatment setting. When the treatment variable $T \in \mathcal{T}$ is continuous, one common approach is to smooth the indicator function $\mathbbm{1}_{\{T=t\}}$ by a kernel function $K:\mathbb{R}\to [0,\infty)$, yielding the following IPW estimator as:
\begin{equation}
	\label{m_IPW}
	\hat{m}_{\mathrm{IPW}}(t) = \frac{1}{nh}\sum_{i=1}^n \frac{K\left(\frac{T_i-t}{h}\right)}{\hat{p}_{T|\bm{S}}(T_i|\bm{S}_i)}\cdot Y_i,
\end{equation}
where $h>0$ is a smoothing bandwidth and $\hat{p}_{T|\bm{S}}(t|\bm{s})$ is a (consistent) estimator of the conditional density $p_{T|\bm{S}}(t|\bm{s})$. In practice, without loss of its consistency, one can implement a self-normalized IPW estimator \eqref{m_IPW_selfnorm} of $m(t)$ as shown in \autoref{app:self_norm} to reduce the variance of \eqref{m_IPW}. Note that there are other approaches than kernel smoothing to smoothly approximate the non-regular target parameter $m(t)$, such as the series method outlined in Section 2.4 of \cite{luedtke2024one}.

$\bullet$ {\bf Doubly Robust (DR) Estimator:} The above RA estimator \eqref{m_RA} can be combined with the IPW estimator \eqref{m_IPW} to obtain the following DR estimator as:
	\begin{equation}
		\label{m_DR}
		\hat{m}_{\mathrm{DR}}(t) =\frac{1}{nh}\sum_{i=1}^n \left\{\frac{K\left(\frac{T_i-t}{h}\right)}{\hat{p}_{T|\bm{S}}(T_i|\bm{S}_i)}\cdot \left[Y_i - \hat \mu(t,\bm{S}_i)\right]+ h\cdot \hat{\mu}(t,\bm{S}_i) \right\},
	\end{equation}
where $\hat{\mu}(t,\bm{s})$ and $\hat{p}_{T|\bm{S}}(t,\bm{s})$ are (consistent) estimators of $\mu(t,\bm{s})$ and $p_{T|\bm{S}}(t,\bm{s})$ respectively. For completeness, we state and prove the asymptotic properties of the above estimators in \autoref{app:proof_m_pos}.

\begin{remark}
	\label{remark:m_IPW}
	There exists a slightly different formulation of the IPW estimator of $m(t)$ in the literature \citep{colangelo2020double,klosin2021automatic} as:
	\begin{equation}
		\label{m_IPW2}
		\hat{m}_{\mathrm{IPW,2}}(t) = \frac{1}{nh}\sum_{i=1}^n \frac{K\left(\frac{T_i-t}{h}\right)}{\hat{p}_{T|\bm{S}}(t|\bm{S}_i)}\cdot Y_i,
	\end{equation}
	in which the (estimated) inverse probability weight $\frac{1}{\hat{p}_{T|\bm{S}}(t|\bm{S}_i)}$ is evaluated at query point $t$ conditioning on each $\bm{S}_i$. We demonstrate in \autoref{app:m_ipw_diff} that the asymptotic difference between the oracle versions of \eqref{m_IPW} and \eqref{m_IPW2} will be of order $O(h^2) + O_P\left(\sqrt{\frac{h}{n}}\right)$ under some regularity conditions, which thus shrinks to 0 as $h\to 0$ and $n\to \infty$. In practice, we recommend using the form \eqref{m_IPW} for the IPW estimator of $m(t)$, because the estimated conditional density $\hat{p}_{T|\bm{S}}$ is more likely to be positive at sample points $(T_i,\bm{S}_i),i=1,...,n$ than at the (query) points $(t,\bm{S}_i),i=1,...,n$. 
\end{remark}

\begin{remark}
\label{remark:m_dr_pseudo_outcome}
Another important class of DR estimators for $m(t)$ constructs an (estimated) pseudo-outcome based on the efficient influence function of $\E\left[m(T)\right]$ as:
$$\varphi(Y,T,\bm{S};\hat{\mu},\hat{p}_{T|\bm{S}}) = \frac{Y-\hat{\mu}(T,\bm{S})}{\hat{p}_{T|\bm{S}}(T|\bm{S})} \int_{\mathcal{S}} \hat{p}_{T|\bm{S}}(T|\bm{s}) \, d\mathbb{P}_n(\bm{s}) + \int_{\mathcal{S}} \hat{\mu}(T,\bm{s})\, d\mathbb{P}_n(\bm{s}),$$
which is then regressed on the treatment variable $T$ \citep{kennedy2017non}. As mentioned in \autoref{subsec:related_works}, the DR estimator in \cite{zeng2025nonparametric} relies on this pseudo-outcome. A rigorous comparison between our proposed DR estimator of $\theta(t)$ in \autoref{sec:theta_pos} and theirs under positivity will be an interesting future direction.
\end{remark}

\section{Nonparametric Inference on $\theta(t)$ With Positivity}
\label{sec:theta_pos}

In this section, analogous to the estimation of $m(t)$ in \autoref{subsec:m_pos}, we study three different methods for estimating the derivative effect curve $t\mapsto \theta(t)=\frac{d}{dt}\E\left[Y(t)\right]$ with kernel smoothing under Assumptions~\ref{assump:id_cond} and \ref{assump:positivity}. Notably, both the IPW and DR estimators of $\theta(t)$ are novel contribution to the existing literature and exhibit distinct insights.

$\bullet$ {\bf Regression Adjustment (RA) Estimator:} Assumption~\ref{assump:id_cond}(d), together with other conditions in \ref{assump:id_cond} and \ref{assump:positivity}, guarantees the identification of $\theta(t)$ as $\E\left[\frac{\partial}{\partial t} \mu(t,\bm{S})\right]$ and provides a natural RA estimator as:
\begin{equation}
	\label{theta_RA}
	\hat{\theta}_{\mathrm{RA}}(t) = \frac{1}{n}\sum_{i=1}^n \hat{\beta}(t,\bm{S}_i),
\end{equation}
where $\hat{\beta}(t,\bm{s})$ is a (consistent) estimator of $\beta(t,\bm{s})=\frac{\partial}{\partial t} \mu(t,\bm{s})$.

$\bullet$ {\bf Inverse Probability Weighting (IPW) Estimator:} Inspired by the nonparametric derivative estimator in \cite{mack1989derivative}, we propose the following IPW estimator of $\theta(t)$ as:
\begin{equation}
	\label{theta_IPW}
	\hat{\theta}_{\mathrm{IPW}}(t) =\frac{1}{nh^2} \sum_{i=1}^n \frac{Y_i\left(\frac{T_i-t}{h}\right) K\left(\frac{T_i-t}{h}\right)}{\kappa_2 \cdot \hat{p}_{T|\bm{S}}(T_i|\bm{S}_i)},
\end{equation}
where $K:\mathbb{R}\to [0,\infty)$ is a kernel function with $\kappa_2=\int u^2 K(u)\,du$, $h>0$ is a smoothing bandwidth, and $\hat{p}_{T|\bm{S}}(t|\bm{s})$ is a (consistent) estimator of the conditional density $p_{T|\bm{S}}(t|\bm{s})$. One can implement the self-normalized IPW estimator \eqref{theta_IPW_selfnorm} of $\theta(t)$ in \autoref{app:self_norm} to reduce the variance of \eqref{theta_IPW}.

\begin{remark}
	\label{remark:theta_IPW}
	One might define the IPW estimator by evaluating the estimated inverse probability weights at points $(t,\bm{S}_i),i=1,...,n$ as:
	\begin{equation}
		\label{theta_IPW2}
		\hat{\theta}_{\mathrm{IPW,2}}(t) =\frac{1}{nh^2} \sum_{i=1}^n \frac{Y_i\left(\frac{T_i-t}{h}\right) K\left(\frac{T_i-t}{h}\right)}{\kappa_2 \cdot \hat{p}_{T|\bm{S}}(t|\bm{S}_i)}.
	\end{equation}
	However, different from \eqref{theta_IPW} in Remark~\ref{remark:m_IPW}, this IPW estimator $\hat{\theta}_{\mathrm{IPW,2}}(t)$ of $\theta(t)$ is (asymptotically) biased even when $h\to 0$ and $n\to \infty$; see \autoref{app:theta_ipw_diff} for details. Hence, our proposed IPW form \eqref{theta_IPW} is preferable not only due to the practical reason as stated in Remark~\ref{remark:m_IPW} but also because of its statistical consistency as justified in \autoref{thm:theta_pos} below.
\end{remark}

$\bullet$ {\bf Doubly Robust (DR) Estimator:} To achieve the doubly robust property like $\hat{m}_{\mathrm{DR}}(t)$ in \eqref{m_DR} (see also \autoref{app:proof_m_pos}), we propose the following DR estimator of $\theta(t)$ as:
\begin{equation}
	\label{theta_DR}
	\hat{\theta}_{\mathrm{DR}}(t) = \frac{1}{nh}\sum_{i=1}^n \left\{ \frac{\left(\frac{T_i-t}{h}\right)K\left(\frac{T_i-t}{h}\right) }{h\cdot \kappa_2\cdot \hat{p}_{T|\bm{S}}(T_i|\bm{S}_i)} \left[Y_i - \hat{\mu}(t,\bm{S}_i) - (T_i-t)\cdot \hat{\beta}(t,\bm{S}_i)
	\right]+ h\cdot \hat{\beta}(t,\bm{S}_i) \right\},
\end{equation}
where $\hat{\mu}(t,\bm{s}),\hat{\beta}(t,\bm{s}), \hat{p}_{T|\bm{S}}(t,\bm{s})$ are (consistent) estimators of $\mu(t,\bm{s}),\beta(t,\bm{s}), p_{T|\bm{S}}(t,\bm{s})$, respectively. We discuss how these nuisance functions can be estimated in \autoref{app:nuisance_est}. The key insight of why $\hat{\theta}_{\mathrm{DR}}(t)$ in \eqref{theta_DR} embraces the doubly robust property is that we leverage a local polynomial approximation \citep{fan1996local} to push the residual of the IPW component in \eqref{theta_DR} to at least second order before combining with the RA component. In other words, it can be shown that the Neyman orthogonality holds as $h\to 0$ \citep{neyman1959optimal,neyman1979c,chernozhukov2018double}.
As pointed out in Remark~\ref{remark:theta_IPW}, we need to compute the inverse probability weights at the sample points as $\frac{1}{\hat{p}_{T|\bm{S}}(T_i|\bm{S}_i)}, i=1,...,n$ for the above DR estimator \eqref{theta_DR}. If we otherwise compute the inverse probability weights at the (query) points as $\frac{1}{\hat{p}_{T|\bm{S}}(t|\bm{S}_i)}$ for $i=1,...,n$, then the resulting $\hat{\theta}_{\mathrm{DR}}(t)$ will be asymptotically biased even when both of the conditional density model $p_{T|\bm{S}}$ and the outcome model $\mu,\beta$ are correctly specified. Finally, we also outline a self-normalized version of \eqref{theta_DR} in \autoref{app:self_norm} for stabilizing its variance.


\subsection{Asymptotic Theory}
\label{subsec:asymp_theta_pos}

We introduce some regularity conditions for our subsequent theoretical analysis. Let $\mathcal{J} \subset \mathcal{T}\times \mathcal{S}$ be the support of the joint density $p(t,\bm{s})$, $\mathcal{J}^{\circ}$ be the interior of $\mathcal{J}$, and $\partial\mathcal{J}$ be the boundary of $\mathcal{J}$.

\begin{assump}[Differentiability of the conditional mean outcome function]
	\label{assump:reg_diff}
	For any $(t,\bm{s}) \in \mathcal{T}\times \mathcal{S}$ so that $\mu(t,\bm{s})=\E(Y|T=t,\bm{S}=\bm{s})$ is well-defined, it holds that
	\begin{enumerate}[label=(\alph*)]
		\item $\mu(t,\bm{s})$ is at least four times continuously differentiable with respect to $t$.
		\item $\mu(t,\bm{s})$ and all of its partial derivatives are uniformly bounded on $\mathcal{T}\times \mathcal{S}$.
		
		\item There exist constants $\sigma,c_1>0$ such that $\mathrm{Var}(Y|T=t,\bm{S}=\bm{s}) >\sigma^2$ and $\E|Y|^{2+c_1} <\infty$.
	\end{enumerate} 
\end{assump}

\begin{assump}[Differentiability of the density functions]
	\label{assump:den_diff}
	For any $(t,\bm{s})\in \mathcal{J}$, it holds that
	\begin{enumerate}[label=(\alph*)]
		\item The joint density $p(t,\bm{s})$ and the conditional density $p_{T|\bm{S}}(t|\bm{s})$ are at least three times continuously differentiable with respect to $t$.
		
		\item $p(t,\bm{s})$, $p_{T|\bm{S}}(t|\bm{s})$, $p_{\bm{S}|T}(\bm{s}|t)$, as well as all of the partial derivatives of $p(t,\bm{s})$ and $p_{T|\bm{S}}(t|\bm{s})$ are bounded and continuous up to the boundary $\partial \mathcal{J}$.
		
		\item The support $\mathcal{T}$ of the marginal density $p_T(t)$ is compact and $p_T(t)$ is uniformly bounded away from 0 within $\mathcal{T}$.
	\end{enumerate}
\end{assump}

\begin{assump}[Regular kernel conditions]
	\label{assump:reg_kernel}
A kernel function $K:\mathbb{R} \to [0,\infty)$ is bounded and compactly supported on $[-1,1]$ with $\int_{\mathbb{R}} K(t)\,dt =1$ and $K(t)=K(-t)$. In addition, it holds that 
\begin{enumerate}[label=(\alph*)]
	\item $\kappa_j := \int_{\mathbb{R}} u^j K(u) \, du < \infty$ and $\nu_j := \int_{\mathbb{R}} u^j K^2(u) \, du < \infty$ for all $j=1,2,...$.
	\item $K$ is a second-order kernel, \emph{i.e.}, $\kappa_1 = 0$ and $\kappa_2 >0$.
	\item $\mathcal{K} = \left\{t'\mapsto \left(\frac{t'-t}{h}\right)^{k_1} K\left(\frac{t'-t}{h}\right): t\in \mathcal{T}, h>0, k_1=0,1\right\}$ is a bounded VC-type class of measurable functions on $\mathbb{R}$.
\end{enumerate}
\end{assump}

Assumptions~\ref{assump:reg_diff} and \ref{assump:den_diff} are common smoothness conditions for derivative estimation with kernel smoothing methods \citep{gasser1984estimating,mack1989derivative,wand1994kernel,wasserman2006all}. These assumptions can be relaxed by the H\"older continuity condition. The uniform lower bound on $p_T(t)$ within its support $\mathcal{T}$ in Assumption~\ref{assump:den_diff}(c) is only needed when we establish the uniform consistency of our proposed estimators and identify the derivative effect curve $\theta(t)$ when the positivity condition is violated. Assumption~\ref{assump:reg_kernel}(a,b) are more like properties than regularity conditions on those commonly used kernel functions, such as the triangular kernel $K(u)=(1-|u|)\,\mathbbm{1}_{\{|u|\leq 1\}}$ and Epanechnikov kernel $K(u)=\frac{3}{4}(1-|u|)\, \mathbbm{1}_{\{|u|\leq 1\}}$. Finally, the VC-type condition in Assumption~\ref{assump:reg_kernel}(c) is only required when we are interested in the uniform consistency of our proposed estimators over $\mathcal{T}$.

The following theorem summarizes the asymptotic properties of our proposed DR estimator $\hat{\theta}_{\mathrm{DR}}(t)$, and we defer its proof and other consistency results for $\hat{\theta}_{\mathrm{RA}}(t), \hat{\theta}_{\mathrm{IPW}}(t)$ in \autoref{app:proof_theta_pos}. 

\begin{theorem}[Asymptotic properties of $\hat{\theta}_{\mathrm{DR}}(t)$ under positivity]
	\label{thm:theta_pos}
	Suppose that Assumptions~\ref{assump:id_cond}, \ref{assump:reg_diff}, \ref{assump:den_diff}, \ref{assump:reg_kernel}, and \ref{assump:positivity} hold and $\hat{\mu},\hat{\beta}, \hat{p}_{T|\bm{S}}$ are constructed on a data sample independent of $\{(Y_i,T_i,\bm{S}_i)\}_{i=1}^n$. For any fixed $t\in \mathcal{T}$, we let $\bar{\mu}(t,\bm{s})$, $\bar{\beta}(t,\bm{s})$, and $\bar{p}_{T|\bm{S}}(t|\bm{s})$ be fixed bounded functions to which $\hat{\mu}(t,\bm{s})$, $\hat{\beta}(t,\bm{s})$ and $\hat{p}_{T|\bm{S}}(t|\bm{s})$ converge.
	If, in addition, we assume that 
	\begin{enumerate}[label=(\alph*)]
		\item $\bar{p}_{T|\bm{S}}$ satisfies Assumptions~\ref{assump:den_diff} and \ref{assump:positivity};
		\item either (i) ``$\,\bar{\mu}=\mu$ and $\bar{\beta}=\beta$'' with only $h\norm{\bar{\beta}(t,\bm{S}) - \beta(t,\bm{S})}_{L_2} \to 0$ or (ii) ``$\,\bar{p}_{T|\bm{S}} = p_{T|\bm{S}}$'';
		\item $\sqrt{nh} \sup\limits_{|u-t|\leq h} \norm{\hat{p}_{T|\bm{S}}(u|\bm{S}) - p_{T|\bm{S}}(u|\bm{S})}_{L_2} \left[\norm{\hat{\mu}(t,\bm{S}) - \mu(t,\bm{S})}_{L_2} + h \norm{\hat{\beta}(t,\bm{S}) - \beta(t,\bm{S})}_{L_2}\right] = o_P(1)$,
	\end{enumerate}
	then 
	$$\sqrt{nh^3}\left[\hat{\theta}_{\mathrm{DR}}(t) - \theta(t)\right] = \frac{1}{\sqrt{n}} \sum_{i=1}^n \left\{\phi_{h,t}\left(Y_i,T_i,\bm{S}_i;\bar{\mu}, \bar{\beta}, \bar{p}_{T|\bm{S}}\right) + \sqrt{h^3}\left[\bar{\beta}(t,\bm{S}_i) -  \mathbb{E}\left[\beta(t,\bm{S})\right] \right]\right\} +o_P(1)$$
	when $nh^7\to c_3$ for some finite number $c_3\geq 0$, where $$\phi_{h,t}\left(Y,T,\bm{S}; \bar{\mu},\bar{\beta}, \bar{p}_{T|\bm{S}}\right) = \frac{\left(\frac{T-t}{h}\right) K\left(\frac{T-t}{h}\right)}{\sqrt{h}\cdot \kappa_2\cdot \bar{p}_{T|\bm{S}}(T|\bm{S})}\cdot \left[Y - \bar{\mu}(t,\bm{S}) - (T-t)\cdot \bar{\beta}(t,\bm{S})\right].$$
	Furthermore, 
	$$\sqrt{nh^3}\left[\hat{\theta}_{\mathrm{DR}}(t) - \theta(t) - h^2 B_{\theta}(t)\right] \stackrel{d}{\to} \mathcal{N}\left(0,V_{\theta}(t)\right)$$
	with $V_{\theta}(t) = \mathbb{E}\left[\phi_{h,t}^2\left(Y,T,\bm{S};\bar{\mu}, \bar{\beta}, \bar{p}_{T|\bm{S}}\right)\right]$ and 
	\begin{align*}
		B_{\theta}(t) = 
		\begin{cases}
			\frac{\kappa_4}{6\kappa_2} \cdot \mathbb{E}_{\bm{S}}\left\{\frac{3\frac{\partial}{\partial t} p_{T|\bm{S}}(t|\bm{S}) \cdot \frac{\partial^2}{\partial t^2} \mu(t,\bm{S}) + p_{T|\bm{S}}(t|\bm{S})\left[ \frac{\partial^3}{\partial t^3} \mu(t,\bm{S}) - 3\frac{\partial}{\partial t} \log\bar{p}_{T|\bm{S}}(t|\bm{S}) \cdot \frac{\partial^2}{\partial t^2} \mu(t,\bm{S}) \right]}{\bar{p}_{T|\bm{S}}(t|\bm{S})} \right\} \; \text{ when } \bar{\mu}=\mu \text{ and } \bar{\beta}=\beta,\\
			\frac{\kappa_4}{6\kappa_2} \cdot \mathbb{E}_{\bm{S}}\left[\frac{\partial^3}{\partial t^3} \mu(t,\bm{S})\right] \quad\; \text{ when }\; \bar{p}_{T|\bm{S}} = p_{T|\bm{S}}.
		\end{cases}
	\end{align*}
\end{theorem}

As established by \autoref{thm:theta_pos}, the proposed estimator $\hat{\theta}_{\mathrm{DR}}(t)$ achieves doubly robust consistency for $\theta(t)$, provided that either the conditional density model $\bar{p}_{T|\bm{S}}$ or the outcome model $\bar{\mu},\bar{\beta}$ is correctly specified. Unlike the DR estimator $\hat{m}_{\mathrm{DR}}(t)$ of the dose-response curve $m(t)$, which only requires the specification of $\mu(t,\bm{s})$ in the outcome model, the DR estimator $\hat{\theta}_{\mathrm{DR}}(t)$ of the derivative effect $\theta(t)$ necessitates specifying both $\mu(t,\bm{s})$ and its partial derivative $\beta(t,\bm{s})=\frac{\partial}{\partial t}\mu(t,\bm{s})$ in the outcome model. This added complexity is essential for accurately estimating derivatives.

We require in \autoref{thm:theta_pos} and other subsequent results that $\hat{\mu},\hat{\beta},\hat{p}_{T|\bm{S}}$ are obtained from a data sample independent of $\left\{(Y_i,T_i,\bm{S}_i)\right\}_{i=1}^n$. This requirement avoids the need for uniform entropy conditions on $\bar{\mu}, \bar{p}_{T|\bm{S}}$ imposed by \cite{kennedy2017non}. When no additional data sample is available, these nuisance function estimators $\hat{\mu},\hat{\beta},\hat{p}_{T|\bm{S}}$ can still be estimated using cross-fitting techniques, allowing for valid construction of the associated estimators of $\theta(t)$; see \autoref{app:imple} for the detailed procedures. Importantly, the established rates of convergence in \autoref{thm:theta_pos} remain unchanged for the cross-fitted estimators.

Finally, the estimation bias of $\hat{\theta}_{\mathrm{DR}}(t)$ is of order $O(h^2)$ due to the use of a standard second-order kernel. While a lower-bias estimator could be derived under Assumptions~\ref{assump:reg_diff} and \ref{assump:den_diff} by employing a higher-order kernel, we still recommend our proposed estimator $\hat{\theta}_{\mathrm{DR}}(t)$ for practical use, as higher-order kernels introduce greater complexity and are more challenging to implement.

\subsection{Statistical Inference on $\theta(t)$}
\label{subsec:theta_inference_pos}

To leverage the asymptotic normality of $\hat{\theta}_{\mathrm{DR}}(t)$ for pointwise inference on $\theta(t)$ in practice, we need to address two additional challenges: (i) estimate the asymptotic variance $V_{\theta}(t)$; and (ii) select a proper bandwidth parameter $h>0$.

For challenge (i), we estimate $V_{\theta}(t)$ in \autoref{thm:theta_pos} by the sample variance of the influence function $\phi_{h,t}$ or the asymptotic linear form as:
\begin{equation}
\label{theta_var_est}
\hat{V}_{\theta}(t) = \frac{1}{n} \sum_{i=1}^n \left\{\phi_{h,t}\left(Y_i,T_i,\bm{S}_i;\hat{\mu}, \hat{\beta}, \hat{p}_{T|\bm{S}}\right) + \sqrt{h^3}\left[\hat{\beta}(t,\bm{S}_i) - \hat{\theta}_{\mathrm{DR}}(t) \right]\right\}^2.
\end{equation}
The cross-fitted version of $\hat{V}_{\theta}(t)$ can be found in \eqref{theta_var_est_crossfit} of \autoref{app:imple}. Notice that the second part $\sqrt{h^3}\left[\hat{\beta}(t,\bm{S}_i) - \hat{\theta}_{\mathrm{DR}}(t) \right]$ in \eqref{theta_var_est} is asymptotically negligible. We keep this part mainly for a more conservative estimate of the asymptotic variance $V_{\theta}(t)$ to guarantee a better empirical coverage of the resulting pointwise confidence interval.

For challenge (ii), the optimal bandwidth that minimizes the asymptotic mean squared error of $\hat{\theta}_{\mathrm{DR}}(t)$ is of order $O\left(n^{-\frac{1}{7}}\right)$. However, to construct a valid Wald-type confidence interval, an undersmoothing bandwidth $h$ is typically required for the first-order bias of $\hat{\theta}_{\mathrm{DR}}(t)$ to be asymptotically negligible, \emph{i.e.}, $h^2\sqrt{nh^3}=o(1)$ \citep[Section~5.7]{wasserman2006all}. Therefore, we recommend choosing the bandwidth $h$ to be of order $O\left(n^{-\frac{1}{5}}\right)$, aligning with the outputs of standard bandwidth selection methods for nonparametric regression \citep{wand1994kernel,li2004cross}.

Finally, the $(1-\tau)$-level confidence interval for $\theta(t)$ is thus given by $\left[\hat{\theta}_{\mathrm{DR}}(t) \pm q_{1-\frac{\tau}{2}} \sqrt{\frac{\hat{V}_{\theta}(t)}{nh^3}} \right]$, where $q_{1-\frac{\tau}{2}}$ is the $\left(1-\frac{\tau}{2}\right)$ quantile of the standard normal distribution $\mathcal{N}(0,1)$.

\begin{remark}[Uniform inference via multiplier bootstrap]
	It is also statistically valid to conduct uniform inference on $\theta(t)$ over $t\in \mathcal{T}$ via multiplier bootstrap under our regularity conditions in \autoref{thm:theta_pos}. Specifically, let $\left\{Z_i\right\}_{i=1}^n$ be a sequence of i.i.d. random variables independent of the observed data $\left\{(Y_i,T_i,\bm{S}_i)\right\}_{i=1}^n$ with $\E(Z_i)=\mathrm{Var}(Z_i)=1$ and sub-exponential tails. Then, we sample $B$ different i.i.d. datasets $\left\{Z_i^{(b)}\right\}_{i=1}^n, b=1,...,B$ and compute the bootstrap DR estimators of $\theta(t)$ as:
	$$\hat{\theta}_{\mathrm{DR}}^{(b)*}(t) = \frac{1}{nh}\sum_{i=1}^n Z_i^{(b)}\left\{ \frac{\left(\frac{T_i-t}{h}\right)K\left(\frac{T_i-t}{h}\right) }{h\cdot \kappa_2\cdot \hat{p}_{T|\bm{S}}(T_i|\bm{S}_i)} \left[Y_i - \hat{\mu}(t,\bm{S}_i) - (T_i-t)\cdot \hat{\beta}(t,\bm{S}_i)
	\right]+ h\cdot \hat{\beta}(t,\bm{S}_i) \right\}$$
	for $b=1,...,B$. If $\hat{Q}(1-\tau)$ is the $(1-\tau)$ quantile of the sequence $\left\{\sup_{t\in \mathcal{T}} \sqrt{nh^3}\left|\frac{\hat{\theta}_{\mathrm{DR}}^{(b)*}(t) - \hat{\theta}_{\mathrm{DR}}(t)}{\sqrt{\hat{V}_{\theta}(t)}}\right|\right\}_{b=1}^B$, then the $(1-\tau)$ uniform confidence band of $\theta(t)$ is given by
	$\left[\hat{\theta}_{\mathrm{DR}}(t) \pm \hat{Q}(1-\tau) \sqrt{\frac{\hat{V}_{\theta}(t)}{nh^3}} \right]$. The asymptotic validity of this confidence band under cross-fitting follows from Theorem 4.2 in \cite{fan2022estimation}; see also Section S4 in \cite{colangelo2020double}.
\end{remark}

\subsection{Nonparametric Efficiency Guarantee for $\hat{\theta}_{\mathrm{DR}}(t)$}
\label{subsec:nonp_eff}

We now study the nonparametric efficiency of our DR estimator $\hat{\theta}_{\mathrm{DR}}(t)$, providing extra insights into its formulation. As discussed in \autoref{subsec:related_works}, a key challenge in deriving a nonparametric efficiency bound for $\theta(t)$ is its lack of pathwise differentiability \citep{bickel1998efficient,diaz2013targeted}, which implies that the efficient influence function relative to a nonparametric model does not always exist. Our solution is to derive the efficient influence function for a smooth functional $\P\mapsto \varpi_{h,t}(\P):= \mathbb{E}\left[\frac{Y \cdot \left(\frac{T-t}{h}\right) K\left(\frac{T-t}{h}\right)}{h^2\cdot \kappa_2 \cdot p_{T|\bm{S}}(T|\bm{S})}\right]$ for a fixed bandwidth $h>0$ \citep{van2018cv,takatsu2022debiased}. This functional, defined by the IPW form of $\theta(t)$, smoothly approximates $\theta(t)$ with $\varpi_{h,t}(\P) - \theta(t) = O(h^2)$ as shown in \autoref{app:theta_IPW_pos} and remains pathwise differentiable for any fixed $h>0$. We establish in the following theorem that our DR estimator $\hat{\theta}_{\mathrm{DR}}(t)$ attains the same asymptotic variance $V_{\theta}(t)$ as the (approximated) one obtained from the efficient influence function of $\varpi_{h,t}(\P)$, up to a finite-sample bias of order $O(h^2)$.

\begin{theorem}[Efficient influence function]
\label{thm:nonp_eff}
Suppose that Assumptions~\ref{assump:positivity}, \ref{assump:reg_diff}, and \ref{assump:reg_kernel} hold for the nonparametric model containing $\P$. For any fixed bandwidth $h>0$ and $t\in \mathcal{T}$, the efficient influence function of $\varpi_{h,t}(\P)$ relative to this model is given by
\begin{align}
\label{theta_eif}
\begin{split}
&\int_{\mathcal{T}} \frac{\mu(t_1,\bm{S})\left(\frac{t_1-t}{h}\right)  K\left(\frac{t_1-t}{h}\right)}{h^2 \cdot \kappa_2} \, dt_1 + \frac{\left[Y - \mu(T,\bm{S})\right] \left(\frac{T-t}{h}\right) K\left(\frac{T-t}{h}\right)}{h^2 \cdot \kappa_2\cdot p_{T|\bm{S}}(T|\bm{S})} - \varpi_{h,t}(\P) \\
&= \beta(t,\bm{S}) + \frac{\left[Y - \mu(T,\bm{S})\right] \left(\frac{T-t}{h}\right) K\left(\frac{T-t}{h}\right)}{h^2 \cdot \kappa_2\cdot p_{T|\bm{S}}(T|\bm{S})} - \varpi_{h,t}(\P) + h^2 \cdot \tilde{B}_{\theta}(t),
\end{split}
\end{align}
where $\tilde{B}_{\theta}(t) = \frac{\kappa_4}{6\kappa_2} \cdot \frac{\partial^3}{\partial t^3} \mu(t,\bm{s}) + o(h)$. Under the setup of \autoref{thm:theta_pos}, the resulting estimator of $\theta(t)$ can be written as:
\begin{equation}
\label{theta_DR2}
\hat{\theta}_{\mathrm{DR,2}}(t) = \frac{1}{nh}\sum_{i=1}^n \left\{ \frac{\left(\frac{T_i-t}{h}\right)K\left(\frac{T_i-t}{h}\right) }{h\cdot \kappa_2\cdot \hat{p}_{T|\bm{S}}(T_i|\bm{S}_i)} \left[Y_i - \hat{\mu}(T_i,\bm{S}_i)
\right]+ h\cdot \hat{\beta}(t,\bm{S}_i) \right\},
\end{equation}
which has the same doubly robust properties and asymptotic variance as our proposed one $\hat{\theta}_{\mathrm{DR}}(t)$.
\end{theorem}

The proof of \autoref{thm:nonp_eff} is in \autoref{app:nonp_eff_proof}. We recommend defining the DR estimator of $\theta(t)$ as \eqref{theta_DR} or \eqref{theta_DR2} rather than the one-step estimator based on the efficient influence function \eqref{theta_eif} because of two reasons. First, evaluating the integral in \eqref{theta_eif} is computationally challenging in practice. Second, the approximation error of \eqref{theta_DR} or \eqref{theta_DR2} relative to \eqref{theta_eif} is of the same order $O(h^2)$ as the approximation error of \eqref{theta_eif} to the target parameter $\theta(t)$.

\section{Identification and Inconsistency Issues Without Positivity}
\label{sec:inconsistency}

This section discusses the general identification issue on the dose-response curve $t\mapsto m(t)=\E\left[Y(t)\right]$ and its derivative effect curve $t\mapsto \theta(t)=\frac{d}{dt}\E\left[Y(t)\right]$ when the positivity condition (Assumption~\ref{assump:positivity}) is violated. 
We propose an additive structural assumption on the outcome model in \eqref{add_conf_model} to address the identification issue. However, even under this additive confounding model \eqref{add_conf_model}, the IPW and DR estimators of $m(t)$ and $\theta(t)$ remain inconsistent without the positivity condition due to the support discrepancy. 
To resolve this inconsistency, we leverage techniques from nonparametric set estimation to propose our bias-corrected IPW and DR estimators.

\subsection{Identification Issue Without Positivity}
\label{subsec:id_issue_nopos}

When the positivity condition (Assumption~\ref{assump:positivity}) fails to hold, the conditional mean outcome (or regression) function $\mu(t,\bm{s}) = \mathbb{E}(Y|T=t,\bm{S}=\bm{s})$ is not well-defined in those regions of $\mathcal{T}\times\mathcal{S}$ that lie outside the support $\mathcal{J}$ of the joint density $p(t,\bm{s})$. Hence, the G-computation formulae $\E\left[\mu(t,\bm{S})\right]$ and $\E\left[\frac{\partial}{\partial t} \mu(t,\bm{S})\right]$ are ill-defined and cannot be used to identify $m(t)$ and $\theta(t)$, respectively.

Similarly, identifying $m(t)$ and $\theta(t)$ through the IPW formulae requires the positivity condition as well, because we demonstrate in the proofs of \autoref{thm:theta_pos} and Proposition~\ref{prop:m_pos} that 
\begin{equation}
\label{IPW_forms}
\lim\limits_{h\to 0} \E\left[\frac{Y \cdot K\left(\frac{T-t}{h}\right)}{h\cdot p_{T|\bm{S}}(T|\bm{S})} \right] = \E\left[\mu(t,\bm{S})\right] \quad \text{ and } \quad \lim\limits_{h\to 0} \E\left[\frac{Y \left(\frac{T-t}{h}\right) K\left(\frac{T-t}{h}\right)}{\kappa_2 h^2\cdot p_{T|\bm{S}}(T|\bm{S})} \right] = \E\left[\frac{\partial}{\partial t}\mu(t,\bm{S})\right].
\end{equation}

Therefore, it is impossible in general to identify the causal dose-response curve $t\mapsto m(t)=\E\left[Y(t)\right]$ and its derivative effect curve $t\mapsto \theta(t) = \frac{d}{dt}\E\left[Y(t)\right]$ without further identification or structural assumptions when the positivity condition is violated.

\subsubsection{Remedy: Identification Under an Additive Structural Model}
\label{subsec:id_add_mod}

While the identifications of $m(t)$ and $\theta(t)$ are infeasible without positivity in general, they are indeed identifiable under an additive structural assumption on the potential outcome model as $Y(t) = \bar{m}(t) + \eta(\bm{S}) +\epsilon$ for any $t\in \mathcal{T}$ \citep{zhang2024nonparametric},
which, under the consistency condition (Assumption~\ref{assump:id_cond}(a)), is equivalent to the following additive confounding model
\begin{align}
	\label{add_conf_model}
	\begin{split}
		Y &= \bar{m}(T) + \eta(\bm{S}) +\epsilon,
	\end{split}
\end{align}
where $\bar{m}:\mathcal{T}\to \mathbb{R}$ and $\eta:\mathcal{S}\to \mathbb{R}$ are deterministic functions, $\E(\epsilon|T,\bm{S})=0$, $\mathrm{Var}(\epsilon|T,\bm{S}) > \sigma^2 >0$, and $\E|\epsilon|^{2+c_1} <\infty$ as in Assumption~\ref{assump:reg_diff}(c). Such an additive model is a common working model in the context of spatial statistics \citep{paciorek2010importance,schnell2020}, where the covariates $\bm{S}\in \mathcal{S}\subset \mathbb{R}^d$ consist of spatial locations or other spatially correlated confounding variables. More broadly, it also appears in the literature of nonparametric \citep{stone1985additive} and high-dimensional statistics \citep{meier2009high,guo2019decorrelated}. 

Under model \eqref{add_conf_model}, the dose-response curve $m(t)$ and its derivative $\theta(t)$ become
\begin{equation}
\label{m_theta_additive}
m(t)=\mathbb{E}\left[Y(t)\right] = \bar{m}(t) + \E\left[\eta(\bm{S})\right] \quad \text{ and } \quad \theta(t)=m'(t)=\bar{m}'(t).
\end{equation}
They are identifiable from the observable data through the formulas
\begin{align}
\label{id_m}
\begin{split}
	\theta(t) &= \bar{m}'(t) = \mathbb{E}\left[\frac{\partial}{\partial t}\mu(T,\bm{S})\Big|T=t\right],\\
	m(t) &= \mathbb{E}\left[Y+\int_T^t\theta(\tilde{t})\, d\tilde{t}\right] = \mathbb{E}\left\{Y+\int_T^t\mathbb{E}\left[\frac{\partial}{\partial t}\mu(T,\bm{S})\Big|T=\tilde{t}\right]\, d\tilde{t}\right\}.
\end{split}
\end{align}
For completeness, we also summarize this identification theory as Proposition~\ref{prop:id_additive} in \autoref{app:id_additive}. As a result, the RA estimator of $\theta(t)$ under model \eqref{add_conf_model} without assuming the positivity condition is given by
\begin{equation}
	\label{theta_RA_corrected}
	\hat{\theta}_{\mathrm{C,RA}}(t) = \int \hat{\beta}(t,\bm{s}) \, d\hat{F}_{\bm{S}|T}(\bm{s}|t),
\end{equation}
where $\hat{\beta}(t,\bm{s})$ and $\hat{F}_{\bm{S}|T}(\bm{s}|t)$ are (consistent) estimators of of $\beta(t,\bm{s}) = \frac{\partial}{\partial t} \mu(t,\bm{S})$ and the conditional cumulative distribution function (CDF) $\P_{\bm{S}|T}(\bm{s}|t):= F_{\bm{S}|T}(\bm{s}|t)$, respectively. By \eqref{id_m}, the integral RA estimator of $m(t)$ under model \eqref{add_conf_model} can be written as:
\begin{equation}
	\label{m_RA_corrected}
	\hat{m}_{\mathrm{C,RA}}(t) = \frac{1}{n}\sum_{i=1}^n \left[Y_i + \int_{\tilde{t}=T_i}^{\tilde{t}=t} \hat{\theta}_{\mathrm{C,RA}}(\tilde{t})\, d\tilde{t} \right].
\end{equation}
Both estimators \eqref{theta_RA_corrected} and \eqref{m_RA_corrected} are consistent even when the positivity condition is violated \citep{zhang2024nonparametric}; see also \autoref{thm:theta_nopos} and \autoref{app:m_nopos_proof}. In the sequel, we will discuss both the challenges and solutions for extending these RA estimators to IPW and DR estimators of $\theta(t)$ and $m(t)$ under model \eqref{add_conf_model}.

\subsection{Estimation Issues of IPW Estimators Under the Additive Confounding Model \eqref{add_conf_model}}
\label{subsec:IPW_bias}

Although the causal quantities $m(t)$ and $\theta(t)$ are identifiable under the additive confounding model \eqref{add_conf_model}, the IPW formulae \eqref{IPW_forms} are indeed biased without positivity due to the support discrepancy between the conditional density $p_{\bm{S}|T}(\bm{s}|t)$ for $t\in \mathcal{T}$ and the marginal density $p_{\bm{S}}(\bm{s})$. To examine these biases, we can equivalently analyze the following oracle IPW estimators of $m(t)$ and $\theta(t)$ defined as:
\begin{equation}
	\label{IPW_oracle}
	\tilde{m}_{\mathrm{IPW}}(t) = \frac{1}{nh}\sum_{i=1}^n \frac{Y_i\cdot K\left(\frac{T_i-t}{h}\right)}{p_{T|\bm{S}}(T_i|\bm{S}_i)} \quad \text{ and } \quad \tilde{\theta}_{\mathrm{IPW}}(t) = \frac{1}{nh^2} \sum_{i=1}^n \frac{Y_i\left(\frac{T_i-t}{h}\right) K\left(\frac{T_i-t}{h}\right)}{\kappa_2 \cdot p_{T|\bm{S}}(T_i|\bm{S}_i)},
\end{equation}
where the estimated conditional density $\hat{p}_{T|\bm{S}}(t|\bm{s})$ is replaced by the true one $p_{T|\bm{S}}(t|\bm{s})$.

\begin{proposition}[Inconsistency of IPW estimators]
	\label{prop:IPW_bias_m}
	Suppose that Assumptions~\ref{assump:id_cond}(a-c), \ref{assump:reg_diff}, \ref{assump:den_diff}(c), and \ref{assump:reg_kernel}(a-b) hold under the additive confounding model \eqref{add_conf_model}. Assume also that when the bandwidth $h$ is small, the Lebesgue measure of the symmetric difference set satisfies 
	$$\left|\mathcal{S}(t+uh)\triangle \mathcal{S}(t)\right| = \left|\left[\mathcal{S}(t+uh)\setminus \mathcal{S}(t)\right] \cup \left[\mathcal{S}(t)\setminus \mathcal{S}(t+uh)\right]\right|=o(1)$$
	for any $t\in \mathcal{T}$ and $u\in \mathbb{R}$. Then, when $h$ is small, the expectation of $\tilde{m}_{\mathrm{IPW}}(t)$ in \eqref{IPW_oracle} is given by
	$$\E\left[\tilde{m}_{\mathrm{IPW}}(t)\right] = \bar{m}(t)\cdot \rho(t) + \omega(t) + o(1),$$
	where $\rho(t) = \P\left(\bm{S}\in \mathcal{S}(t)\right)$ and $\omega(t) = \E\left[\eta(\bm{S}) \mathbbm{1}_{\{\bm{S}\in \mathcal{S}(t)\}}\right]$. If, in addition, there exists a constant $A_h>0$ depending on $h$ such that
	\begin{equation}
		\label{quadratic_cond}
		\int_{\mathbb{R}}\E\left\{\left[\bar{m}(t) +\eta(\bm{S})\right]\left[\mathbbm{1}_{\{\bm{S}\in \mathcal{S}(t+uh)\setminus \mathcal{S}(t)\}} - \mathbbm{1}_{\{\bm{S}\in \mathcal{S}(t)\setminus \mathcal{S}(t+uh)\}}\right] \right\} u\cdot K(u)\, du=O(A_h)
	\end{equation}
	for any $t\in \mathcal{T}$ and $u\in \mathbb{R}$ when $h$ is small, then the expectation of $\tilde{\theta}_{\mathrm{IPW}}(t)$ in \eqref{IPW_forms} is given by
	$$\E\left[\tilde{\theta}_{\mathrm{IPW}}(t)\right]=\bar{m}'(t)\cdot \rho(t) + O\left(\frac{A_h}{h}\right).$$
\end{proposition}

The proof of Proposition~\ref{prop:IPW_bias_m} is in \autoref{app:IPW_bias_proof_m}. We emphasize that the IPW estimators in \eqref{IPW_oracle} have two layers of bias. First, if $\frac{A_h}{h}\to 0$ as $h\to 0$ (see also Remark~\ref{remark:reg_cond} below), then the results in Proposition~\ref{prop:IPW_bias_m} will imply that 
\begin{align*}
	\lim_{h\to 0} \mathbb{E}\left[\tilde{m}_{\mathrm{IPW}}(t)\right] &= \lim_{h\to 0} \E\left[\frac{Y \cdot K\left(\frac{T-t}{h}\right)}{h\cdot p_{T|\bm{S}}(T|\bm{S})} \right] = \bar{m}(t)\cdot \rho(t) + \omega(t)\neq m(t), \\
	\lim_{h\to 0} \mathbb{E}\left[\tilde{\theta}_{\mathrm{IPW}}(t)\right] &= \lim_{h\to 0} \E\left[\frac{Y \left(\frac{T-t}{h}\right) K\left(\frac{T-t}{h}\right)}{h^2\cdot\kappa_2\cdot p_{T|\bm{S}}(T|\bm{S})} \right] = \bar{m}'(t)\cdot \rho(t) \neq \theta(t),
\end{align*}
where we recall that $m(t)=\bar{m}(t) + \E\left[\eta(\bm{S})\right]$ and $\theta(t)=\bar{m}'(t)$ from \eqref{m_theta_additive}. Second, if $\frac{A_h}{h}$ does not converge to 0, then the bias of $\tilde{\theta}_{\mathrm{IPW}}(t)$ will be larger or even diverging to infinity as $h\to 0$. In reality, the estimation biases or inconsistencies of IPW estimators in \eqref{IPW_oracle} are due to the discrepancy between the conditional support $\mathcal{S}(t)$ of $p_{\bm{S}|T}(\bm{s}|t)$ and the marginal support $\mathcal{S}$ of $p_{\bm{S}}(\bm{s})$. To correct for the bias of IPW estimators, it is necessary to address the geometric discrepancy, a solution to which will be elaborated upon in \autoref{sec:theta_nopos}.

Finally, since both RA and IPW estimators cannot be used to identify and estimate $m(t)$ and $\theta(t)$ due to identification and inconsistency issues, the previously studied DR estimators \eqref{m_DR} and \eqref{theta_DR} will be pointless without the positivity condition. 

\begin{remark}
	\label{remark:reg_cond}
	The regularity condition \eqref{quadratic_cond} is indeed not an assumption but rather a natural property. This is because as $h\to 0$, the differences between two sets $\mathcal{S}(t+uh)\setminus \mathcal{S}(t)$ and $\mathcal{S}(t)\setminus \mathcal{S}(t+uh)$ shrink to 0 for any $t\in \mathcal{T}$ and $u\in \mathbb{R}$. Additionally, when the expectation in \eqref{quadratic_cond} is independent of $u$, one can deduce by the second-order kernel property of $K$ that the left-hand side of \eqref{quadratic_cond} is 0. Hence, as $h\to 0$, the left-hand side of \eqref{quadratic_cond} should converge to 0 in a certain rate depending on $h$. 
\end{remark}

\section{Nonparametric Inference on $\theta(t)$ Without Positivity}
\label{sec:theta_nopos}

In this section, we present our solution for addressing the estimation biases of IPW estimators for the dose-response curve $m(t)$ and its derivative $\theta(t)$, as described in \autoref{subsec:IPW_bias}, when the positivity condition (Assumption~\ref{assump:positivity}) is violated. Specifically, our proposed IPW and DR estimators for $\theta(t)$ under the additive confounding model \eqref{add_conf_model} rely on a consistent estimation of the interior region of the support of the conditional density $p_{\bm{S}|T}(\bm{s}|t)$. Our approach establishes a connection between the classical support estimation problem and a contemporary causal inference challenge, namely the dose-response curve estimation problem.

\subsection{Bias-Corrected IPW and DR Estimators of $\theta(t)$}
\label{subsec:bias_corrected_IPW_DR}

Recall from \eqref{IPW_oracle} and Proposition~\ref{prop:IPW_bias_m} that the oracle IPW estimator of $\theta(t)$ is the sample average of the IPW quantity $\Xi_t(Y,T,\bm{S}) = \frac{Y\left(\frac{T-t}{h}\right) K\left(\frac{T-t}{h}\right)}{h^2 \cdot \kappa_2 \cdot p_{T|\bm{S}}(T|\bm{S})}$, and it is biased for estimating the quantity of interest $\theta(t)=\bar{m}'(t)$ even under model \eqref{add_conf_model}. In particular, $\E\left[\Xi_t(Y,T,\bm{S}) \right]$ converges to $\bar{m}'(t)\cdot \rho(t)$ as $h\to 0$ under some mild regularity conditions, where $\rho(t)=\P\left(\bm{S}\in \mathcal{S}(t)\right)$ for any $t\in \mathcal{T}$. The first step toward removing the bias of $\mathbb{E}\left[\Xi_t(Y,T,\bm{S})\right]$ is to decouple the quantity of interest $\theta(t)=\bar{m}'(t)$ from the nuisance function $\rho(t)$. To this end, we consider a modified IPW quantity defined as:
\begin{equation}
	\label{IPW_quantity}
	\tilde{\Xi}_t(Y,T,\bm{S}) = \Xi_t(Y,T,\bm{S}) \cdot \frac{p_{\bm{S}|T}(\bm{S}|t)}{p_S(\bm{S})} = \frac{Y\left(\frac{T-t}{h}\right) K\left(\frac{T-t}{h}\right) p_{\bm{S}|T}(\bm{S}|t)}{h^2\cdot \kappa_2 \cdot p(T,\bm{S})},
\end{equation}
in which we multiply the original IPW quantity $\Xi_t(Y,T,\bm{S})$ by a density ratio $\frac{p_{\bm{S}|T}(\bm{S}|t)}{p_S(\bm{S})}$. The following proposition demonstrates that the remaining bias in $\mathbb{E}\left[\tilde{\Xi}_t(Y,T,\bm{S})\right]$ can be disentangled from the quantity of interest $\theta(t)=\bar{m}'(t)$ in an additive form.

\begin{proposition}
	\label{prop:theta_IPW_new}
Suppose that Assumptions~\ref{assump:id_cond}(a-c), \ref{assump:reg_diff}, \ref{assump:den_diff}(c), and \ref{assump:reg_kernel}(a-b) hold under the additive confounding model \eqref{add_conf_model}. Then, when the bandwidth $h$ is small, the expectation of the modified IPW quantity \eqref{IPW_quantity} is given by
\begin{align*}
	&\mathbb{E}\left[\tilde{\Xi}_t(Y,T,\bm{S})\right] = \bar{m}'(t) + O(h^2) \\
	&\quad\quad + \int_{\mathbb{R}}\E\left\{\left[\bar{m}(t+uh) +\eta(\bm{S})\right]\left[\mathbbm{1}_{\{\bm{S}\in \mathcal{S}(t+uh)\setminus \mathcal{S}(t)\}} - \mathbbm{1}_{\{\bm{S}\in \mathcal{S}(t)\setminus \mathcal{S}(t+uh)\}}\right] \Big| T=t\right\} u\cdot K(u)\, du.
\end{align*}
\end{proposition}

\begin{remark}
\label{remark:IPW_modified}
Different from Remarks~\ref{remark:m_IPW} and \ref{remark:theta_IPW}, the conditional density $p_{\bm{S}|T}$ should be evaluated at the (query) point $(t,\bm{S})$ instead of the sample point $(T,\bm{S})$ in the modified IPW quantity \eqref{IPW_quantity}. Otherwise, the expectation of \eqref{IPW_quantity} will have an asymptotically non-vanishing additive bias; see the proof of Proposition~\ref{prop:theta_IPW_new} in \autoref{app:IPW_new_proof} for details.
\end{remark}

Proposition~\ref{prop:theta_IPW_new} reveals that the estimation bias of the modified IPW quantity \eqref{IPW_quantity} results from the support discrepancy between $\mathcal{S}(t)$ and the integration range $\mathcal{S}(t+uh)$ for a given integration variable $u\in \mathbb{R}$; see \autoref{fig:support_diff} for an illustration. As shown in Proposition~\ref{prop:IPW_bias_m}, this additive bias may not always shrink at the rate $O(h^2)$ as $h\to 0$. To further reduce the bias of the modified IPW quantity \eqref{IPW_quantity} to $O(h^2)$ without assuming positivity, we address the support discrepancy of \eqref{IPW_quantity} by restricting the conditional density $p_{\bm{S}|T}(\bm{s}|t)$ to its interior region, defining it as $p_{\zeta}(\bm{s}|t)$, and refining \eqref{IPW_quantity} as:
\begin{equation}
	\label{IPW_quantity2}
	\tilde{\Xi}_{t,\zeta}(Y,T,\bm{S}) = \frac{Y\left(\frac{T-t}{h}\right) K\left(\frac{T-t}{h}\right) p_{\zeta}(\bm{S}|t)}{h^2\cdot \kappa_2 \cdot p(T,\bm{S})}.
\end{equation}
Essentially, the only requirement for defining the $\zeta$-interior conditional density $p_{\zeta}(\bm{s}|t)$ is that its support satisfies the following condition:
\begin{equation}
	\label{int_supp}
	\{\bm{s} \in \mathcal{S}(t): p_{\zeta}(\bm{s}|t) > 0\} \subset \mathcal{S}(t+\delta)\quad \text{ for any }\quad  \delta\in[-h, h].
\end{equation}
Here, we propose two approaches for defining $p_{\zeta}(\bm{s}|t)$ and leave other options to interested readers.

{\bf 1. Support Shrinking Approach:} Let $\mathcal{S}(t) \ominus \zeta = \left\{\bm{s}\in \mathcal{S}(t): \inf_{\bm{x}\in \partial \mathcal{S}(t)}\norm{\bm{s}-\bm{x}}_2 \geq \zeta\right\}$ denote the set of interior points of $\mathcal{S}(t)$ that are at least a distance $\zeta$ away from the boundary $\mathcal{S}(t)$. Then, we define the $\zeta$-interior conditional density with $\zeta>0$ being a tuning parameter as:
\begin{equation}
	\label{interior_cond_density}
	p_{\zeta}(\bm{s}|t) = \frac{p_{\bm{S}|T}(\bm{s}|t) \cdot \mathbbm{1}_{\left\{\bm{s}\in \mathcal{S}(t)\ominus \zeta\right\}}}{\int_{\mathcal{S}(t)\ominus \zeta} p_{\bm{S}|T}(\bm{s}_1|t) \,d\bm{s}_1} \propto p_{\bm{S}|T}(\bm{s}|t) \cdot \mathbbm{1}_{\left\{\bm{s}\in \mathcal{S}(t)\ominus \zeta\right\}}.
\end{equation}
This interior density is indeed the conditional density $p_{\bm{S}|T}(\bm{s}|t)$ restricted to the interior of its support $\mathcal{S}(t)$. Its estimator $\hat{p}_{\zeta}(\bm{s}|t)$ can be constructed using a support estimator $\hat{\mathcal{S}}(t)$ and constraining the conditional density estimator $\hat{p}_{\bm{S}|T}(\bm{s}|t)$ within the region $\hat{\mathcal{S}}(t)\ominus \zeta$. 

{\bf 2. Level Set Approach:} Let $\mathcal{L}_{\zeta}(t) = \left\{\bm{s}\in \mathcal{S}(t): p_{\bm{S}|T}(\bm{s}|t) \geq \zeta\right\}$ be the $\zeta$-upper level set of the conditional density $p_{\bm{S}|T}(\bm{s}|t)$. Then, we define the $\zeta$-interior conditional density as:
\begin{equation}
\label{interior_cond_density_levelset}
p_{\zeta}(\bm{s}|t) = \frac{p_{\bm{S}|T}(\bm{s}|t) \cdot \mathbbm{1}_{\left\{\bm{s}\in \mathcal{L}_{\zeta}(t)\right\}}}{\int_{\mathcal{L}_{\zeta}(t)} p_{\bm{S}|T}(\bm{s}_1|t) \,d\bm{s}_1} \propto p_{\bm{S}|T}(\bm{s}|t) \cdot \mathbbm{1}_{\left\{\bm{s}\in \mathcal{L}_{\zeta}(t)\right\}}.
\end{equation}
The level set approach restricts the conditional density $p_{\bm{S}|T}(\bm{s}|t)$ to the high-density region, which is generally located away from the support boundary. We may construct the estimator $\hat p_{\zeta}(\bm{s}|t)$ using a level set estimator
$\hat{\mathcal{L}}_{\zeta}(t) = \left\{\bm{s}\in \mathcal{S}(t): \hat{p}_{\bm{S}|T}(\bm{s}|t) \geq \zeta\right\}$ and constraining $\hat{p}_{\bm{S}|T}(\bm{s}|t)$ to $\hat{\mathcal{L}}_{\zeta}(t)$. 

We further specialize condition \eqref{int_supp} for the above two approaches by introducing the following smoothness condition on the conditional support $\mathcal{S}(t)$.

\begin{figure}
	\centering
	\includegraphics[width=1\linewidth]{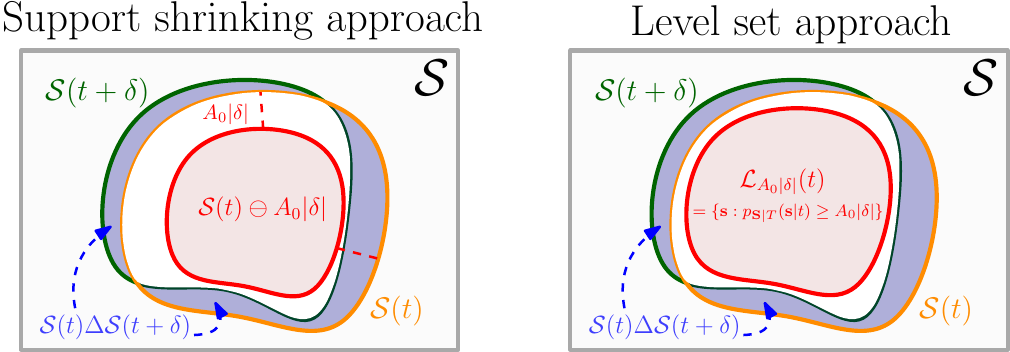}
	\caption{Graphical illustrations of the support discrepancy between $\mathcal{S}(t)$ and $\mathcal{S}(t+\delta)$ for $t\in \mathcal{T}$ as well as Assumption~\ref{assump:cond_support_smooth}, where $\delta$ can take its value as $uh\in \mathbb{R}$.}
	\label{fig:support_diff}
\end{figure}

\begin{assump}[Smoothness condition on $\mathcal{S}(t)$]
	\label{assump:cond_support_smooth}
	For any $\delta \in \mathbb{R}$ and $t\in \mathcal{T}$, there exists an absolute constant $A_0>0$ such that either (i) ``$\,\mathcal{S}(t) \ominus \left(A_0|\delta| \right) \subset \mathcal{S}(t+\delta)$'' for the support shrinking approach or (ii) ``$\,\mathcal{L}_{A_0|\delta|}(t) \subset \mathcal{S}(t+\delta)$'' for the level set approach.
\end{assump} 

To some extent, Assumption~\ref{assump:cond_support_smooth} can be viewed as a Lipschitz condition of the conditional support $\mathcal{S}(t)$. It can be satisfied when the Euclidean norm of the gradient $\norm{\nabla_{\bm{s}} p_{\bm{S}|T}(\bm{s}|t)}_2$ is bounded away from 0 at the boundary of $\mathcal{S}(t)$ \citep{cadre2006kernel}. This assumption allows us to ignore the boundary discrepancy as long as we do not evaluate our IPW quantity \eqref{IPW_quantity2} near the boundary; see \autoref{fig:support_diff} for a graphical illustration.

\begin{proposition}
	\label{prop:theta_IPW_new2}
	Suppose that Assumptions~\ref{assump:id_cond}(a-c), \ref{assump:reg_diff}, \ref{assump:den_diff}(c), \ref{assump:reg_kernel}(a-b), and \ref{assump:cond_support_smooth} hold under the additive confounding model \eqref{add_conf_model}. Then, when the bandwidth $h>0$ is small, the expectation of the modified IPW quantity \eqref{IPW_quantity2} is given by
	$$\mathbb{E}\left[\tilde{\Xi}_{t,\zeta}(Y,T,\bm{S}) \right] = \bar{m}'(t) + \frac{h^2\kappa_4}{6\kappa_2}\cdot \bar{m}^{(3)}(t) + O\left(h^3\right).$$
\end{proposition}

The proof of Proposition~\ref{prop:theta_IPW_new2} is in \autoref{app:IPW_new_proof2}. This result demonstrates that the expectation of our newly modified IPW quantity $\tilde{\Xi}_{t,\zeta}(Y,T,\bm{S})$ in \eqref{IPW_quantity2} converges to the quantity of interest $\theta(t)=\bar{m}'(t)$ in the standard order $O(h^2)$ as $h\to 0$ under the additive confounding model \eqref{add_conf_model}. Notice that the tuning parameter $\zeta=\zeta_n >0$ in \eqref{interior_cond_density} is allowed to converge to 0 as $n\to \infty$, as long as the condition $h=h_n < \frac{\zeta_n}{A_0}$ holds under Assumption~\ref{assump:cond_support_smooth}. 

Given this newly modified IPW quantity \eqref{IPW_quantity2}, we propose the bias-corrected IPW estimator of $\theta(t)$ without the positivity condition as:
\begin{equation}
	\label{theta_IPW_bnd_corrected}
	\hat{\theta}_{\mathrm{C,IPW}}(t) = \frac{1}{nh^2} \sum_{i=1}^n \frac{Y_i\left(\frac{T_i-t}{h}\right) K\left(\frac{T_i-t}{h}\right) \hat{p}_{\zeta}(\bm{S}_i|t)}{\kappa_2 \cdot \hat{p}(T_i,\bm{S}_i)},
\end{equation}
where $\hat{p}(t,\bm{s})$ is a consistent estimator of the joint density $p(t,\bm{s})$ and $\hat{p}_{\zeta}(\bm{s}|t)$
is an estimated $\zeta$-interior conditional density.

Finally, we combine the modified RA estimator \eqref{theta_RA_corrected} with our bias-corrected IPW estimator \eqref{theta_IPW_bnd_corrected} to propose our bias-corrected DR estimator of $\theta(t)$ as:
\begin{equation}
	\label{theta_DR_bnd_corrected}
	\hat{\theta}_{\mathrm{C,DR}}(t) = \frac{1}{nh^2} \sum_{i=1}^n \frac{\left(\frac{T_i-t}{h}\right) K\left(\frac{T_i-t}{h}\right) \hat{p}_{\zeta}(\bm{S}_i|t)}{\kappa_2\cdot \hat{p}(T_i,\bm{S}_i)} \left[Y_i - \hat{\mu}(t,\bm{S}_i) - (T_i-t)\cdot \hat{\beta}(t,\bm{S}_i)\right] + \int \hat{\beta}(t,\bm{s})\cdot \hat{p}_{\zeta}(\bm{s}|t)\, d\bm{s}.
\end{equation}
Notice that for the RA component of $\hat{\theta}_{\mathrm{C,DR}}(t)$, we replace the original conditional CDF estimator $\hat{F}_{\bm{S}|T}$ in \eqref{theta_RA_corrected} with the estimated $\zeta$-interior conditional density $\hat{p}_{\zeta}$. This modification is necessary because the IPW component of $\hat{\theta}_{\mathrm{C,DR}}(t)$ is defined through $\hat{p}_{\zeta}$. Both the RA and IPW components need to match up with each other in the definition of $\hat{\theta}_{\mathrm{C,DR}}(t)$ for its consistency.

\begin{remark}
	\label{remark:interior_density_est}
	While both the support shrinking and level set approaches are valid, we recommend the level set approach in practice, because support estimation is a notoriously challenging problem in nonparametric statistics \citep{devroye1980detection}. Additionally, selecting an appropriate $\zeta$ for the support shrinking method is nontrivial. In contrast, level set estimation has been studied over decades \citep{cuevas1997plug,cadre2006kernel},
	and the threshold can be set as $\zeta =0.5\cdot \max\left\{\hat{p}_{\bm{S}|T}(\bm{S}_i|t): i=1,...,n\right\}$.
	Notice that users may adjust the multiplier 0.5 in this rule, where a smaller value generally increases the effective sample size but also raises the risk of violating condition \eqref{int_supp}.
\end{remark}

\subsection{Asymptotic Theory}
\label{subsec:asymp_theta_nopos}

The following theorem summarizes the asymptotic properties of our DR estimator \eqref{theta_DR_bnd_corrected} of $\theta(t)$ under the additive confounding model \eqref{add_conf_model} without assuming the positivity condition. We again defer its proof and other consistency results for our RA \eqref{theta_RA_corrected} and IPW \eqref{theta_IPW_bnd_corrected} estimators of $\theta(t)$ to \autoref{app:theta_nopos_proof}.

\begin{theorem}[Asymptotic properties of $\hat{\theta}_{\mathrm{C,DR}}(t)$ without positivity]
	\label{thm:theta_nopos}
	Suppose that Assumptions~\ref{assump:id_cond}(a-c), \ref{assump:reg_diff}, \ref{assump:den_diff}, \ref{assump:reg_kernel}, and \ref{assump:cond_support_smooth} hold under the additive confounding model \eqref{add_conf_model}, and the support $\mathcal{S} \subset \mathbb{R}^d$ of the marginal density $p_{\bm{S}}$ is compact. In addition, $\hat{\mu},\hat{\beta}, \hat{p}_{\zeta}, \hat{p}$ are constructed on a data sample independent of $\{(Y_i,T_i,\bm{S}_i)\}_{i=1}^n$. For any fixed $t\in \mathcal{T}$, we let $\bar{\mu}(t,\bm{s})$, $\bar{\beta}(t,\bm{s})$, $\bar{p}_{\zeta}(\bm{s}|t)$, and $\bar{p}(t,\bm{s})$ be fixed bounded functions to which $\hat{\mu}(t,\bm{s})$, $\hat{\beta}(t,\bm{s})$, $\hat{p}_{\zeta}(\bm{s}|t)$, and $\hat{p}(t,\bm{s})$ converge.
	If, in addition, we assume that 
	\begin{enumerate}[label=(\alph*)]
		\item $\bar{p},\bar{p}_{\zeta}$ satisfy Assumptions~\ref{assump:den_diff} and \ref{assump:cond_support_smooth} as well as $\sqrt{nh^3} \norm{\hat{p}_{\zeta}(\bm{S}|t) - \bar{p}_{\zeta}(\bm{S}|t)}_{L_2}=o(1)$;
		\item either (i) ``$\,\bar{\mu}=\mu$ and $\bar{\beta}=\beta$'' or (ii) ``$\,\bar{p} = p$'';
		\item {\small $\sqrt{nh} \left[\norm{\hat{p}_{\zeta}(\bm{S}|t) - \bar{p}_{\zeta}(\bm{S}|t)}_{L_2} + \sup\limits_{|u-t|\leq h} \norm{\hat{p}(u,\bm{S}) - p(u,\bm{S})}_{L_2} \right]\left[\norm{\hat{\mu}(t,\bm{S}) - \mu(t,\bm{S})}_{L_2} + h \norm{\hat{\beta}(t,\bm{S}) - \beta(t,\bm{S})}_{L_2}\right] = o_P(1)$},
	\end{enumerate}
	then 
	\begin{align*}
		&\sqrt{nh^3}\left[\hat{\theta}_{\mathrm{C,DR}}(t) - \theta(t)\right] \\
		&= \frac{1}{\sqrt{n}} \sum_{i=1}^n \left\{\phi_{C,h,t}\left(Y_i,T_i,\bm{S}_i;\bar{\mu}, \bar{\beta}, \bar{p},\bar{p}_{\zeta}\right) + \sqrt{h^3}\left[\int \bar{\beta}(t,\bm{s})\cdot \bar{p}_{\zeta}(\bm{s}|t)\, d\bm{s} - \theta(t)\right] \right\} + o_P(1)
	\end{align*}
	when $nh^7\to c_3$ for some finite number $c_3\geq 0$, where $$\phi_{C,h,t}\left(Y,T,\bm{S}; \bar{\mu},\bar{\beta}, \bar{p},\bar{p}_{\zeta}\right) = \frac{\left(\frac{T-t}{h}\right) K\left(\frac{T-t}{h}\right) \cdot \bar{p}_{\zeta}(\bm{S}|t)}{\sqrt{h}\cdot \kappa_2\cdot \bar{p}(T,\bm{S})}\cdot \left[Y - \bar{\mu}(t,\bm{S}) - (T-t)\cdot \bar{\beta}(t,\bm{S})\right].$$
	Furthermore, 
	$$\sqrt{nh^3}\left[\hat{\theta}_{\mathrm{C,DR}}(t) - \theta(t) - h^2 B_{C,\theta}(t)\right] \stackrel{d}{\to} \mathcal{N}\left(0,V_{C,\theta}(t)\right)$$
	with $V_{C,\theta}(t) = \mathbb{E}\left[\phi_{C,h,t}^2\left(Y,T,\bm{S};\bar{\mu}, \bar{\beta}, \bar{p},\bar{p}_{\zeta}\right)\right]$ and 
	\begin{align*}
		B_{C,\theta}(t) &= \begin{cases}
			\frac{\kappa_4}{6\kappa_2} \int \left\{\frac{3\frac{\partial}{\partial t} p(t,\bm{s}) \cdot \bar{m}''(t) + p(t,\bm{s})\left[\bar{m}^{(3)}(t) - 3\frac{\partial}{\partial t} \log\bar{p}(t,\bm{s}) \cdot \bar{m}''(t) \right]}{\bar{p}(t,\bm{s})} \right\} \bar{p}_{\zeta}(\bm{s}|t)\, d\bm{s}& \text{ when } \bar{\mu}=\mu \text{ and } \bar{\beta}=\beta,\\
			\frac{\kappa_4}{6\kappa_2} \cdot \bar{m}^{(3)}(t) & \text{ when } \bar{p} = p.
		\end{cases}
	\end{align*}
\end{theorem}

Similar to our discussions in \autoref{subsec:theta_inference_pos} and \autoref{subsec:nonp_eff}, $V_{C,\theta}(t)$ in \autoref{thm:theta_nopos} attains the nonparametric efficiency bound derived from a smooth functional $\P\mapsto \varpi_{C,h,t}(\P):= \mathbb{E}\left[\frac{Y \cdot \left(\frac{T-t}{h}\right) K\left(\frac{T-t}{h}\right) \cdot \bar{p}_{\zeta}(S|t)}{h^2\cdot \kappa_2 \cdot p(T,\bm{S})}\right]$, and we can estimate it by
$$\hat{V}_{C,\theta}(t) = \frac{1}{n} \sum_{i=1}^n \left\{\phi_{C,h,t}\left(Y_i,T_i,\bm{S}_i;\hat{\mu}, \hat{\beta}, \hat{p}, \hat{p}_{\zeta}\right) + \sqrt{h^3}\left[\int \hat{\beta}(t,\bm{s}) \cdot \hat{p}_{\zeta}(\bm{s}|t)\, d\bm{s} - \hat{\theta}_{\mathrm{C,DR}}(t) \right]\right\}^2$$
and choose the bandwidth $h$ to be of order $O\left(n^{-\frac{1}{5}}\right)$ to ensure valid inference. 
As a corollary, we can plug either IPW \eqref{theta_IPW_bnd_corrected} or DR \eqref{theta_DR_bnd_corrected} estimators into our integral formula \eqref{m_RA_corrected} to obtain the integral IPW or DR estimators of the dose-response curve $m(t)$ under model \eqref{add_conf_model}. We establish the asymptotic theory for these integral estimators in Corollary~\ref{cor:m_nopos} of \autoref{app:m_nopos_proof}.

\begin{remark}
Under standard regularity conditions in nonparametric estimation \citep{wasserman2006all}, the rates of convergence for $\hat{\mu}, \hat{p}_{T|\bm{S}},\hat{p}$ in \autoref{thm:theta_pos} and \autoref{thm:theta_nopos} would be of order $O\left(n^{-\frac{2}{4 + d}}\right)$ up to some possible $\log n$ factors, while $\hat{\beta}$ converges at rate $O\left(n^{-\frac{2}{6 + d}}\right)$. 
As shown by \cite{farrell2021deep,colangelo2020double}, these rates are attainable by neural network models. Additionally, the rate of convergence for $\hat{F}_{\bm{S}|T}$ can be dimensionally independent, achieving $O\left(\left(\frac{\log n}{n}\right)^{\frac{2}{5}}\right)$ \citep{einmahl2005uniform}. Faster rates can be obtained under higher-order smoothness conditions and with higher-order kernel functions. Finally, the typical rate of convergence for $\hat{p}_{\zeta}$ to $p_{\zeta}$ is of order $O\left(n^{-\frac{2}{5+d}}\right)$ \citep{cuevas1997plug,tsybakov1997nonparametric}, which seems to be slower than the requirement $\sqrt{nh}\norm{\hat{p}_{\zeta}(\bm{S}|t) - \bar{p}_{\zeta}(\bm{S}|t)}_{L_2} = o_P\left(1\right)$. However, we emphasize that the limiting quantity $\bar{p}_{\zeta}$ of $\hat{p}_{\zeta}$ in \autoref{thm:theta_nopos} needs not be the true interior conditional density $p_{\zeta}$ but rather its smooth surrogate. This flexibility allows for constructing $\hat{p}_{\zeta}$ with a smaller bandwidth (distinct from $h$) or using a more data-adaptive method that ensures sufficiently fast convergence to $\bar{p}_{\zeta}$.
\end{remark}

\section{Numerical Experiments}
\label{sec:experiments}

In this section, we evaluate the finite-sample performances of our proposed estimators of $\theta(t)=\frac{d}{dt}\E\left[Y(t)\right]$ in \autoref{sec:theta_pos} and compare them with the finite-difference approach in \cite{colangelo2020double} under the positivity condition through simulation studies and an analysis of the Job Corps program in the United States. Furthermore, we compare the bias-corrected estimators of $\theta(t)$ in \autoref{subsec:bias_corrected_IPW_DR} with their counterparts via simulation studies when the positivity condition is violated. 

\subsection{Simulation Studies With Positivity}
\label{subsec:simulation}

We generate i.i.d. observations $\{(Y_i,T_i,\bm{S}_i)\}_{i=1}^n$ from the following data-generating model as in \cite{colangelo2020double,klosin2021automatic}:
\begin{align}
\label{dgp}
\begin{split}
&Y=1.2\,T+T^2 + TS_1+1.2\,\bm{\xi}^T\bm{S} +\epsilon\cdot \sqrt{0.5+F_{\mathcal{N}(0,1)}(S_1)}, \quad \epsilon\sim \mathcal{N}(0,1),\\ 
&T= F_{\mathcal{N}(0,1)}\left(3\bm{\xi}^T\bm{S}\right) - 0.5 + 0.75E, \quad  \bm{S}=(S_1,...,S_d)^T \sim \mathcal{N}_d\left(\bm{0}, \Sigma\right), \quad E\sim \mathcal{N}(0,1),
\end{split}
\end{align}
where $F_{\mathcal{N}(0,1)}$ is the CDF of $\mathcal{N}\left(0, 1\right)$, $\bm{\xi}=(\xi_1,...,\xi_d)^T \in \mathbb{R}^d$ has its entry $\xi_j=\frac{1}{j^2}$ for $j=1,...,d$ as well as $\Sigma_{ii}=1$, $\Sigma_{ij}=0.5$ when $|i-j|=1$, and $\Sigma_{ij}=0$ when $|i-j|>1$ for $i,j=1,...,d$. Here, $d=20$ unless stated otherwise. The dose-response curve is thus given by $m(t)= 1.2t+t^2$, and our parameter of interest is the derivative effect curve $\theta(t)=1.2+2t$.

We evaluate our proposed estimators of $\theta(t)$ in \autoref{sec:theta_pos} alongside the finite-difference estimator by \cite{colangelo2020double} with 5-fold cross-fitting. In particular, we replicate their finite-difference estimators  using their neural network (NN) and kernel neural network (KNN) models for estimating the nuisance functions $\mu(t,\bm{s})$ and $p_{T|\bm{S}}(t|\bm{s})$, which yield their best performances. Additionally, similar to the setups in \cite{colangelo2020double,klosin2021automatic}, we use the Epanechnikov kernel $K(u)=\frac{3}{4}(1-|u|)\, \mathbbm{1}_{\{|u|\leq 1\}}$ under a bandwidth choice $h=1.25\, \hat{\sigma}_T\cdot n^{-\frac{1}{5}}$, where $\hat{\sigma}_T$ is the sample standard deviation of $\{T_1,...,T_n\}$. Furthermore, for our proposed estimators, the nuisance functions $\mu(t,\bm{s})$ and $\beta(t,\bm{s})$ are estimated by neural network models as well, while $p_{T|\bm{S}}(t|\bm{s})$ is estimated by either the method of kernel density estimation (KDE) on residuals or the approach of regressing kernel-smoothed outcomes (RKS); see \autoref{app:nuisance_est} for details. To prevent division by zero, all estimated conditional density values $\hat{p}_{T|\bm{S}}(T_i|\bm{S}_i), i=1,...,n$ smaller than 0.001 are set to this value. For comparison, we also implement our proposed DR estimator of $\theta(t)$ under the true conditional density (``True''). All our DR estimators are self-normalized as described in \autoref{app:self_norm} to reduce their variances. The nominal levels of all the yielded pointwise confidence intervals are set to 95\%.

\begin{figure}[t]
	\centering
	\includegraphics[width=1\linewidth]{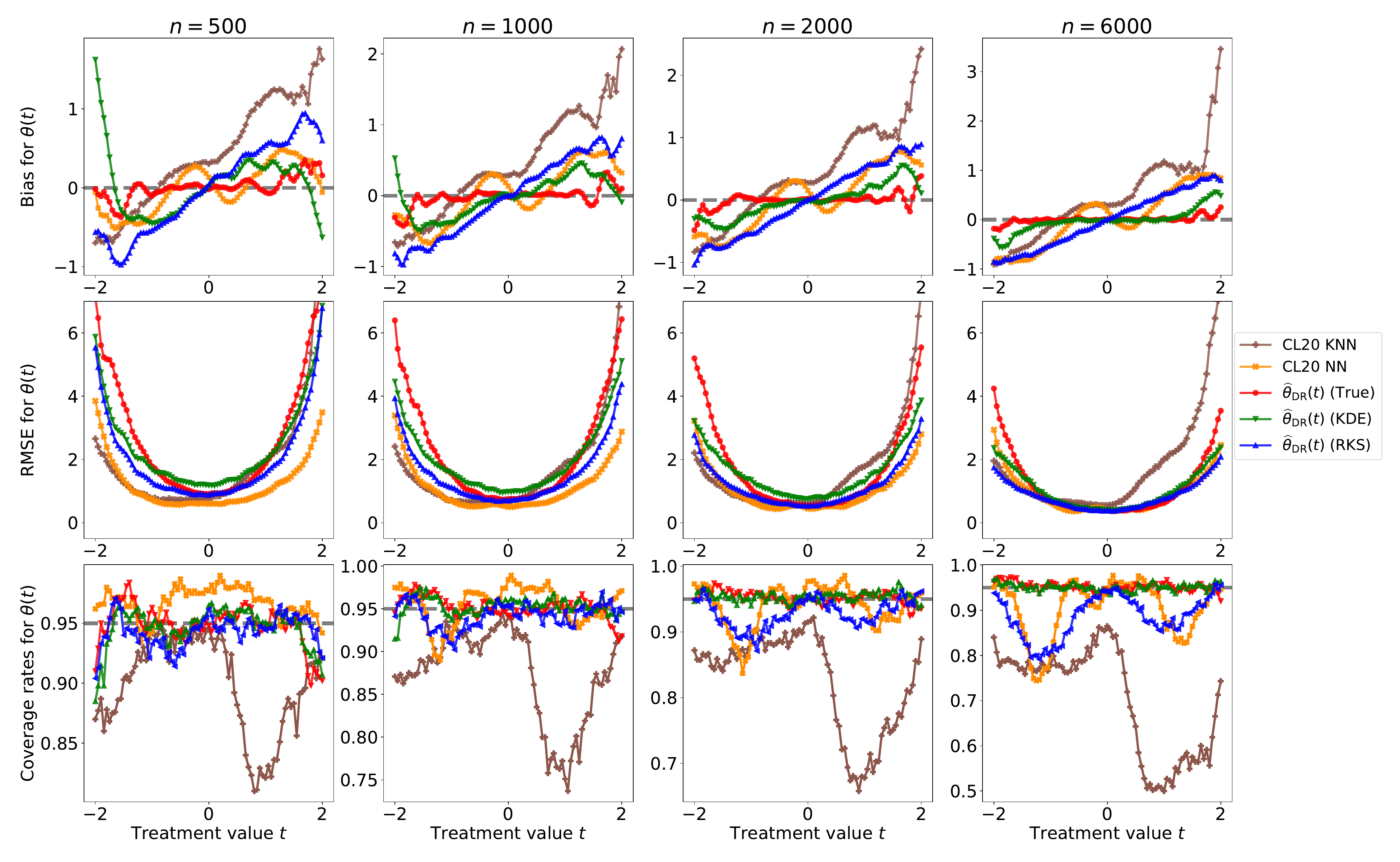}
	\caption{Comparisons between our proposed estimators and the finite-difference approaches by \cite{colangelo2020double} (``CL20'') under positivity and with 5-fold cross-fitting across various sample sizes. Rows present estimation bias, RMSE, and coverage probability for each estimator of $\theta(t)$, while columns correspond to different values for $n$.}
	\label{fig:theta_L5}
\end{figure}

The simulation results are shown in \autoref{fig:theta_L5} for various sample sizes, where the estimation biases, root mean square errors (RMSEs), and coverage rates of confidence intervals are calculated by averaging over 1000 Monte Carlo replications. Additional results when the bandwidth parameter varies or cross-fitting is not employed are in \autoref{app:add_pos_bandwidth} and \autoref{app:add_pos_nocrossfitting}. Unlike prior studies in \cite{colangelo2020double,klosin2021automatic}, which focus solely on $t=0$, our comparative simulations evaluate 81 treatment values across $t\in [-2,2]$. 
Overall, our proposed DR estimators, using either true or KDE-estimated conditional density values, outperform the finite-difference methods of \cite{colangelo2020double} in terms of estimation bias while maintaining comparable RMSE. When it comes to statistical inference, the confidence intervals from our DR estimators consistently show better empirical coverages than those from \cite{colangelo2020double}. These performance advantages of our DR estimators arise from directly estimating and inferring $\theta(t)$ without requiring a step-size parameter for finite-difference approximations.

\subsection{Simulation Studies Without Positivity}
\label{subsec:sim_nopos}

We now assess the finite-sample performances of our bias-corrected IPW and DR estimators of $\theta(t)$ in \autoref{subsec:bias_corrected_IPW_DR} and compare them with those counterparts in \autoref{sec:theta_pos} when the positivity condition is violated. To this end, we generate i.i.d. data $\{(Y_i, T_i,S_i)\}_{i=1}^n$ from the following data-generating model
\begin{align}
	\label{dgp2}
	\begin{split}
		&Y=T^3 +T^2 + 10S+\epsilon, \quad T= \sin(\pi S) + E, \quad S\sim \mathrm{Uniform}[-1,1] \subset \mathbb{R}, 
	\end{split}
\end{align}
where $E\sim \mathrm{Uniform}[-0.3,0.3]$ is an independent treatment variation and $\epsilon\sim \mathcal{N}(0,1)$ is an independent noise variable. The marginal supports of $T$ and $S$ are $\mathcal{T}=[-1.3, 1.3]$ and $\mathcal{S}=[-1,1]$ respectively, while the joint support of $(T,S)$ only covers a thin band region of the product space $\mathcal{T}\times \mathcal{S}$; see Figure 1 in \cite{zhang2024nonparametric} for illustration. The true derivative effect curve is thus given by $\theta(t)=3t^2 + 2t$.

We evaluate our bias-corrected estimators of $\theta(t)$ in \autoref{subsec:bias_corrected_IPW_DR} on the simulated dataset, alongside those estimators from \autoref{sec:theta_pos} that assumes the positivity condition. All these estimators are assessed with 5-fold cross-fitting. Again, we use the Epanechnikov kernel $K(u)=\frac{3}{4}(1-|u|)\, \mathbbm{1}_{\{|u|\leq 1\}}$ under a bandwidth choice $h=2\,\hat{\sigma}_T\cdot n^{-\frac{1}{5}}$. For those estimators assuming positivity, we estimate the nuisance functions $\mu(t,s)$ and $\beta(t,s)$ by neural network models in \autoref{app:nuisance_est} and utilize the true conditional density function $p_{T|S}$ evaluated at the observations $\{(T_i,S_i)\}_{i=1}^n$. For the bias-corrected estimators, we estimate the joint density $p(t,s)$ and conditional density $p_{S|T}(s|t)$ using kernel density estimation with a Gaussian kernel $K(u)=\frac{1}{\sqrt{2\pi}}\exp\left(-\frac{u^2}{2}\right)$. The estimated interior densities $\hat{p}_{\zeta}(S_i|t), i=1,...,n$ are computed via the trimming method outlined in Remark~\ref{remark:interior_density_est}. All the estimators are self-normalized as described in \autoref{app:self_norm} to reduce their variances, and the nominal levels of all the yielded pointwise confidence intervals are set to 95\%.

\begin{figure}[t]
	\centering
	\includegraphics[width=1\linewidth]{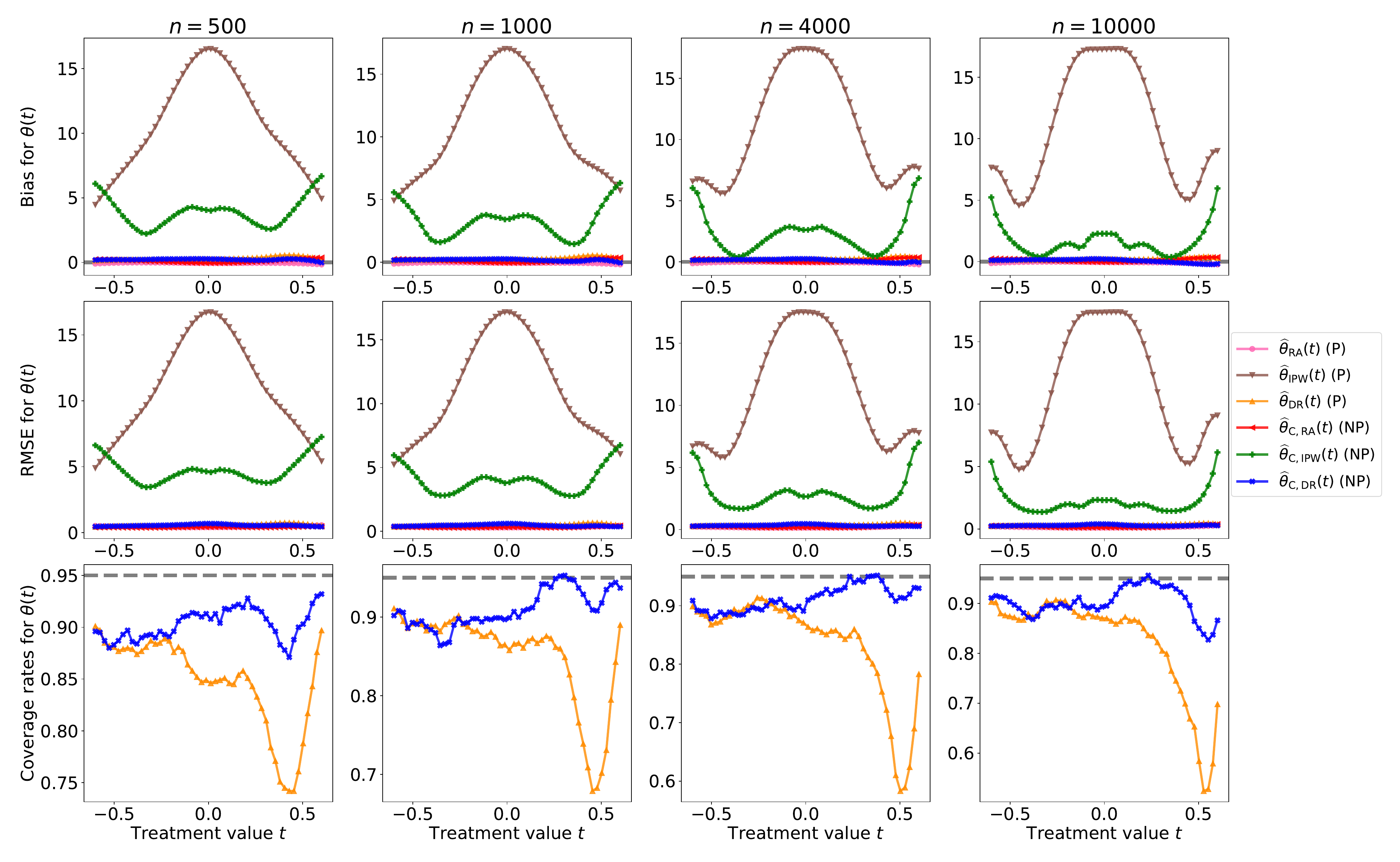}
	\caption{Comparisons between our bias-corrected estimators (NP) in \autoref{subsec:bias_corrected_IPW_DR} and their counterparts (P) in \autoref{sec:theta_pos} under the violation of positivity and with 5-fold cross-fitting across different sample sizes. Rows present estimation bias, RMSE, and coverage probability for each estimator of $\theta(t)$, while columns correspond to different values for $n$.}
	\label{fig:theta_L5_nopos}
\end{figure}

The simulation results for different sample sizes are presented in \autoref{fig:theta_L5_nopos}, where the estimation biases, root mean square errors (RMSEs), and coverage rates of confidence intervals are calculated by averaging over 1000 Monte Carlo replications. Additional results when the bandwidth parameter varies or cross-fitting is not employed are in \autoref{app:add_sim_nopos_bandwidth} and \autoref{app:add_nopos_nocrossfitting}. The bias-corrected IPW estimator \eqref{theta_IPW_bnd_corrected} effectively reduces the estimation biases of the standard IPW estimator \eqref{theta_IPW} of $\theta(t)$ across when the positivity condition is violated. Furthermore, the bias-corrected DR estimator \eqref{theta_DR_bnd_corrected} achieves comparable biases and RMSEs to its standard counterpart \eqref{theta_DR}, even when \eqref{theta_DR} uses the oracle conditional density $p_{T|S}$. Notably, the confidence intervals yielded by the bias-corrected DR estimator \eqref{theta_DR_bnd_corrected} exhibit better coverage probabilities compared to its counterpart \eqref{theta_DR}. These findings support the theoretical properties of our proposed bias-corrected IPW and DR estimators in \autoref{subsec:bias_corrected_IPW_DR}. Nonetheless, the bias-corrected RA estimator \eqref{theta_RA_corrected} remains the preferred choice when it comes to estimation accuracy due to its simplicity under violations of the positivity condition.

\subsection{Case Study: An Analysis of the Job Corps Program}
\label{subsec:job_corps}

We demonstrate the applicability of our proposed DR estimators for $\theta(t)$ by extending the analysis of \cite{colangelo2020double} on the Job Corps program in the United States (U.S.). This program aims at providing academic and vocational training to U.S. legal residents aged 16--24 who come from low-income households \citep{schochet2001national}. The data used in our analysis originated from the National Job Corps Study, which conducted some randomized experiments on first-time applicants in the 48 contiguous states and the District of Columbia between November 1994 and February 1996 \citep{schochet2008does}. 

\begin{figure}[t]
	\centering
	\includegraphics[width=1\linewidth]{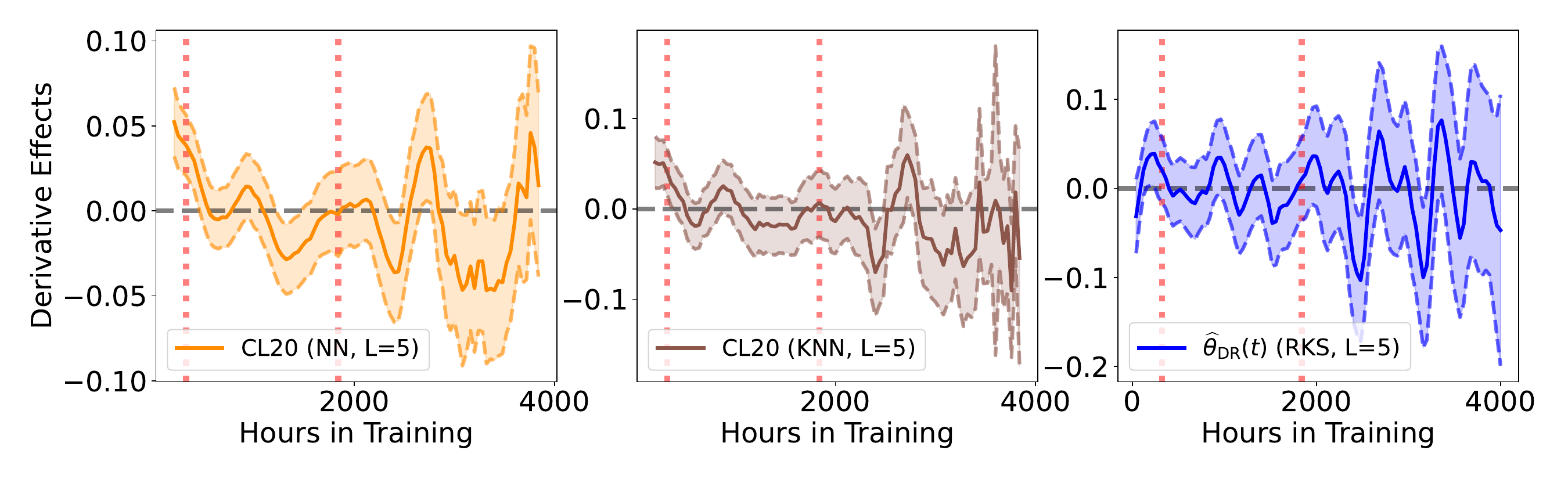}
	\caption{Estimated derivative effect curves with 95\% confidence intervals using our proposed estimators and the finite-difference approaches by \cite{colangelo2020double} (``CL20'') under 5-fold cross-fitting. The vertical red dotted lines mark the original treatment range $[320,1840]$ analyzed in \cite{colangelo2020double}.}
	\label{fig:job_corps}
\end{figure}

Numerous studies have examined the causal effects of the Job Corps program from various angles \citep{flores2009identification,flores2012estimating,huber2014identifying,lee2018partial,huber2020direct,lee2024lee}. Following \cite{colangelo2020double}, we analyze the relationship between employment outcomes and the duration of academic and vocational training, focusing on the derivative effect curve $\theta(t)=\frac{d}{dt}\E\left[Y(t)\right]$. The data sample includes 4,024 individuals who received at least 40 hours of training. The outcome variable $Y$ represents the proportions of weeks employed in the second year following the program assignment, and the treatment variable $T$ is the total hours of academic and vocational training received. The covariate vector $\bm{S}$, comprising 49 socioeconomic characteristics, ensures the validity of the ignorability assumption \citep{flores2012estimating}; see Table 4 in \cite{huber2020direct} for detailed descriptions of the covariates. Before applying derivative effect estimation methods, categorical covariates were converted to dummy variables, and all variables were standardized to have mean 0 and variance 1.

We apply our proposed DR estimator \eqref{theta_DR} with the same setup as in \autoref{subsec:simulation} to the standardized data, extending the range of queried treatment values from $[320,1840]$ to $[40, 4000]$. For consistency, we use the same bandwidth parameter $h=223$ and apply the neural network model for conditional density estimation as in \cite{colangelo2020double}. The estimated derivative effect curves with 95\% confidence intervals under 5-fold cross-fitting are shown in \autoref{fig:job_corps}. Overall, our DR estimator produces similar patterns to the finite-difference estimates from \cite{colangelo2020double}. However, our confidence intervals are more conservative and include 0 for nearly all treatment values, suggesting insufficient evidence to confirm the program’s effectiveness. Additional results when cross-fitting is not employed are shown in \autoref{app:job_corps_nocrossfitting}.


\section{Discussion}

In summary, this paper studies nonparametric DR inference methods for the derivative function of the dose-response curve with and without the positivity condition. We establish the asymptotic properties of our proposed estimators under mild conditions, permitting the use of machine learning methods for nuisance function estimation with cross-fitting. Furthermore, our identification theory and refinements of IPW and DR estimators without positivity open up a novel link between the dose-response curve inference challenge and the nonparametric set estimation problem. Simulation studies and empirical applications demonstrate the advantages of our DR estimator over the existing finite-difference method for derivative effect inference. This work also highlights several avenues for future research.


{\bf 1. Bias correction for DR estimators:} As shown in \autoref{thm:theta_pos} and \autoref{thm:theta_nopos}, our DR estimators of $\theta(t)$ contain bias terms of order $O(h^2)$. These biases become asymptotically negligible when the bandwidth is chosen as $h\asymp n^{-\frac{1}{5}}$ that matches up the standard rate of convergence for nonparametric regression. To guarantee valid inference, an alternative approach is to explicitly estimate and correct these bias terms, as demonstrated by \cite{calonico2018effect,cheng2019nonparametric,takatsu2022debiased}. A rigorous investigation of this bias-corrected approach for our DR estimators would be a valuable direction for future research.


{\bf 2. Derivative estimation in other causal contexts:} Our proposed DR inference methods for $\theta(t)$ can be naturally extended to conduct inference on other causal estimands of interest, such as the instantaneous causal effect $\frac{d}{d t} \E\left[Y(t)|\bm{S}=\bm{s}\right]$ \citep{stolzenberg1980measurement,ratkovic2023estimation} or the marginal direct and indirect effects in causal mediation analysis \citep{huber2020direct,xu2021multiply}.



\section*{Acknowledgement}

We thank Alex Luedtke and Jon A. Wellner for their helpful comments. YZ is supported in part by YC's NSF grant DMS-2141808. YC is supported by NSF grants DMS-1952781, 2112907, 2141808, and NIH U24-AG07212.

\vspace{5mm}
\singlespacing
\bibliography{IPWDR_bib}

\begin{thebibliography}{98}
\providecommand{\natexlab}[1]{#1}
\providecommand{\url}[1]{\texttt{#1}}
\expandafter\ifx\csname urlstyle\endcsname\relax
  \providecommand{\doi}[1]{doi: #1}\else
  \providecommand{\doi}{doi: \begingroup \urlstyle{rm}\Url}\fi

\bibitem[Bang and Robins(2005)]{bang2005doubly}
H.~Bang and J.~M. Robins.
\newblock Doubly robust estimation in missing data and causal inference models.
\newblock \emph{Biometrics}, 61\penalty0 (4):\penalty0 962--973, 2005.

\bibitem[Baydin et~al.(2018)Baydin, Pearlmutter, Radul, and
  Siskind]{baydin2018automatic}
A.~G. Baydin, B.~A. Pearlmutter, A.~A. Radul, and J.~M. Siskind.
\newblock Automatic differentiation in machine learning: a survey.
\newblock \emph{Journal of Machine Learning Research}, 18\penalty0
  (153):\penalty0 1--43, 2018.

\bibitem[Bickel et~al.(1998)Bickel, Klaassen, Ritov, and
  Wellner]{bickel1998efficient}
P.~Bickel, C.~Klaassen, Y.~Ritov, and J.~Wellner.
\newblock \emph{Efficient and Adaptive Estimation for Semiparametric Models}.
\newblock Springer New York, 1998.

\bibitem[Blondel and Roulet(2024)]{blondel2024elements}
M.~Blondel and V.~Roulet.
\newblock The elements of differentiable programming.
\newblock \emph{arXiv preprint arXiv:2403.14606}, 2024.

\bibitem[Bong and Lee(2023)]{bong2023local}
S.~Bong and K.~Lee.
\newblock Local causal effects with continuous exposures: A matching estimator
  for the average causal derivative effect.
\newblock \emph{arXiv preprint arXiv:2311.18532}, 2023.

\bibitem[Bonvini and Kennedy(2022)]{bonvini2022fast}
M.~Bonvini and E.~H. Kennedy.
\newblock Fast convergence rates for dose-response estimation.
\newblock \emph{arXiv preprint arXiv:2207.11825}, 2022.

\bibitem[Bonvini et~al.(2023)Bonvini, McClean, Branson, and
  Kennedy]{bonvini2023incremental}
M.~Bonvini, A.~McClean, Z.~Branson, and E.~H. Kennedy.
\newblock Incremental causal effects: an introduction and review.
\newblock In \emph{Handbook of matching and weighting adjustments for causal
  inference}, pages 349--372. Chapman and Hall/CRC, 2023.

\bibitem[Branson et~al.(2023)Branson, Kennedy, Balakrishnan, and
  Wasserman]{branson2023causal}
Z.~Branson, E.~H. Kennedy, S.~Balakrishnan, and L.~Wasserman.
\newblock Causal effect estimation after propensity score trimming with
  continuous treatments.
\newblock \emph{arXiv preprint arXiv:2309.00706}, 2023.

\bibitem[Cadre(2006)]{cadre2006kernel}
B.~Cadre.
\newblock Kernel estimation of density level sets.
\newblock \emph{Journal of Multivariate Analysis}, 97\penalty0 (4):\penalty0
  999--1023, 2006.

\bibitem[Calonico et~al.(2018)Calonico, Cattaneo, and
  Farrell]{calonico2018effect}
S.~Calonico, M.~D. Cattaneo, and M.~H. Farrell.
\newblock On the effect of bias estimation on coverage accuracy in
  nonparametric inference.
\newblock \emph{Journal of the American Statistical Association}, 113\penalty0
  (522):\penalty0 767--779, 2018.

\bibitem[Carone et~al.(2019)Carone, Luedtke, and van~der
  Laan]{carone2019toward}
M.~Carone, A.~R. Luedtke, and M.~J. van~der Laan.
\newblock Toward computerized efficient estimation in infinite-dimensional
  models.
\newblock \emph{Journal of the American Statistical Association}, 114\penalty0
  (527):\penalty0 1174--1190, 2019.

\bibitem[Cattaneo et~al.(2010)Cattaneo, Crump, and Jansson]{cattaneo2010robust}
M.~D. Cattaneo, R.~K. Crump, and M.~Jansson.
\newblock Robust data-driven inference for density-weighted average
  derivatives.
\newblock \emph{Journal of the American Statistical Association}, 105\penalty0
  (491):\penalty0 1070--1083, 2010.

\bibitem[Chen and Liao(2014)]{chen2014sieveM}
X.~Chen and Z.~Liao.
\newblock Sieve m inference on irregular parameters.
\newblock \emph{Journal of Econometrics}, 182\penalty0 (1):\penalty0 70--86,
  2014.

\bibitem[Chen et~al.(2014)Chen, Liao, and Sun]{chen2014sieve}
X.~Chen, Z.~Liao, and Y.~Sun.
\newblock Sieve inference on possibly misspecified semi-nonparametric time
  series models.
\newblock \emph{Journal of Econometrics}, 178:\penalty0 639--658, 2014.

\bibitem[Cheng and Chen(2019)]{cheng2019nonparametric}
G.~Cheng and Y.-C. Chen.
\newblock {Nonparametric inference via bootstrapping the debiased estimator}.
\newblock \emph{Electronic Journal of Statistics}, 13\penalty0 (1):\penalty0
  2194 -- 2256, 2019.

\bibitem[Chernozhukov et~al.(2018)Chernozhukov, Chetverikov, Demirer, Duflo,
  Hansen, Newey, and Robins]{chernozhukov2018double}
V.~Chernozhukov, D.~Chetverikov, M.~Demirer, E.~Duflo, C.~Hansen, W.~Newey, and
  J.~Robins.
\newblock {Double/debiased machine learning for treatment and structural
  parameters}.
\newblock \emph{The Econometrics Journal}, 21\penalty0 (1):\penalty0 C1--C68,
  01 2018.

\bibitem[Chernozhukov et~al.(2022)Chernozhukov, Newey, and
  Singh]{chernozhukov2022automatic}
V.~Chernozhukov, W.~K. Newey, and R.~Singh.
\newblock Automatic debiased machine learning of causal and structural effects.
\newblock \emph{Econometrica}, 90\penalty0 (3):\penalty0 967--1027, 2022.

\bibitem[Colangelo and Lee(2020)]{colangelo2020double}
K.~Colangelo and Y.-Y. Lee.
\newblock Double debiased machine learning nonparametric inference with
  continuous treatments.
\newblock \emph{arXiv preprint arXiv:2004.03036}, 2020.

\bibitem[Cole and Hern{\'a}n(2008)]{cole2008constructing}
S.~R. Cole and M.~A. Hern{\'a}n.
\newblock Constructing inverse probability weights for marginal structural
  models.
\newblock \emph{American Journal of Epidemiology}, 168\penalty0 (6):\penalty0
  656--664, 2008.

\bibitem[Cuevas(2009)]{cuevas2009set}
A.~Cuevas.
\newblock Set estimation: Another bridge between statistics and geometry.
\newblock \emph{Bolet{\'\i}n de Estad{\'\i}stica e Investigaci{\'o}n
  Operativa}, 25\penalty0 (2):\penalty0 71--85, 2009.

\bibitem[Cuevas and Fraiman(1997)]{cuevas1997plug}
A.~Cuevas and R.~Fraiman.
\newblock {A plug-in approach to support estimation}.
\newblock \emph{The Annals of Statistics}, 25\penalty0 (6):\penalty0 2300 --
  2312, 1997.

\bibitem[Devroye and Wise(1980)]{devroye1980detection}
L.~Devroye and G.~L. Wise.
\newblock Detection of abnormal behavior via nonparametric estimation of the
  support.
\newblock \emph{SIAM Journal on Applied Mathematics}, 38\penalty0 (3):\penalty0
  480--488, 1980.

\bibitem[D{\'\i}az and Hejazi(2020)]{diaz2020causal}
I.~D{\'\i}az and N.~S. Hejazi.
\newblock Causal mediation analysis for stochastic interventions.
\newblock \emph{Journal of the Royal Statistical Society Series B: Statistical
  Methodology}, 82\penalty0 (3):\penalty0 661--683, 2020.

\bibitem[D{\'\i}az and van~der Laan(2013)]{diaz2013targeted}
I.~D{\'\i}az and M.~J. van~der Laan.
\newblock Targeted data adaptive estimation of the causal dose--response curve.
\newblock \emph{Journal of Causal Inference}, 1\penalty0 (2):\penalty0
  171--192, 2013.

\bibitem[Einmahl and Mason(2005)]{einmahl2005uniform}
U.~Einmahl and D.~M. Mason.
\newblock {Uniform in bandwidth consistency of kernel-type function
  estimators}.
\newblock \emph{The Annals of Statistics}, 33\penalty0 (3):\penalty0 1380 --
  1403, 2005.

\bibitem[Fan and Gijbels(1996)]{fan1996local}
J.~Fan and I.~Gijbels.
\newblock \emph{Local polynomial modelling and its applications}, volume~66.
\newblock Chapman \& Hall/CRC, 1996.

\bibitem[Fan et~al.(2022)Fan, Hsu, Lieli, and Zhang]{fan2022estimation}
Q.~Fan, Y.-C. Hsu, R.~P. Lieli, and Y.~Zhang.
\newblock Estimation of conditional average treatment effects with
  high-dimensional data.
\newblock \emph{Journal of Business \& Economic Statistics}, 40\penalty0
  (1):\penalty0 313--327, 2022.

\bibitem[Farrell et~al.(2021)Farrell, Liang, and Misra]{farrell2021deep}
M.~H. Farrell, T.~Liang, and S.~Misra.
\newblock Deep neural networks for estimation and inference.
\newblock \emph{Econometrica}, 89\penalty0 (1):\penalty0 181--213, 2021.

\bibitem[Flores(2007)]{flores2007estimation}
C.~Flores.
\newblock Estimation of dose-response functions and optimal doses with a
  continuous treatment.
\newblock Technical report, Department of Economics, University of Miami, 2007.
\newblock URL \url{https://core.ac.uk/download/pdf/7169663.pdf}.

\bibitem[Flores and Flores-Lagunes(2009)]{flores2009identification}
C.~A. Flores and A.~Flores-Lagunes.
\newblock Identification and estimation of causal mechanisms and net effects of
  a treatment under unconfoundedness.
\newblock IZA Discussion Papers 4237, Institute of Labor Economics (IZA), 2009.

\bibitem[Flores et~al.(2012)Flores, Flores-Lagunes, Gonzalez, and
  Neumann]{flores2012estimating}
C.~A. Flores, A.~Flores-Lagunes, A.~Gonzalez, and T.~C. Neumann.
\newblock Estimating the effects of length of exposure to instruction in a
  training program: The case of job corps.
\newblock \emph{Review of Economics and Statistics}, 94\penalty0 (1):\penalty0
  153--171, 2012.

\bibitem[Galvao and Wang(2015)]{galvao2015uniformly}
A.~F. Galvao and L.~Wang.
\newblock Uniformly semiparametric efficient estimation of treatment effects
  with a continuous treatment.
\newblock \emph{Journal of the American Statistical Association}, 110\penalty0
  (512):\penalty0 1528--1542, 2015.

\bibitem[Gasser and M{\"u}ller(1984)]{gasser1984estimating}
T.~Gasser and H.-G. M{\"u}ller.
\newblock Estimating regression functions and their derivatives by the kernel
  method.
\newblock \emph{Scandinavian Journal of Statistics}, pages 171--185, 1984.

\bibitem[Gill and Robins(2001)]{gill2001causal}
R.~D. Gill and J.~M. Robins.
\newblock Causal inference for complex longitudinal data: the continuous case.
\newblock \emph{Annals of Statistics}, 29\penalty0 (6):\penalty0 1785--1811,
  2001.

\bibitem[Godambe and Joshi(1965)]{godambe1965admissibility}
V.~Godambe and V.~Joshi.
\newblock Admissibility and bayes estimation in sampling finite populations. i.
\newblock \emph{The Annals of Mathematical Statistics}, 36\penalty0
  (6):\penalty0 1707--1722, 1965.

\bibitem[Guo et~al.(2019)Guo, Yuan, and Zhang]{guo2019decorrelated}
Z.~Guo, W.~Yuan, and C.-H. Zhang.
\newblock Decorrelated local linear estimator: Inference for non-linear effects
  in high-dimensional additive models.
\newblock \emph{arXiv preprint arXiv:1907.12732}, 2019.

\bibitem[H{\"a}rdle and Stoker(1989)]{hardle1989investigating}
W.~H{\"a}rdle and T.~M. Stoker.
\newblock Investigating smooth multiple regression by the method of average
  derivatives.
\newblock \emph{Journal of the American Statistical Association}, 84\penalty0
  (408):\penalty0 986--995, 1989.

\bibitem[Hart and Vieu(1990)]{hart1990data}
J.~D. Hart and P.~Vieu.
\newblock Data-driven bandwidth choice for density estimation based on
  dependent data.
\newblock \emph{The Annals of Statistics}, pages 873--890, 1990.

\bibitem[Hines et~al.(2023)Hines, Diaz-Ordaz, and
  Vansteelandt]{hines2023optimally}
O.~Hines, K.~Diaz-Ordaz, and S.~Vansteelandt.
\newblock Optimally weighted average derivative effects.
\newblock \emph{arXiv preprint arXiv:2308.05456}, 2023.

\bibitem[Hirano and Imbens(2004)]{hirano2004propensity}
K.~Hirano and G.~W. Imbens.
\newblock \emph{The Propensity Score with Continuous Treatments}, chapter~7,
  pages 73--84.
\newblock John Wiley \& Sons, Ltd, 2004.

\bibitem[Hirshberg and Wager(2020)]{hirshberg2020debiased}
D.~A. Hirshberg and S.~Wager.
\newblock Debiased inference of average partial effects in single-index models:
  Comment on wooldridge and zhu.
\newblock \emph{Journal of Business \& Economic Statistics}, 38\penalty0
  (1):\penalty0 19--24, 2020.

\bibitem[Huber(2014)]{huber2014identifying}
M.~Huber.
\newblock Identifying causal mechanisms (primarily) based on inverse
  probability weighting.
\newblock \emph{Journal of Applied Econometrics}, 29\penalty0 (6):\penalty0
  920--943, 2014.

\bibitem[Huber et~al.(2020)Huber, Hsu, Lee, and Lettry]{huber2020direct}
M.~Huber, Y.-C. Hsu, Y.-Y. Lee, and L.~Lettry.
\newblock Direct and indirect effects of continuous treatments based on
  generalized propensity score weighting.
\newblock \emph{Journal of Applied Econometrics}, 35\penalty0 (7):\penalty0
  814--840, 2020.

\bibitem[Ichimura and Newey(2022)]{ichimura2022influence}
H.~Ichimura and W.~K. Newey.
\newblock The influence function of semiparametric estimators.
\newblock \emph{Quantitative Economics}, 13\penalty0 (1):\penalty0 29--61,
  2022.

\bibitem[Imai and van Dyk(2004)]{imai2004causal}
K.~Imai and D.~A. van Dyk.
\newblock Causal inference with general treatment regimes: Generalizing the
  propensity score.
\newblock \emph{Journal of the American Statistical Association}, 99\penalty0
  (467):\penalty0 854--866, 2004.

\bibitem[Kallus and Zhou(2018)]{kallus2018policy}
N.~Kallus and A.~Zhou.
\newblock Policy evaluation and optimization with continuous treatments.
\newblock In \emph{International Conference on Artificial Intelligence and
  Statistics}, pages 1243--1251. PMLR, 2018.

\bibitem[Kennedy(2024)]{kennedy2024semiparametric}
E.~H. Kennedy.
\newblock Semiparametric doubly robust targeted double machine learning: a
  review.
\newblock \emph{Handbook of Statistical Methods for Precision Medicine}, pages
  207--236, 2024.

\bibitem[Kennedy et~al.(2017)Kennedy, Ma, McHugh, and Small]{kennedy2017non}
E.~H. Kennedy, Z.~Ma, M.~D. McHugh, and D.~S. Small.
\newblock Nonparametric methods for doubly robust estimation of continuous
  treatment effects.
\newblock \emph{Journal of the Royal Statistical Society Series B: Statistical
  Methodology}, 79\penalty0 (4):\penalty0 1229--1245, 2017.

\bibitem[Klosin(2021)]{klosin2021automatic}
S.~Klosin.
\newblock Automatic double machine learning for continuous treatment effects.
\newblock \emph{arXiv preprint arXiv:2104.10334}, 2021.

\bibitem[Lee(2009)]{lee2009training}
D.~S. Lee.
\newblock Training, wages, and sample selection: Estimating sharp bounds on
  treatment effects.
\newblock \emph{The Review of Economic Studies}, 76\penalty0 (3):\penalty0
  1071--1102, 2009.

\bibitem[Lee(2018)]{lee2018partial}
Y.-Y. Lee.
\newblock Partial mean processes with generated regressors: Continuous
  treatment effects and nonseparable models.
\newblock \emph{arXiv preprint arXiv:1811.00157}, 2018.

\bibitem[Lee and Liu(2024)]{lee2024lee}
Y.-Y. Lee and C.-A. Liu.
\newblock Lee bounds with a continuous treatment in sample selection.
\newblock \emph{arXiv preprint arXiv:2411.04312}, 2024.

\bibitem[Lehmann(1999)]{lehmann1999elements}
E.~L. Lehmann.
\newblock \emph{Elements of large-sample theory}.
\newblock Springer, 1999.

\bibitem[Li and Racine(2004)]{li2004cross}
Q.~Li and J.~Racine.
\newblock Cross-validated local linear nonparametric regression.
\newblock \emph{Statistica Sinica}, 14:\penalty0 485--512, 2004.

\bibitem[Luedtke(2024)]{luedtke2024simplifying}
A.~Luedtke.
\newblock Simplifying debiased inference via automatic differentiation and
  probabilistic programming.
\newblock \emph{arXiv preprint arXiv:2405.08675}, 2024.

\bibitem[Luedtke and Chung(2024)]{luedtke2024one}
A.~Luedtke and I.~Chung.
\newblock One-step estimation of differentiable hilbert-valued parameters.
\newblock \emph{The Annals of Statistics}, 52\penalty0 (4):\penalty0
  1534--1563, 2024.

\bibitem[Mack and M{\"u}ller(1989)]{mack1989derivative}
Y.~Mack and H.-G. M{\"u}ller.
\newblock Derivative estimation in nonparametric regression with random
  predictor variable.
\newblock \emph{Sankhy{\=a}: The Indian Journal of Statistics, Series A}, pages
  59--72, 1989.

\bibitem[McClean et~al.(2024)McClean, Li, Bae, McAdams-DeMarco, D{\'\i}az, and
  Wu]{mcclean2024fair}
A.~McClean, Y.~Li, S.~Bae, M.~A. McAdams-DeMarco, I.~D{\'\i}az, and W.~Wu.
\newblock Fair comparisons of causal parameters with many treatments and
  positivity violations.
\newblock \emph{arXiv preprint arXiv:2410.13522}, 2024.

\bibitem[Meier et~al.(2009)Meier, van~de Geer, and B{\"u}hlmann]{meier2009high}
L.~Meier, S.~van~de Geer, and P.~B{\"u}hlmann.
\newblock {High-dimensional additive modeling}.
\newblock \emph{The Annals of Statistics}, 37\penalty0 (6B):\penalty0 3779 --
  3821, 2009.

\bibitem[Meloche(1990)]{meloche1990asymptotic}
J.~Meloche.
\newblock Asymptotic behaviour of the mean integrated squared error of kernel
  density estimators for dependent observations.
\newblock \emph{The Canadian Journal of Statistics/La Revue Canadienne de
  Statistique}, pages 205--211, 1990.

\bibitem[Neugebauer and van~der Laan(2007)]{neugebauer2007nonparametric}
R.~Neugebauer and M.~van~der Laan.
\newblock Nonparametric causal effects based on marginal structural models.
\newblock \emph{Journal of Statistical Planning and Inference}, 137\penalty0
  (2):\penalty0 419--434, 2007.

\bibitem[Newey and Robins(2018)]{newey2018cross}
W.~K. Newey and J.~R. Robins.
\newblock Cross-fitting and fast remainder rates for semiparametric estimation.
\newblock \emph{arXiv preprint arXiv:1801.09138}, 2018.

\bibitem[Newey and Stoker(1993)]{newey1993efficiency}
W.~K. Newey and T.~M. Stoker.
\newblock Efficiency of weighted average derivative estimators and index
  models.
\newblock \emph{Econometrica}, 61\penalty0 (5):\penalty0 1199--1223, 1993.

\bibitem[Neyman(1959)]{neyman1959optimal}
J.~Neyman.
\newblock Optimal asymptotic tests of composite hypotheses.
\newblock \emph{Probability and Statsitics}, pages 213--234, 1959.

\bibitem[Neyman(1979)]{neyman1979c}
J.~Neyman.
\newblock C($\alpha$) tests and their use.
\newblock \emph{Sankhy{\=a}: The Indian Journal of Statistics, Series A},
  41\penalty0 (1/2):\penalty0 1--21, 1979.

\bibitem[Paciorek(2010)]{paciorek2010importance}
C.~J. Paciorek.
\newblock The importance of scale for spatial-confounding bias and precision of
  spatial regression estimators.
\newblock \emph{Statistical Science}, 25\penalty0 (1):\penalty0 107--125, 2010.

\bibitem[Paszke et~al.(2017)Paszke, Gross, Chintala, Chanan, Yang, DeVito, Lin,
  Desmaison, Antiga, and Lerer]{paszke2017automatic}
A.~Paszke, S.~Gross, S.~Chintala, G.~Chanan, E.~Yang, Z.~DeVito, Z.~Lin,
  A.~Desmaison, L.~Antiga, and A.~Lerer.
\newblock Automatic differentiation in pytorch.
\newblock In \emph{NIPS 2017 Workshop on Autodiff}, 2017.

\bibitem[Paszke et~al.(2019)Paszke, Gross, Massa, Lerer, Bradbury, Chanan,
  Killeen, Lin, Gimelshein, Antiga, Desmaison, Kopf, Yang, DeVito, Raison,
  Tejani, Chilamkurthy, Steiner, Fang, Bai, and Chintala]{paszke2019pytorch}
A.~Paszke, S.~Gross, F.~Massa, A.~Lerer, J.~Bradbury, G.~Chanan, T.~Killeen,
  Z.~Lin, N.~Gimelshein, L.~Antiga, A.~Desmaison, A.~Kopf, E.~Yang, Z.~DeVito,
  M.~Raison, A.~Tejani, S.~Chilamkurthy, B.~Steiner, L.~Fang, J.~Bai, and
  S.~Chintala.
\newblock Pytorch: An imperative style, high-performance deep learning library.
\newblock In H.~Wallach, H.~Larochelle, A.~Beygelzimer, F.~d\textquotesingle
  Alch\'{e}-Buc, E.~Fox, and R.~Garnett, editors, \emph{Advances in Neural
  Information Processing Systems}, volume~32. Curran Associates, Inc., 2019.

\bibitem[Powell et~al.(1989)Powell, Stock, and
  Stoker]{powell1989semiparametric}
J.~L. Powell, J.~H. Stock, and T.~M. Stoker.
\newblock Semiparametric estimation of index coefficients.
\newblock \emph{Econometrica}, 57\penalty0 (6):\penalty0 1403--1430, 1989.

\bibitem[Ratkovic and Tingley(2023)]{ratkovic2023estimation}
M.~Ratkovic and D.~Tingley.
\newblock Estimation and inference on nonlinear and heterogeneous effects.
\newblock \emph{The Journal of Politics}, 85\penalty0 (2):\penalty0 421--435,
  2023.

\bibitem[Robins(1986)]{robins1986new}
J.~Robins.
\newblock A new approach to causal inference in mortality studies with a
  sustained exposure period—application to control of the healthy worker
  survivor effect.
\newblock \emph{Mathematical Modelling}, 7\penalty0 (9-12):\penalty0
  1393--1512, 1986.

\bibitem[Robins et~al.(2000)Robins, Hernan, and Brumback]{robins2000marginal}
J.~M. Robins, M.~A. Hernan, and B.~Brumback.
\newblock Marginal structural models and causal inference in epidemiology.
\newblock \emph{Epidemiology}, 11\penalty0 (5):\penalty0 550--560, 2000.

\bibitem[Rubin(1974)]{rubin1974estimating}
D.~B. Rubin.
\newblock Estimating causal effects of treatments in randomized and
  nonrandomized studies.
\newblock \emph{Journal of Educational Psychology}, 66\penalty0 (5):\penalty0
  688--701, 1974.

\bibitem[Schick(1986)]{schick1986on}
A.~Schick.
\newblock {On Asymptotically Efficient Estimation in Semiparametric Models}.
\newblock \emph{The Annals of Statistics}, 14\penalty0 (3):\penalty0 1139 --
  1151, 1986.

\bibitem[Schindl et~al.(2024)Schindl, Shen, and
  Kennedy]{schindl2024incremental}
K.~Schindl, S.~Shen, and E.~H. Kennedy.
\newblock Incremental effects for continuous exposures.
\newblock \emph{arXiv preprint arXiv:2409.11967}, 2024.

\bibitem[Schnell and Papadogeorgou(2020)]{schnell2020}
P.~Schnell and G.~Papadogeorgou.
\newblock Mitigating unobserved spatial confounding when estimating the effect
  of supermarket access on cardiovascular disease deaths.
\newblock \emph{Annals of Applied Statistics}, 14:\penalty0 2069--2095, 12
  2020.

\bibitem[Schochet et~al.(2001)Schochet, Burghardt, and
  Glazerman]{schochet2001national}
P.~Z. Schochet, J.~Burghardt, and S.~Glazerman.
\newblock National job corps study: The impacts of job corps on participants'
  employment and related outcomes.
\newblock Mathematica policy research reports, Mathematica Policy Research,
  2001.

\bibitem[Schochet et~al.(2008)Schochet, Burghardt, and
  McConnell]{schochet2008does}
P.~Z. Schochet, J.~Burghardt, and S.~McConnell.
\newblock Does job corps work? impact findings from the national job corps
  study.
\newblock \emph{American Economic Review}, 98\penalty0 (5):\penalty0
  1864--1886, 2008.

\bibitem[Shao(2003)]{shao2003mathematical}
J.~Shao.
\newblock \emph{Mathematical Statistics}.
\newblock Springer Science \& Business Media, 2003.

\bibitem[Stolzenberg(1980)]{stolzenberg1980measurement}
R.~M. Stolzenberg.
\newblock The measurement and decomposition of causal effects in nonlinear and
  nonadditive models.
\newblock \emph{Sociological Methodology}, 11:\penalty0 459--488, 1980.

\bibitem[Stone(1985)]{stone1985additive}
C.~J. Stone.
\newblock Additive regression and other nonparametric models.
\newblock \emph{The Annals of Statistics}, 13\penalty0 (2):\penalty0 689--705,
  1985.

\bibitem[Su et~al.(2019)Su, Ura, and Zhang]{su2019non}
L.~Su, T.~Ura, and Y.~Zhang.
\newblock Non-separable models with high-dimensional data.
\newblock \emph{Journal of Econometrics}, 212\penalty0 (2):\penalty0 646--677,
  2019.

\bibitem[Swaminathan and Joachims(2015)]{swaminathan2015self}
A.~Swaminathan and T.~Joachims.
\newblock The self-normalized estimator for counterfactual learning.
\newblock In C.~Cortes, N.~Lawrence, D.~Lee, M.~Sugiyama, and R.~Garnett,
  editors, \emph{Advances in Neural Information Processing Systems}, volume~28,
  2015.

\bibitem[Takatsu and Westling(2025)]{takatsu2022debiased}
K.~Takatsu and T.~Westling.
\newblock Debiased inference for a covariate-adjusted regression function.
\newblock \emph{Journal of the Royal Statistical Society Series B: Statistical
  Methodology}, 87\penalty0 (1):\penalty0 33--55, 2025.

\bibitem[Trotter and Tukey(1956)]{tukey1956conditional}
H.~F. Trotter and J.~W. Tukey.
\newblock Conditional monte carlo for normal samples.
\newblock In \emph{Symposium on Monte Carlo Methods}, pages 64--79. John Wiley
  and Sons, 1956.

\bibitem[Tsybakov(1997)]{tsybakov1997nonparametric}
A.~B. Tsybakov.
\newblock On nonparametric estimation of density level sets.
\newblock \emph{The Annals of Statistics}, 25\penalty0 (3):\penalty0 948--969,
  1997.

\bibitem[van~der Laan and Robins(2003)]{van2003unified}
M.~J. van~der Laan and J.~M. Robins.
\newblock \emph{Unified methods for censored longitudinal data and causality}.
\newblock Springer, 2003.

\bibitem[van~der Laan et~al.(2018)van~der Laan, Bibaut, and Luedtke]{van2018cv}
M.~J. van~der Laan, A.~Bibaut, and A.~R. Luedtke.
\newblock Cv-tmle for nonpathwise differentiable target parameters.
\newblock In M.~J. van~der Laan and S.~Rose, editors, \emph{Targeted Learning
  in Data Science: Causal Inference for Complex Longitudinal Studies}, pages
  455--481. Springer, 2018.

\bibitem[van~der Vaart(1991)]{van1991differentiable}
A.~van~der Vaart.
\newblock On differentiable functionals.
\newblock \emph{The Annals of Statistics}, 19\penalty0 (1):\penalty0 178--204,
  1991.

\bibitem[van~der Vaart(1998)]{VDV1998}
A.~W. van~der Vaart.
\newblock \emph{Asymptotic Statistics}.
\newblock Cambridge Series in Statistical and Probabilistic Mathematics.
  Cambridge University Press, 1998.

\bibitem[Wand and Jones(1994)]{wand1994kernel}
M.~P. Wand and M.~C. Jones.
\newblock \emph{Kernel Smoothing}.
\newblock CRC press, 1994.

\bibitem[Wasserman(2006)]{wasserman2006all}
L.~Wasserman.
\newblock \emph{All of nonparametric statistics}.
\newblock Springer Science \& Business Media, 2006.

\bibitem[Westreich and Cole(2010)]{westreich2010invited}
D.~Westreich and S.~R. Cole.
\newblock Invited commentary: positivity in practice.
\newblock \emph{American Journal of Epidemiology}, 171\penalty0 (6):\penalty0
  674--677, 2010.

\bibitem[Wu et~al.(2024)Wu, Mealli, Kioumourtzoglou, Dominici, and
  Braun]{wu2024matching}
X.~Wu, F.~Mealli, M.-A. Kioumourtzoglou, F.~Dominici, and D.~Braun.
\newblock Matching on generalized propensity scores with continuous exposures.
\newblock \emph{Journal of the American Statistical Association}, 119\penalty0
  (545):\penalty0 757--772, 2024.

\bibitem[Xu et~al.(2021)Xu, Sani, Ghassami, and Shpitser]{xu2021multiply}
Y.~Xu, N.~Sani, A.~Ghassami, and I.~Shpitser.
\newblock Multiply robust causal mediation analysis with continuous treatments.
\newblock \emph{arXiv preprint arXiv:2105.09254}, 2021.

\bibitem[Zeng et~al.(2025)Zeng, Levis, Lee, Kennedy, and
  Keele]{zeng2025nonparametric}
Z.~Zeng, A.~W. Levis, J.~Lee, E.~H. Kennedy, and L.~Keele.
\newblock Nonparametric estimation of local treatment effects with continuous
  instruments.
\newblock \emph{arXiv preprint arXiv:2504.03063}, 2025.

\bibitem[Zhang et~al.(2024)Zhang, Chen, and Giessing]{zhang2024nonparametric}
Y.~Zhang, Y.-C. Chen, and A.~Giessing.
\newblock Nonparametric inference on dose-response curves without the
  positivity condition.
\newblock \emph{arXiv preprint arXiv:2405.09003}, 2024.

\bibitem[Zhou and Wolfe(2000)]{zhou2000derivative}
S.~Zhou and D.~A. Wolfe.
\newblock On derivative estimation in spline regression.
\newblock \emph{Statistica Sinica}, 10\penalty0 (1):\penalty0 93--108, 2000.

\end{thebibliography}

\newpage 
\onehalfspacing

\begin{center}
{\LARGE \textbf{Supplementary Materials to ``Doubly Robust Inference on Causal Derivative Effects for Continuous Treatments''}}\\~\\

\end{center}

\appendix

\setcounter{page}{1}

\noindent{\bf\LARGE Contents}

\startcontents[sections]
\printcontents[sections]{l}{1}{\setcounter{tocdepth}{2}}

\vspace{5mm}

\section{Practical Considerations}
\label{app:practical}

In this section, we outline some practical aspects involved in implementing our proposed estimators of dose-response and derivative effect curves.

\subsection{Self-Normalized IPW and DR Estimators}
\label{app:self_norm}

The classical IPW estimators of $m(t)$ and $\theta(t)$ suffer from the variance blowup when some estimated conditional densities $\hat{p}_{T|\bm{S}}(T_i|\bm{S}_i), i=1,...,n$ are close to 0. While truncating these estimates in a threshold value can reduce the instability of IPW estimators \citep{branson2023causal,wu2024matching}, determining the appropriate threshold value in practice is not straightforward. Alternatively, one can reduce the variances of IPW estimators and maintain its consistency by implementing the self-normalized version of IPW estimators
\citep{swaminathan2015self,kallus2018policy}. This idea was originally from the importance sampling literature \citep{tukey1956conditional} and also known as H\'ajek estimator \citep{godambe1965admissibility}. 

$\bullet$ {\bf Estimators Under Positivity:} For the IPW estimator \eqref{m_IPW} of $m(t)$, its self-normalized version takes the form
\begin{equation}
	\label{m_IPW_selfnorm}
	\hat{m}_{\mathrm{IPW}}^{\mathrm{norm}}(t) = \frac{\hat{m}_{\mathrm{IPW}}(t)}{\frac{1}{nh}\sum_{j=1}^n \frac{K\left(\frac{T_j-t}{h}\right)}{\hat{p}_{T|\bm{S}}(T_j|\bm{S}_j)}} = \frac{\sum_{i=1}^n \frac{Y_i\cdot K\left(\frac{T_i-t}{h}\right)}{\hat{p}_{T|\bm{S}}(T_i|\bm{S}_i)}}{\sum_{j=1}^n \frac{K\left(\frac{T_j-t}{h}\right)}{\hat{p}_{T|\bm{S}}(T_j|\bm{S}_j)}}.
\end{equation}
The self-normalized IPW estimator \eqref{m_IPW_selfnorm} maintains the consistency of the original IPW estimator \eqref{m_IPW}, because $\hat{p}_{T|\bm{S}}$ is a consistent estimator of $p_{T|\bm{S}}$ and the (oracle) denominator of \eqref{m_IPW_selfnorm} has its expectation as:
\begin{align*}
\mathbb{E}\left[\frac{1}{nh}\sum_{j=1}^n \frac{K\left(\frac{T_j-t}{h}\right)}{p_{T|\bm{S}}(T_j|\bm{S}_j)}\right] &= \frac{1}{h} \int_{\mathcal{T}\times \mathcal{S}} \frac{K\left(\frac{t_1-t}{h}\right)}{p_{T|\bm{S}}(t_1|\bm{s}_1)} \cdot p(t_1,\bm{s}_1)\, d\bm{s}_1 dt_1 \\
&= \int_{\mathbb{R}\times \mathcal{S}} K(u) \cdot p_{\bm{S}}(\bm{s}_1)\, d\bm{s}_1 du =1
\end{align*}
under Assumption~\ref{assump:reg_kernel}.

Similarly, for the IPW estimator \eqref{theta_IPW} of $\theta(t)$, its self-normalized version can be written as:
\begin{equation}
	\label{theta_IPW_selfnorm}
	\hat{\theta}_{\mathrm{IPW}}^{\mathrm{norm}}(t) = \frac{\hat{\theta}_{\mathrm{IPW}}(t)}{\frac{1}{nh}\sum_{j=1}^n \frac{K\left(\frac{T_j-t}{h}\right)}{\hat{p}_{T|\bm{S}}(T_j|\bm{S}_j)}} = \frac{\sum_{i=1}^n \frac{Y_i \left(\frac{T_i-t}{h}\right) K\left(\frac{T_i-t}{h}\right)}{\hat{p}_{T|\bm{S}}(T_i|\bm{S}_i)}}{\kappa_2h\sum_{j=1}^n \frac{K\left(\frac{T_j-t}{h}\right)}{\hat{p}_{T|\bm{S}}(T_j|\bm{S}_j)}}.
\end{equation}
This self-normalized technique can also be applied to the IPW component of the DR estimators \eqref{m_DR} and \eqref{theta_DR} to stabilize their variances, leading to self-normalized DR estimators of $m(t)$ and $\theta(t)$ as:
\begin{equation}
\label{m_DR_selfnorm}
\hat{m}_{\mathrm{DR}}^{\mathrm{norm}}(t)= \frac{\sum_{i=1}^n \frac{\left[Y_i - \hat \mu(t,\bm{S}_i)\right] K\left(\frac{T_i-t}{h}\right)}{\hat{p}_{T|\bm{S}}(T_i|\bm{S}_i)}}{\sum_{j=1}^n \frac{K\left(\frac{T_j-t}{h}\right)}{\hat{p}_{T|\bm{S}}(T_j|\bm{S}_j)}} + \frac{1}{n}\sum_{i=1}^n \hat{\mu}(t,\bm{S}_i)
\end{equation}
and
\begin{equation}
	\label{theta_DR_selfnorm}
\hat{\theta}_{\mathrm{DR}}^{\mathrm{norm}}(t) = \frac{\sum_{i=1}^n \frac{\left[Y_i - \hat{\mu}(t,\bm{S}_i) - (T_i-t)\cdot \hat{\beta}(t,\bm{S}_i) \right]\left(\frac{T_i-t}{h}\right) K\left(\frac{T_i-t}{h}\right)}{\hat{p}_{T|\bm{S}}(T_i|\bm{S}_i)}}{\kappa_2h \sum_{j=1}^n \frac{K\left(\frac{T_j-t}{h}\right)}{\hat{p}_{T|\bm{S}}(T_j|\bm{S}_j)}} + \frac{1}{n}\sum_{i=1}^n \hat{\beta}(t,\bm{S}_i),
\end{equation}
respectively. Compared to \eqref{theta_var_est}, the estimated asymptotic variance of the self-normalized DR estimator \eqref{theta_DR_selfnorm} thus becomes
\begin{equation}
	\label{theta_var_est_selfnorm}
	\hat{V}_{\theta}(t) = \frac{1}{n} \sum_{i=1}^n \left\{\frac{nh\cdot \phi_{h,t}\left(Y_i,T_i,\bm{S}_i;\hat{\mu}, \hat{\beta}, \hat{p}_{T|\bm{S}}\right)}{\sum_{j=1}^n \frac{K\left(\frac{T_j-t}{h}\right)}{\hat{p}_{T|\bm{S}}(T_j|\bm{S}_j)}} + \sqrt{h^3}\left[\hat{\beta}(t,\bm{S}_i) - \hat{\theta}_{\mathrm{DR}}^{\mathrm{norm}}(t) \right]\right\}^2.
\end{equation}

$\bullet$ {\bf Estimators Without Positivity:} For the bias-corrected IPW estimator \eqref{theta_IPW_bnd_corrected} of $\theta(t)$, we also adopt the self-normalized technique to deduce that
\begin{equation}
	\label{theta_IPW_bnd_corrected_selfnorm}
	\hat{\theta}_{\mathrm{C,IPW}}^{\mathrm{norm}}(t) = \frac{\hat{\theta}_{\mathrm{C,IPW}}(t)}{\frac{1}{nh}\sum_{j=1}^n \frac{K\left(\frac{T_j-t}{h}\right)\cdot \hat{p}_{\zeta}(\bm{S}_j|t)}{\hat{p}(T_j,\bm{S}_j)}} 
	= \frac{\sum_{i=1}^n \frac{Y_i \left(\frac{T_i-t}{h}\right) K\left(\frac{T_i-t}{h}\right)\cdot \hat{p}_{\zeta}(\bm{S}_i|t)}{\hat{p}(T_i,\bm{S}_i)}}{\kappa_2h\sum_{j=1}^n \frac{K\left(\frac{T_j-t}{h}\right)\cdot \hat{p}_{\zeta}(\bm{S}_j|t)}{\hat{p}(T_j,\bm{S}_j)}}.
\end{equation}
The self-normalized IPW estimator \eqref{theta_IPW_bnd_corrected_selfnorm} again maintains the consistency of the original IPW estimator \eqref{theta_IPW_bnd_corrected}, because $\hat{p},\hat{p}_{\zeta}$ are consistent estimators of $p,\bar{p}_{\zeta}$ respectively and the (oracle) denominator of \eqref{theta_IPW_bnd_corrected_selfnorm} has its expectation as:
\begin{align*}
	\mathbb{E}\left[\frac{1}{nh}\sum_{j=1}^n \frac{K\left(\frac{T_j-t}{h}\right)\cdot \bar{p}_{\zeta}(\bm{S}_j|t)}{p(T_j,\bm{S}_j)}\right] &= \frac{1}{h} \int_{\mathcal{T}\times \mathcal{S}} K\left(\frac{t_1-t}{h}\right) \cdot \bar{p}_{\zeta}(\bm{s}_1|t)\, d\bm{s}_1 dt_1 \\
	&= \int_{\mathbb{R}\times \mathcal{S}} K(u) \cdot \bar{p}_{\zeta}(\bm{s}_1|t)\, d\bm{s}_1 du =1
\end{align*}
under Assumption~\ref{assump:reg_kernel}.

Analogously, the self-normalized bias-corrected DR estimator of $\theta(t)$ is given by
\begin{equation}
	\label{theta_DR_bnd_corrected_selfnorm}
	\hat{\theta}_{\mathrm{C,DR}}^{\mathrm{norm}}(t) = \frac{\sum_{i=1}^n \frac{\left[Y_i - \hat{\mu}(t,\bm{S}_i) - (T_i-t)\cdot \hat{\beta}(t,\bm{S}_i)\right] \left(\frac{T_i-t}{h}\right) K\left(\frac{T_i-t}{h}\right)\cdot \hat{p}_{\zeta}(\bm{S}_i|t)}{\hat{p}(T_i,\bm{S}_i)}}{\kappa_2h\sum_{j=1}^n \frac{K\left(\frac{T_j-t}{h}\right)\cdot \hat{p}_{\zeta}(\bm{S}_j|t)}{\hat{p}(T_j,\bm{S}_j)}} + \int \hat{\beta}(t,\bm{s})\cdot \hat{p}_{\zeta}(\bm{s}|t)\, d\bm{s},
\end{equation}
whose estimated asymptotic variance becomes
$$\hat{V}_{C,\theta}(t) = \frac{1}{n} \sum_{i=1}^n \left\{\frac{nh\cdot \phi_{C,h,t}\left(Y_i,T_i,\bm{S}_i;\hat{\mu}, \hat{\beta}, \hat{p}, \hat{p}_{\zeta}\right)}{\sum_{j=1}^n \frac{K\left(\frac{T_j-t}{h}\right)\cdot \hat{p}_{\zeta}(\bm{S}_j|t)}{\hat{p}(T_j,\bm{S}_j)}} + \sqrt{h^3}\left[\int \hat{\beta}(t,\bm{s}) \cdot \hat{p}_{\zeta}(\bm{s}|t)\, d\bm{s} - \hat{\theta}_{\mathrm{C,DR}}(t) \right]\right\}^2.$$

\subsection{Implementation of Proposed Estimators in \autoref{sec:theta_pos} with Cross-Fitting}
\label{app:imple}

We explain the implementation details for our proposed estimators of $\theta(t)$ in \autoref{sec:theta_pos} with cross-fitting \citep{schick1986on,newey2018cross,chernozhukov2018double} as follows. The same procedures can be applied to the estimators of $m(t)$ in \autoref{subsec:m_pos} as well.

{\bf 1. Partitioning the Data:} The observed data $\{(Y_i,T_i,\bm{S}_i)\}_{i=1}^n$ are partitioned into $L$ distinct subsets of approximately equal size. Commonly, the 5-fold ($L=5$) or 10-fold ($L=10$) cross-fitting is applied in practice, and no cross-fitting is used when $L=1$ by convention. Let $I_{\ell}, \ell=1,...,L$ be the index sets of such a partition so that $\cup_{\ell=1}^L I_{\ell}=\{1,...,n\}$. 

{\bf 2. Estimating the Nuisance Functions:} For each index set $I_{\ell}$, we estimate the nuisance functions $\mu,\beta,p_{T|\bm{S}}$ using the observations that are \emph{not} in $I_{\ell}$; see \autoref{app:nuisance_est} for details. The estimated nuisance functions are denoted by $\hat{\mu}^{(\ell)},\hat{\beta}^{(\ell)},\hat{p}_{T|\bm{S}}^{(\ell)}$, respectively, for $\ell=1,...,L$. Recall that $\mu(t,\bm{s})$ is the conditional mean outcome function, $\beta(t,\bm{s})$ is the partial derivative of $\mu(t,\bm{s})$ with respect to $t$, and $p_{T|\bm{S}}(t|\bm{s})$ is the conditional density function of $T$ given $\bm{S}=\bm{s}$.

{\bf 3. Constructing the Final Estimators:} The RA, IPW, and DR estimators of $\theta(t)$ under cross-fitting are given by
\begin{align}
\label{theta_est_crossfit}
\begin{split}
\hat{\theta}_{\mathrm{RA}}(t) &= \frac{1}{n}\sum_{\ell=1}^L \sum_{i\in I_{\ell}} \hat{\beta}^{(\ell)}(t,\bm{S}_i), \\
\hat{\theta}_{\mathrm{IPW}}(t) &= \frac{1}{nh^2} \sum_{\ell=1}^L \sum_{i\in I_{\ell}} \frac{Y_i\left(\frac{T_i-t}{h}\right) K\left(\frac{T_i-t}{h}\right)}{\kappa_2 \cdot \hat{p}_{T|\bm{S}}^{(\ell)}(T_i|\bm{S}_i)},\\
\hat{\theta}_{\mathrm{DR}}(t) &= \frac{1}{nh}\sum_{\ell=1}^L \sum_{i\in I_{\ell}} \left\{ \frac{\left(\frac{T_i-t}{h}\right)K\left(\frac{T_i-t}{h}\right) }{h\cdot \kappa_2\cdot \hat{p}_{T|\bm{S}}^{(\ell)}(T_i|\bm{S}_i)} \left[Y_i - \hat{\mu}^{(\ell)}(t,\bm{S}_i) - (T_i-t)\cdot \hat{\beta}^{(\ell)}(t,\bm{S}_i)
\right]+ h\cdot \hat{\beta}^{(\ell)}(t,\bm{S}_i) \right\},
\end{split}
\end{align}
where the bandwidth parameter $h>0$ is chosen beforehand. The estimated asymptotic variance of $\hat{\theta}_{\mathrm{DR}}(t)$ under cross-fitting is given by
\begin{equation}
	\label{theta_var_est_crossfit}
	\hat{V}_{\theta}(t) = \frac{1}{n} \sum_{\ell=1}^L \sum_{i\in I_{\ell}} \left\{\phi_{h,t}\left(Y_i,T_i,\bm{S}_i;\hat{\mu}^{(\ell)}, \hat{\beta}^{(\ell)}, \hat{p}_{T|\bm{S}}^{(\ell)}\right) + \sqrt{h^3}\left[\hat{\beta}^{(\ell)}(t,\bm{S}_i) - \hat{\theta}_{\mathrm{DR}}(t) \right]\right\}^2.
\end{equation}
The self-normalized technique in \autoref{app:self_norm} can be applied to these cross-fitted estimators accordingly.

\subsection{Nuisance Function Estimation}
\label{app:nuisance_est}

The implementation of the DR estimator \eqref{theta_DR} of $\theta(t)$ in \autoref{sec:theta_pos} requires the estimation of three nuisance functions: (i) the conditional mean outcome function $\mu(t,\bm{s})=\E\left(Y|T=t,\bm{S}=\bm{s}\right)$; (ii) the partial derivative function $\beta(t,\bm{s})=\frac{\partial}{\partial t} \mu(t,\bm{s})$; and (iii) the conditional density $p_{T|\bm{S}}(t|\bm{s})$. Below, we discuss how these nuisance functions are estimated in our numerical experiments of \autoref{sec:experiments}.

$\bullet$ {\bf Estimations of $\mu(t,\bm{s})$ and $\beta(t,\bm{s})$:} We apply a fully connected neural network model with two hidden layers of size $100\times 50$ and use the sigmoid linear unit function $u\mapsto \frac{u}{1+e^{-u}}$ as the activation function to ensure the smoothness of resulting estimators $\hat{\mu}(t,\bm{s})$ and $\hat{\beta}(t,\bm{s})$. Other choices of the neural network architectures and activation functions also works for our proposed DR estimator \eqref{theta_DR}. Theoretically, Theorem 1 in \cite{farrell2021deep} and Section 3.1 in \cite{colangelo2020double} discuss some regularity conditions under which our requirements on the rates of convergence are satisfied by neural network models. Practically, our neural network model is implemented via \texttt{PyTorch} \citep{paszke2019pytorch}, and we use its automatic differentiation engine \citep{paszke2017automatic} to compute the estimated partial derivative $\hat{\beta}(t,\bm{s})$ from the fitted conditional mean outcome function $\hat{\mu}(t,\bm{s})$. In contrast to numerical differentiation, automatic differentiation offers the key advantage of being hyperparameter-free and inherently accurate to working precision \citep{baydin2018automatic,blondel2024elements}. Recently, it has been employed to compute the semi-parametric or non-parametric efficient influence function for any statistical functional \citep{luedtke2024simplifying}.

$\bullet$ {\bf Estimation of $p_{T|\bm{S}}(t|\bm{s})$:} Given the data-generating model \eqref{dgp}, we consider two different methods for estimating the conditional density $p_{T|\bm{S}}(t|\bm{s})$ with kernel smoothing techniques.
\begin{enumerate}
	\item \emph{Method 1 (Kernel density estimation (KDE) on residuals):} Notice that the relationship between the covariate vector $\bm{S}$ and the treatment variation variable $E$ is additive in model \eqref{dgp}, \emph{i.e.}, 
	$$T=g_{\bm{S}}(\bm{S})+ g_E(E) \quad \text{ with } \quad g_{\bm{S}}(\bm{S}) = \frac{1}{\sqrt{2\pi}}\int_{-\infty}^{3\bm{\xi}^T\bm{S}} \exp\left(-\frac{u^2}{2}\right) du - 0.5 \; \text{ and } \; g_E(E)=0.75E.$$
	In addition, the regression function of $T$ against $\bm{S}$ is given by $\E\left(T|\bm{S}=\bm{s}\right) = g_{\bm{S}}(\bm{s})$. To estimate the regression function $g_{\bm{S}}$, we can apply any machine learning method to the data $\{(T_i,\bm{S}_i)\}_{i=1}^n$ using cross-fitting. In the actual implementation, we use a neural network model with one hidden layer of size 20 and rectified linear unit function $u\mapsto\max\{0,u\}$ as the activation function. Based on the fitted regression function $\hat{g}_{\bm{S}}$, we construct an estimator of $p_{T|\bm{S}}(t|\bm{s})$ as:
	$$\hat{p}_{T|\bm{S}}(t|\bm{s}) = \frac{1}{nh_e} \sum_{i=1}^n K_e\left[\frac{t-\hat{g}_{\bm{S}}(\bm{s}) - \left(T_i-\hat{g}_{\bm{S}}(\bm{S}_i)\right)}{h_e}\right],$$
	where the kernel function $K_e$ and bandwidth $h_e>0$ may differ from those in our DR estimator \eqref{theta_DR}. It is worth noting that the observations $\{T_1-\hat{g}_{\bm{S}}(\bm{S}_1),...,T_n-\hat{g}_{\bm{S}}(\bm{S}_n)\}$ are not i.i.d., necessitating additional analysis for asymptotic theory and bandwidth selection \citep{hart1990data,meloche1990asymptotic}. For simplicity, we use the Epanechnikov kernel $K_e(u)=\frac{3}{4}(1-|u|)\, \mathbbm{1}_{\{|u|\leq 1\}}$ and choose the bandwidth via Silverman's rule of thumb as $\hat{h}_e=\left(\frac{4}{3}\right)^{\frac{1}{5}}\hat{\sigma}_e n^{-\frac{1}{5}}$, where $\hat{\sigma}_e$ is the sample standard deviation of $\{T_1-\hat{g}_{\bm{S}}(\bm{S}_1),...,T_n-\hat{g}_{\bm{S}}(\bm{S}_n)\}$.
	
	\item \emph{Method 2 (Regression on kernel-smoothed outcomes (RKS))} The validity of Method 1 relies on the additive relation between $\bm{S}$ and $E$ in the model for $T$. Since this additive structure may not hold in general, we consider another kernel smoothing method for estimating $p_{T|\bm{S}}(t|\bm{s})$. Specifically, we estimate a kernel-smoothed regression function $g(t,\bm{s})=\E\left[K_r\left(\frac{T-t}{h_r}\right) \big| \bm{S}=\bm{s}\right]$ by regressing kernel-smoothed outcomes $\left\{K_r\left(\frac{T_i-t}{h_r}\right)\right\}_{i=1}^n$ against the covariate vectors $\left\{\bm{S}_i\right\}_{i=1}^n$ via any machine learning method. The fitted kernel-smoothed regression function $\hat{g}(t,\bm{s})$ is a consistent estimator of $p_{T|\bm{S}}(t|\bm{s})$ when the regression method is accurate, because
	\begin{align*}
		g(t,\bm{s})&=\E\left[K_r\left(\frac{T-t}{h_r}\right) \Big| \bm{S}=\bm{s}\right]\\
		&=\int_{\mathcal{T}} K_r\left(\frac{t_1-t}{h_r}\right) p_{T|\bm{S}}(t_1|\bm{s})\, dt_1\\
		&= \int_{\mathbb{R}} K_r(u)\cdot p_{T|\bm{S}}(t+uh_r|\bm{s})\, du\\
		&= \int_{\mathbb{R}} K_r(u)\left[p_{T|\bm{S}}(t|\bm{s}) + uh_r\cdot p_{T|\bm{S}}'(t|\bm{s}) + \frac{u^2h_r^2}{2}\cdot p_{T|\bm{S}}''(t|\bm{s}) + o\left(h_r^3\right)\right]\, du\\
		&= p_{T|\bm{S}}(t|\bm{s}) + O\left(h_r^2\right)\\
		&\to p_{T|\bm{S}}(t|\bm{s})
	\end{align*}
    as $h_r\to 0$ under Assumptions~\ref{assump:den_diff} and \ref{assump:reg_kernel}(a-b). In the actual implementation, we again use a neural network model with one hidden layer of size 20 and rectified linear unit function $u\mapsto\max\{0,u\}$ as the activation function.
	Here, the kernel function $K_r$ and bandwidth $h_r>0$ can be different from those in our DR estimator \eqref{theta_DR}. To ensure a relatively large effective sample size for fitting $g$, we use the Gaussian kernel $K_r(u)=\frac{1}{\sqrt{2\pi}}\exp\left(-\frac{u^2}{2}\right)$ and choose the bandwidth by Silverman's rule of thumb as $\hat{h}_r=\left(\frac{4}{3}\right)^{\frac{1}{5}}\hat{\sigma}_T n^{-\frac{1}{5}}$, where $\hat{\sigma}_T$ is the sample standard deviation of $\{T_1,...,T_n\}$.
\end{enumerate}
Besides that, \cite{klosin2021automatic} proposed another method with kernel smoothing that directly estimates the reciprocal $\frac{1}{p_{T|\bm{S}}(t|\bm{s})}$ of the conditional density using a minimum distance Lasso approach \citep{chernozhukov2022automatic}. This method employs polynomial basis functions of the covariate vector $\bm{S}$ and a kernel-smoothed $L_2$ loss function. We briefly experimented with this approach and found that its performance and computational efficiency are inferior to the two methods above. In addition, this approach is very sensitive to the choice of its tuning parameter as shown in Section 6 of \cite{klosin2021automatic}. Thus, we choose not to report its results.

\section{Additional Simulation Results}
\label{app:add_sim}

This section provides supplementary simulation results assessing the impact of varying the bandwidth parameter on the performance of our proposed estimators of $\theta(t)$ and the finite-difference method by \cite{colangelo2020double}. Furthermore, we evaluate the finite-sample performances of our proposed estimators without cross-fitting in both simulation studies and the empirical analysis of the U.S. Job Corps Program dataset.

\subsection{Simulation Studies With Positivity Across Different Bandwidth Choices}
\label{app:add_pos_bandwidth}

We follow the same data-generating process and experimental setup in \autoref{subsec:simulation} to evaluate the performances of different estimators of $\theta(t)$ under the positivity condition, varying the bandwidth parameter $h$. In line with the bandwidth choices in \cite{colangelo2020double,klosin2021automatic}, we examine four scaling factors for the bandwidth parameter as $h=C_h\cdot \hat{\sigma}_T\cdot n^{-\frac{1}{5}}$ with $C_T\in \left\{0.75, 1, 1.25, 1.5\right\}$, where $\hat{\sigma}_T$ is the sample standard deviation of $\{T_1,...,T_n\}$. For supplementary purposes, we only present the simulation results with 5-fold cross-fitting when the sample size is $n=4000$ in \autoref{fig:theta_L5_bandwidth}. The results are mostly consistent with our findings in \autoref{subsec:simulation}.
Our proposed DR estimators, leveraging either true or KDE-estimated conditional densities, demonstrate lower estimation biases and superior empirical coverage probabilities for their confidence intervals compared to the finite-difference method of \cite{colangelo2020double}. At the same time, they maintain RMSEs that are comparable to the finite-difference method. This additional results further demonstrate the robustness of our proposed DR estimator \eqref{theta_DR} to variations in its bandwidth parameter.

\begin{figure}[t]
	\centering
	\includegraphics[width=1\linewidth]{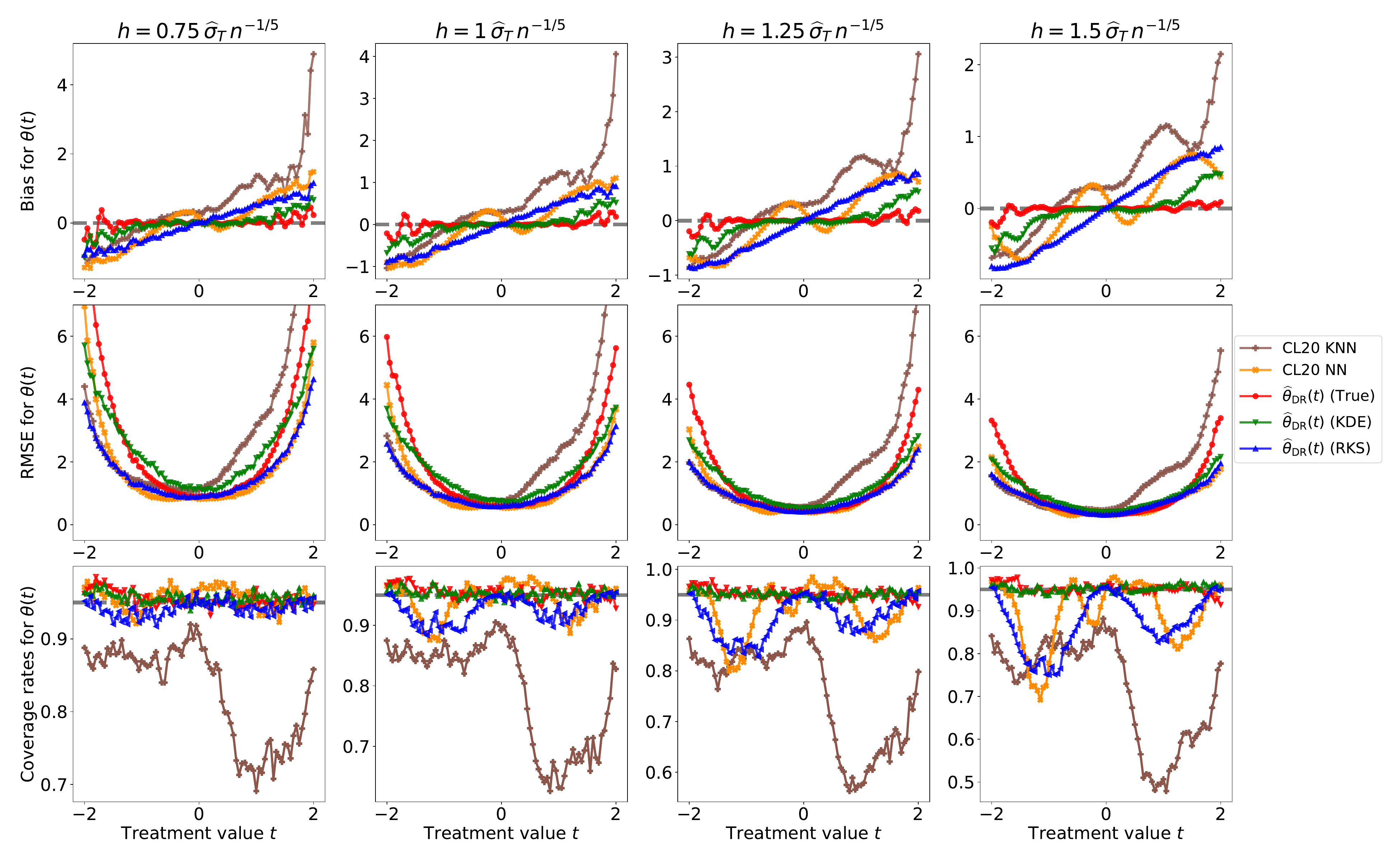}
	\caption{Comparisons between our proposed estimators and the finite-difference approaches by \cite{colangelo2020double} (``CL20'') under positivity and with 5-fold cross-fitting across different bandwidth values ($h$). Rows present estimation bias, RMSE, and coverage probability for each estimator of $\theta(t)$, while columns correspond to different scaling factors for $h$.}
	\label{fig:theta_L5_bandwidth}
\end{figure}

\subsection{Simulation Studies With Positivity and No Cross-Fitting}
\label{app:add_pos_nocrossfitting}

For exploratory purposes, we conduct additional simulations to compare the performances of our proposed estimators of $\theta(t)$ in \autoref{sec:theta_pos} with the finite-difference method by \cite{colangelo2020double} when cross-fitting is not employed.

We replicate the experimental setup in \autoref{subsec:simulation} to generate the simulation results shown in \autoref{fig:theta_L1} across various sample sizes without using any cross-fitting. When the sample size is small, our proposed DR estimators without cross-fitting exhibit lower RMSEs but higher estimation biases than the finite-difference method by \cite{colangelo2020double}, resulting in inferior empirical coverage probabilities for the associated confidence intervals. However, as the sample size increases, the estimation biases of our DR estimators diminish, and the empirical coverage probabilities of their confidence intervals improve, ultimately surpassing the finite-sample performance of the finite-difference method by \cite{colangelo2020double}.

\begin{figure}[t]
	\centering
	\includegraphics[width=1\linewidth]{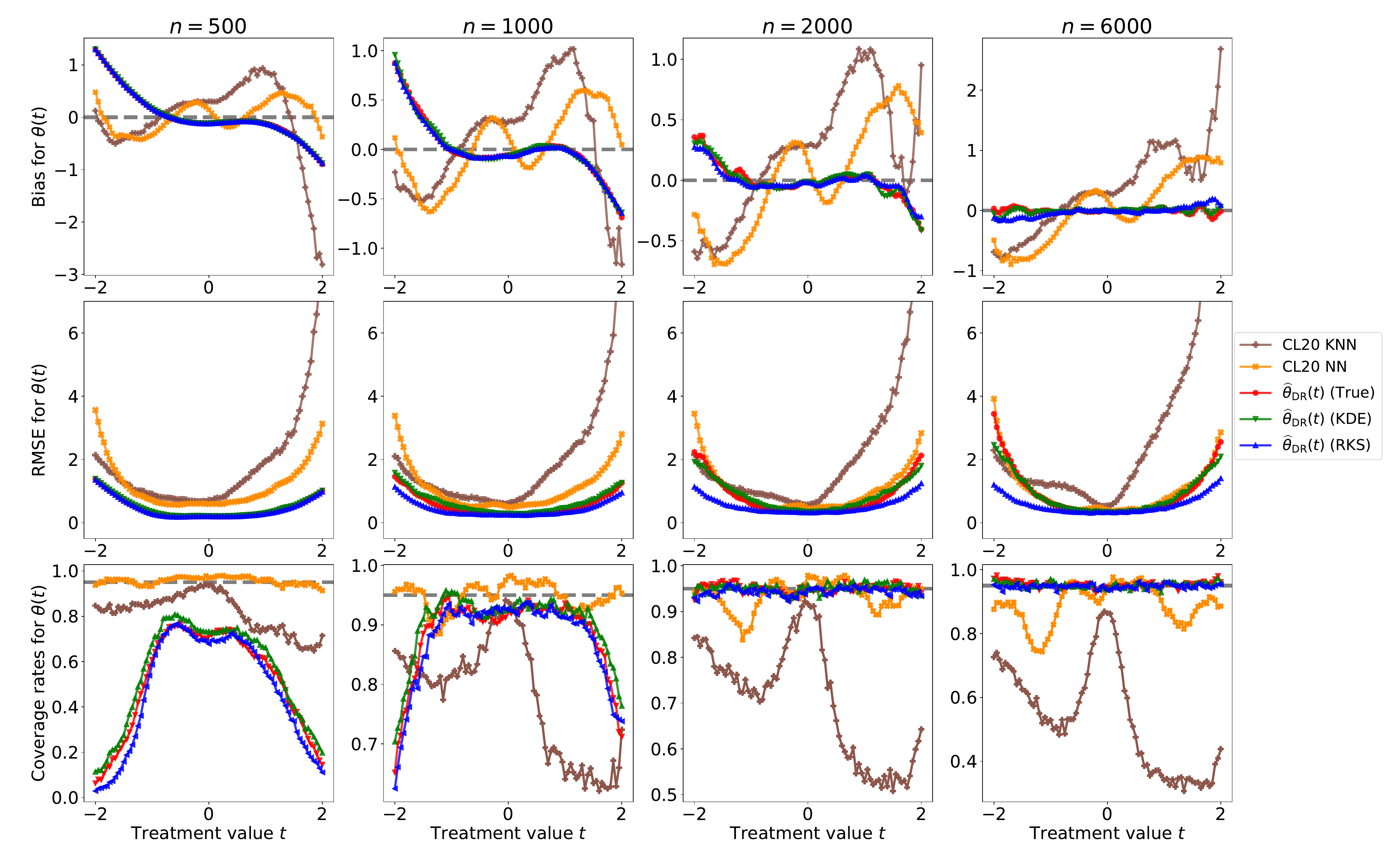}
	\caption{Comparisons between our proposed estimators and the finite-difference approaches by \cite{colangelo2020double} (``CL20'') under positivity and without cross-fitting across various sample sizes. Rows present estimation bias, RMSE, and coverage probability for each estimator of $\theta(t)$, while columns correspond to different values for $n$. This figure follows an identical simulation setup as \autoref{fig:theta_L5} but without using any cross-fitting.}
	\label{fig:theta_L1}
\end{figure}


These results without cross-fitting again highlight the practical utility of our proposed estimators in \autoref{sec:theta_pos} under the positivity condition. However, developing rigorous theoretical guarantees for these estimators without cross-fitting is beyond the scope of this paper and will be addressed in future work.


\subsection{Simulation Studies Without Positivity Across Different Bandwidth Choices}
\label{app:add_sim_nopos_bandwidth}

We adopt the same data-generating process and experimental setup in \autoref{subsec:sim_nopos} to evaluate the performances of different estimators of $\theta(t)$ under various choices of the bandwidth parameter without assuming the positivity condition. Specifically, we test four scaling factors for the bandwidth parameter as $h=C_h\cdot \hat{\sigma}_T\cdot n^{-\frac{1}{5}}$ with $C_T\in \left\{0.75, 1, 1.5, 2\right\}$, where $\hat{\sigma}_T$ is the sample standard deviation of $\{T_1,...,T_n\}$. For supplementary purposes, we only present the simulation results with 5-fold cross-fitting when the sample size is $n=2000$ in \autoref{fig:theta_L5_nopos_bandwidth}. Again, our proposed bias-corrected estimators of $\theta(t)$ demonstrate significant improvements by reducing bias and enhancing the empirical coverage probabilities of the resulting confidence intervals compared to their counterparts in \autoref{sec:theta_pos} across various bandwidth parameter choices.

\begin{figure}[t]
	\centering
	\includegraphics[width=1\linewidth]{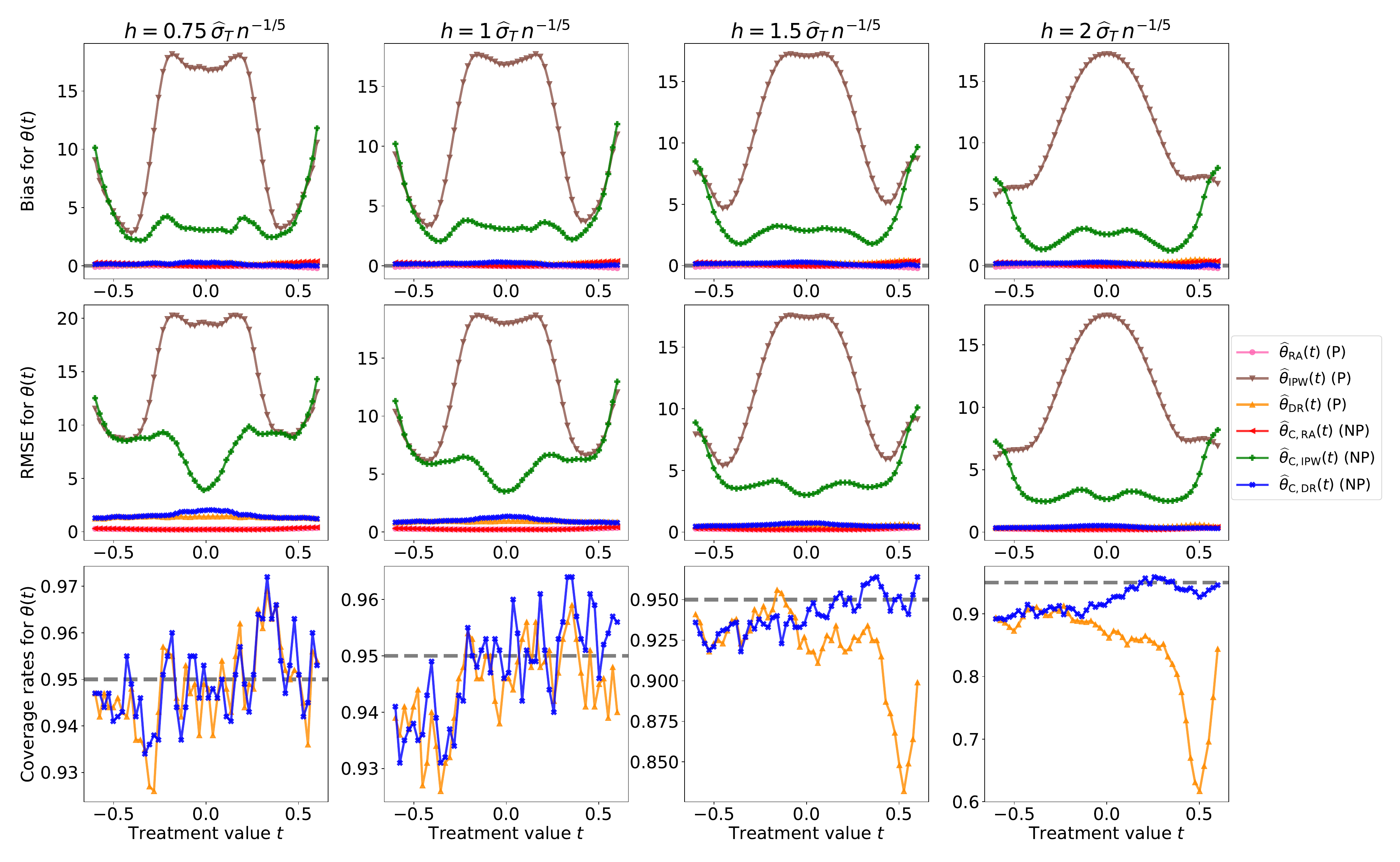}
	\caption{Comparisons between our bias-corrected estimators (NP) in \autoref{subsec:bias_corrected_IPW_DR} and their counterparts (P) under the violation of positivity and with 5-fold cross-fitting ($L=5$) across different bandwidth values ($h$). Rows present estimation bias, RMSE, and coverage probability for each estimator of $\theta(t)$, while columns correspond to different scaling factors for $h$.}
	\label{fig:theta_L5_nopos_bandwidth}
\end{figure}

\subsection{Simulation Studies Without Positivity and No Cross-Fitting}
\label{app:add_nopos_nocrossfitting}

For exploratory purposes, we conduct additional simulations for our bias-corrected estimators of $\theta(t)$ in \autoref{subsec:bias_corrected_IPW_DR} when the positivity condition is violated and cross-fitting is not employed.

Using the same experimental setup described in \autoref{subsec:sim_nopos}, we generate simulation results, shown in \autoref{fig:theta_L1_nopos}, across various sample sizes without using any cross-fitting. As expected, the estimation biases and RMSEs of our bias-corrected estimators improve as the sample size increases, consistently outperforming their standard counterparts. However, in comparison to the results obtained with 5-fold cross-fitting in \autoref{fig:theta_L5_nopos}, the performance of our bias-corrected estimators without cross-fitting deteriorates, particularly in terms of the empirical coverage probabilities of the resulting confidence intervals. These results consolidate the need of cross-fitting for constructing our bias-corrected estimators as \autoref{thm:theta_nopos} suggests.

\begin{figure}[t]
	\centering
	\includegraphics[width=1\linewidth]{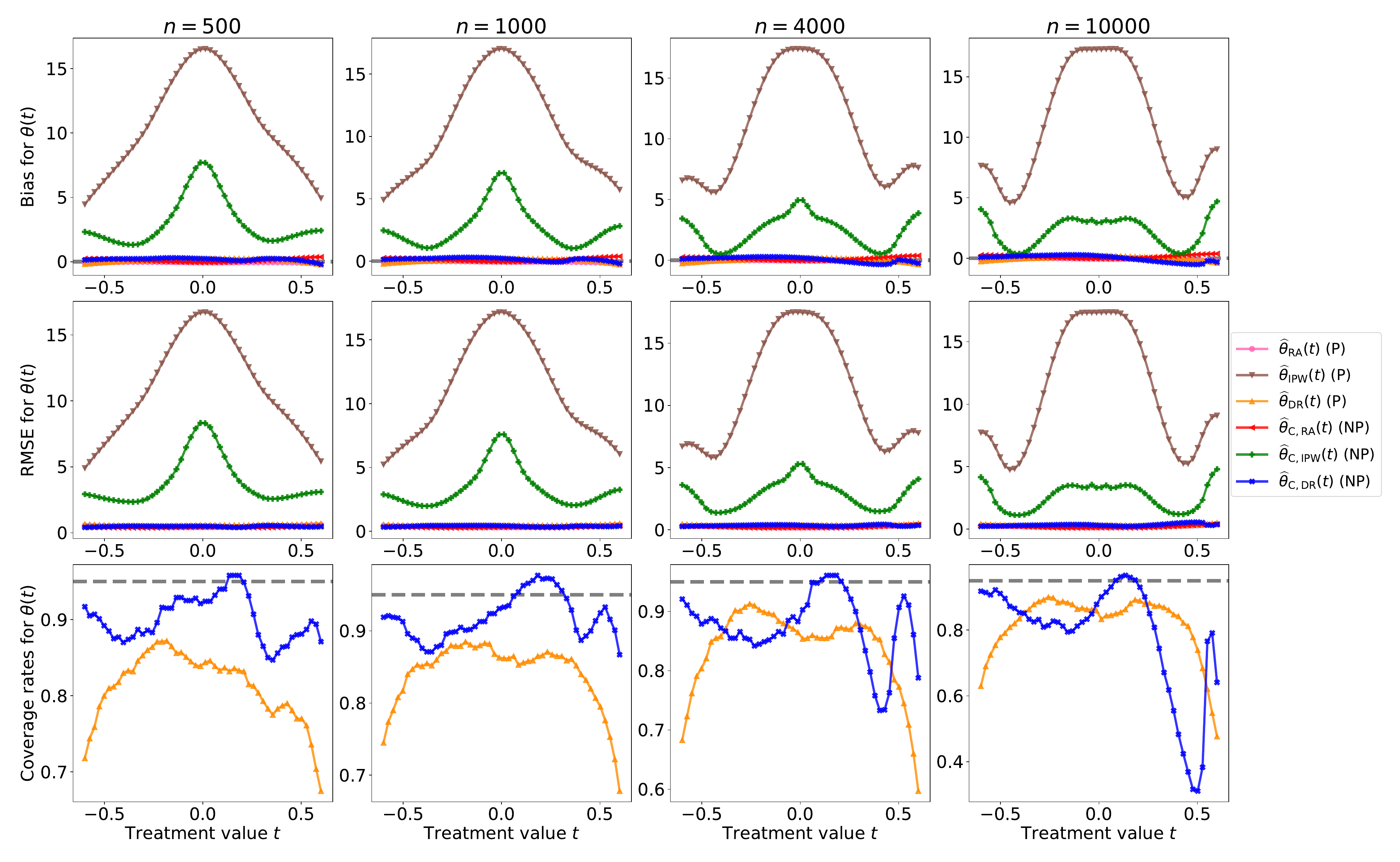}
	\caption{Comparisons between our bias-corrected estimators (NP) in \autoref{subsec:bias_corrected_IPW_DR} and their counterparts (P) in \autoref{sec:theta_pos} under the violation of positivity and without cross-fitting across different sample sizes. Rows present estimation bias, RMSE, and coverage probability for each estimator of $\theta(t)$, while columns correspond to different values for $n$. This figure follows an identical simulation setup as in \autoref{fig:theta_L5_nopos} but without using any cross-fitting.}
	\label{fig:theta_L1_nopos}
\end{figure}

\subsection{Analysis of the Job Corps Program With No Cross-fitted Estimators}
\label{app:job_corps_nocrossfitting}

Finally, we explore the behaviors of our proposed DR estimator \eqref{theta_DR} and the finite-difference method by \cite{colangelo2020double} when cross-fitting is not employed. Following the same analysis pipeline described in \autoref{subsec:job_corps}, but without employing cross-fitting, we produce the results shown in \autoref{fig:job_corps_nocrossfit}. Interestingly, our DR estimator without cross-fitting reveals some distinct trends compared to its counterpart with 5-fold cross-fitting. Specifically, it suggests a positive impact on employment during the first 20 weeks ($\sim 800$ hours), diminishing benefits after 23 weeks ($\sim 920$ hours), and statistically significant negative effects beyond 43 weeks ($\sim 1720$ hours). These trends align with prior research (\emph{e.g.}, Figure 2 of \citealt{lee2009training}), which documented short-term negative impacts of the program on employment propensities (104 weeks after the program assignment). However, it is worth mentioning that the analysis in \cite{lee2009training} was based on a binary treatment variable of being in the program or not. Since we do not establish any theoretical guarantees for our DR estimator of $\theta(t)$ when cross-fitting is not applied in this paper, more thorough investigations are necessary in the future to substantiate these short-term negative impacts of the Job Corps program.

\begin{figure}[t]
	\centering
	\includegraphics[width=1\linewidth]{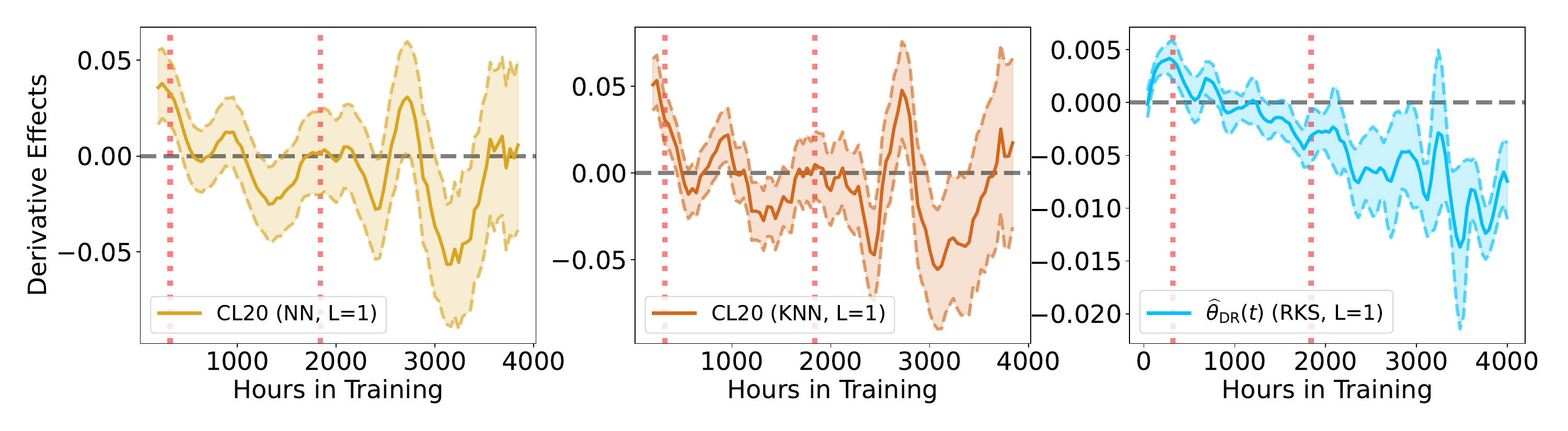}
	\caption{Estimated derivative effect curves with 95\% confidence intervals using our proposed estimators and the finite-difference approaches by \cite{colangelo2020double} (``CL20'') without cross-fitting. The vertical red dotted lines mark the original treatment range $[320,1840]$ analyzed in \cite{colangelo2020double}. This figure follows an identical analysis pipeline as in \autoref{fig:job_corps} but without using any cross-fitting.}
	\label{fig:job_corps_nocrossfit}
\end{figure}

\section{Identification of $m(t)$ and $\theta(t)$ Under the Additive Confounding Model \eqref{add_conf_model}}
\label{app:id_additive}

\begin{proposition}[Identifications of $m(t)$ and $\theta(t)$]
	\label{prop:id_additive}
	Suppose that Assumptions~\ref{assump:id_cond}(a-c) and \ref{assump:den_diff}(c) holds under model \eqref{add_conf_model}. Then, for any $t\in \mathcal{T}$ with $p_{\bm{S}|T}(\bm{s}|t)>0$ for some $\bm{s}\in \mathcal{S}$, we have that
	$$\theta(t) = \bar{m}'(t) = \mathbb{E}\left[\frac{\partial}{\partial t}\mu(T,\bm{S})\Big|T=t\right],$$
	where $\mu(t,\bm{s})=\mathbb{E}(Y|T=t,\bm{S}=\bm{s})$. If, in addition, the marginal support $\mathcal{T}$ of $p_T(t)$ is connected, then 
	$$m(t) = \mathbb{E}\left[Y+\int_T^t\theta(\tilde{t})\, d\tilde{t}\right] = \mathbb{E}\left\{Y+\int_T^t\mathbb{E}\left[\frac{\partial}{\partial t}\mu(T,\bm{S})\Big|T=\tilde{t}\right]\, d\tilde{t}\right\}.$$
\end{proposition}

\begin{proof}[Proof of Proposition~\ref{prop:id_additive}]
	We first study the identification of $\theta(t)$. By \eqref{m_theta_additive}, $\theta(t)=\bar{m}'(t)$, and the conditional mean outcome function $\mu(t,\bm{s})=\E\left(Y|T=t,\bm{S}=\bm{s}\right)=\bar{m}(t)+\eta(\bm{s})$ is well-defined within the support $\mathcal{J}$ of the joint density $p(t,\bm{s})$. In particular, for any $t\in \mathcal{T}$ with $p_{\bm{S}|T}(\tilde{\bm{s}}|t)>0$ for some $\tilde{\bm{s}}\in \mathcal{S}$, we know that $p(t,\tilde{\bm{s}})>0$ so that $\frac{\partial}{\partial t}\mu(t,\tilde{\bm{s}})=\bar{m}'(t)$ is also well-defined for these $\tilde{\bm{s}}\in \mathcal{S}(t)$. Furthermore, under Assumption~\ref{assump:den_diff}(c), the support $\mathcal{S}(t)$ of the conditional distribution $p_{\bm{S}|T}(\bm{s}|t)$ is non-degenerate (\emph{i.e.}, has nonzero Lebesgue measure). Thus,
	$$\theta(t) = \bar{m}'(t) = \mathbb{E}\left[\frac{\partial}{\partial t}\mu(T,\bm{S})\Big|T=t\right]$$
	is valid. 
	
	As for the identification of $m(t)$, we apply the fundamental theorem of calculus and argue that
	$$m(t)=m(T) + \int_T^t \theta(\tilde{t})\, d\tilde{t}.$$
	Taking the expectation over $T$ yields that
	$$m(t)=\mathbb{E}\left[m(T) + \int_T^t \theta(\tilde{t})\, d\tilde{t}\right]=\E\left\{Y + \int_T^t\mathbb{E}\left[\frac{\partial}{\partial t}\mu(T,\bm{S})\Big|T=\tilde{t}\right]\, d\tilde{t}\right\},$$
	where the second equality follows from the fact that $\E\left[m(T)\right] = \mathbb{E}\left[\bar{m}(T)\right] + \E\left[\eta(\bm{S})\right]=\E(Y)$ by \eqref{m_theta_additive}. Here, the connectedness of $\mathcal{T}$ ensures that the integration of $\theta(t)$ is only over the region where it is identifiable. When $\mathcal{T}$ has multiple connected components, the integral formula \eqref{id_m} as well as the observations should be restricted to the connected component in which the point of interest $t\in \mathcal{T}$ lies.
\end{proof}

\section{Asymptotic Differences Between Two Variants of IPW Estimators}

In this section, we study the asymptotic differences between the IPW estimators when the inverse probability weights are evaluated at the sample points $(T_i,\bm{S}_i),i=1,...,n$ or at the (query) points $(t,\bm{S}_i),i=1,...,n$. Specifically, for estimating the dose-response curve $m(t)$, we have two variants of the IPW estimators as \eqref{m_IPW} and \eqref{m_IPW2}. Similarly, for estimating the derivative effect $\theta(t)=m'(t)$, we also consider two different versions of the IPW estimators as \eqref{theta_IPW} and \eqref{theta_IPW2}. For the sake of illustrations, we assume that the conditional density $p_{T|\bm{S}}$ is known and only consider the oracle IPW estimators.

\subsection{Asymptotic Difference Between IPW Estimators \eqref{m_IPW} and \eqref{m_IPW2} of $m(t)$}
\label{app:m_ipw_diff}

We define the difference between two oracle IPW estimators of $m(t)$ as:
\begin{equation}
	\label{m_IPW_diff}
	\tilde{\Delta}_{\mathrm{IPW,m}}(t) = \tilde{m}_{\mathrm{IPW,2}}(t) - \tilde{m}_{\mathrm{IPW}}(t) = \frac{1}{nh}\sum_{i=1}^n \left[\frac{1}{p_{T|\bm{S}}(t|\bm{S}_i)} - \frac{1}{p_{T|\bm{S}}(T_i|\bm{S}_i)}\right] Y_i \cdot K\left(\frac{T_i-t}{h}\right).
\end{equation}

\begin{proposition}
	\label{prop:m_ipw_diff}
	Suppose that Assumptions~\ref{assump:positivity}, \ref{assump:reg_diff}, \ref{assump:den_diff}, and \ref{assump:reg_kernel} hold with $\mu(t,\bm{s})=\E(Y|T=t,\bm{S}=\bm{s})$. Then, for any fixed $t\in \mathcal{T}$, we have that
	\begin{align*}
		\tilde{\Delta}_{\mathrm{IPW,m}}(t) &= h^2\kappa_2\cdot \mathbb{E}\left[\frac{\frac{\partial}{\partial t} p_{T|\bm{S}}(t|\bm{S}) \cdot \frac{\partial}{\partial t} \mu(t,\bm{S})}{p_{T|\bm{S}}(t|\bm{S})} + \frac{\mu(t,\bm{S})\cdot \frac{\partial^2}{\partial t^2} p_{T|\bm{S}}(t|\bm{S})}{2p_{T|\bm{S}}(t|\bm{S})} \right] + O(h^3) + O_P\left(\sqrt{\frac{h}{n}}\right)\\
		&= O(h^2) + O_P\left(\sqrt{\frac{h}{n}}\right)
	\end{align*}
	as $h\to 0$ and $n\to \infty$.
\end{proposition}

\begin{proof}[Proof of Proposition~\ref{prop:m_ipw_diff}]
	By Chebyshev's inequality, we know that
	\begin{align*}
		\tilde{\Delta}_{\mathrm{IPW,m}}(t) &= \mathbb{E}\left[\tilde{\Delta}_{\mathrm{IPW,m}}(t)\right] + \tilde{\Delta}_{\mathrm{IPW,m}}(t) - \mathbb{E}\left[\tilde{\Delta}_{\mathrm{IPW,m}}(t)\right]\\
		&= \mathbb{E}\left[\tilde{\Delta}_{\mathrm{IPW,m}}(t)\right] + O_P\left(\sqrt{\mathrm{Var}\left[\tilde{\Delta}_{\mathrm{IPW,m}}(t)\right]}\right).
	\end{align*}
	On one hand, we calculate that
	\begin{align*}
		&\mathbb{E}\left[\tilde{\Delta}_{\mathrm{IPW,m}}(t)\right] \\
		&= \mathbb{E}\left\{\frac{1}{h}\left[\frac{1}{p_{T|\bm{S}}(t|\bm{S})} - \frac{1}{p_{T|\bm{S}}(T|\bm{S})}\right] Y \cdot K\left(\frac{T-t}{h}\right) \right\} \\
		&= \mathbb{E}\left\{\int_{\mathcal{T}} \left[\frac{1}{p_{T|\bm{S}}(t|\bm{S})} - \frac{1}{p_{T|\bm{S}}(t_1|\bm{S})}\right] \mu(t_1,\bm{S}) \cdot K\left(\frac{t_1-t}{h}\right)\cdot p_{T|\bm{S}}(t_1|\bm{S})\, dt_1\right\}\\
		&\stackrel{\text{(i)}}{=} \mathbb{E}\left\{\int_{\mathbb{R}} \left[\frac{p_{T|\bm{S}}(t+uh|\bm{S})}{p_{T|\bm{S}}(t|\bm{S})} - 1\right] \mu(t+uh,\bm{S}) \cdot K(u) \, du\right\}\\
		&\stackrel{\text{(ii)}}{=} \mathbb{E}\bigg\{\int_{\mathbb{R}} \left[\frac{uh\cdot \frac{\partial}{\partial t}p_{T|\bm{S}}(t|\bm{S}) + \frac{u^2h^2}{2} \cdot \frac{\partial^2}{\partial t^2} p_{T|\bm{S}}(t|\bm{S})}{p_{T|\bm{S}}(t|\bm{S})} + O(h^3)\right]\\
		&\quad\quad\quad \times \left[\mu(t,\bm{S}) +uh\cdot \frac{\partial}{\partial t}\mu(t,\bm{S}) + \frac{u^2h^2}{2}\frac{\partial^2}{\partial t^2}\mu(t,\bm{S}) + O(h^3)\right] K(u) \, du\bigg\}\\
		&= h^2\kappa_2\cdot \mathbb{E}\left[\frac{\frac{\partial}{\partial t} p_{T|\bm{S}}(t|\bm{S}) \cdot \frac{\partial}{\partial t} \mu(t,\bm{S})}{p_{T|\bm{S}}(t|\bm{S})} + \frac{\mu(t,\bm{S})\cdot \frac{\partial^2}{\partial t^2} p_{T|\bm{S}}(t|\bm{S})}{2p_{T|\bm{S}}(t|\bm{S})} \right] + O(h^3)
	\end{align*}
	where (i) uses a change of variable $u=\frac{t_1-t}{h}$ while (ii) applies Taylor's expansions on $p_{T|\bm{S}}$ and $\mu$ under Assumptions~\ref{assump:reg_diff} and \ref{assump:den_diff}. On the other hand, we also compute that
	\begin{align*}
		&\mathrm{Var}\left[\tilde{\Delta}_{\mathrm{IPW,m}}(t) \right] \\
		&= \frac{1}{nh^2} \cdot \mathrm{Var}\left\{\left[\frac{1}{p_{T|\bm{S}}(t|\bm{S})} - \frac{1}{p_{T|\bm{S}}(T|\bm{S})}\right] Y \cdot K\left(\frac{T-t}{h}\right) \right\}\\
		&\stackrel{\text{(iii)}}{=}\frac{1}{nh^2} \cdot \mathbb{E}\left\{\left[\frac{1}{p_{T|\bm{S}}(t|\bm{S})} - \frac{1}{p_{T|\bm{S}}(T|\bm{S})}\right]^2 Y^2 K^2\left(\frac{T-t}{h}\right)\right\} + O\left(\frac{h^4}{n}\right)\\
		&\stackrel{\text{(iv)}}{=} \frac{1}{nh}\cdot \mathbb{E}\left\{\left[\frac{1}{p_{T|\bm{S}}(t|\bm{S})} - \frac{1}{p_{T|\bm{S}}(t+uh|\bm{S})}\right]^2 \left[\mu(t+uh,\bm{S})^2 + \sigma^2\right] K^2(u) \cdot p_{T|\bm{S}}(t+uh|\bm{S})\, du \right\}\\
		&\stackrel{\text{(v)}}{=} \frac{1}{nh}\cdot \mathbb{E}\left\{\frac{\left[uh\cdot \frac{\partial}{\partial t} p_{T|\bm{S}}(t|\bm{S}) + \frac{u^2h^2}{2}\frac{\partial^2}{\partial t^2}p_{T|\bm{S}}(t|\bm{S}) + O(h^3)\right]^2}{p_{T|\bm{S}}^2(t|\bm{S})\cdot p_{T|\bm{S}}(t+uh|\bm{S})} \cdot \left[\mu(t+uh,\bm{S})^2 + \sigma^2\right] K^2(u) \, du \right\}\\
		&= \frac{h}{n}\cdot \nu_2 \cdot \mathbb{E}\left\{\frac{\left[\frac{\partial}{\partial t} \log p_{T|\bm{S}}(t|\bm{S})\right]^2 \left[\mu(t,\bm{S})^2 + \sigma^2\right]}{p_{T|\bm{S}}(t|\bm{S})}\right\} + O\left(\frac{h^2}{n}\right)
	\end{align*}
	where (iii) leverages our above calculation on $\mathbb{E}\left[\tilde{\Delta}_{\mathrm{IPW,m}}(t) \right] = O(h^2)$, (iv) applies a change of variable, and (v) utilizes Taylor's expansion on $p_{T|\bm{S}}$ under Assumption~\ref{assump:den_diff}.
	
	In total, we conclude that
	\begin{align*}
		\tilde{\Delta}_{\mathrm{IPW,m}}(t) &= \mathbb{E}\left[\tilde{\Delta}_{\mathrm{IPW,m}}(t)\right] + O_P\left(\sqrt{\mathrm{Var}\left[\tilde{\Delta}_{\mathrm{IPW,m}}(t)\right]}\right) \\
		&= h^2\kappa_2\cdot \mathbb{E}\left[\frac{\frac{\partial}{\partial t} p_{T|\bm{S}}(t|\bm{S}) \cdot \frac{\partial}{\partial t} \mu(t,\bm{S})}{p_{T|\bm{S}}(t|\bm{S})} + \frac{\mu(t,\bm{S})\cdot \frac{\partial^2}{\partial t^2} p_{T|\bm{S}}(t|\bm{S})}{2p_{T|\bm{S}}(t|\bm{S})} \right] + O(h^3) + O_P\left(\sqrt{\frac{h}{n}}\right)\\
		&= O(h^2) + O_P\left(\sqrt{\frac{h}{n}}\right)
	\end{align*}
	as $h\to 0$ and $n\to \infty$. The result follows.
\end{proof}

\subsection{Asymptotic Difference Between IPW Estimators \eqref{theta_IPW} and \eqref{theta_IPW2} of $\theta(t)$}
\label{app:theta_ipw_diff}

We define the difference between two oracle IPW estimators of $\theta(t)$ as:
\begin{equation}
	\label{theta_IPW_diff}
	\tilde{\Delta}_{\mathrm{IPW,\theta}}(t) = \tilde{\theta}_{\mathrm{IPW,2}}(t) - \tilde{\theta}_{\mathrm{IPW}}(t) = \frac{1}{nh}\sum_{i=1}^n \left[\frac{1}{p_{T|\bm{S}}(t|\bm{S}_i)} - \frac{1}{p_{T|\bm{S}}(T_i|\bm{S}_i)}\right] \frac{Y_i\left(\frac{T_i-t}{h^2}\right) K\left(\frac{T_i-t}{h}\right)}{\kappa_2}.
\end{equation}

\begin{proposition}
	\label{prop:theta_ipw_diff}
	Suppose that Assumptions~\ref{assump:positivity}, \ref{assump:reg_diff}, \ref{assump:den_diff}, and \ref{assump:reg_kernel} hold with $\mu(t,\bm{s})=\E(Y|T=t,\bm{S}=\bm{s})$. Then, for any fixed $t\in \mathcal{T}$, we have that
	\begin{align*}
		\tilde{\Delta}_{\mathrm{IPW,\theta}}(t) &= \mathbb{E}\left[\mu(t,\bm{S})\cdot \frac{\partial}{\partial t} \log p_{T|\bm{S}}(t|\bm{S})\right] \\
		&\quad + \frac{h^2\kappa_4}{2\kappa_2}\cdot \mathbb{E}\left[\frac{\partial^2}{\partial t^2}\mu(t,\bm{S})\cdot \frac{\partial}{\partial t} \log p_{T|\bm{S}}(t|\bm{S}) + \frac{\frac{\partial}{\partial t}\mu(t,\bm{S}) \cdot \frac{\partial^2}{\partial t^2} p_{T|\bm{S}}(t|\bm{S})}{p_{T|\bm{S}}(t|\bm{S})} + \frac{\mu(t,\bm{S}) \cdot \frac{\partial^3}{\partial t^3} p_{T|\bm{S}}(t|\bm{S})}{3p_{T|\bm{S}}(t|\bm{S})}\right] \\
		&\quad + O(h^3) + O_P\left(\sqrt{\frac{1}{nh}}\right)\\
		&= \mathbb{E}\left[\mu(t,\bm{S})\cdot \frac{\partial}{\partial t} \log p_{T|\bm{S}}(t|\bm{S})\right] + O(h^2) + O_P\left(\sqrt{\frac{1}{nh}}\right)
	\end{align*}
	as $h\to 0$ and $nh\to \infty$.
\end{proposition}

\begin{proof}[Proof of Proposition~\ref{prop:theta_ipw_diff}]
	By Chebyshev's inequality, we know that
	\begin{align*}
		\tilde{\Delta}_{\mathrm{IPW,\theta}}(t) &= \mathbb{E}\left[\tilde{\Delta}_{\mathrm{IPW,\theta}}(t)\right] + \tilde{\Delta}_{\mathrm{IPW,\theta}}(t) - \mathbb{E}\left[\tilde{\Delta}_{\mathrm{IPW,\theta}}(t)\right]\\
		&= \mathbb{E}\left[\tilde{\Delta}_{\mathrm{IPW,\theta}}(t)\right] + O_P\left(\sqrt{\mathrm{Var}\left[\tilde{\Delta}_{\mathrm{IPW,\theta}}(t)\right]}\right).
	\end{align*}
	On one hand, we calculate that
	\begin{align*}
		&\mathbb{E}\left[\tilde{\Delta}_{\mathrm{IPW,\theta}}(t)\right] \\
		&= \mathbb{E}\left\{\left[\frac{1}{p_{T|\bm{S}}(t|\bm{S})} - \frac{1}{p_{T|\bm{S}}(T|\bm{S})}\right] \frac{Y\left(\frac{T-t}{h^2}\right) K\left(\frac{T-t}{h}\right)}{h\cdot \kappa_2}\right\}\\
		&= \mathbb{E}\left\{\int_{\mathcal{T}} \left[\frac{p_{T|\bm{S}}(t_1|\bm{S})}{p_{T|\bm{S}}(t|\bm{S})} - 1\right] \frac{\mu(t_1,\bm{S}) \left(\frac{t_1-t}{h^2}\right) K\left(\frac{t_1-t}{h}\right)}{h\cdot \kappa_2} \, dt_1 \right\} \\
		&\stackrel{\text{(i)}}{=} \mathbb{E}\left\{\int_{\mathbb{R}} \left[\frac{p_{T|\bm{S}}(t+uh|\bm{S})}{p_{T|\bm{S}}(t|\bm{S})} - 1\right] \frac{\mu(t+uh,\bm{S}) \cdot u \cdot K(u)}{h\cdot \kappa_2} \, du \right\} \\
		&\stackrel{\text{(ii)}}{=} \mathbb{E}\bigg\{\int_{\mathbb{R}} \left[\frac{uh\cdot \frac{\partial}{\partial t}p_{T|\bm{S}}(t|\bm{S}) + \frac{u^2h^2}{2} \frac{\partial^2}{\partial t^2}p_{T|\bm{S}}(t|\bm{S}) + \frac{u^3h^3}{6} \frac{\partial^3}{\partial t^3}p_{T|\bm{S}}(t|\bm{S}) + O(h^4)}{p_{T|\bm{S}}(t|\bm{S})} \right] \\
		&\quad\quad \quad \times \left[\mu(t,\bm{S}) + uh\cdot \frac{\partial}{\partial t} \mu(t,\bm{S}) + \frac{u^2h^2}{2} \frac{\partial^2}{\partial t^2} \mu(t,\bm{S}) + O(h^3)\right]\frac{u \cdot K(u)}{h\cdot \kappa_2} \, du \bigg\} \\
		&= \mathbb{E}\left[\mu(t,\bm{S})\cdot \frac{\partial}{\partial t} \log p_{T|\bm{S}}(t|\bm{S})\right] \\
		&\quad + \frac{h^2\kappa_4}{2\kappa_2}\cdot \mathbb{E}\left[\frac{\partial^2}{\partial t^2}\mu(t,\bm{S})\cdot \frac{\partial}{\partial t} \log p_{T|\bm{S}}(t|\bm{S}) + \frac{\frac{\partial}{\partial t}\mu(t,\bm{S}) \cdot \frac{\partial^2}{\partial t^2} p_{T|\bm{S}}(t|\bm{S})}{p_{T|\bm{S}}(t|\bm{S})} + \frac{\mu(t,\bm{S}) \cdot \frac{\partial^3}{\partial t^3} p_{T|\bm{S}}(t|\bm{S})}{3p_{T|\bm{S}}(t|\bm{S})}\right] + O(h^3),
	\end{align*}
	where (i) uses a change of variable $u=\frac{t_1-t}{h}$ while (ii) applies Taylor's expansions on $p_{T|\bm{S}}$ and $\mu$ under Assumptions~\ref{assump:reg_diff} and \ref{assump:den_diff}. On the other hand, we also compute that
	\begin{align*}
		&\mathrm{Var}\left[\tilde{\Delta}_{\mathrm{IPW,\theta}}(t)\right] \\
		&= \frac{1}{nh^2} \cdot \mathrm{Var}\left\{\left[\frac{1}{p_{T|\bm{S}}(t|\bm{S})} - \frac{1}{p_{T|\bm{S}}(T|\bm{S})}\right] \frac{Y\left(\frac{T-t}{h^2}\right) K\left(\frac{T-t}{h}\right)}{\kappa_2} \right\}\\
		&= \frac{1}{nh^4} \cdot \mathbb{E}\left\{\left[\frac{1}{p_{T|\bm{S}}(t|\bm{S})} - \frac{1}{p_{T|\bm{S}}(T|\bm{S})}\right]^2 \frac{Y^2\left(\frac{T-t}{h}\right)^2 K^2\left(\frac{T-t}{h}\right)}{\kappa_2^2} \right\} + O\left(\frac{1}{n}\right)\\
		&\stackrel{\text{(iii)}}{=} \frac{1}{nh^4} \cdot \mathbb{E}\left\{\left[\frac{1}{p_{T|\bm{S}}(t|\bm{S})} - \frac{1}{p_{T|\bm{S}}(T|\bm{S})}\right]^2 \frac{\left[\mu(T,\bm{S})^2 + \sigma^2\right]\left(\frac{T-t}{h}\right)^2 K^2\left(\frac{T-t}{h}\right)}{\kappa_2^2} \right\} + O\left(\frac{1}{n}\right)\\
		&\stackrel{\text{(iv)}}{=} \frac{1}{nh^3} \cdot \mathbb{E}\left\{\int_{\mathbb{R}}\left[\frac{1}{p_{T|\bm{S}}(t|\bm{S})} - \frac{1}{p_{T|\bm{S}}(t+uh|\bm{S})}\right]^2 \frac{\left[\mu(t+uh,\bm{S})^2 + \sigma^2\right]u^2 K^2(u)}{\kappa_2^2} \cdot p_{T|\bm{S}}(t+uh|\bm{S})\, du\right\} + O\left(\frac{1}{n}\right)\\
		&\stackrel{\text{(v)}}{=} \frac{1}{nh^3} \cdot \mathbb{E}\bigg\{\int_{\mathbb{R}} \frac{\left[uh\cdot \frac{\partial}{\partial t}p_{T|\bm{S}}(t|\bm{S}) + \frac{u^2h^2}{2} \cdot \frac{\partial^2}{\partial t^2} p_{T|\bm{S}}(t|\bm{S}) + \frac{u^3h^3}{6} \cdot \frac{\partial^3}{\partial t^3} p_{T|\bm{S}}(t|\bm{S}) + O(h^4)\right]^2}{p_{T|\bm{S}}^2(t|\bm{S})\cdot p_{T|\bm{S}}(t+uh|\bm{S})} \\
		&\quad \quad \quad \times \frac{\left[\mu(t+uh,\bm{S})^2 + \sigma^2\right]u^2 K^2(u)}{\kappa_2^2} \, du\bigg\} + O\left(\frac{1}{n}\right)\\
		&= \frac{\nu_4}{nh\cdot \kappa_2^2} \cdot \mathbb{E}\left\{\frac{\left[\frac{\partial}{\partial t} \log p_{T|\bm{S}}(t|\bm{S})\right]^2 \left[\mu(t,\bm{S})^2 + \sigma^2\right]}{p_{T|\bm{S}}(t|\bm{S})} \right\} + O\left(\frac{1}{n}\right),
	\end{align*}
	where (iii) leverages our above calculation on $\mathbb{E}\left[\tilde{\Delta}_{\mathrm{IPW,\theta}}(t) \right] = O(1)$, (iv) applies a change of variable, and (v) utilizes Taylor's expansion on $p_{T|\bm{S}}$ under Assumption~\ref{assump:den_diff}.
	
	In total, we conclude that
	\begin{align*}
		\tilde{\Delta}_{\mathrm{IPW,\theta}}(t) &= \mathbb{E}\left[\tilde{\Delta}_{\mathrm{IPW,\theta}}(t)\right] + O_P\left(\sqrt{\mathrm{Var}\left[\tilde{\Delta}_{\mathrm{IPW,\theta}}(t)\right]}\right) \\
		&= \mathbb{E}\left[\mu(t,\bm{S})\cdot \frac{\partial}{\partial t} \log p_{T|\bm{S}}(t|\bm{S})\right] \\
		&\quad + \frac{h^2\kappa_4}{2\kappa_2}\cdot \mathbb{E}\left[\frac{\partial^2}{\partial t^2}\mu(t,\bm{S})\cdot \frac{\partial}{\partial t} \log p_{T|\bm{S}}(t|\bm{S}) + \frac{\frac{\partial}{\partial t}\mu(t,\bm{S}) \cdot \frac{\partial^2}{\partial t^2} p_{T|\bm{S}}(t|\bm{S})}{p_{T|\bm{S}}(t|\bm{S})} + \frac{\mu(t,\bm{S}) \cdot \frac{\partial^3}{\partial t^3} p_{T|\bm{S}}(t|\bm{S})}{3p_{T|\bm{S}}(t|\bm{S})}\right] \\
		&\quad + O(h^3) + O_P\left(\sqrt{\frac{1}{nh}}\right)\\
		&= \mathbb{E}\left[\mu(t,\bm{S})\cdot \frac{\partial}{\partial t} \log p_{T|\bm{S}}(t|\bm{S})\right] + O(h^2) + O_P\left(\sqrt{\frac{1}{nh}}\right)
	\end{align*}
	as $h\to 0$ and $nh\to \infty$. The result follows.
\end{proof}

\begin{remark}
	Given our convergence analysis for $\hat{\theta}_{\mathrm{IPW}}(t)$ in \autoref{thm:theta_pos}, one can easily calculate that the non-vanishing bias term $\mathbb{E}\left[\mu(t,\bm{S})\cdot \frac{\partial}{\partial t} \log p_{T|\bm{S}}(t|\bm{S})\right]$ for $\tilde{\Delta}_{\mathrm{IPW,\theta}}(t)$ in Proposition~\ref{prop:theta_ipw_diff} results from the IPW estimator $\tilde{\theta}_{\mathrm{IPW,2}}(t)$ in \eqref{theta_IPW2}. Therefore, unless $\mu(t,\bm{s}) = 0$ or $\frac{\partial}{\partial t} \log p_{T|\bm{S}}(t|\bm{s})=0$, the IPW estimator \eqref{theta_IPW2} of $\theta(t)$ is asymptotically biased and inconsistent.
\end{remark}

\section{Consistency of Estimating $m(t)$ Under Positivity}
\label{app:proof_m_pos}

In this section, we review and prove the consistency results of $\hat{m}_{\mathrm{RA}}(t)$, $\hat{m}_{\mathrm{IPW}}(t)$, and $\hat{m}_{\mathrm{DR}}(t)$ in \eqref{m_RA}, \eqref{m_IPW}, \eqref{m_DR} for estimating the dose-response curve $t\mapsto m(t)=\E\left[Y(t)\right]$ under the positivity condition.

\begin{proposition}[Consistency of Estimating $m(t)$ Under Positivity]
	\label{prop:m_pos}
	Suppose that Assumptions~\ref{assump:id_cond}, \ref{assump:reg_diff}, \ref{assump:den_diff}, \ref{assump:reg_kernel}, and \ref{assump:positivity} hold as well as $\hat{\mu}, \hat{p}_{T|\bm{S}}$ are constructed on a data sample independent of $\{(Y_i,T_i,\bm{S}_i)\}_{i=1}^n$. For any fixed $t\in \mathcal{T}$, we let $\bar{\mu}(t,\bm{s})$ and $\bar{p}_{T|\bm{S}}(t|\bm{s})$ be fixed bounded functions to which $\hat{\mu}(t,\bm{s})$ and $\hat{p}_{T|\bm{S}}(t|\bm{s})$ converge under the rates of convergence as:
	$$\norm{\hat{\mu}(t,\bm{S}) - \bar{\mu}(t,\bm{S})}_{L_2} = O_P\left(\Upsilon_{1,n}\right) \quad \text{ and } \quad \sup_{|u-t|\leq h}\norm{\hat{p}_{T|\bm{S}}(u|\bm{S}) - \bar{p}_{T|\bm{S}}(u|\bm{S})}_{L_2}=O_P\left(\Upsilon_{2,n}\right),$$
	where $\Upsilon_{1,n},\Upsilon_{2,n}\to 0$ as $n\to \infty$. Then, as $h\to 0$ and $nh\to \infty$, we have that
	\begin{align*}
		&\hat{m}_{\mathrm{RA}}(t) - m(t) = O_P\left(\Upsilon_{1,n} + \norm{\bar{\mu}(t,\bm{S}) - \mu(t,\bm{S})}_{L_2} + \frac{1}{\sqrt{n}}\right),\\
		& \hat{m}_{\mathrm{IPW}}(t) - m(t) = O(h^2) + O_P\left(\sqrt{\frac{1}{nh}} + \Upsilon_{2,n} + \sup\limits_{|u-t|\leq h}\norm{\bar{p}_{T|\bm{S}}(u|\bm{S}) - p_{T|\bm{S}}(u|\bm{S})}_{L_2}\right). 
	\end{align*}
	If, in addition, we assume that 
	\begin{enumerate}[label=(\alph*)]
		\item $\bar{p}_{T|\bm{S}}$ satisfies Assumptions~\ref{assump:den_diff} and \ref{assump:positivity};
		\item either (i) ``$\,\bar{\mu} = \mu$'' or ``$\,\bar{p}_{T|\bm{S}}=p_{T|\bm{S}}$'' almost surely;
		\item $\sqrt{nh} \norm{\hat{\mu}(t,\bm{S}) - \mu(t,\bm{S})}_{L_2} \sup_{|u-t|\leq h}\norm{\hat{p}_{T|\bm{S}}(u|\bm{S}) - p_{T|\bm{S}}(u|\bm{S})}_{L_2}=o_P(1)$,
	\end{enumerate}
	then 
	$$\sqrt{nh}\left[\hat{m}_{\mathrm{DR}}(t) - m(t)\right] = \frac{1}{\sqrt{n}} \sum_{i=1}^n \left\{\psi_{h,t}\left(Y_i,T_i,\bm{S}_i; \bar{\mu}, \bar{p}_{T|\bm{S}}\right) + \sqrt{h}\left[\bar{\mu}(t,\bm{S}_i) -\mathbb{E}\left[\mu(t,\bm{S}) \right]\right]\right\} +o_P(1)$$
	when $nh^5\to c_2$ for some finite number $c_2\geq 0$, where $\psi_{h,t}\left(Y,T,\bm{S}; \bar{\mu}, \bar{p}_{T|\bm{S}}\right) = \frac{K\left(\frac{T-t}{h}\right)}{\sqrt{h} \cdot \bar{p}_{T|\bm{S}}(T|\bm{S})}\left[Y-\bar{\mu}(t,\bm{S})\right]$ and 
	$$\sqrt{nh}\left[\hat{m}_{\mathrm{DR}}(t) - m(t) - h^2 B_m(t)\right] \stackrel{d}{\to} \mathcal{N}\left(0,V_m(t)\right)$$
	with $V_m(t) = \mathbb{E}\left[\psi_{h,t}^2\left(Y,T,\bm{S}; \bar{\mu}, \bar{p}_{T|\bm{S}}\right)\right]$ and
	\begin{align*}
		B_m(t) &= 
		\begin{cases}
			\frac{\kappa_2}{2}\cdot \mathbb{E}_{\bm{S}}\left\{\frac{2\frac{\partial}{\partial t}\mu(t,\bm{S})\cdot \frac{\partial}{\partial t}p_{T|\bm{S}}(t|\bm{S}) + p_{T|\bm{S}}(t|\bm{S})\left[\frac{\partial^2}{\partial t^2}\mu(t,\bm{S}) - 2 \frac{\partial}{\partial t}\log \bar{p}_{T|\bm{S}}(t|\bm{S}) \cdot \frac{\partial}{\partial t}\mu(t,\bm{S})\right]}{\bar{p}_{T|\bm{S}}(t|\bm{S})}\right\} \quad \text{ when } \bar{\mu}=\mu,\\
			\frac{\kappa_2}{2}\cdot \mathbb{E}_{\bm{S}}\left[\frac{\partial^2}{\partial t^2} \mu(t,\bm{S}) \right] \quad\quad  \text{ when } \bar{p}_{T|\bm{S}} = p_{T|\bm{S}}.
		\end{cases}
	\end{align*}
\end{proposition}

\begin{remark}[Uniform asymptotic theory for estimating $m(t)$]
If we assume that 
\[
\begin{cases}
\sup_{t\in \mathcal{T}}\norm{\hat{\mu}(t,\bm{S}) - \bar{\mu}(t,\bm{S})}_{L_2}=O_P(\Upsilon_{1,n}),\\
\sup_{t\in \mathcal{T}}\norm{\hat{p}_{T|\bm{S}}(t|\bm{S}) - \bar{p}_{T|\bm{S}}(t|\bm{S})}_{L_2}=O_P\left(\Upsilon_{2,n}\right),\\
\sqrt{nh} \sup_{t\in \mathcal{T}} \norm{\hat{\mu}(t,\bm{S}) - \mu(t,\bm{S})}_{L_2} \norm{\hat{p}_{T|\bm{S}}(t|\bm{S}) - p_{T|\bm{S}}(t|\bm{S})}_{L_2}=o_P(1),
\end{cases}
\]
then the pointwise convergence results in Proposition~\ref{prop:m_pos} can be strengthened to the uniform ones; see our side notes in the proof below.
\end{remark}

\begin{proof}[Proof of Proposition~\ref{prop:m_pos}]
	We derive the rates of convergence of $\hat{m}_{\mathrm{RA}}(t)$ given by \eqref{m_RA} and $\hat{m}_{\mathrm{IPW}}(t)$ given by \eqref{m_IPW} in \autoref{app:m_RA_pos} and \autoref{app:m_IPW_pos}, respectively. We also prove the asymptotic linearity, double robustness, and asymptotic normality of $\hat{m}_{\mathrm{DR}}(t)$ given by \eqref{m_DR} in \autoref{app:m_DR_pos}.
	
	\subsection{Rate of Convergence of $\hat{m}_{\mathrm{RA}}(t)$}
	\label{app:m_RA_pos}
	
	Firstly, we derive the rate of convergence for $\hat{m}_{\mathrm{RA}}(t)$ in \eqref{m_RA}. Notice that
	\begin{align*}
		\hat{m}_{\mathrm{RA}}(t) - m(t) &= \frac{1}{n} \sum_{i=1}^n \hat{\mu}(t,\bm{S}_i) - \mathbb{E}\left[\mu(t,\bm{S}_1)\right]\\
		&= \underbrace{\frac{1}{n} \sum_{i=1}^n \left[\hat{\mu}(t,\bm{S}_i) - \bar{\mu}(t,\bm{S}_i) \right]}_{\textbf{Term I}} + \underbrace{\frac{1}{n} \sum_{i=1}^n \left[\bar{\mu}(t,\bm{S}_i) - \mu(t,\bm{S}_i) \right]}_{\textbf{Term II}} + \underbrace{\frac{1}{n} \sum_{i=1}^n \left\{\mu(t,\bm{S}_i) - \mathbb{E}\left[\mu(t,\bm{S}_i) \right] \right\}}_{\textbf{Term III}}.
	\end{align*}
	
	$\bullet$ \textbf{Term I:} By Markov's inequality (and H\"older's inequality), we know that
	\begin{align*}
		\textbf{Term I} &\leq \frac{1}{n} \sum_{i=1}^n \left|\hat{\mu}(t,\bm{S}_i) - \bar{\mu}(t,\bm{S}_i) \right| \\
		&= O_P\left(\mathbb{E}\left|\hat{\mu}(t,\bm{S}_1) - \bar{\mu}(t,\bm{S}_1) \right|\right) = O_P\left(\norm{\hat{\mu}(t,\bm{S}) - \bar{\mu}(t,\bm{S})}_{L_2}\right) = O_P\left(\Upsilon_{1,n}\right).
	\end{align*}
	
	$\bullet$ \textbf{Term II:} We similarly derive that
	\begin{align*}
		\textbf{Term II} &\leq \frac{1}{n} \sum_{i=1}^n \left|\bar{\mu}(t,\bm{S}_i) - \mu(t,\bm{S}_i) \right| = O_P\left(\norm{\bar{\mu}(t,\bm{S}) - \mu(t,\bm{S})}_{L_2}\right).
	\end{align*}
	
	$\bullet$ \textbf{Term III:} By the central limit theorem and the boundedness of $\mu(t,\bm{s})$ on $\mathcal{T}\times \mathcal{S}$ under Assumption~\ref{assump:reg_diff}, we know that
	$$\textbf{Term III} = \frac{1}{n} \sum_{i=1}^n \left\{\mu(t,\bm{S}_i) - \mathbb{E}\left[\mu(t,\bm{S}_i) \right] \right\} = O_P\left(\frac{1}{\sqrt{n}}\right).$$
	As a side note, under Assumption~\ref{assump:reg_diff}, we know that $\left|\mu(t_1,\bm{s}) - \mu(t_2,\bm{s})\right| \leq A_1 |t_1-t_2|$ for some absolute constant $A_1>0$. Together with the compactness of $\mathcal{T}$ and Example 19.7 in \cite{VDV1998}, we also deduce that
	$$ \sup_{t\in \mathcal{T}}\left|\frac{1}{n} \sum_{i=1}^n \left\{\mu(t,\bm{S}_i) - \mathbb{E}\left[\mu(t,\bm{S}_i) \right] \right\} \right| = O_P\left(\frac{1}{\sqrt{n}}\right).$$
	
	In summary, we conclude that
	$$\hat{m}_{\mathrm{RA}}(t) - m(t) = O_P\left(\Upsilon_{1,n} + \norm{\bar{\mu}(t,\bm{S}) - \mu(t,\bm{S})}_{L_2} + \frac{1}{\sqrt{n}}\right).$$
	
	\subsection{Rate of Convergence of $\hat{m}_{\mathrm{IPW}}(t)$}
	\label{app:m_IPW_pos}
	
	Secondly, we derive the rate of convergence for $\hat{m}_{\mathrm{IPW}}(t)$ in \eqref{m_IPW}. Note that
	\begin{align*}
		&\hat{m}_{\mathrm{IPW}}(t) - m(t) \\
		&= \tilde{m}_{\mathrm{IPW}}(t) - m(t) + \hat{m}_{\mathrm{IPW}}(t) - \tilde{m}_{\mathrm{IPW}}(t)\\
		&= \underbrace{\frac{1}{nh}\sum_{i=1}^n \frac{K\left(\frac{T_i-t}{h}\right)}{p_{T|\bm{S}}(T_i|\bm{S}_i)}\cdot Y_i - \mathbb{E}\left[\mu(t,\bm{S}_1)\right]}_{\textbf{Term IV}} + \underbrace{\frac{1}{nh}\sum_{i=1}^n \frac{K\left(\frac{T_i-t}{h}\right)}{\hat{p}_{T|\bm{S}}(T_i|\bm{S}_i)}\cdot Y_i - \frac{1}{nh}\sum_{i=1}^n \frac{K\left(\frac{T_i-t}{h}\right)}{p_{T|\bm{S}}(T_i|\bm{S}_i)}\cdot Y_i }_{\textbf{Term V}},
	\end{align*}
	where $\tilde{m}_{\mathrm{IPW}}(t) = \frac{1}{nh}\sum_{i=1}^n \frac{K\left(\frac{T_i-t}{h}\right)}{p_{T|\bm{S}}(T_i|\bm{S}_i)}Y_i$ is the oracle IPW estimator of $m(t)$ defined in \eqref{IPW_oracle}. We shall handle \textbf{Term IV} and \textbf{Term V} in \autoref{app:m_pos_term_IV} and \autoref{app:m_pos_term_V}, respectively.
	
	\subsubsection{Rate of Convergence of \textbf{Term IV} for $\hat{m}_{\mathrm{IPW}}(t)$}
	\label{app:m_pos_term_IV}
	
	Under Assumptions~\ref{assump:reg_diff} and \ref{assump:reg_kernel}, we calculate the bias of $\tilde{m}_{\mathrm{IPW}}(t)$ as:
	\begin{align*}
		&\mathbb{E}\left[\tilde{m}_{\mathrm{IPW}}(t)\right] - m(t) \\
		&= \mathbb{E}\left[\frac{1}{h} \frac{K\left(\frac{T_i-t}{h}\right)}{p_{T|\bm{S}}(T_i|\bm{S}_i)}\cdot Y_i\right] - \mathbb{E}\left[\mu(t,\bm{S}_1)\right]\\
		&= \frac{1}{h} \int_{\mathcal{S}\times \mathcal{T}} \frac{K\left(\frac{t_1-t}{h}\right)}{p_{T|\bm{S}}(t_1|\bm{s}_1)} \cdot \mu(t_1,\bm{s}_1) \cdot p(t_1,\bm{s}_1) \, dt_1 d\bm{s}_1 - \mathbb{E}\left[\mu(t,\bm{S}_1)\right] \\
		&\stackrel{\text{(i)}}{=} \int_{\mathcal{S}} \int_{\mathbb{R}} K(u)\cdot \mu(t+uh,\bm{s}_1) \cdot p_S(\bm{s}_1) \, du d\bm{s}_1 - \mathbb{E}\left[\mu(t,\bm{S}_1)\right]\\
		& \stackrel{\text{(ii)}}{=} \int_{\mathcal{S}} \int_{\mathbb{R}} K(u) \left[\mu(t,\bm{s}_1) + uh\cdot \frac{\partial}{\partial t}\mu(t,\bm{s}_1) + \frac{u^2h^2}{2} \frac{\partial^2}{\partial t^2}\mu(\tilde{t},\bm{s}_1) \right] p_S(\bm{s}_1) \, du d\bm{s}_1 - \mathbb{E}\left[\mu(t,\bm{S}_1)\right] \\
		&\stackrel{\text{(iii)}}{=} \int_{\mathcal{S}} \mu(t,\bm{s}_1) \cdot p_S(\bm{s}_1)\, d\bm{s}_1 - \mathbb{E}\left[\mu(t,\bm{S}_1)\right] + \int_{\mathcal{S}} \int_{\mathbb{R}} K(u)\cdot \frac{u^2h^2}{2} \cdot \frac{\partial^2}{\partial t^2}\mu(\tilde{t},\bm{s}_1) \cdot p_S(\bm{s}_1) \, du d\bm{s}_1\\
		&= O(h^2),
	\end{align*}
	where (i) uses a change of variable $u = \frac{t_1-t}{h}$, (ii) applies Taylor's expansion with some $\tilde{t}$ that lies between $t$ and $t+uh$, and (iii) utilizes the properties of the second-order kernel function $K$. Similarly, we compute the variance of $\tilde{m}_{\mathrm{IPW}}(t)$ as:
	\begin{align*}
		&\mathrm{Var}\left[\tilde{m}_{\mathrm{IPW}}(t)\right] \\
		&= \frac{1}{nh^2}\cdot \mathrm{Var}\left[\frac{K\left(\frac{T_i-t}{h}\right)}{p_{T|\bm{S}}(T_i|\bm{S}_i)}Y_i\right] \\
		&= \frac{1}{nh^2}\cdot \mathbb{E}\left[\frac{K^2\left(\frac{T_i-t}{h}\right)}{\left[p_{T|\bm{S}}(T_i|\bm{S}_i)\right]^2}\cdot Y_i^2 \right] - \frac{1}{nh^2} \left\{\mathbb{E}\left[\frac{K\left(\frac{T_i-t}{h}\right)}{p_{T|\bm{S}}(T_i|\bm{S}_i)}\cdot Y_i\right]\right\}^2\\
		&= \frac{1}{nh^2} \int_{\mathcal{S}\times \mathcal{T}} \frac{K^2\left(\frac{t_1-t}{h}\right)}{\left[p_{T|\bm{S}}(t_1|\bm{s}_1)\right]^2} \cdot \left[\mu(t_1,\bm{s}_1)^2 +\sigma^2\right] \cdot p(t_1,\bm{s}_1)\, dt_1 d\bm{s}_1 - \frac{\left\{\mathbb{E}\left[\mu(t,\bm{S}_1)\right]\right\}^2}{n} + O\left(\frac{h^2}{n}\right)\\
		&\stackrel{\text{(i)}}{=} \frac{1}{nh} \int_{\mathcal{S}} \int_{\mathbb{R}} \frac{K^2(u)}{p_{T|\bm{S}}(t+uh|\bm{s}_1)} \cdot \left[\mu(t+uh,\bm{s}_1)^2 +\sigma^2\right] p_S(\bm{s}_1)\, du d\bm{s}_1 + O\left(\frac{1}{n}\right) \\
		&\stackrel{\text{(ii)}}{=} \frac{1}{nh} \int_{\mathcal{S}} \int_{\mathbb{R}} \frac{K^2(u)}{p_{T|\bm{S}}(t|\bm{s}_1) + uh \cdot \frac{\partial}{\partial t}p_{T|\bm{S}}(t'|\bm{s}_1)} \left[\mu(t,\bm{s}_1)^2 + 2uh\cdot \mu(t'',\bm{s}_1)\cdot \frac{\partial}{\partial t} \mu(t'',\bm{s}_1) + \sigma^2\right] p_S(\bm{s}_1)\, du d\bm{s}_1 \\
		&\quad + O\left(\frac{1}{n}\right)\\
		&\stackrel{\text{(iii)}}{=} \frac{1}{nh} \int_{\mathcal{S}} \int_{\mathbb{R}} \frac{K^2(u)}{p_{T|\bm{S}}(t|\bm{s}_1)} \cdot \left[\mu(t,\bm{s}_1)^2 + \sigma^2\right] p_S(\bm{s}_1)\, du d\bm{s}_1 + O\left(\frac{1}{n}\right)\\
		&\stackrel{\text{(iv)}}{=} O\left(\frac{1}{nh}\right),
	\end{align*}
	where (i) uses a change of variable $u = \frac{t_1-t}{h}$ and the boundedness of $\mu(t,\bm{s})$, (ii) applies the Taylor's expansion under Assumptions~\ref{assump:reg_diff} and \ref{assump:den_diff} with $t',t''$ being two points between $t$ and $t+uh$, (iii) absorbs the higher order terms to $O\left(\frac{1}{n}\right)$, and (iv) utilizes the properties of $K$ under Assumption~\ref{assump:reg_kernel} and the positivity condition (Assumption~\ref{assump:positivity}). Now, by Chebyshev's inequality and our above calculations, we obtain that
	\begin{align*}
		\tilde{m}_{\mathrm{IPW}}(t) - m(t) &= \tilde{m}_{\mathrm{IPW}}(t) - \mathbb{E}\left[\tilde{m}_{\mathrm{IPW}}(t)\right] + \mathbb{E}\left[\tilde{m}_{\mathrm{IPW}}(t)\right] - m(t)\\
		&= O_P\left(\sqrt{\mathrm{Var}\left[\tilde{m}_{\mathrm{IPW}}(t)\right]}\right) + O(h^2)\\
		&= O_P\left(\sqrt{\frac{1}{nh}}\right) + O(h^2)
	\end{align*}
	as $h\to 0$ and $nh\to \infty$. As a side note, under the VC-type condition on $K$ (Assumption~\ref{assump:reg_kernel}(c)), we can apply Theorem 2 in \cite{einmahl2005uniform} to strengthen the above pointwise rate of convergence to the uniform one as:
	$$\sup_{t\in \mathcal{T}}\left|\tilde{m}_{\mathrm{IPW}}(t) - m(t) \right| = O_P\left(\sqrt{\frac{|\log h|}{nh}}\right) + O(h^2).$$
	
	\subsubsection{Rate of Convergence of \textbf{Term V} for $\hat{m}_{\mathrm{IPW}}(t)$}
	\label{app:m_pos_term_V}
	
	By direct calculations, we have that
	\begin{align*}
		&\textbf{Term V} \\
		&= \frac{1}{nh}\sum_{i=1}^n \frac{K\left(\frac{T_i-t}{h}\right)}{p_{T|\bm{S}}(T_i|\bm{S}_i)} \cdot Y_i \left[\frac{p_{T|\bm{S}}(T_i|\bm{S}_i) - \hat{p}_{T|\bm{S}}(T_i|\bm{S}_i)}{\hat{p}_{T|\bm{S}}(T_i|\bm{S}_i)}\right]\\
		&= \frac{1}{nh}\sum_{i=1}^n \frac{K\left(\frac{T_i-t}{h}\right)}{p_{T|\bm{S}}(T_i|\bm{S}_i)} \cdot Y_i  \left\{\frac{p_{T|\bm{S}}(T_i|\bm{S}_i) -\bar{p}_{T|\bm{S}}(T_i|\bm{S}_i) + \bar{p}_{T|\bm{S}}(T_i|\bm{S}_i) - \hat{p}_{T|\bm{S}}(T_i|\bm{S}_i)}{p_{T|\bm{S}}(T_i|\bm{S}_i) - \left[p_{T|\bm{S}}(T_i|\bm{S}_i) - \bar{p}_{T|\bm{S}}(T_i|\bm{S}_i)\right] - \left[\bar{p}_{T|\bm{S}}(T_i|\bm{S}_i) - \hat{p}_{T|\bm{S}}(T_i|\bm{S}_i)\right]}\right\}\\
		&\stackrel{\text{(i)}}{=}\left\{\mathbb{E}\left[\mu(t,\bm{S}) \right] + O(h^2) + O_P\left(\sqrt{\frac{1}{nh}}\right)\right\}\\
		&\quad \times \frac{O_P\left(\sup\limits_{|u-t|\leq h}\norm{\hat{p}_{T|\bm{S}}(u|\bm{S}) - \bar{p}_{T|\bm{S}}(u|\bm{S})}_{L_2} + \sup\limits_{|u-t|\leq h}\norm{\bar{p}_{T|\bm{S}}(u|\bm{S}) - p_{T|\bm{S}}(u|\bm{S})}_{L_2} \right)}{\inf_{(t,\bm{s})\in \mathcal{T}\times \mathcal{S}} p_{T|\bm{S}}(t|\bm{s}) - O_P\left(\sup\limits_{|u-t|\leq h}\norm{\hat{p}_{T|\bm{S}}(u|\bm{S}) - \bar{p}_{T|\bm{S}}(u|\bm{S})}_{L_2} + \sup\limits_{|u-t|\leq h}\norm{\bar{p}_{T|\bm{S}}(u|\bm{S}) - p_{T|\bm{S}}(u|\bm{S})}_{L_2} \right)}\\
		&= O_P\left(\Upsilon_{2,n} + \sup\limits_{|u-t|\leq h}\norm{\bar{p}_{T|\bm{S}}(u|\bm{S}) - p_{T|\bm{S}}(u|\bm{S})}_{L_2}\right) \left[O(1 + h^2) + O_P\left(\sqrt{\frac{1}{nh}}\right)\right]\\
		&= O_P\left(\Upsilon_{2,n} + \sup\limits_{|u-t|\leq h}\norm{\bar{p}_{T|\bm{S}}(u|\bm{S}) - p_{T|\bm{S}}(u|\bm{S})}_{L_2}\right)
	\end{align*}
	as $h\to 0$ and $nh\to \infty$, where (i) utilizes our results for \textbf{Term IV} and Markov's inequality.\\
	
	Combining our results for \textbf{Term IV} and \textbf{Term V} in \autoref{app:m_pos_term_IV} and \autoref{app:m_pos_term_V}, we conclude that 
	$$\hat{m}_{\mathrm{IPW}}(t) - m(t) = O(h^2) + O_P\left(\sqrt{\frac{1}{nh}} + \Upsilon_{2,n} + \sup\limits_{|u-t|\leq h}\norm{\bar{p}_{T|\bm{S}}(u|\bm{S}) - p_{T|\bm{S}}(u|\bm{S})}_{L_2}\right).$$
	
	\subsection{Asymptotic Properties and Double Robustness of $\hat{m}_{\mathrm{DR}}(t)$}
	\label{app:m_DR_pos}
	
	Finally, we establish the double robustness and asymptotic properties of $\hat{m}_{\mathrm{DR}}(t)$ in \eqref{m_DR}. Some parts of the following proof are inspired by the proof of Theorem 3.1 in \cite{colangelo2020double}. Notice that under Assumption~\ref{assump:id_cond},
	\begin{align*}
		&\hat{m}_{\mathrm{DR}}(t) - m(t) \\
		&= \frac{1}{nh}\sum_{i=1}^n \left\{\frac{K\left(\frac{T_i-t}{h}\right)}{\hat{p}_{T|\bm{S}}(T_i|\bm{S}_i)}\cdot \left[Y_i - \hat \mu(t,\bm{S}_i)\right]+ h\cdot \hat{\mu}(t,\bm{S}_i) \right\} - \mathbb{E}\left[\mu(t,\bm{S})\right] \\
		&= \mathbb{P}_n \Psi_{h,t}\left(Y,T,\bm{S}; \bar{\mu}, \bar{p}_{T|\bm{S}}\right) - \mathbb{E}\left[\mu(t,\bm{S})\right] + \mathbb{P}_n \left[\Psi_{h,t}\left(Y,T,\bm{S}; \hat{\mu}, \hat{p}_{T|\bm{S}}\right) - \Psi_{h,t}\left(Y,T,\bm{S}; \bar{\mu}, \bar{p}_{T|\bm{S}}\right) \right] \\
		&= \underbrace{\mathbb{P}_n \Psi_{h,t}\left(Y,T,\bm{S}; \bar{\mu}, \bar{p}_{T|\bm{S}}\right) - \mathbb{E}\left[\mu(t,\bm{S})\right]}_{\textbf{Term VI}} \\
		&\quad + \underbrace{\left(\mathbb{P}_n - \P\right)\left[\hat{\mu}(t,\bm{S}) - \bar{\mu}(t,\bm{S})\right]}_{\textbf{Term VII}} + \underbrace{\left(\mathbb{P}_n -\P\right)\left\{\frac{K\left(\frac{T-t}{h}\right)}{h}\left[\frac{1}{\hat{p}_{T|\bm{S}}(T|\bm{S})} - \frac{1}{\bar{p}_{T|\bm{S}}(T|\bm{S})}\right]\left[Y -\bar{\mu}(t,\bm{S})\right]\right\}}_{\textbf{Term VIII}}\\
		&\quad + \underbrace{\left(\mathbb{P}_n -\P\right)\left\{\frac{K\left(\frac{T-t}{h}\right)}{h\cdot \bar{p}_{T|\bm{S}}(T|\bm{S})}\left[\bar{\mu}(t,\bm{S}) - \hat{\mu}(t,\bm{S})\right]\right\}}_{\textbf{Term IX}} \\
		&\quad + \underbrace{\mathbb{P}_n\left\{\frac{K\left(\frac{T-t}{h}\right)}{h}\left[\frac{1}{\hat{p}_{T|\bm{S}}(T|\bm{S})} - \frac{1}{\bar{p}_{T|\bm{S}}(T|\bm{S})}\right]\left[\bar{\mu}(t,\bm{S}) - \hat{\mu}(t,\bm{S})\right]\right\}}_{\textbf{Term X}}\\
		&+ \underbrace{\P\left\{\left[1- \frac{K\left(\frac{T-t}{h}\right)}{h \cdot \bar{p}_{T|\bm{S}}(T|\bm{S})}\right] \left[\hat{\mu}(t,\bm{S}) - \bar{\mu}(t,\bm{S})\right] \right\} + \P\left\{\frac{K\left(\frac{T-t}{h}\right)}{h}\left[\frac{1}{\hat{p}_{T|\bm{S}}(T|\bm{S})} - \frac{1}{\bar{p}_{T|\bm{S}}(T|\bm{S})}\right]\left[Y -\bar{\mu}(t,\bm{S})\right]\right\}}_{\textbf{Term XI}},
	\end{align*}
	where $\Psi_{h,t}\left(Y,T,\bm{S}; \mu,p_{T|\bm{S}}\right) = \frac{K\left(\frac{T-t}{h}\right)}{h\cdot p_{T|\bm{S}}(T|\bm{S})}\cdot \left[Y-\mu(t,\bm{S})\right] + \mu(t,\bm{S})$. It remains to show that the dominating \textbf{Term VI} is of order $O(h^2)+O_P\left(\sqrt{\frac{1}{nh}}\right)$ in \autoref{app:m_pos_Term_VI} and the remainder terms are of order $o_P\left(\sqrt{\frac{1}{nh}}\right)$ for any $t\in \mathcal{T}$ in \autoref{app:m_pos_Term_VII}, \autoref{app:m_pos_Term_VIII}, \autoref{app:m_pos_Term_X}, and \autoref{app:m_pos_Term_XI}. We shall also derive the asymptotic normality of $\hat{m}_{\mathrm{DR}}(t)$ in \autoref{app:m_pos_asym_norm}.
	
	\subsubsection{Analysis of \textbf{Term VI} for $\hat{m}_{\mathrm{DR}}(t)$}
	\label{app:m_pos_Term_VI}
	
	We analyze the variance and bias of \textbf{Term VI} separately as follows. Notice that
	\begin{align*}
		&\mathrm{Var}\left[\textbf{Term VI}\right]\\
		&=\mathrm{Var}\left[\mathbb{P}_n \Psi_{h,t}\left(Y,T,\bm{S}; \bar{\mu}, \bar{p}_{T|\bm{S}}\right) \right] \\
		&= \frac{1}{nh^2} \mathrm{Var}\left[\frac{K\left(\frac{T-t}{h}\right)}{\bar{p}_{T|\bm{S}}(T|\bm{S})}\cdot \left[Y-\bar{\mu}(t,\bm{S})\right] + h\cdot \bar{\mu}(t,\bm{S}) \right]\\
		&\stackrel{\text{(i)}}{\lesssim} \frac{1}{nh^2} \mathrm{Var}\left[\frac{K\left(\frac{T-t}{h}\right)}{\bar{p}_{T|\bm{S}}(T|\bm{S})}\cdot \left[Y-\bar{\mu}(t,\bm{S})\right] \right] + \frac{1}{n} \mathrm{Var}\left[\bar{\mu}(t,\bm{S})\right]\\
		&\stackrel{\text{(ii)}}{=} \frac{1}{nh^2}\cdot \mathbb{E}\left[\frac{K^2\left(\frac{T-t}{h}\right)}{\bar{p}_{T|\bm{S}}^2(T|\bm{S})}\cdot \left[Y-\bar{\mu}(t,\bm{S})\right]^2 \right] + O\left(\frac{1}{n}\right)\\
		&= \frac{1}{nh^2} \int_{\mathcal{S}} \int_{\mathcal{T}} \frac{K^2\left(\frac{t_1-t}{h}\right)}{\bar{p}_{T|\bm{S}}^2(t_1|\bm{s}_1)}\left\{\left[\mu(t_1,\bm{s}_1)-\bar{\mu}(t,\bm{s}_1)\right]^2 + \sigma^2 \right\} p(t_1,\bm{s}_1)\, dt_1 d\bm{s}_1 + O\left(\frac{1}{n}\right)\\
		&\stackrel{\text{(iii)}}{=} \frac{1}{nh} \int_{\mathcal{S}} \int_{\mathbb{R}} \frac{K^2(u)}{\bar{p}_{T|\bm{S}}^2(t+uh|\bm{s}_1)}\left\{\left[\mu(t+uh,\bm{s}_1)-\bar{\mu}(t,\bm{s}_1)\right]^2 + \sigma^2 \right\} p(t+uh,\bm{s}_1)\, du d\bm{s}_1 + O\left(\frac{1}{n}\right)\\
		& \stackrel{\text{(iv)}}{=} \frac{1}{nh} \int_{\mathcal{S}} \int_{\mathbb{R}} \frac{K^2(u)}{\bar{p}_{T|\bm{S}}^2(t|\bm{s}_1) + O(h^2)}\left\{\left[\mu(t,\bm{s}_1)-\bar{\mu}(t,\bm{s}_1)\right]^2 +O(h^2)+ \sigma^2 \right\} \left[p(t,\bm{s}_1)+O(h)\right] du d\bm{s}_1 + O\left(\frac{1}{n}\right)\\
		&=O\left(\frac{1}{nh}\right),
	\end{align*}
	where (i) uses Cauchy-Schwarz inequality on the covariance, (ii) uses the boundedness of $\bar{\mu}$ under Assumption~\ref{assump:reg_diff} to derive the term $O\left(\frac{1}{n}\right)$, (iii) leverages a change of variable $u=\frac{t_1-t}{h}$, and (iv) applies the Taylor's expansion under Assumptions~\ref{assump:reg_diff} and \ref{assump:den_diff}. In the above calculations, we also note from the line (i) that the second part $\bar{\mu}(t,\bm{S})$ of $\Psi_{h,t}\left(Y,T,\bm{S}; \bar{\mu}, \bar{p}_{T|\bm{S}}\right)$ is of smaller order than the first term $\frac{K\left(\frac{T-t}{h}\right)}{h\cdot \bar{p}_{T|\bm{S}}(T|\bm{S})}\left[Y-\bar{\mu}(t,\bm{S})\right]$. Thus, we can only keep the first term in the final asymptotically linear form of $\hat{m}_{\mathrm{DR}}(t)$. Now, by Chebyshev's inequality, we conclude that
	\begin{align*}
		\left(\mathbb{P}_n -\P\right) \Psi_{h,t}\left(Y,T,\bm{S}; \bar{\mu}, \bar{p}_{T|\bm{S}}\right) &= O_P\left(\sqrt{\mathrm{Var}\left[\mathbb{P}_n\Psi_{h,t}\left(Y,T,\bm{S}; \bar{\mu}, \bar{p}_{T|\bm{S}}\right) \right]}\right) \\
		&=O_P\left(\sqrt{\mathrm{Var}\left[\frac{1}{\sqrt{h}}\cdot \mathbb{P}_n\psi_{h,t}\left(Y,T,\bm{S}; \bar{\mu}, \bar{p}_{T|\bm{S}}\right) \right]}\right) = O_P\left(\sqrt{\frac{1}{nh}}\right),
	\end{align*}
	where $\psi_{h,t}\left(Y,T,\bm{S}; \bar{\mu}, \bar{p}_{T|\bm{S}}\right) = \frac{K\left(\frac{T-t}{h}\right)}{\sqrt{h} \cdot \bar{p}_{T|\bm{S}}(T|\bm{S})}\left[Y-\bar{\mu}(t,\bm{S})\right]$. In addition, by direct calculations and Taylor's expansions, we derive that
	\begin{align*}
		&\mathrm{Bias}\left[\textbf{Term VI}\right] \\
		&= \P\left[\Psi_{h,t}\left(Y,T,\bm{S}; \bar{\mu}, \bar{p}_{T|\bm{S}}\right) \right] - \mathbb{E}\left[\mu(t,\bm{S}_1)\right]\\
		&= \mathbb{E}\left[\frac{K\left(\frac{T-t}{h}\right)}{h \cdot \bar{p}_{T|\bm{S}}(T|\bm{S})}\left[Y-\bar{\mu}(t,\bm{S})\right] \right] + \mathbb{E}\left[\bar{\mu}(t,\bm{S}_1) - \mu(t,\bm{S}_1)\right]\\
		&= \int_{\mathcal{S}} \int_{\mathcal{T}} \frac{K\left(\frac{t_1-t}{h}\right)}{h \cdot \bar{p}_{T|\bm{S}}(t_1|\bm{s}_1)}\left[\mu(t_1,\bm{s}_1)-\bar{\mu}(t,\bm{s}_1)\right] p(t_1,\bm{s}_1) \, dt_1d\bm{s}_1 + \mathbb{E}\left[\bar{\mu}(t,\bm{S}_1) - \mu(t,\bm{S}_1)\right]\\
		&= \int_{\mathcal{S}} \int_{\mathbb{R}} \frac{K(u)}{\bar{p}_{T|\bm{S}}(t+uh|\bm{s}_1)}\left[\mu(t+uh,\bm{s}_1)-\bar{\mu}(t,\bm{s}_1)\right] p(t+uh,\bm{s}_1) \, du d\bm{s}_1 + \mathbb{E}\left[\bar{\mu}(t,\bm{S}_1) - \mu(t,\bm{S}_1)\right]\\
		&= \int_{\mathcal{S}} \int_{\mathbb{R}} \frac{K(u) \left[\mu(t,\bm{s}_1)-\bar{\mu}(t,\bm{s}_1) + uh \frac{\partial}{\partial t}\mu(t,\bm{s}_1) +\frac{u^2h^2}{2} \frac{\partial^2}{\partial t^2}\mu(t,\bm{s}_1) + O(h^3)\right]}{\bar{p}_{T|\bm{S}}(t|\bm{s}_1) + uh\frac{\partial}{\partial t} \bar{p}_{T|\bm{S}}(t|\bm{s}_1) + \frac{u^2h^2}{2} \cdot \frac{\partial^2}{\partial t^2}\bar{p}_{T|\bm{S}}(t|\bm{s}_1) + O(h^3)} \\
		&\quad\quad \times \left[p(t,\bm{s}_1) + uh \frac{\partial}{\partial t}p(t,\bm{s}_1) + \frac{u^2h^2}{2} \frac{\partial^2}{\partial t^2}p(t,\bm{s}_1) +O(h^3)\right] \, du d\bm{s}_1 + \mathbb{E}\left[\bar{\mu}(t,\bm{S}_1) - \mu(t,\bm{S}_1)\right]\\
		&= \int_{\mathcal{S}} \int_{\mathbb{R}} K(u) \left[\mu(t,\bm{s}_1)-\bar{\mu}(t,\bm{s}_1) + uh \frac{\partial}{\partial t}\mu(t,\bm{s}_1) +\frac{u^2h^2}{2} \frac{\partial^2}{\partial t^2}\mu(t,\bm{s}_1) + O(h^3)\right] \\
		&\quad \quad \times \left[\frac{1}{\bar{p}_{T|\bm{S}}(t|\bm{s}_1)} - \frac{uh\frac{\partial}{\partial t} \bar{p}_{T|\bm{S}}(t|\bm{s}_1)}{\bar{p}_{T|\bm{S}}^2(t|\bm{s}_1)} - \frac{u^2h^2 \cdot \frac{\partial^2}{\partial t^2}\bar{p}_{T|\bm{S}}(t|\bm{s}_1)}{2 \bar{p}_{T|\bm{S}}^2(t|\bm{s}_1)} + \frac{u^2 h^2\left[\frac{\partial}{\partial t} \bar{p}_{T|\bm{S}}(t|\bm{s}_1)\right]^2}{\bar{p}_{T|\bm{S}}^3(t|\bm{s}_1)} + O(h^3)\right]\\
		&\quad\quad \times \left[p(t,\bm{s}_1) + uh \frac{\partial}{\partial t}p(t,\bm{s}_1) + \frac{u^2h^2}{2} \frac{\partial^2}{\partial t^2}p(t,\bm{s}_1) +O(h^3)\right] \, du d\bm{s}_1 + \mathbb{E}\left[\bar{\mu}(t,\bm{S}_1) - \mu(t,\bm{S}_1)\right]\\
		&= \int_{\mathcal{S}} \frac{\left[\mu(t,\bm{s}_1)-\bar{\mu}(t,\bm{s}_1)\right]}{\bar{p}_{T|\bm{S}}(t|\bm{s}_1)} \cdot p(t,\bm{s}_1) \, d\bm{s}_1 + \int_{\mathcal{S}} \left[\bar{\mu}(t,\bm{s}_1)- \mu(t,\bm{s}_1)\right] p_S(\bm{s}_1)\, d\bm{s}_1 \\
		&\quad + \frac{h^2\kappa_2}{2} \int_{\mathcal{S}} \Bigg\{\frac{\frac{\partial^2}{\partial t^2} \mu(t,\bm{s}_1)}{\bar{p}_{T|\bm{S}}(t|\bm{s}_1)} \cdot p(t,\bm{s}_1) + \frac{2\left[\frac{\partial}{\partial t} \bar{p}_{T|\bm{S}}(t|\bm{s}_1)\right]^2/\bar{p}_{T|\bm{S}}(t|\bm{s}_1) - \frac{\partial^2}{\partial t^2} \bar{p}_{T|\bm{S}}(t|\bm{s}_1)}{\bar{p}_{T|\bm{S}}^2(t|\bm{s}_1)}\left[\bar{\mu}(t,\bm{s}_1)- \mu(t,\bm{s}_1)\right] p(t,\bm{s}_1) \\
		&\quad\quad + \frac{\left[\bar{\mu}(t,\bm{s}_1)- \mu(t,\bm{s}_1)\right]}{\bar{p}_{T|\bm{S}}(t|\bm{s}_1)}\cdot \frac{\partial^2}{\partial t^2} p(t,\bm{s}_1) - \frac{2\left[\frac{\partial}{\partial t}\mu(t,\bm{s}_1)\right]\left[\frac{\partial}{\partial t}\bar{p}_{T|\bm{S}}(t|\bm{s}_1)\right]}{\bar{p}_{T|\bm{S}}^2(t|\bm{s}_1)} \cdot p(t,\bm{s}_1) + \frac{2\left[\frac{\partial}{\partial t}\mu(t,\bm{s}_1)\right]\left[\frac{\partial}{\partial t}p(t,\bm{s}_1)\right]}{\bar{p}_{T|\bm{S}}(t|\bm{s}_1)}\\
		&\quad\quad - \frac{2\left[\frac{\partial}{\partial t}p(t,\bm{s}_1)\right]\left[\frac{\partial}{\partial t}\bar{p}_{T|\bm{S}}(t|\bm{s}_1)\right] \left[\bar{\mu}(t,\bm{s}_1)- \mu(t,\bm{s}_1)\right]}{\bar{p}_{T|\bm{S}}^2(t|\bm{s}_1)}\Bigg\} d\bm{s}_1 + o(h^2)\\
		&= \int_{\mathcal{S}} \frac{\left[\mu(t,\bm{s}_1)-\bar{\mu}(t,\bm{s}_1)\right]\left[p_{T|\bm{S}}(t|\bm{s}_1) - \bar{p}_{T|\bm{S}}(t|\bm{s}_1)\right]}{\bar{p}_{T|\bm{S}}(t|\bm{s}_1) \cdot p_{T|\bm{S}}(t|\bm{s}_1)} \cdot p(t,\bm{s}_1) \, d\bm{s}_1 + h^2B_m(t) + o(h^2)\\
		&= \mathbb{E}_{\bm{S}}\left\{\frac{\left[\mu(t,\bm{S})-\bar{\mu}(t,\bm{S})\right]\left[p_{T|\bm{S}}(t|\bm{S}) - \bar{p}_{T|\bm{S}}(t|\bm{S})\right]}{\bar{p}_{T|\bm{S}}(t|\bm{S})} \right\} + h^2B_m(t) + o(h^2),
	\end{align*}
	where the complicated bias term $B_m(t)$ is given by
	\begin{align*}
		B_m(t) &= \frac{\kappa_2}{2}\cdot \mathbb{E}_{\bm{S}}\Bigg\{\frac{\left[\mu(t,\bm{S})-\bar{\mu}(t,\bm{S})\right]}{\bar{p}_{T|\bm{S}}(t|\bm{S})}\Bigg[\frac{2\frac{\partial}{\partial t} \bar{p}_{T|\bm{S}}(t|\bm{S})\cdot \frac{\partial}{\partial t} \log\bar{p}_{T|\bm{S}}(t|\bm{S}) - \frac{\partial^2}{\partial t^2} \bar{p}_{T|\bm{S}}(t|\bm{S})}{\bar{p}_{T|\bm{S}}(t|\bm{S})} \cdot p_{T|\bm{S}}(t|\bm{S}) + \frac{\partial^2}{\partial t^2} p_{T|\bm{S}}(t|\bm{S}) \\
		&\quad\quad\quad\quad - \frac{2\left[\frac{\partial}{\partial t}p_{T|\bm{S}}(t|\bm{S})\right]\left[\frac{\partial}{\partial t}\bar{p}_{T|\bm{S}}(t|\bm{S})\right]}{\bar{p}_{T|\bm{S}}(t|\bm{S})}\Bigg] \Bigg\} \\
		&\quad + \frac{\kappa_2}{2}\cdot \mathbb{E}_{\bm{S}}\left\{\frac{\left[\bar{p}_{T|\bm{S}}(t|\bm{S}) \cdot \frac{\partial^2}{\partial t^2}\mu(t,\bm{S}) - 2\left[\frac{\partial}{\partial t}\mu(t,\bm{S})\right]\left[\frac{\partial}{\partial t}\bar{p}_{T|\bm{S}}(t|\bm{S})\right]\right]}{\bar{p}_{T|\bm{S}}^2(t|\bm{S})}\cdot p_{T|\bm{S}}(t|\bm{S}) \right\}\\
		&\quad + \kappa_2\cdot \mathbb{E}_{\bm{S}}\left[\frac{\frac{\partial}{\partial t}\mu(t,\bm{S})\cdot \frac{\partial}{\partial t}p_{T|\bm{S}}(t|\bm{S})}{\bar{p}_{T|\bm{S}}(t|\bm{S})}\right]
	\end{align*}
	and $\bar{p}(t,\bm{s}) = \bar{p}_{T|\bm{S}}(t|\bm{s})\cdot p_S(\bm{s})$. Under the condition that either $\bar{\mu} = \mu$ or $\bar{p}_{T|\bm{S}}=p_{T|\bm{S}}$, we have that
	$$\mathbb{E}_{\bm{S}}\left\{\frac{\left[\mu(t,\bm{S})-\bar{\mu}(t,\bm{S})\right]\left[p_{T|\bm{S}}(t|\bm{S}) - \bar{p}_{T|\bm{S}}(t|\bm{S})\right]}{\bar{p}_{T|\bm{S}}(t|\bm{S})} \right\} = 0$$
	and
	\begin{align*}
		B_m(t) &= 
		\begin{cases}
			\frac{\kappa_2}{2}\cdot \mathbb{E}_{\bm{S}}\left\{\frac{\left[\bar{p}_{T|\bm{S}}(t|\bm{S}) \cdot \frac{\partial^2}{\partial t^2}\mu(t,\bm{S}) - 2\left[\frac{\partial}{\partial t}\mu(t,\bm{S})\right]\left[\frac{\partial}{\partial t}\bar{p}_{T|\bm{S}}(t|\bm{S})\right]\right]}{\bar{p}_{T|\bm{S}}^2(t|\bm{S})}\cdot p_{T|\bm{S}}(t|\bm{S}) + \frac{2\frac{\partial}{\partial t}\mu(t,\bm{S})\cdot \frac{\partial}{\partial t}p_{T|\bm{S}}(t|\bm{S})}{\bar{p}_{T|\bm{S}}(t|\bm{S})}\right\} \, \text{ when } \bar{\mu}=\mu,\\
			\frac{\kappa_2}{2}\cdot \mathbb{E}_{\bm{S}}\left[\frac{\partial^2}{\partial t^2} \mu(t,\bm{S}) \right] \quad\quad  \text{ when } \bar{p}_{T|\bm{S}} = p_{T|\bm{S}},
		\end{cases}\\
		&= \begin{cases}
			\frac{\kappa_2}{2}\cdot \mathbb{E}_{\bm{S}}\left\{\frac{2\frac{\partial}{\partial t}\mu(t,\bm{S})\cdot \frac{\partial}{\partial t}p_{T|\bm{S}}(t|\bm{S}) + p_{T|\bm{S}}(t|\bm{S})\left[\frac{\partial^2}{\partial t^2}\mu(t,\bm{S}) - 2 \frac{\partial}{\partial t}\log \bar{p}_{T|\bm{S}}(t|\bm{S}) \cdot \frac{\partial}{\partial t}\mu(t,\bm{S})\right]}{\bar{p}_{T|\bm{S}}(t|\bm{S})}\right\} \quad \text{ when } \bar{\mu}=\mu,\\
			\frac{\kappa_2}{2}\cdot \mathbb{E}_{\bm{S}}\left[\frac{\partial^2}{\partial t^2} \mu(t,\bm{S}) \right] \quad\quad  \text{ when } \bar{p}_{T|\bm{S}} = p_{T|\bm{S}}.
		\end{cases}
	\end{align*}
	As a result, as $h\to 0$ and $nh\to \infty$, we have that 
	\begin{align*}
		\textbf{Term VI} &= \mathbb{P}_n \Psi_{h,t}\left(Y,T,\bm{S}; \bar{\mu}, \bar{p}_{T|\bm{S}}\right) - \mathbb{E}\left[\mu(t,\bm{S})\right] \\
		&= h^2B_m(t) + o(h^2) + O_P\left(\sqrt{\frac{1}{nh}}\right) \\
		&= O(h^2) + O_P\left(\sqrt{\frac{1}{nh}}\right).
	\end{align*}
	As a side note, under some VC-type condition on the kernel function $K$ \citep{einmahl2005uniform}, we can strengthen the above pointwise rate of convergence to the following uniform one as:
	$$\sup_{t\in \mathcal{T}} \left|\textbf{Term VI}\right| = O(h^2) + O_P\left(\sqrt{\frac{|\log h|}{nh}}\right);$$
	see Theorem 4 in \cite{einmahl2005uniform} for details.
	
	\subsubsection{Analysis of \textbf{Term VII} for $\hat{m}_{\mathrm{DR}}(t)$}
	\label{app:m_pos_Term_VII}
	
	By Markov's inequality, we know that
	\begin{align*}
		\sqrt{nh}\cdot \textbf{Term VII} &= \sqrt{h} \cdot \mathbb{G}_n\left[\hat{\mu}(t,\bm{S}) - \bar{\mu}(t,\bm{S}) \right]\\
		&= O_P\left(\sqrt{h}\cdot \Upsilon_{1,n}\right) = o_P(1)
	\end{align*}
	because $\mathbb{E}\left\{h\cdot \left[\hat{\mu}(t,\bm{S}) - \bar{\mu}(t,\bm{S})\right]^2\right\} = h\cdot \norm{\hat{\mu}(t,\bm{S}) - \bar{\mu}(t,\bm{S})}_{L_2}^2 =O_P\left(h\cdot \Upsilon_{1,n}^2\right)$ and $\Upsilon_{1,n}\to 0$ as $n\to\infty$. As a side note, under Assumption~\ref{assump:reg_diff} on $\bar{\mu}$ and $\hat{\mu}$, we know that the function $\bm{s} \mapsto \hat{\mu}(t,\bm{s}) - \bar{\mu}(t,\bm{s})$ is Lipschitz continuous with respect to $t\in \mathcal{T}$. Together with the compactness of $\mathcal{T}$ and Example 19.7 in \cite{VDV1998}, we can also deduce that
	$$\sup_{t\in \mathcal{T}}\left|\sqrt{h} \cdot \mathbb{G}_n\left[\hat{\mu}(t,\bm{S}) - \bar{\mu}(t,\bm{S}) \right] \right| = O_P\left(\sqrt{h}\cdot \sup_{t\in \mathcal{T}}\norm{\hat{\mu}(t,\bm{S}) - \bar{\mu}(t,\bm{S})}_{L_2}\right),$$
	which will be $o_P(1)$ as well if $\sup_{t\in \mathcal{T}}\norm{\hat{\mu}(t,\bm{S}) - \bar{\mu}(t,\bm{S})}_{L_2}=o_P(1)$.
	
	\subsubsection{Analyses of \textbf{Term VIII} and \textbf{Term IX} for $\hat{m}_{\mathrm{DR}}(t)$}
	\label{app:m_pos_Term_VIII}
	
	The argument for showing \textbf{Term VIII} and \textbf{Term IX} to be $o_P\left(\sqrt{\frac{1}{nh}}\right)$ will be similar to the one for \textbf{Term VII} above. By Markov's inequality, we know that
	\begin{align*}
		\sqrt{nh}\cdot \textbf{Term VIII} &=  \mathbb{G}_n\left[\frac{K\left(\frac{T-t}{h}\right)}{\sqrt{h}}\left[\frac{1}{\hat{p}_{T|\bm{S}}(T|\bm{S})} - \frac{1}{\bar{p}_{T|\bm{S}}(T|\bm{S})}\right]\left[Y -\bar{\mu}(t,\bm{S})\right] \right]\\
		&= O_P\left(\Upsilon_{2,n}\right) = o_P(1)
	\end{align*}
	because 
	\begin{align*}
		&\mathbb{E}\left\{\frac{K^2\left(\frac{T-t}{h}\right)}{h}\cdot \frac{\left[\hat{p}_{T|\bm{S}}(T|\bm{S}) - \bar{p}_{T|\bm{S}}(T|\bm{S})\right]^2}{\hat{p}_{T|\bm{S}}^2(T|\bm{S})\cdot \bar{p}_{T|\bm{S}}^2(T|\bm{S})}\cdot \left[Y -\bar{\mu}(t,\bm{S})\right]^2 \right\} \\
		&= \mathbb{E}\left\{\frac{K^2\left(\frac{T-t}{h}\right)}{h}\cdot \frac{\left[\hat{p}_{T|\bm{S}}(T|\bm{S}) - \bar{p}_{T|\bm{S}}(T|\bm{S})\right]^2}{\hat{p}_{T|\bm{S}}^2(T|\bm{S})\cdot \bar{p}_{T|\bm{S}}^2(T|\bm{S})}\cdot \left[\left(\mu(T,\bm{S}) -\bar{\mu}(t,\bm{S}) \right)^2 + \sigma^2\right] \right\}\\
		&\stackrel{\text{(i)}}{=} \mathbb{E}\left\{\int_{\mathbb{R}} K^2(u)\cdot \frac{\left[\hat{p}_{T|\bm{S}}(t+uh|\bm{S}) - \bar{p}_{T|\bm{S}}(t+uh|\bm{S})\right]^2 p_{T|\bm{S}}(t+uh|\bm{S})}{\hat{p}_{T|\bm{S}}^2(T|\bm{S})\cdot \bar{p}_{T|\bm{S}}^2(T|\bm{S})} \cdot \left[\left(\mu(t+uh,\bm{S}) -\bar{\mu}(t,\bm{S}) \right)^2 + \sigma^2\right]\right\}\\
		&\stackrel{\text{(ii)}}{\lesssim} \sup_{|u-t|\leq h}\norm{\hat{p}_{T|\bm{S}}(u|\bm{S}) - \bar{p}_{T|\bm{S}}(u|\bm{S})}_{L_2}^2\\
		&\stackrel{\text{(iii)}}{=}O_P\left(\Upsilon_{2,n}^2\right) = o_P(1),
	\end{align*}
	where (i) uses the change of variable $u=\frac{T-t}{h}$ in the integration, (ii) leverages the boundedness of $\mu,\bar{\mu}$ under Assumption~\ref{assump:reg_diff} and the positivity condition (Assumption~\ref{assump:positivity}) on $\bar{p}_{T|\bm{S}}$, as well as (iii) applies $\sup\limits_{|u-t|\leq h}\norm{\hat{p}_{T|\bm{S}}(u|\bm{S}) - \bar{p}_{T|\bm{S}}(u|\bm{S})}_{L_2}=O_P\left(\Upsilon_{2,n}\right)$ with $\Upsilon_{2,n}\to 0$ as $n\to \infty$. As a side note again, under the VC-type condition on the kernel function $K$ \citep{einmahl2005uniform} and $\sup_{t\in \mathcal{T}}\norm{\hat{p}_{T|\bm{S}}(t|\bm{S}) - \bar{p}_{T|\bm{S}}(t|\bm{S})}_{L_2}=O_P\left(\Upsilon_{2,n}\right)=o_P(1)$, we can strengthen the above pointwise rate of convergence to the following uniform result as:
	$$\sup_{t\in \mathcal{T}}\left|\mathbb{G}_n\left[\frac{K\left(\frac{T-t}{h}\right)}{\sqrt{h}}\left[\frac{1}{\hat{p}_{T|\bm{S}}(T|\bm{S})} - \frac{1}{\bar{p}_{T|\bm{S}}(T|\bm{S})}\right]\left[Y -\bar{\mu}(t,\bm{S})\right] \right]\right|=o_P(1).$$
	Similarly, by Markov's inequality, we have that
	\begin{align*}
		\sqrt{nh} \cdot \textbf{Term IX} &= \mathbb{G}_n\left\{\frac{K\left(\frac{T-t}{h}\right)}{\sqrt{h}\cdot \bar{p}_{T|\bm{S}}(T|\bm{S})}\left[\bar{\mu}(t,\bm{S}) - \hat{\mu}(t,\bm{S})\right]\right\} = O_P\left(\Upsilon_{1,n}\right) = o_P(1)
	\end{align*}
	because
	\begin{align*}
		\mathbb{E}\left\{\frac{K^2\left(\frac{T-t}{h}\right)}{h\cdot \bar{p}_{T|\bm{S}}^2(T|\bm{S})}\left[\bar{\mu}(t,\bm{S}) - \hat{\mu}(t,\bm{S})\right]^2\right\} &=\mathbb{E}\left\{\int_{\mathcal{T}} \frac{K^2\left(\frac{t_1-t}{h}\right) \cdot p_{T|\bm{S}}(t_1|\bm{S})}{h\cdot \bar{p}_{T|\bm{S}}^2(t_1|\bm{S})} \left[\bar{\mu}(t,\bm{S}) - \hat{\mu}(t,\bm{S})\right]^2 dt_1 \right\}\\
		&\stackrel{\text{(i)}}{=} \mathbb{E}\left\{\int_{\mathbb{R}} \frac{K^2(u) \cdot p_{T|\bm{S}}(t+uh|\bm{S})}{\bar{p}_{T|\bm{S}}^2(t+uh|\bm{S})} \left[\bar{\mu}(t,\bm{S}) - \hat{\mu}(t,\bm{S})\right]^2 du \right\} \\
		&\stackrel{\text{(ii)}}{\lesssim} \norm{\hat{\mu}(t,\bm{S}) - \bar{\mu}(t,\bm{S})}_{L_2}^2\\
		&= O_P\left(\Upsilon_{1,n}^2\right) = o_P(1),
	\end{align*}
	where (i) uses the change of variable $u=\frac{t_1-t}{h}$ and (ii) leverages the boundedness of $p_{T|\bm{S}}$ under Assumption~\ref{assump:reg_diff}, the positivity condition (Assumption~\ref{assump:positivity}) on $\bar{p}_{T|\bm{S}}$, the boundedness condition on $K$ under Assumption~\ref{assump:reg_kernel}, as well as $\norm{\hat{\mu}(t,\bm{S}) - \bar{\mu}(t,\bm{S})}_{L_2}=O_P\left(\Upsilon_{1,n}\right)$ with $\Upsilon_{1,n}\to 0$ as $n\to \infty$. In addition, if $\sup_{t\in \mathcal{T}}\norm{\hat{\mu}(t,\bm{S}) - \bar{\mu}(t,\bm{S})}_{L_2}=o_P(1)$, then the above pointwise rate of convergence can be strengthened to the uniform one as:
	$$\sup_{t\in \mathcal{T}}\left|\mathbb{G}_n\left\{\frac{K\left(\frac{T-t}{h}\right)}{\sqrt{h}\cdot \bar{p}_{T|\bm{S}}(T|\bm{S})}\left[\bar{\mu}(t,\bm{S}) - \hat{\mu}(t,\bm{S})\right]\right\}\right|=o_P(1).$$
	
	\subsubsection{Analysis of \textbf{Term X} for $\hat{m}_{\mathrm{DR}}(t)$}
	\label{app:m_pos_Term_X}
	
	We first calculate that
	\begin{align*}
		&\mathbb{E}\left|\sqrt{nh}\cdot\textbf{Term X}\right|\\
		&\mathbb{E}\left|\sqrt{\frac{n}{h}} \cdot K\left(\frac{T-t}{h}\right)\cdot \frac{\left[\bar{p}_{T|\bm{S}}(T|\bm{S}) - \hat{p}_{T|\bm{S}}(T|\bm{S})\right]}{\hat{p}_{T|\bm{S}}(T|\bm{S})\cdot \bar{p}_{T|\bm{S}}(T|\bm{S})}\cdot \left[\bar{\mu}(t,\bm{S}) - \hat{\mu}(t,\bm{S})\right] \right| \\
		&\stackrel{\text{(i)}}{\leq} \sqrt{nh} \cdot \sqrt{\mathbb{E}\left\{\frac{K\left(\frac{T-t}{h}\right) \left[\bar{p}_{T|\bm{S}}(T|\bm{S}) - \hat{p}_{T|\bm{S}}(T|\bm{S})\right]^2 }{h\cdot \hat{p}_{T|\bm{S}}^2(T|\bm{S})\cdot \bar{p}_{T|\bm{S}}^2(T|\bm{S})}\right\} \cdot \mathbb{E}\left\{\frac{K\left(\frac{T-t}{h}\right)}{h} \cdot \left[\bar{\mu}(t,\bm{S}) - \hat{\mu}(t,\bm{S})\right]^2\right\}}\\
		&= \sqrt{nh} \cdot \sqrt{\mathbb{E}\left\{\int_{\mathbb{R}}\frac{K\left(u\right) \left[\bar{p}_{T|\bm{S}}(t+uh|\bm{S}) - \hat{p}_{T|\bm{S}}(t+uh|\bm{S})\right]^2 }{\hat{p}_{T|\bm{S}}^2(t+uh|\bm{S})\cdot \bar{p}_{T|\bm{S}}^2(t+uh|\bm{S})} \cdot p_{T|\bm{S}}(t+uh|\bm{S})\, du\right\}}\\
		&\quad \times \sqrt{\mathbb{E}\left\{\int_{\mathbb{R}} K\left(u\right) \cdot \left[\bar{\mu}(t,\bm{S}) - \hat{\mu}(t,\bm{S})\right]^2 p_{T|\bm{S}}(t+uh|\bm{S})\, du\right\}} \\
		&\lesssim \sqrt{nh} \sup_{|u-t|\leq h} \norm{\hat{p}_{T|\bm{S}}(u|\bm{S}) - p_{T|\bm{S}}(u|\bm{S})}_{L_2} \norm{\hat{\mu}(t,\bm{S}) - \mu(t,\bm{S})}_{L_2}\\
		&\stackrel{\text{(ii)}}{=} o_P(1),
	\end{align*}
	where (i) uses Cauchy-Schwarz inequality and (ii) leverages our assumption (c) on the doubly robust rate of convergence in the proposition statement. As a result, by Markov's inequality, we obtain that
	$$\sqrt{nh}\cdot \textbf{Term X} = \sqrt{\frac{n}{h}}\cdot \mathbb{P}_n\left\{K\left(\frac{T-t}{h}\right)\left[\frac{1}{\hat{p}_{T|\bm{S}}(T|\bm{S})} - \frac{1}{\bar{p}_{T|\bm{S}}(T|\bm{S})}\right]\left[\bar{\mu}(t,\bm{S}) - \hat{\mu}(t,\bm{S})\right]\right\}=o_P(1).$$
	
	\subsubsection{Analysis of \textbf{Term XI} for $\hat{m}_{\mathrm{DR}}(t)$}
	\label{app:m_pos_Term_XI}
	
	By direct calculations with some change of variables, we have that
	\begin{align*}
		&\textbf{Term XI} \\
		&=\mathbb{E}\left\{\left[1- \frac{K\left(\frac{T-t}{h}\right)}{h \cdot \bar{p}_{T|\bm{S}}(T|\bm{S})}\right] \left[\hat{\mu}(t,\bm{S}) - \bar{\mu}(t,\bm{S})\right] \right\} + \mathbb{E}\left\{\frac{K\left(\frac{T-t}{h}\right)}{h}\left[\frac{1}{\hat{p}_{T|\bm{S}}(T|\bm{S})} - \frac{1}{\bar{p}_{T|\bm{S}}(T|\bm{S})}\right]\left[Y -\bar{\mu}(t,\bm{S})\right]\right\} \\
		&= \mathbb{E}\left\{\mathbb{E}\left[1- \frac{K\left(\frac{T-t}{h}\right)}{h \cdot \bar{p}_{T|\bm{S}}(T|\bm{S})} \bigg| \bm{S}\right] \left[\hat{\mu}(t,\bm{S}) - \bar{\mu}(t,\bm{S})\right] \right\} \\
		&\quad + \mathbb{E}\left\{\frac{K\left(\frac{T-t}{h}\right)\left[\bar{p}_{T|\bm{S}}(T|\bm{S}) - \hat{p}_{T|\bm{S}}(T|\bm{S})\right]}{h\cdot \bar{p}_{T|\bm{S}}(T|\bm{S}) \cdot \hat{p}_{T|\bm{S}}(T|\bm{S})} \left[\mu(T,\bm{S}) -\bar{\mu}(t,\bm{S})\right]\right\} \\
		&= \underbrace{\mathbb{E}\left\{\left[1- \int_{\mathbb{R}} \frac{K(u)\cdot p_{T|\bm{S}}(t+uh|\bm{S})}{\bar{p}_{T|\bm{S}}(t+uh|\bm{S})} \,du\right]\left[\hat{\mu}(t,\bm{S}) - \bar{\mu}(t,\bm{S})\right] \right\}}_{\textbf{Term XIa}} \\
		&\quad + \underbrace{\mathbb{E}\left\{\int_{\mathbb{R}} \frac{K(u)\left[\bar{p}_{T|\bm{S}}(t+uh|\bm{S}) - \hat{p}_{T|\bm{S}}(t+uh|\bm{S})\right]}{\bar{p}_{T|\bm{S}}(t+uh|\bm{S}) \cdot \hat{p}_{T|\bm{S}}(t+uh|\bm{S})} \left[\mu(t+uh,\bm{S}) -\bar{\mu}(t,\bm{S})\right] p_{T|\bm{S}}(t+uh|\bm{S})\, du\right\}}_{\textbf{Term XIb}}.
	\end{align*}
	On one hand, when $\bar{p}_{T|\bm{S}} = p_{T|\bm{S}}$, we know from Assumption~\ref{assump:reg_kernel} that $\textbf{Term XIa} = 0$ and
	\begin{align*}
		\textbf{Term XIb} &= \mathbb{E}\left\{\int_{\mathbb{R}} \frac{K(u)\left[p_{T|\bm{S}}(t+uh|\bm{S}) - \hat{p}_{T|\bm{S}}(t+uh|\bm{S})\right]}{\hat{p}_{T|\bm{S}}(t+uh|\bm{S})} \left[\mu(t+uh,\bm{S}) -\bar{\mu}(t,\bm{S})\right] \, du\right\} \\
		&\lesssim \sup_{|u-t|\leq h} \norm{\hat{p}_{T|\bm{S}}(u|\bm{S}) - p_{T|\bm{S}}(u|\bm{S})}_{L_2} \\
		&=o_P\left(\sqrt{\frac{1}{nh}}\right)
	\end{align*}
	by the boundedness of $\mu,\bar{\mu}$ under Assumption~\ref{assump:reg_diff}, the positivity condition (Assumption~\ref{assump:positivity}), and our assumption (c) on the doubly robust rate of convergence in the proposition statement. Specifically, since $\norm{\hat{\mu}(t,\bm{S}) - \mu(t,\bm{S})}_{L_2}= O_P(1)$ when $\bar{\mu}\neq \mu$, our assumption (c) entails that $\sup_{|u-t|\leq h} \norm{\hat{p}_{T|\bm{S}}(u|\bm{S}) - p_{T|\bm{S}}(u|\bm{S})}_{L_2} =o_P\left(\sqrt{\frac{1}{nh}}\right)$.
	
	On the other hand, when $\bar{\mu}=\mu$, we know from Assumption~\ref{assump:positivity} on $\bar{p}_{T|\bm{S}}$ and the boundedness of $p_{T|\bm{S}}$ by Assumption~\ref{assump:den_diff} that
	\begin{align*}
		\textbf{Term XIa} &= \mathbb{E}\left\{\left[1- \int_{\mathbb{R}} \frac{K(u)\cdot p_{T|\bm{S}}(t+uh|\bm{S})}{\bar{p}_{T|\bm{S}}(t+uh|\bm{S})} \,du\right]\left[\hat{\mu}(t,\bm{S}) - \bar{\mu}(t,\bm{S})\right] \right\}\\
		&\lesssim \norm{\hat{\mu}(t,\bm{S}) - \mu(t,\bm{S})}_{L_2}\\
		&=o_P\left(\sqrt{\frac{1}{nh}}\right),
	\end{align*}
	where we again argue from our assumption (c) on the doubly robust rate of convergence in the proposition statement that $\norm{\hat{\mu}(t,\bm{S}) - \mu(t,\bm{S})}_{L_2}=o_P\left(\sqrt{\frac{1}{nh}}\right)$ if $\bar{p}_{T|\bm{S}} \neq p_{T|\bm{S}}$ and 
	$$\sup_{|u-t|\leq h} \norm{\hat{p}_{T|\bm{S}}(u|\bm{S}) - p_{T|\bm{S}}(u|\bm{S})}_{L_2}=O_P(1).$$ 
	In addition, we have that
	\begin{align*}
		\textbf{Term XIb} &= \mathbb{E}\left\{\int_{\mathbb{R}} \frac{K(u)\left[\bar{p}_{T|\bm{S}}(t+uh|\bm{S}) - \hat{p}_{T|\bm{S}}(t+uh|\bm{S})\right]}{\bar{p}_{T|\bm{S}}(t+uh|\bm{S}) \cdot \hat{p}_{T|\bm{S}}(t+uh|\bm{S})} \left[\mu(t+uh,\bm{S}) -\mu(t,\bm{S})\right] p_{T|\bm{S}}(t+uh|\bm{S})\, du\right\}\\
		&\leq \sqrt{\mathbb{E}\left\{\int_{\mathbb{R}} K(u) \left[\bar{p}_{T|\bm{S}}(t+uh|\bm{S}) - \hat{p}_{T|\bm{S}}(t+uh|\bm{S})\right]^2 p_{T|\bm{S}}(t+uh|\bm{S})\, du\right\}}\\
		&\quad \times \sqrt{\mathbb{E}\left\{\int_{\mathbb{R}} \frac{K(u)\left[\mu(t+uh,\bm{S}) -\mu(t,\bm{S})\right]^2}{\bar{p}_{T|\bm{S}}^2(t+uh|\bm{S}) \cdot \hat{p}_{T|\bm{S}}^2(t+uh|\bm{S})}\cdot p_{T|\bm{S}}(t+uh|\bm{S})\, du\right\}}	\\
		&\stackrel{\text{(i)}}{=} O_P\left(\sup_{|u-t|\leq h} \norm{\hat{p}_{T|\bm{S}}(u|\bm{S}) - p_{T|\bm{S}}(u|\bm{S})}_{L_2} \right)\\
		&\quad \times \sqrt{\mathbb{E}\Bigg\{\int_{\mathbb{R}} \frac{K(u) \left[uh\cdot \frac{\partial}{\partial t}\mu(t,\bm{S}) + O(h^2)\right]^2 \left[p_{T|\bm{S}}(t|\bm{S}) + uh\cdot \frac{\partial}{\partial t} p_{T|\bm{S}}(t|\bm{S}) + O(h^2) \right]^2}{\left[\bar{p}^2_{T|\bm{S}}(t|\bm{S}) + 2uh\cdot \bar{p}_{T|\bm{S}}(t|\bm{S}) \cdot \frac{\partial}{\partial t} \bar{p}^2_{T|\bm{S}}(t|\bm{S}) + O(h^2)\right]^2\left[1+O_P\left(\Upsilon_{1,n}^2\right)\right]}\, du\Bigg\}}\\
		& = O_P\left(h^2\cdot  \sup_{|u-t|\leq h} \norm{\hat{p}_{T|\bm{S}}(u|\bm{S}) - p_{T|\bm{S}}(u|\bm{S})}_{L_2} \right) \\
		&= O_P\left(h^2\cdot \Upsilon_{2,n}\right) \\
		&\stackrel{\text{(ii)}}{=} o_P\left(\sqrt{\frac{1}{nh}}\right),
	\end{align*}
	where (i) applies Taylor's expansion and uses the fact that the difference between $\bar{p}_{T|\bm{S}}$ and $\hat{p}_{T|\bm{S}}$ is small when $\sup_{|u-t|\leq h} \norm{\hat{p}_{T|\bm{S}}(u|\bm{S}) - p_{T|\bm{S}}(u|\bm{S})}_{L_2}=O_P\left(\Upsilon_{2,n}\right)$ as well as (i) leverages the arguments that $\sqrt{nh}\cdot h^2 = \sqrt{nh^5} \to \sqrt{c_2}\in [0,\infty)$ and $\Upsilon_{2,n}\to 0$ as $n\to \infty$.
	
	\subsubsection{Asymptotic Normality of $\hat{m}_{\mathrm{DR}}(t)$}
	\label{app:m_pos_asym_norm}
	
	For the asymptotic normality of $\hat{m}_{\mathrm{DR}}(t)$, it follows from the Lyapunov central limit theorem. Specifically, we already show in \autoref{app:m_pos_Term_VI} and subsequent subsections that
	\begin{align*}
		\sqrt{nh}\left[\hat{m}_{\mathrm{DR}}(t) - m(t)\right] &= \frac{1}{\sqrt{n}} \sum_{i=1}^n \left\{\psi_{h,t}\left(Y_i,T_i,\bm{S}_i; \bar{\mu}, \bar{p}_{T|\bm{S}}\right) + \sqrt{h}\left[\bar{\mu}(t,\bm{S}_i) -  \mathbb{E}\left[\mu(t,\bm{S})\right] \right] \right\} +o_P(1)\\
		&= \frac{1}{\sqrt{n}} \sum_{i=1}^n \psi_{h,t}\left(Y_i,T_i,\bm{S}_i; \bar{\mu}, \bar{p}_{T|\bm{S}}\right) +o_P(1)
	\end{align*}
	with $\psi_{h,t}\left(Y,T,\bm{S}; \bar{\mu}, \bar{p}_{T|\bm{S}}\right) = \frac{K\left(\frac{T-t}{h}\right)}{\sqrt{h} \cdot \bar{p}_{T|\bm{S}}(T|\bm{S})}\left[Y-\bar{\mu}(t,\bm{S})\right]$ and $V_m(t) = \mathbb{E}\left[\psi_{h,t}^2\left(Y,T,\bm{S}; \bar{\mu}, \bar{p}_{T|\bm{S}}\right)\right] = O(1)$ by our calculation in \textbf{Term VI}. Then, $\sum_{i=1}^n \mathrm{Var}\left[\frac{1}{\sqrt{n}} \cdot \psi_{h,t}\left(Y_i,T_i,\bm{S}_i; \bar{\mu}, \bar{p}_{T|\bm{S}}\right) \right] = O(1)$ and 
	\begin{align*}
		&\sum_{i=1}^n \mathbb{E}\left|\frac{1}{\sqrt{n}} \cdot \psi_{h,t}\left(Y_i,T_i,\bm{S}_i; \bar{\mu}, \bar{p}_{T|\bm{S}}\right)\right|^{2+c_1} \\
		&= \mathbb{E}\left|\frac{K^{2+c_1}\left(\frac{T-t}{h}\right) \cdot \left[Y-\bar{\mu}(t,\bm{S})\right]^{2+c_1}}{n^{\frac{c_1}{2}} h^{1+\frac{c_1}{2}} \cdot \bar{p}_{T|\bm{S}}^{2+c_1}(T|\bm{S})} \right| \\
		&\lesssim \mathbb{E}\left\{\int_{\mathbb{R}} \frac{K^{2+c_1}(u) \cdot \left[\left[ \mu(t+uh,\bm{S})-\bar{\mu}(t,\bm{S})\right]^{2+c_1} + \mathbb{E}|Y-\mu(t+uh,\bm{S})|^{2+c_1}\right]}{\sqrt{(nh)^{c_1}}\cdot \bar{p}_{T|\bm{S}}^{2+c_1}(t+uh|\bm{S})} \cdot p_{T|\bm{S}}(t+uh|\bm{S})\, du \right\}\\
		&=O\left(\sqrt{\frac{1}{(nh)^{c_1}}}\right) =o(1)
	\end{align*}
	by the boundedness of $\mu,\bar{\mu},p_{T|\bm{S}}$, the positivity condition on $\bar{p}_{T|\bm{S}}$, $\mathbb{E}|Y|^{2+c_1}<\infty$ by Assumption~\ref{assump:reg_diff}(c), and the requirement that $nh\to\infty$ as $n\to \infty$. Hence, the Lyapunov condition holds, and we have that
	$$\sqrt{nh}\left[\hat{m}_{\mathrm{DR}}(t) - m(t) - h^2 B_m(t)\right] \stackrel{d}{\to} \mathcal{N}\left(0,V_m(t)\right)$$
	after subtracting the dominating bias term $h^2 B_m(t)$ of $\psi_{h,t}\left(Y,T,\bm{S}; \bar{\mu}, \bar{p}_{T|\bm{S}}\right)$ that we have computed in \textbf{Term VI}. The proof is thus completed.
\end{proof}

\section{Proof of \autoref{thm:theta_pos}}
\label{app:proof_theta_pos}

\begin{customthm}{1}[Consistency of estimating $\theta(t)$ under positivity]
	Suppose that Assumptions~\ref{assump:id_cond}, \ref{assump:reg_diff}, \ref{assump:den_diff}, \ref{assump:reg_kernel}, and \ref{assump:positivity} hold and $\hat{\mu},\hat{\beta}, \hat{p}_{T|\bm{S}}$ are constructed on a data sample independent of $\{(Y_i,T_i,\bm{S}_i)\}_{i=1}^n$. For any fixed $t\in \mathcal{T}$, we let $\bar{\mu}(t,\bm{s})$, $\bar{\beta}(t,\bm{s})$, and $\bar{p}_{T|\bm{S}}(t|\bm{s})$ be fixed bounded functions to which $\hat{\mu}(t,\bm{s})$, $\hat{\beta}(t,\bm{s})$ and $\hat{p}_{T|\bm{S}}(t|\bm{s})$ converge under the rates of convergence as:
	\begin{align*}
		&\norm{\hat{\mu}(t,\bm{S}) - \bar{\mu}(t,\bm{S})}_{L_2} = O_P\left(\Upsilon_{1,n}\right), \quad \norm{\hat{\beta}(t,\bm{S}) - \bar{\beta}(t,\bm{S})}_{L_2} = O_P\left(\Upsilon_{3,n}\right),\\
		& \text{ and } \quad \sup_{|u-t|\leq h}\norm{\hat{p}_{T|\bm{S}}(u|\bm{S}) - \bar{p}_{T|\bm{S}}(u|\bm{S})}_{L_2}=O_P\left(\Upsilon_{2,n}\right),
	\end{align*}
	where $\Upsilon_{1,n}, \Upsilon_{3,n},\Upsilon_{2,n}\to 0$ as $n\to \infty$. Then, as $h\to 0$ and $nh^3\to \infty$, we have that
	\begin{align*}
		&\hat{\theta}_{\mathrm{RA}}(t) - \theta(t) = O_P\left(\Upsilon_{3,n} + \norm{\bar{\beta}(t,\bm{S}) - \beta(t,\bm{S})}_{L_2} + \frac{1}{\sqrt{n}}\right), \\ 
		&\hat{\theta}_{\mathrm{IPW}}(t) - \theta(t) = O(h^2)+ O_P\left(\sqrt{\frac{1}{nh^3}} + \Upsilon_{2,n} + \sup_{|u-t|\leq h}\norm{\bar{p}_{T|\bm{S}}(u|\bm{S}) - p_{T|\bm{S}}(u|\bm{S})}_{L_2}\right). 
	\end{align*}
	If, in addition, we assume that 
	\begin{enumerate}[label=(\alph*)]
		\item $\bar{p}_{T|\bm{S}}$ satisfies Assumptions~\ref{assump:den_diff} and \ref{assump:positivity};
		\item either (i) ``$\,\bar{\mu}=\mu$ and $\bar{\beta}=\beta$'' with only $h\cdot \Upsilon_{3,n}\to 0$ or (ii) ``$\,\bar{p}_{T|\bm{S}} = p_{T|\bm{S}}$'';
		\item $\sqrt{nh} \sup\limits_{|u-t|\leq h} \norm{\hat{p}_{T|\bm{S}}(u|\bm{S}) - p_{T|\bm{S}}(u|\bm{S})}_{L_2} \left[\norm{\hat{\mu}(t,\bm{S}) - \mu(t,\bm{S})}_{L_2} + h \norm{\hat{\beta}(t,\bm{S}) - \beta(t,\bm{S})}_{L_2}\right] = o_P(1)$,
	\end{enumerate}
	then 
	$$\sqrt{nh^3}\left[\hat{\theta}_{\mathrm{DR}}(t) - \theta(t)\right] = \frac{1}{\sqrt{n}} \sum_{i=1}^n \left\{\phi_{h,t}\left(Y_i,T_i,\bm{S}_i;\bar{\mu}, \bar{\beta}, \bar{p}_{T|\bm{S}}\right) + \sqrt{h^3}\left[\bar{\beta}(t,\bm{S}_i) -  \mathbb{E}\left[\beta(t,\bm{S})\right] \right]\right\} +o_P(1)$$
	when $nh^7\to c_3$ for some finite number $c_3\geq 0$, where $$\phi_{h,t}\left(Y,T,\bm{S}; \bar{\mu},\bar{\beta}, \bar{p}_{T|\bm{S}}\right) = \frac{\left(\frac{T-t}{h}\right) K\left(\frac{T-t}{h}\right)}{\sqrt{h}\cdot \kappa_2\cdot \bar{p}_{T|\bm{S}}(T|\bm{S})}\cdot \left[Y - \bar{\mu}(t,\bm{S}) - (T-t)\cdot \bar{\beta}(t,\bm{S})\right].$$
	Furthermore, 
	$$\sqrt{nh^3}\left[\hat{\theta}_{\mathrm{DR}}(t) - \theta(t) - h^2 B_{\theta}(t)\right] \stackrel{d}{\to} \mathcal{N}\left(0,V_{\theta}(t)\right)$$
	with $V_{\theta}(t) = \mathbb{E}\left[\phi_{h,t}^2\left(Y,T,\bm{S};\bar{\mu}, \bar{\beta}, \bar{p}_{T|\bm{S}}\right)\right]$ and 
	\begin{align*}
		B_{\theta}(t) = 
		\begin{cases}
			\frac{\kappa_4}{6\kappa_2} \cdot \mathbb{E}_{\bm{S}}\left\{\frac{3\frac{\partial}{\partial t} p_{T|\bm{S}}(t|\bm{S}) \cdot \frac{\partial^2}{\partial t^2} \mu(t,\bm{S}) + p_{T|\bm{S}}(t|\bm{S})\left[ \frac{\partial^3}{\partial t^3} \mu(t,\bm{S}) - 3\frac{\partial}{\partial t} \log\bar{p}_{T|\bm{S}}(t|\bm{S}) \cdot \frac{\partial^2}{\partial t^2} \mu(t,\bm{S}) \right]}{\bar{p}_{T|\bm{S}}(t|\bm{S})} \right\} \; \text{ when } \bar{\mu}=\mu \text{ and } \bar{\beta}=\beta,\\
			\frac{\kappa_4}{6\kappa_2} \cdot \mathbb{E}_{\bm{S}}\left[\frac{\partial^3}{\partial t^3} \mu(t,\bm{S})\right] \quad\; \text{ when }\; \bar{p}_{T|\bm{S}} = p_{T|\bm{S}}.
		\end{cases}
	\end{align*}
\end{customthm}

\begin{remark}[Uniform asymptotic theory for estimating $\theta(t)$]
If we assume that 
\[
\begin{cases}
	\sup_{t\in \mathcal{T}}\norm{\hat{\mu}(t,\bm{S}) - \bar{\mu}(t,\bm{S})}_{L_2}=O_P(\Upsilon_{1,n}),\\
	\sup_{t\in \mathcal{T}}\norm{\hat{p}_{T|\bm{S}}(t|\bm{S}) - \bar{p}_{T|\bm{S}}(t|\bm{S})}_{L_2}=O_P\left(\Upsilon_{2,n}\right),\\
	\sup_{t\in \mathcal{T}}\norm{\hat{\beta}(t,\bm{S}) - \bar{\beta}(t,\bm{S})}_{L_2}=O_P(\Upsilon_{3,n}),\\
	\sqrt{nh} \sup_{t\in \mathcal{T}} \norm{\hat{p}_{T|\bm{S}}(u|\bm{S}) - p_{T|\bm{S}}(u|\bm{S})}_{L_2} \left[\norm{\hat{\mu}(t,\bm{S}) - \mu(t,\bm{S})}_{L_2} + h \norm{\hat{\beta}(t,\bm{S}) - \beta(t,\bm{S})}_{L_2}\right]=o_P(1),
\end{cases}
\]
then the pointwise convergence results in \autoref{thm:theta_pos} can be strengthened to the uniform ones; see our side notes in the proof below.
\end{remark}

\begin{proof}[Proof of \autoref{thm:theta_pos}]
	We derive the rates of convergence of $\hat{\theta}_{\mathrm{RA}}(t)$ given by \eqref{theta_RA} and $\hat{\theta}_{\mathrm{IPW}}(t)$ given by \eqref{theta_IPW} in \autoref{app:theta_RA_pos} and \autoref{app:theta_IPW_pos}, respectively. We also prove the asymptotic linearity, double robustness, and asymptotic normality of $\hat{\theta}_{\mathrm{DR}}(t)$ given by \eqref{theta_DR} in \autoref{app:theta_DR_pos}.
	
	\subsection{Rate of Convergence of $\hat{\theta}_{\mathrm{RA}}(t)$}
	\label{app:theta_RA_pos}
	
	Firstly, we derive the rate of convergence for $\hat{\theta}_{\mathrm{RA}}(t)$ in \eqref{theta_RA}. Under Assumption~\ref{assump:id_cond}(d), we have that
	\begin{align*}
		\hat{\theta}_{\mathrm{RA}}(t) - \theta(t) &= \frac{1}{n} \sum_{i=1}^n \hat{\beta}(t,\bm{S}_i) - \frac{d}{dt}\mathbb{E}\left[\mu(t,\bm{S}_1)\right]\\
		&= \underbrace{\frac{1}{n} \sum_{i=1}^n \left[\hat{\beta}(t,\bm{S}_i) - \bar{\beta}(t,\bm{S}_i) \right]}_{\textbf{Term I}} + \underbrace{\frac{1}{n} \sum_{i=1}^n \left[\bar{\beta}(t,\bm{S}_i) - \beta(t,\bm{S}_i) \right]}_{\textbf{Term II}} + \underbrace{\frac{1}{n} \sum_{i=1}^n \left\{\beta(t,\bm{S}_i) - \mathbb{E}\left[\beta(t,\bm{S}_i) \right] \right\}}_{\textbf{Term III}}.
	\end{align*}
	
	$\bullet$ \textbf{Term I:} By Markov's inequality (and H\"older's inequality), we know that
	\begin{align*}
		\textbf{Term I} &\leq \frac{1}{n} \sum_{i=1}^n \left|\hat{\beta}(t,\bm{S}_i) - \bar{\beta}(t,\bm{S}_i) \right| \\
		&= O_P\left(\mathbb{E}\left|\hat{\beta}(t,\bm{S}_1) - \bar{\beta}(t,\bm{S}_1) \right|\right) = O_P\left(\norm{\hat{\beta}(t,\bm{S}) - \bar{\beta}(t,\bm{S})}_{L_2}\right) = O_P\left(\Upsilon_{3,n}\right).
	\end{align*}
	
	$\bullet$ \textbf{Term II:} Analogously, we derive that
	\begin{align*}
		\textbf{Term II} &\leq \frac{1}{n} \sum_{i=1}^n \left|\bar{\beta}(t,\bm{S}_i) - \beta(t,\bm{S}_i) \right| = O_P\left(\norm{\bar{\beta}(t,\bm{S}) - \beta(t,\bm{S})}_{L_2}\right).
	\end{align*}
	
	$\bullet$ \textbf{Term III:} By the central limit theorem and the boundedness of $\beta(t,\bm{s})$ on $\mathcal{T}\times \mathcal{S}$ under Assumption~\ref{assump:reg_diff}, we know that
	$$\textbf{Term III} = \frac{1}{n} \sum_{i=1}^n \left\{\beta(t,\bm{S}_i) - \mathbb{E}\left[\beta(t,\bm{S}_i) \right] \right\} = O_P\left(\frac{1}{\sqrt{n}}\right).$$
	As a side note, under Assumption~\ref{assump:reg_diff}, we know that $\left|\beta(t_1,\bm{s}) - \beta(t_2,\bm{s})\right| \leq A_2 |t_1-t_2|$ for some absolute constant $A_2 >0$. Together with the compactness of $\mathcal{T}$ and Example 19.7 in \cite{VDV1998}, we also deduce that
	$$ \sup_{t\in \mathcal{T}}\left|\frac{1}{n} \sum_{i=1}^n \left\{\beta(t,\bm{S}_i) - \mathbb{E}\left[\beta(t,\bm{S}_i) \right] \right\} \right| = O_P\left(\frac{1}{\sqrt{n}}\right).$$
	
	In summary, we conclude that
	$$\hat{\theta}_{\mathrm{RA}}(t) - \theta(t) = O_P\left(\Upsilon_{3,n} + \norm{\bar{\beta}(t,\bm{S}) - \beta(t,\bm{S})}_{L_2} + \frac{1}{\sqrt{n}}\right).$$
	
	\subsection{Rate of Convergence of $\hat{\theta}_{\mathrm{IPW}}(t)$}
	\label{app:theta_IPW_pos}
	
	Secondly, we derive the rate of convergence for $\hat{\theta}_{\mathrm{IPW}}(t)$ in \eqref{theta_IPW}. Note that
	\begin{align*}
		&\hat{\theta}_{\mathrm{IPW}}(t) - \theta(t) \\
		&= \tilde{\theta}_{\mathrm{IPW}}(t) - \theta(t) + \hat{\theta}_{\mathrm{IPW}}(t) - \tilde{\theta}_{\mathrm{IPW}}(t)\\
		&= \underbrace{\frac{1}{nh}\sum_{i=1}^n \frac{\left(\frac{T_i-t}{h^2}\right)K\left(\frac{T_i-t}{h}\right)}{\kappa_2\cdot p_{T|\bm{S}}(T_i|\bm{S}_i)} \cdot Y_i - \mathbb{E}\left[\beta(t,\bm{S})\right]}_{\textbf{Term IV}} + \underbrace{\frac{1}{nh}\sum_{i=1}^n \frac{\left(\frac{T_i-t}{h^2}\right)K\left(\frac{T_i-t}{h}\right)}{\kappa_2\cdot \hat{p}_{T|\bm{S}}(T_i|\bm{S}_i)} \cdot Y_i - \frac{1}{nh}\sum_{i=1}^n \frac{\left(\frac{T_i-t}{h^2}\right)K\left(\frac{T_i-t}{h}\right)}{\kappa_2\cdot p_{T|\bm{S}}(T_i|\bm{S}_i)} \cdot Y_i}_{\textbf{Term V}},
	\end{align*}
	where $\tilde{\theta}_{\mathrm{IPW}}(t) = \frac{1}{nh}\sum_{i=1}^n \frac{\left(\frac{T_i-t}{h^2}\right)K\left(\frac{T_i-t}{h}\right)}{\kappa_2\cdot p_{T|\bm{S}}(T_i|\bm{S}_i)} \cdot Y_i$ is the oracle IPW estimator of $\theta(t)$ defined in \eqref{IPW_oracle} and $\beta(t,\bm{s})=\frac{\partial}{\partial t}\mu(t,\bm{s})$. We shall handle \textbf{Term IV} and \textbf{Term V} in \autoref{app:theta_pos_term_IV} and \autoref{app:theta_pos_term_V}, respectively.
	
	\subsubsection{Rate of Convergence of \textbf{Term IV} for $\hat{\theta}_{\mathrm{IPW}}(t)$}
	\label{app:theta_pos_term_IV}
	
	Under Assumptions~\ref{assump:reg_diff} and \ref{assump:reg_kernel}, we calculate the bias of $\tilde{\theta}_{\mathrm{IPW}}(t)$ as:
	\begin{align*}
		&\mathbb{E}\left[\tilde{\theta}_{\mathrm{IPW}}(t)\right] - \theta(t) \\
		&= \mathbb{E}\left[\frac{1}{h} \frac{\left(\frac{T_i-t}{h^2}\right)K\left(\frac{T_i-t}{h}\right)}{\kappa_2\cdot p_{T|\bm{S}}(T_i|\bm{S}_i)} \cdot Y_i\right] - \mathbb{E}\left[\beta(t,\bm{S})\right]\\
		&= \frac{1}{h} \int_{\mathcal{S}\times \mathcal{T}} \frac{\left(\frac{t_1-t}{h^2}\right)K\left(\frac{t_1-t}{h}\right)}{\kappa_2\cdot p_{T|\bm{S}}(t_1|\bm{s}_1)} \cdot \mu(t_1,\bm{s}_1) \cdot p(t_1,\bm{s}_1) \, dt_1 d\bm{s}_1 - \mathbb{E}\left[\beta(t,\bm{S})\right] \\
		&\stackrel{\text{(i)}}{=} \frac{1}{h\cdot \kappa_2}\int_{\mathcal{S}} \int_{\mathbb{R}} uK(u)\cdot \mu(t+uh,\bm{s}_1) \cdot p_S(\bm{s}_1) \, du d\bm{s}_1 - \mathbb{E}\left[\beta(t,\bm{S})\right]\\
		& \stackrel{\text{(ii)}}{=} \frac{1}{h\cdot \kappa_2}\int_{\mathcal{S}} \int_{\mathbb{R}} u K(u) \left[\mu(t,\bm{s}_1) + uh\cdot \frac{\partial}{\partial t}\mu(t,\bm{s}_1) + \frac{u^2h^2}{2} \frac{\partial^2}{\partial t^2}\mu(t,\bm{s}_1) + \frac{u^3h^3}{6} \frac{\partial^3}{\partial t^3}\mu(\tilde{t},\bm{s}_1) \right] p_S(\bm{s}_1) \, du d\bm{s}_1 \\
		&\quad - \mathbb{E}\left[\frac{\partial}{\partial t}\mu(t,\bm{S})\right] \\
		&\stackrel{\text{(iii)}}{=} \int_{\mathcal{S}} \frac{\partial}{\partial t}\mu(t,\bm{s}_1) \cdot p_S(\bm{s}_1)\, d\bm{s}_1 - \mathbb{E}\left[\frac{\partial}{\partial t}\mu(t,\bm{S})\right] + \int_{\mathcal{S}} \int_{\mathbb{R}} K(u)\cdot \frac{u^3h^2}{6\kappa_2} \cdot \frac{\partial^3}{\partial t^3}\mu(\tilde{t},\bm{s}_1) \cdot p_S(\bm{s}_1) \, du d\bm{s}_1\\
		&= O(h^2),
	\end{align*}
	where (i) uses a change of variable $u = \frac{t_1-t}{h}$, (ii) applies Taylor's expansion with some $\tilde{t}$ that lies between $t$ and $t+uh$, and (iii) utilizes the properties of the second-order symmetric kernel function $K$. Similarly, we compute the variance of $\tilde{\theta}_{\mathrm{IPW}}(t)$ as:
	\begin{align*}
		&\mathrm{Var}\left[\tilde{\theta}_{\mathrm{IPW}}(t)\right] \\
		&= \frac{1}{nh^2\kappa_2^2}\cdot \mathrm{Var}\left[\frac{\left(\frac{T_i-t}{h^2}\right)K\left(\frac{T_i-t}{h}\right)}{p_{T|\bm{S}}(T_i|\bm{S}_i)}Y_i\right] \\
		&= \frac{1}{nh^2\kappa_2^2}\cdot \mathbb{E}\left[\frac{\left(\frac{T_i-t}{h^2}\right)^2K^2\left(\frac{T_i-t}{h}\right)}{\left[p_{T|\bm{S}}(T_i|\bm{S}_i)\right]^2}\cdot Y_i^2 \right] - \frac{1}{nh^2\kappa_2^2} \left\{\mathbb{E}\left[\frac{\left(\frac{T_i-t}{h^2}\right)K\left(\frac{T_i-t}{h}\right)}{p_{T|\bm{S}}(T_i|\bm{S}_i)}\cdot Y_i\right]\right\}^2\\
		&= \frac{1}{nh^4\kappa_2^2} \int_{\mathcal{S}\times \mathcal{T}} \frac{\left(\frac{t_1-t}{h}\right)^2K^2\left(\frac{t_1-t}{h}\right)}{\left[p_{T|\bm{S}}(t_1|\bm{s}_1)\right]^2} \cdot \left[\mu(t_1,\bm{s}_1)^2 +\sigma^2\right] p(t_1,\bm{s}_1)\, dt_1 d\bm{s}_1 - \frac{\left\{\mathbb{E}\left[\beta(t,\bm{S})\right]\right\}^2}{n} + O\left(\frac{h^2}{n}\right)\\
		&\stackrel{\text{(i)}}{=} \frac{1}{nh^3\kappa_2^2} \int_{\mathcal{S}} \int_{\mathbb{R}} \frac{u^2K^2(u)}{p_{T|\bm{S}}(t+uh|\bm{s}_1)} \cdot \left[\mu(t+uh,\bm{s}_1)^2 +\sigma^2\right] p_S(\bm{s}_1)\, du d\bm{s}_1 + O\left(\frac{1}{n}\right) \\
		&\stackrel{\text{(ii)}}{=} \frac{1}{nh^3\kappa_2^2} \int_{\mathcal{S}} \int_{\mathbb{R}} \frac{u^2K^2(u)}{p_{T|\bm{S}}(t|\bm{s}_1) + uh \cdot \frac{\partial}{\partial t}p_{T|\bm{S}}(t'|\bm{s}_1)} \left[\mu(t,\bm{s}_1)^2 + 2uh\cdot \mu(t'',\bm{s}_1)\cdot \frac{\partial}{\partial t} \mu(t'',\bm{s}_1) + \sigma^2\right] p_S(\bm{s}_1)\, du d\bm{s}_1 \\
		&\quad + O\left(\frac{1}{n}\right)\\
		&\stackrel{\text{(iii)}}{=} \frac{1}{nh^3\kappa_2^2} \int_{\mathcal{S}} \int_{\mathbb{R}} \frac{u^2 K^2(u)}{p_{T|\bm{S}}(t|\bm{s}_1)} \cdot \left[\mu(t,\bm{s}_1)^2 + \sigma^2\right] p_S(\bm{s}_1)\, du d\bm{s}_1 + O\left(\frac{1}{n}\right)\\
		&\stackrel{\text{(iv)}}{=} O\left(\frac{1}{nh^3}\right),
	\end{align*}
	where (i) uses a change of variable $u = \frac{t_1-t}{h}$ and the boundedness of $\beta(t,\bm{s})$, (ii) applies the Taylor's expansion under Assumptions~\ref{assump:reg_diff} and \ref{assump:den_diff} with $t',t''$ being two points between $t$ and $t+uh$, (iii) absorbs the higher order terms to $O\left(\frac{1}{n}\right)$, and (iv) utilizes the properties of $K$ under Assumption~\ref{assump:reg_kernel} and the positivity condition (Assumption~\ref{assump:positivity}). Now, by Chebyshev's inequality and our above calculations, we obtain that
	\begin{align*}
		\tilde{\theta}_{\mathrm{IPW}}(t) - \theta(t) &= \tilde{\theta}_{\mathrm{IPW}}(t) - \mathbb{E}\left[\tilde{\theta}_{\mathrm{IPW}}(t)\right] + \mathbb{E}\left[\tilde{\theta}_{\mathrm{IPW}}(t)\right] - \theta(t)\\
		&= O_P\left(\sqrt{\mathrm{Var}\left[\tilde{\theta}_{\mathrm{IPW}}(t)\right]}\right) + O(h^2)\\
		&= O_P\left(\sqrt{\frac{1}{nh^3}}\right) + O(h^2)
	\end{align*}
	as $h\to 0$ and $nh^3\to \infty$. As a side note, under the VC-type condition on $K$ (Assumption~\ref{assump:reg_kernel}(c)), we can apply Theorem 2 in \cite{einmahl2005uniform} to strengthen the above pointwise rate of convergence to the uniform one as:
	$$\sup_{t\in \mathcal{T}}\left|\tilde{\theta}_{\mathrm{IPW}}(t) - \theta(t) \right| = O_P\left(\sqrt{\frac{|\log h|}{nh^3}}\right) + O(h^2).$$
	
	\subsubsection{Rate of Convergence of \textbf{Term V} for $\hat{\theta}_{\mathrm{IPW}}(t)$}
	\label{app:theta_pos_term_V} 
	
	By direct calculations, we have that
	\begin{align*}
		&\textbf{Term V} \\
		&= \frac{1}{nh}\sum_{i=1}^n \frac{\left(\frac{T_i-t}{h^2}\right)K\left(\frac{T_i-t}{h}\right)}{\kappa_2\cdot p_{T|\bm{S}}(T_i|\bm{S}_i)} \cdot Y_i \left[\frac{p_{T|\bm{S}}(T_i|\bm{S}_i) - \hat{p}_{T|\bm{S}}(T_i|\bm{S}_i)}{\hat{p}_{T|\bm{S}}(T_i|\bm{S}_i)}\right]\\
		&= \frac{1}{nh}\sum_{i=1}^n \frac{\left(\frac{T_i-t}{h^2}\right)K\left(\frac{T_i-t}{h}\right)}{\kappa_2\cdot p_{T|\bm{S}}(T_i|\bm{S}_i)} \cdot Y_i \left\{\frac{p_{T|\bm{S}}(T_i|\bm{S}_i) -\bar{p}_{T|\bm{S}}(T_i|\bm{S}_i) + \bar{p}_{T|\bm{S}}(T_i|\bm{S}_i) - \hat{p}_{T|\bm{S}}(T_i|\bm{S}_i)}{p_{T|\bm{S}}(T_i|\bm{S}_i) - \left[p_{T|\bm{S}}(T_i|\bm{S}_i) - \bar{p}_{T|\bm{S}}(T_i|\bm{S}_i)\right] - \left[\bar{p}_{T|\bm{S}}(T_i|\bm{S}_i) - \hat{p}_{T|\bm{S}}(T_i|\bm{S}_i)\right]}\right\}\\
		&\stackrel{\text{(i)}}{=}\left\{\mathbb{E}\left[\beta(t,\bm{S}) \right] + O(h^2) + O_P\left(\sqrt{\frac{1}{nh^3}}\right)\right\}\\
		&\quad \times \frac{O_P\left(\sup\limits_{|u-t|\leq h}\norm{\hat{p}_{T|\bm{S}}(u|\bm{S}) - \bar{p}_{T|\bm{S}}(u|\bm{S})}_{L_2} + \sup\limits_{|u-t|\leq h}\norm{\bar{p}_{T|\bm{S}}(u|\bm{S}) - p_{T|\bm{S}}(u|\bm{S})}_{L_2} \right)}{\inf_{(t,\bm{s})\in \mathcal{T}\times \mathcal{S}} p_{T|\bm{S}}(t|\bm{s}) - O_P\left(\sup\limits_{|u-t|\leq h}\norm{\hat{p}_{T|\bm{S}}(u|\bm{S}) - \bar{p}_{T|\bm{S}}(u|\bm{S})}_{L_2} + \sup\limits_{|u-t|\leq h}\norm{\bar{p}_{T|\bm{S}}(u|\bm{S}) - p_{T|\bm{S}}(u|\bm{S})}_{L_2} \right)}\\
		&= O_P\left(\Upsilon_{2,n} + \sup\limits_{|u-t|\leq h}\norm{\bar{p}_{T|\bm{S}}(u|\bm{S}) - p_{T|\bm{S}}(u|\bm{S})}_{L_2}\right) \left[O(1 + h^2) + O_P\left(\sqrt{\frac{1}{nh^3}}\right)\right]\\
		&= O_P\left(\Upsilon_{2,n} + \sup\limits_{|u-t|\leq h}\norm{\bar{p}_{T|\bm{S}}(u|\bm{S}) - p_{T|\bm{S}}(u|\bm{S})}_{L_2}\right)
	\end{align*}
	as $h\to 0$ and $nh^3\to \infty$, where (i) utilizes our results for \textbf{Term IV} and Markov's inequality.\\
	
	Combining our results for \textbf{Term IV} and \textbf{Term V}, we conclude that 
	$$\hat{\theta}_{\mathrm{IPW}}(t) - \theta(t) = O(h^2) + O_P\left(\sqrt{\frac{1}{nh^3}} + \Upsilon_{2,n} + \sup\limits_{|u-t|\leq h}\norm{\bar{p}_{T|\bm{S}}(u|\bm{S}) - p_{T|\bm{S}}(u|\bm{S})}_{L_2}\right).$$
	
	\subsection{Asymptotic Properties of $\hat{\theta}_{\mathrm{DR}}(t)$}
	\label{app:theta_DR_pos}
	
	Finally, using the similar arguments to \autoref{app:m_DR_pos}, we establish the asymptotic properties of $\hat{\theta}_{\mathrm{DR}}(t)$ in \eqref{theta_DR}. Under Assumption~\ref{assump:id_cond}, we have that
	\begin{align*}
		&\hat{\theta}_{\mathrm{DR}}(t) - \theta(t) \\
		&= \frac{1}{nh}\sum_{i=1}^n \left\{\frac{\left(\frac{T_i-t}{h}\right)K\left(\frac{T_i-t}{h}\right)}{h\cdot \kappa_2\cdot \hat{p}_{T|\bm{S}}(T_i|\bm{S}_i)} \left[Y_i - \hat{\mu}(t,\bm{S}_i) - (T_i-t)\cdot \hat{\beta}(t,\bm{S}_i)\right]+ h\cdot \hat{\beta}(t,\bm{S}_i) \right\} - \mathbb{E}\left[\frac{\partial}{\partial t}\mu(t,\bm{S})\right] \\
		&= \mathbb{P}_n \Phi_{h,t}\left(Y,T,\bm{S}; \bar{\mu},\bar{\beta}, \bar{p}_{T|\bm{S}}\right) - \mathbb{E}\left[\frac{\partial}{\partial t}\mu(t,\bm{S})\right] + \mathbb{P}_n \left[\Phi_{h,t}\left(Y,T,\bm{S}; \hat{\mu}, \hat{\beta}, \hat{p}_{T|\bm{S}}\right) - \Phi_{h,t}\left(Y,T,\bm{S}; \bar{\mu},\bar{\beta}, \bar{p}_{T|\bm{S}}\right) \right] \\
		&= \underbrace{\mathbb{P}_n \Phi_{h,t}\left(Y,T,\bm{S}; \bar{\mu},\bar{\beta}, \bar{p}_{T|\bm{S}}\right) - \mathbb{E}\left[\beta(t,\bm{S})\right]}_{\textbf{Term VI}} + \underbrace{\left(\mathbb{P}_n-\P\right)\left[\hat{\beta}(t,\bm{S}) - \bar{\beta}(t,\bm{S}) \right]}_{\textbf{Term VII}} \\
		&\quad + \underbrace{\left(\mathbb{P}_n-\P\right)\left\{\frac{\left(\frac{T-t}{h}\right) K\left(\frac{T-t}{h}\right)}{h^2\kappa_2}\left[\frac{1}{\hat{p}_{T|\bm{S}}(T|\bm{S})} - \frac{1}{\bar{p}_{T|\bm{S}}(T|\bm{S})} \right] \left[Y - \bar{\mu}(t,\bm{S}) - (T-t)\cdot \bar{\beta}(t,\bm{S})\right]\right\}}_{\textbf{Term VIII}}\\
		&\quad + \underbrace{\left(\mathbb{P}_n-\P\right)\left\{\frac{\left(\frac{T-t}{h}\right)K\left(\frac{T-t}{h}\right)}{h^2\kappa_2\cdot \bar{p}_{T|\bm{S}}(T|\bm{S})} \left[\bar{\mu}(t,\bm{S}) - \hat{\mu}(t,\bm{S}) + (T-t)\left[\bar{\beta}(t,\bm{S}) - \hat{\beta}(t,\bm{S})\right] \right]\right\}}_{\textbf{Term IX}} \\
		&\quad + \underbrace{\mathbb{P}_n\left\{\frac{\left(\frac{T-t}{h}\right) K\left(\frac{T-t}{h}\right)}{h^2\kappa_2}\left[\frac{1}{\hat{p}_{T|\bm{S}}(T|\bm{S})} - \frac{1}{\bar{p}_{T|\bm{S}}(T|\bm{S})} \right] \left[\bar{\mu}(t,\bm{S}) - \hat{\mu}(t,\bm{S}) + (T-t)\left[\bar{\beta}(t,\bm{S}) - \hat{\beta}(t,\bm{S})\right]\right]\right\}}_{\textbf{Term X}} \\
		&\quad + \underbrace{\P\left\{\left[1- \frac{\left(\frac{T-t}{h}\right)^2 K\left(\frac{T-t}{h}\right)}{h \cdot \kappa_2\cdot \bar{p}_{T|\bm{S}}(T|\bm{S})}\right] \left[\hat{\beta}(t,\bm{S}) - \bar{\beta}(t,\bm{S})\right]\right\}}_{\textbf{Term XIa}} + \underbrace{\P\left\{\frac{\left(\frac{T-t}{h}\right) K\left(\frac{T-t}{h}\right)}{h^2 \kappa_2\cdot \bar{p}_{T|\bm{S}}(T|\bm{S})}\left[\bar{\mu}(t,\bm{S}) - \hat{\mu}(t,\bm{S})\right]\right\}}_{\textbf{Term XIb}}\\
		&\quad + \underbrace{\P\left\{\frac{\left(\frac{T-t}{h}\right) K\left(\frac{T-t}{h}\right)}{h^2\kappa_2}\left[\frac{1}{\hat{p}_{T|\bm{S}}(T|\bm{S})} - \frac{1}{\bar{p}_{T|\bm{S}}(T|\bm{S})}\right]\left[Y - \bar{\mu}(t,\bm{S}) - (T-t)\cdot \bar{\beta}(t,\bm{S})\right]\right\}}_{\textbf{Term XIc}},
	\end{align*}
	where $\Phi_{h,t}\left(Y,T,\bm{S}; \mu, \beta, p_{T|\bm{S}}\right)= \frac{\left(\frac{T-t}{h}\right) K\left(\frac{T-t}{h}\right)}{h^2\kappa_2\cdot p_{T|\bm{S}}(T|\bm{S})}\cdot \left[Y - \mu(t,\bm{S}) - (T-t)\cdot  \beta(t,\bm{S})\right]+ \beta(t,\bm{S})$. It remains to show that the dominating \textbf{Term VI} is of order $O(h^2)+O_P\left(\sqrt{\frac{1}{nh^3}}\right)$ in \autoref{app:theta_pos_Term_VI} and the remainder terms are of order $o_P\left(\sqrt{\frac{1}{nh^3}}\right)$ for any fixed $t\in \mathcal{T}$ in \autoref{app:theta_pos_Term_VII}, \autoref{app:theta_pos_Term_VIII}, and \autoref{app:theta_pos_Term_X}, and \autoref{app:theta_pos_Term_XI}. We shall also derive the asymptotic normality of $\hat{\theta}_{\mathrm{DR}}(t)$ in \autoref{app:theta_pos_asym_norm}.
	
	\subsubsection{Analysis of \textbf{Term VI} for $\hat{\theta}_{\mathrm{DR}}(t)$}
	\label{app:theta_pos_Term_VI}
	
	We analyze the variance and bias of \textbf{Term VI} separately as follows. By direct calculations, we have that
	\begin{align*}
		&\mathrm{Var}\left[\textbf{Term VI}\right] \\
		&= \mathrm{Var}\left[\mathbb{P}_n \Phi_{h,t}\left(Y,T,\bm{S}; \bar{\mu},\bar{\beta}, \bar{p}_{T|\bm{S}}\right) \right]\\
		&= \frac{1}{nh^2}\cdot \mathrm{Var}\left\{\frac{\left(\frac{T-t}{h}\right) K\left(\frac{T-t}{h}\right)}{h\cdot \kappa_2\cdot \bar{p}_{T|\bm{S}}(T|\bm{S})} \left[Y - \bar{\mu}(t,\bm{S}) - (T-t)\cdot \bar{\beta}(t,\bm{S})\right]+ h\cdot \bar{\beta}(t,\bm{S}) \right\}\\
		&\stackrel{\text{(i)}}{\lesssim} \frac{1}{nh^2} \cdot \mathrm{Var}\left\{\frac{\left(\frac{T-t}{h}\right) K\left(\frac{T-t}{h}\right)}{h\cdot \kappa_2\cdot \bar{p}_{T|\bm{S}}(T|\bm{S})} \left[Y - \bar{\mu}(t,\bm{S}) - (T-t)\cdot \bar{\beta}(t,\bm{S})\right] \right\} + \frac{1}{n} \cdot \mathrm{Var}\left[\bar{\beta}(t,\bm{S})\right]\\
		&\stackrel{\text{(ii)}}{=} \frac{1}{nh^2} \cdot \mathbb{E}\left\{\frac{\left(\frac{T-t}{h}\right)^2 K^2\left(\frac{T-t}{h}\right)}{h^2 \kappa_2^2\cdot \bar{p}_{T|\bm{S}}^2(T|\bm{S})} \left[Y - \bar{\mu}(t,\bm{S}) - (T-t)\cdot \bar{\beta}(t,\bm{S})\right]^2\right\} + O\left(\frac{1}{n}\right)\\
		&= \frac{1}{nh^2}\int_{\mathcal{S}} \int_{\mathcal{T}} \frac{\left(\frac{t_1-t}{h}\right)^2 K^2\left(\frac{t_1-t}{h}\right)}{h^2 \kappa_2^2\cdot \bar{p}_{T|\bm{S}}^2(t_1|\bm{s}_1)} \left\{\sigma^2 + \left[\mu(t_1,\bm{s}_1) - \bar{\mu}(t,\bm{s}_1) - (t_1-t)\cdot \bar{\beta}(t,\bm{s}_1)\right]^2\right\} p(t_1,\bm{s}_1)\, dt_1 d\bm{s}_1 \\
		&\quad + O\left(\frac{1}{n}\right)\\
		&\stackrel{\text{(iii)}}{=} \frac{1}{nh^3}\int_{\mathcal{S}} \int_{\mathbb{R}} \frac{u^2\cdot K^2\left(u\right)}{\kappa_2^2\cdot \bar{p}_{T|\bm{S}}^2(t+uh|\bm{s}_1)}\cdot \left\{\sigma^2 + \left[\mu(t+uh,\bm{s}_1) - \bar{\mu}(t,\bm{s}_1) - hu\cdot \bar{\beta}(t,\bm{s}_1)\right]^2 \right\} p(t+uh,\bm{s}_1)\, du d\bm{s}_1 \\
		&\quad + O\left(\frac{1}{n}\right)\\
		&= \frac{1}{nh^3}\int_{\mathcal{S}} \int_{\mathbb{R}} \frac{u^2\cdot K^2\left(u\right)}{\kappa_2^2\cdot \bar{p}_{T|\bm{S}}^2(t|\bm{s}_1) + O(h^2)}\cdot \left\{\sigma^2 + \left[\mu(t,\bm{s}_1) - \bar{\mu}(t,\bm{s}_1)\right]^2 + O(h^2) \right\} \left[p(t,\bm{s}_1) + O(h)\right]\, du d\bm{s}_1 \\
		&\quad + O\left(\frac{1}{n}\right)\\
		&\stackrel{\text{(iv)}}{=} O\left(\frac{1}{nh^3}\right)
	\end{align*}
	where (i) uses Cauchy-Schwarz inequality on the covariance, (ii) leverages the boundedness of $\bar{\beta}$ under Assumption~\ref{assump:reg_diff} to derive the term $O\left(\frac{1}{n}\right)$, (iii) applies a change of variable $u=\frac{t_1-t}{h}$, as well as (iv) utilizes the boundedness of $\mu,\bar{\beta}$ under Assumption~\ref{assump:reg_diff} and the positivity condition (Assumption~\ref{assump:positivity}) on $\bar{p}_{T|\bm{S}}$. In the above calculation, we also note from the line (i) that the second part $\bar{\beta}(t,\bm{S})$ of $\Phi_{h,t}\left(Y,T,\bm{S};\bar{\mu}, \bar{\beta}, \bar{p}_{T|\bm{S}}\right)$ is of smaller order than the first term $\frac{\left(\frac{T-t}{h}\right) K\left(\frac{T-t}{h}\right)}{h^2\kappa_2\cdot \bar{p}_{T|\bm{S}}(T|\bm{S})}\cdot \left[Y - \bar{\mu}(t,\bm{S}) - (T-t)\cdot \bar{\beta}(t,\bm{S})\right]$. Thus, we can only keep the first term in the final asymptotically linear form of $\hat{\theta}_{\mathrm{DR}}(t)$. Now, by Chebyshev's inequality, we conclude that
	\begin{align*}
		\left(\mathbb{P}_n -\P\right) \Phi_{h,t}\left(Y,T,\bm{S}; \bar{\mu},\bar{\beta}, \bar{p}_{T|\bm{S}}\right) &= O_P\left(\sqrt{\mathrm{Var}\left[\mathbb{P}_n\Phi_{h,t}\left(Y,T,\bm{S}; \bar{\mu}, \bar{\beta}, \bar{p}_{T|\bm{S}}\right) \right]}\right) \\
		&=O_P\left(\sqrt{\mathrm{Var}\left[\frac{1}{\sqrt{h^3}}\cdot \mathbb{P}_n\phi_{h,t}\left(Y,T,\bm{S}; \bar{\mu},\bar{\beta}, \bar{p}_{T|\bm{S}}\right) \right]}\right) \\
		&= O_P\left(\sqrt{\frac{1}{nh^3}}\right),
	\end{align*}
	where $\phi_{h,t}\left(Y,T,\bm{S};\bar{\mu}, \bar{\beta}, \bar{p}_{T|\bm{S}}\right) = \frac{\left(\frac{T-t}{h}\right) K\left(\frac{T-t}{h}\right)}{\sqrt{h}\cdot \kappa_2\cdot \bar{p}_{T|\bm{S}}(T|\bm{S})}\cdot \left[Y - \bar{\mu}(t,\bm{S}) - (T-t)\cdot \bar{\beta}(t,\bm{S})\right]$. In addition, by direct calculations and Taylor's expansions, we derive that
	\begin{align*}
		&\mathrm{Bias}\left[\textbf{Term VI}\right] \\
		&= \P\left[\Phi_{h,t}\left(Y,T,\bm{S}; \bar{\mu},\bar{\beta}, \bar{p}_{T|\bm{S}}\right)\right] - \mathbb{E}\left[\beta(t,\bm{S})\right]\\
		&= \mathbb{E}\left\{\frac{\left(\frac{T-t}{h}\right) K\left(\frac{T-t}{h}\right)}{h^2\kappa_2\cdot \bar{p}_{T|\bm{S}}(T|\bm{S})}\cdot \left[Y - \bar{\mu}(t,\bm{S}) - (T-t)\cdot \bar{\beta}(t,\bm{S})\right]\right\} + \mathbb{E}\left[\bar{\beta}(t,\bm{S}) - \beta(t,\bm{S})\right] \\
		&= \mathbb{E}\left\{\int_{\mathcal{T}} \frac{\left(\frac{t_1-t}{h}\right) K\left(\frac{t_1-t}{h}\right)}{h^2\kappa_2\cdot \bar{p}_{T|\bm{S}}(t_1|\bm{S})}\cdot \left[\mu(t_1,\bm{S}) - \bar{\mu}(t,\bm{S}) - (t_1-t)\cdot \bar{\beta}(t,\bm{S})\right] p_{T|\bm{S}}(t_1|\bm{S}) \, dt_1\right\} + \mathbb{E}\left[\bar{\beta}(t,\bm{S}) - \beta(t,\bm{S})\right] \\
		&\stackrel{\text{(i)}}{=} \mathbb{E}\left\{\int_{\mathbb{R}} \frac{u\cdot K\left(u\right)}{h\cdot \kappa_2\cdot \bar{p}_{T|\bm{S}}(t+uh|\bm{S})}\cdot \left[\mu(t+uh,\bm{S}) - \bar{\mu}(t,\bm{S}) - hu\cdot \bar{\beta}(t,\bm{S})\right] p_{T|\bm{S}}(t+uh|\bm{S}) \, du\right\} \\
		&\quad + \mathbb{E}\left[\bar{\beta}(t,\bm{S}) - \beta(t,\bm{S})\right]\\
		&\stackrel{\text{(ii)}}{=} \mathbb{E}\Bigg\{\int_{\mathbb{R}} \frac{u\cdot K\left(u\right)\left[p_{T|\bm{S}}(t|\bm{S}) + uh\cdot \frac{\partial}{\partial t} p_{T|\bm{S}}(t|\bm{S}) + \frac{u^2h^2}{2} \cdot \frac{\partial^2}{\partial t^2} p_{T|\bm{S}}(t|\bm{S}) + \frac{u^3h^3}{6}\cdot \frac{\partial^3}{\partial t^3} p_{T|\bm{S}}(t|\bm{S}) + O(h^4)\right]}{h\cdot \kappa_2\left[\bar{p}_{T|\bm{S}}(t|\bm{S}) + uh\cdot \frac{\partial}{\partial t} \bar{p}_{T|\bm{S}}(t|\bm{S}) + \frac{u^2h^2}{2} \cdot \frac{\partial^2}{\partial t^2} \bar{p}_{T|\bm{S}}(t|\bm{S}) + \frac{u^3h^3}{6}\cdot \frac{\partial^3}{\partial t^3} \bar{p}_{T|\bm{S}}(t|\bm{S}) + O(h^4)\right]}\\ 
		&\quad \times \left[\left(\mu(t,\bm{S}) - \bar{\mu}(t,\bm{S})\right) + hu\left(\beta(t,\bm{S}) - \bar{\beta}(t,\bm{S})\right) + \frac{u^2 h^2}{2}\cdot \frac{\partial^2}{\partial t^2} \mu(t,\bm{S}) + \frac{u^3h^3}{6} \cdot \frac{\partial^3}{\partial t^3}\mu(t,\bm{S}) + O(h^4)\right] \, du\Bigg\} \\
		&\quad + \mathbb{E}\left[\bar{\beta}(t,\bm{S}) - \beta(t,\bm{S})\right]\\
		&= \mathbb{E}\Bigg\{\int_{\mathbb{R}} \frac{u\cdot K\left(u\right)}{h\cdot \kappa_2}\left[p_{T|\bm{S}}(t|\bm{S}) + uh\cdot \frac{\partial}{\partial t} p_{T|\bm{S}}(t|\bm{S}) + \frac{u^2h^2}{2} \cdot \frac{\partial^2}{\partial t^2} p_{T|\bm{S}}(t|\bm{S}) + \frac{u^3h^3}{6} \cdot \frac{\partial^3}{\partial t^3} p_{T|\bm{S}}(t|\bm{S}) + O(h^4)\right]\\
		&\quad \quad \times\Bigg[\frac{1}{\bar{p}_{T|\bm{S}}(t|\bm{S})} - \frac{uh\cdot \frac{\partial}{\partial t}\bar{p}_{T|\bm{S}}(t|\bm{S})}{\bar{p}_{T|\bm{S}}^2(t|\bm{S})} - \frac{u^2h^2 \cdot \frac{\partial^2}{\partial t^2}\bar{p}_{T|\bm{S}}(t|\bm{S})}{2\bar{p}_{T|\bm{S}}^2(t|\bm{S})} + \frac{u^2h^2 \left[\frac{\partial}{\partial t}\bar{p}_{T|\bm{S}}(t|\bm{S})\right]^2}{\bar{p}_{T|\bm{S}}^3(t|\bm{S})} - \frac{u^3h^3 \cdot \frac{\partial^3}{\partial t^3} \bar{p}_{T|\bm{S}}(t|\bm{S})}{6\bar{p}_{T|\bm{S}}^2(t|\bm{S})} \\
		&\quad \quad \quad\quad + \frac{u^3h^3\left[\frac{\partial}{\partial t} \bar{p}_{T|\bm{S}}(t|\bm{S})\right]\left[\frac{\partial^2}{\partial t^2} \bar{p}_{T|\bm{S}}(t|\bm{S})\right]}{\bar{p}_{T|\bm{S}}^3(t|\bm{S})} + O(h^4)\Bigg]\\  
		&\quad\quad \times \left[\left(\mu(t,\bm{S}) - \bar{\mu}(t,\bm{S})\right) + hu\left(\beta(t,\bm{S}) - \bar{\beta}(t,\bm{S})\right) + \frac{u^2 h^2}{2}\cdot \frac{\partial^2}{\partial t^2} \mu(t,\bm{S}) + \frac{u^3h^3}{6} \cdot \frac{\partial^3}{\partial t^3}\mu(t,\bm{S}) + O(h^4)\right] \, du\Bigg\} \\
		&\quad + \mathbb{E}\left[\bar{\beta}(t,\bm{S}) - \beta(t,\bm{S})\right]\\
		&= \mathbb{E}_{\bm{S}}\left\{\left[\frac{\frac{\partial}{\partial t} p_{T|\bm{S}}(t|\bm{S})}{\bar{p}_{T|\bm{S}}(t|\bm{S})} - \frac{p_{T|\bm{S}}(t|\bm{S}) \cdot \frac{\partial}{\partial t} \bar{p}_{T|\bm{S}}(t|\bm{S})}{\bar{p}_{T|\bm{S}}^2(t|\bm{S})}\right]\left[\mu(t,\bm{S}) - \bar{\mu}(t,\bm{S})\right]\right\} \\
		&\quad + \mathbb{E}_{\bm{S}}\left\{\left[\beta(t,\bm{S}) - \bar{\beta}(t,\bm{S})\right]\left[\frac{p_{T|\bm{S}}(t|\bm{S})}{\bar{p}_{T|\bm{S}}(t|\bm{S})} - 1\right]\right\} + h^2 B_{\theta}(t) + O(h^3),
	\end{align*}
	where (i) uses a change of variable $u=\frac{t_1-t}{h}$ and (ii) applies Taylor's expansion. Here, the complicated bias term $B_{\theta}(t)$ is given by
	\begin{align*}
		B_{\theta}(t) &= \frac{\kappa_4}{2\kappa_2} \cdot \mathbb{E}_{\bm{S}}\Bigg\{\frac{\left[\mu(t,\bm{S}) - \bar{\mu}(t,\bm{s})\right]}{\bar{p}_{T|\bm{S}}(t|\bm{S})}\\
		&\quad \hspace{10mm} \times \Bigg[\frac{1}{3}\cdot \frac{\partial^3}{\partial t^3} p_{T|\bm{S}}(t|\bm{S}) - \frac{\partial^2}{\partial t^2} p_{T|\bm{S}}(t|\bm{S}) \cdot \frac{\partial}{\partial t} \log \bar{p}_{T|\bm{S}}(t|\bm{S}) + 2\frac{\partial}{\partial t} p_{T|\bm{S}}(t|\bm{S}) \cdot \left[\frac{\partial}{\partial t} \log \bar{p}_{T|\bm{S}}(t|\bm{S})\right]^2\\
		&\quad\hspace{15mm} + \frac{6\frac{\partial}{\partial t} \log \bar{p}_{T|\bm{S}}(t|\bm{S}) \cdot \frac{\partial^2}{\partial t^2} \bar{p}_{T|\bm{S}}(t|\bm{S}) - \frac{\partial^3}{\partial t^3} \bar{p}_{T|\bm{S}}(t|\bm{S}) - 3 \frac{\partial}{\partial t} p_{T|\bm{S}}(t|\bm{S}) \cdot \frac{\partial^2}{\partial t^2} \bar{p}_{T|\bm{S}}(t|\bm{S})}{3 \bar{p}_{T|\bm{S}}(t|\bm{S})}\Bigg]\Bigg\} \\
		&\quad + \frac{\kappa_4}{2\kappa_2} \cdot \mathbb{E}_{\bm{S}}\Bigg\{\frac{\left[\beta(t,\bm{S}) - \bar{\beta}(t,\bm{s})\right]}{\bar{p}_{T|\bm{S}}(t|\bm{S})} \Bigg[\frac{\partial^2}{\partial t^2} p_{T|\bm{S}}(t|\bm{S}) - 2\frac{\partial}{\partial t} p_{T|\bm{S}}(t|\bm{S})\cdot \frac{\partial}{\partial t} \log \bar{p}_{T|\bm{S}}(t|\bm{S}) \\
		&\hspace{15mm} + 2p_{T|\bm{S}}(t|\bm{S}) \cdot \left[\frac{\partial}{\partial t} \log \bar{p}_{T|\bm{S}}(t|\bm{S})\right]^2 - \frac{p_{T|\bm{S}}(t|\bm{S})\cdot \frac{\partial^2}{\partial t^2} \bar{p}_{T|\bm{S}}(t|\bm{S})}{\bar{p}_{T|\bm{S}}(t|\bm{S})}\Bigg]\Bigg\}\\
		&\quad + \frac{\kappa_4}{2\kappa_2} \cdot \mathbb{E}_{\bm{S}}\left[\frac{\frac{\partial}{\partial t} p_{T|\bm{S}}(t|\bm{S}) \cdot \frac{\partial^2}{\partial t^2} \mu(t,\bm{S})}{\bar{p}_{T|\bm{S}}(t|\bm{S})} - \frac{p_{T|\bm{S}}(t|\bm{S}) \cdot \frac{\partial}{\partial t} \bar{p}_{T|\bm{S}}(t|\bm{S}) \cdot \frac{\partial^2}{\partial t^2} \mu(t,\bm{S})}{\bar{p}_{T|\bm{S}}^2(t|\bm{S})} + \frac{p_{T|\bm{S}}(t|\bm{S})\cdot \frac{\partial^3}{\partial t^3} \mu(t,\bm{S})}{3\bar{p}_{T|\bm{S}}(t|\bm{S})}\right].
	\end{align*}
	Under the condition that either $\bar{\mu}=\mu$ and $\bar{\beta}=\beta$ or $\bar{p}_{T|\bm{S}} = p_{T|\bm{S}}$, we have that
	$$\mathbb{E}_{\bm{S}}\left\{\frac{p_{T|\bm{S}}(t|\bm{S})\left[\mu(t,\bm{S})-\bar{\mu}(t,\bm{S})\right]}{\bar{p}_{T|\bm{S}}(t|\bm{S})} \cdot \frac{\partial}{\partial t}\log \left[\frac{p_{T|\bm{S}}(t|\bm{S})}{\bar{p}_{T|\bm{S}}(t|\bm{S})}\right]\right\} + \mathbb{E}_{\bm{S}}\left\{\left[\beta(t,\bm{S}) - \bar{\beta}(t,\bm{S})\right]\left[\frac{p_{T|\bm{S}}(t|\bm{S})}{\bar{p}_{T|\bm{S}}(t|\bm{S})} - 1\right]\right\} = 0$$
	and 
	\begin{align*}
		B_{\theta}(t) &= 
		\begin{cases}
			\frac{\kappa_4}{2\kappa_2} \cdot \mathbb{E}_{\bm{S}}\left[\frac{\frac{\partial}{\partial t} p_{T|\bm{S}}(t|\bm{S}) \cdot \frac{\partial^2}{\partial t^2} \mu(t,\bm{S})}{\bar{p}_{T|\bm{S}}(t|\bm{S})} - \frac{p_{T|\bm{S}}(t|\bm{S}) \cdot \frac{\partial}{\partial t} \bar{p}_{T|\bm{S}}(t|\bm{S}) \cdot \frac{\partial^2}{\partial t^2} \mu(t,\bm{S})}{\bar{p}_{T|\bm{S}}^2(t|\bm{S})} + \frac{p_{T|\bm{S}}(t|\bm{S})\cdot \frac{\partial^3}{\partial t^3} \mu(t,\bm{S})}{3\bar{p}_{T|\bm{S}}(t|\bm{S})}\right] \, \text{ when } \bar{\mu}=\mu \text{ and } \bar{\beta}=\beta,\\
			\frac{\kappa_4}{6\kappa_2} \cdot \mathbb{E}_{\bm{S}}\left[\frac{\partial^3}{\partial t^3} \mu(t,\bm{S})\right] \quad  \text{ when } \bar{p}_{T|\bm{S}} = p_{T|\bm{S}},
		\end{cases}\\
		&= \begin{cases}
			\frac{\kappa_4}{6\kappa_2} \cdot \mathbb{E}_{\bm{S}}\left\{\frac{3\frac{\partial}{\partial t} p_{T|\bm{S}}(t|\bm{S}) \cdot \frac{\partial^2}{\partial t^2} \mu(t,\bm{S}) + p_{T|\bm{S}}(t|\bm{S})\left[ \frac{\partial^3}{\partial t^3} \mu(t,\bm{S}) - 3\frac{\partial}{\partial t} \log\bar{p}_{T|\bm{S}}(t|\bm{S}) \cdot \frac{\partial^2}{\partial t^2} \mu(t,\bm{S}) \right]}{\bar{p}_{T|\bm{S}}(t|\bm{S})} \right\} \, \text{ when } \bar{\mu}=\mu \text{ and } \bar{\beta}=\beta,\\
			\frac{\kappa_4}{6\kappa_2} \cdot \mathbb{E}_{\bm{S}}\left[\frac{\partial^3}{\partial t^3} \mu(t,\bm{S})\right] \quad\quad \text{ when } \bar{p}_{T|\bm{S}} = p_{T|\bm{S}}.
		\end{cases}
	\end{align*}
	As a result, as $h\to 0$ and $nh^3\to \infty$, we know that
	\begin{align*}
		\textbf{Term VI} &= \mathbb{P}_n \Phi_{h,t}\left(Y,T,\bm{S}; \bar{\mu},\bar{\beta}, \bar{p}_{T|\bm{S}}\right) - \mathbb{E}\left[\beta(t,\bm{S})\right] \\
		&= h^2 B_{\theta}(t) + o(h^2) + O_P\left(\sqrt{\frac{1}{nh^3}}\right)\\
		&= O(h^2) + O_P\left(\sqrt{\frac{1}{nh^3}}\right).
	\end{align*}
	As a side note, under the VC-type condition on the kernel function $K$ \citep{einmahl2005uniform} (Assumption~\ref{assump:reg_kernel}(c)), we can strengthen the above pointwise rate of convergence to the following uniform one as:
	$$\sup_{t\in \mathcal{T}} \left|\textbf{Term VI}\right| = O(h^2) + O_P\left(\sqrt{\frac{|\log h|}{nh^3}}\right);$$
	see Theorem 4 in \cite{einmahl2005uniform} for details.
	
	\subsubsection{Analysis of \textbf{Term VII} for $\hat{\theta}_{\mathrm{DR}}(t)$}
	\label{app:theta_pos_Term_VII}
	
	By Markov's inequality, we know that
	\begin{align*}
		\sqrt{nh^3}\cdot \textbf{Term VII} &= \sqrt{h^3} \cdot \mathbb{G}_n\left[\hat{\beta}(t,\bm{S}) - \bar{\beta}(t,\bm{S}) \right]\\
		&= O_P\left(\sqrt{h^3}\cdot \Upsilon_{3,n}\right) = o_P(1)
	\end{align*}
	because $\mathbb{E}\left\{h^3\cdot \left[\hat{\beta}(t,\bm{S}) - \bar{\beta}(t,\bm{S})\right]^2\right\} = h^3 \norm{\hat{\beta}(t,\bm{S}) - \bar{\beta}(t,\bm{S})}_{L_2}^2 =O_P\left(h^3\cdot \Upsilon_{3,n}^2\right)$ and $h\cdot \Upsilon_{3,n}\to 0$ as $n\to\infty$. As a side note, under Assumption~\ref{assump:reg_diff} on $\bar{\beta}$ and $\hat{\beta}$, we know that the function $\bm{s} \mapsto \hat{\beta}(t,\bm{s}) - \bar{\beta}(t,\bm{s})$ is Lipschitz continuous with respect to $t\in \mathcal{T}$. Together with the compactness of $\mathcal{T}$ and Example 19.7 in \cite{VDV1998}, we can also deduce that
	$$\sup_{t\in \mathcal{T}}\left|\sqrt{h^3} \cdot \mathbb{G}_n\left[\hat{\beta}(t,\bm{S}) - \bar{\beta}(t,\bm{S}) \right] \right| = O_P\left(\sqrt{h^3}\cdot \sup_{t\in \mathcal{T}}\norm{\hat{\beta}(t,\bm{S}) - \bar{\beta}(t,\bm{S})}_{L_2}\right),$$
	which will be $o_P(1)$ as well if $\sup_{t\in \mathcal{T}}\norm{\hat{\beta}(t,\bm{S}) - \bar{\beta}(t,\bm{S})}_{L_2}=o_P\left(\frac{1}{h}\right)$.
	
	\subsubsection{Analyses of \textbf{Term VIII} and \textbf{Term IX} for $\hat{\theta}_{\mathrm{DR}}(t)$}
	\label{app:theta_pos_Term_VIII}
	
	The argument for showing \textbf{Term VIII} and \textbf{Term IX} to be $o_P\left(\sqrt{\frac{1}{nh^3}}\right)$ will be similar to the one for \textbf{Term VII} above. By Markov's inequality, we know that
	\begin{align*}
		\sqrt{nh^3}\cdot \textbf{Term VIII} &=  \mathbb{G}_n\left\{\frac{\left(\frac{T-t}{h}\right) K\left(\frac{T-t}{h}\right)}{\sqrt{h}\cdot \kappa_2}\left[\frac{1}{\hat{p}_{T|\bm{S}}(T|\bm{S})} - \frac{1}{\bar{p}_{T|\bm{S}}(T|\bm{S})}\right]\left[Y -\bar{\mu}(t,\bm{S}) - (T-t)\cdot \bar{\beta}(t,\bm{S})\right] \right\}\\
		&= O_P\left(\Upsilon_{2,n}\right) = o_P(1)
	\end{align*}
	because 
	\begin{align*}
		&\mathbb{E}\left\{\frac{\left(\frac{T-t}{h}\right)^2K^2\left(\frac{T-t}{h}\right)}{h\cdot \kappa_2^2}\cdot \frac{\left[\hat{p}_{T|\bm{S}}(T|\bm{S}) - \bar{p}_{T|\bm{S}}(T|\bm{S})\right]^2}{\hat{p}_{T|\bm{S}}^2(T|\bm{S})\cdot \bar{p}_{T|\bm{S}}^2(T|\bm{S})}\cdot \left[Y -\bar{\mu}(t,\bm{S}) - (T-t)\cdot \bar{\beta}(t,\bm{S})\right]^2 \right\} \\
		&= \mathbb{E}\left\{\frac{\left(\frac{T-t}{h}\right)^2K^2\left(\frac{T-t}{h}\right)}{h\cdot \kappa_2^2}\cdot \frac{\left[\hat{p}_{T|\bm{S}}(T|\bm{S}) - \bar{p}_{T|\bm{S}}(T|\bm{S})\right]^2}{\hat{p}_{T|\bm{S}}^2(T|\bm{S})\cdot \bar{p}_{T|\bm{S}}^2(T|\bm{S})}\cdot \left[\left(\mu(T,\bm{S}) - \bar{\mu}(t,\bm{S}) - (T-t)\cdot \bar{\beta}(t,\bm{S})\right)^2 + \sigma^2\right] \right\}\\
		&\stackrel{\text{(i)}}{=} \mathbb{E}\Bigg\{\int_{\mathbb{R}} \frac{u^2 K^2(u)\left[\hat{p}_{T|\bm{S}}(t+uh|\bm{S}) - \bar{p}_{T|\bm{S}}(t+uh|\bm{S})\right]^2 p_{T|\bm{S}}(t+uh|\bm{S})}{\kappa_2^2 \cdot \hat{p}_{T|\bm{S}}^2(T|\bm{S})\cdot \bar{p}_{T|\bm{S}}^2(T|\bm{S})} \\
		&\quad\quad \times \left[\left(\mu(t+uh,\bm{S}) -\bar{\mu}(t,\bm{S}) -hu\cdot \bar{\beta}(t,\bm{S}) \right)^2 + \sigma^2\right]\Bigg\}\\
		&\stackrel{\text{(ii)}}{\lesssim} \sup_{|u-t|\leq h}\norm{\hat{p}_{T|\bm{S}}(u|\bm{S}) - \bar{p}_{T|\bm{S}}(u|\bm{S})}_{L_2}^2\\
		&\stackrel{\text{(iii)}}{=}O_P\left(\Upsilon_{2,n}^2\right) = o_P(1),
	\end{align*}
	where (i) uses the change of variable $u=\frac{T-t}{h}$ in the integration, (ii) leverages the boundedness of $\mu,\bar{\mu},\bar{\beta}$ under Assumption~\ref{assump:reg_diff} and the positivity condition (Assumption~\ref{assump:positivity}) on $\bar{p}_{T|\bm{S}}$, as well as (iii) applies $\sup\limits_{|u-t|\leq h}\norm{\hat{p}_{T|\bm{S}}(u|\bm{S}) - \bar{p}_{T|\bm{S}}(u|\bm{S})}_{L_2}=O_P\left(\Upsilon_{2,n}\right)$ with $\Upsilon_{2,n}\to 0$ as $n\to \infty$. As a side note again, under the VC-type condition on the kernel function $K$ \citep{einmahl2005uniform} and $\sup_{t\in \mathcal{T}}\norm{\hat{p}_{T|\bm{S}}(t|\bm{S}) - \bar{p}_{T|\bm{S}}(t|\bm{S})}_{L_2}=o_P(1)$, we can strengthen the above pointwise rate of convergence to the following uniform result as:
	$$\sup_{t\in \mathcal{T}}\left|\mathbb{G}_n\left\{\frac{\left(\frac{T-t}{h}\right) K\left(\frac{T-t}{h}\right)}{\sqrt{h}\cdot \kappa_2}\left[\frac{1}{\hat{p}_{T|\bm{S}}(T|\bm{S})} - \frac{1}{\bar{p}_{T|\bm{S}}(T|\bm{S})}\right]\left[Y -\bar{\mu}(t,\bm{S}) - (T-t)\cdot \bar{\beta}(t,\bm{S})\right] \right\}\right|=o_P(1).$$
	Similarly, by Markov's inequality, we have that
	\begin{align*}
		\sqrt{nh^3} \cdot \textbf{Term IX} &= \mathbb{G}_n\left\{\frac{\left(\frac{T-t}{h}\right) K\left(\frac{T-t}{h}\right)}{\sqrt{h}\cdot \kappa_2\cdot \bar{p}_{T|\bm{S}}(T|\bm{S})}\left[\bar{\mu}(t,\bm{S}) - \hat{\mu}(t,\bm{S}) + (T-t)\left(\bar{\beta}(t,\bm{S}) - \hat{\beta}(t,\bm{S})\right)\right]\right\} \\
		&= O_P\left(\max\left\{\Upsilon_{1,n}, h\cdot \Upsilon_{3,n}\right\}\right) = o_P(1)
	\end{align*}
	because
	\begin{align*}
		&\mathbb{E}\left\{\frac{\left(\frac{T-t}{h}\right)^2K^2\left(\frac{T-t}{h}\right)}{h\cdot \kappa_2^2\cdot  \bar{p}_{T|\bm{S}}^2(T|\bm{S})}\left[\bar{\mu}(t,\bm{S}) - \hat{\mu}(t,\bm{S}) + (T-t)\left(\bar{\beta}(t,\bm{S}) - \hat{\beta}(t,\bm{S})\right)\right]^2\right\}\\ &=\mathbb{E}\left\{\int_{\mathcal{T}} \frac{\left(\frac{t_1-t}{h}\right)^2 K^2\left(\frac{t_1-t}{h}\right) \cdot p_{T|\bm{S}}(t_1|\bm{S})}{h\cdot \kappa_2^2\cdot  \bar{p}_{T|\bm{S}}^2(t_1|\bm{S})} \left[\bar{\mu}(t,\bm{S}) - \hat{\mu}(t,\bm{S}) + (t_1-t)\left(\bar{\beta}(t,\bm{S}) - \hat{\beta}(t,\bm{S})\right)\right]^2 dt_1 \right\}\\
		&\stackrel{\text{(i)}}{=} \mathbb{E}\left\{\int_{\mathbb{R}} \frac{u^2 K^2(u) \cdot p_{T|\bm{S}}(t+uh|\bm{S})}{\kappa_2^2 \cdot \bar{p}_{T|\bm{S}}^2(t+uh|\bm{S})} \left[\bar{\mu}(t,\bm{S}) - \hat{\mu}(t,\bm{S}) + hu\left(\bar{\beta}(t,\bm{S}) - \hat{\beta}(t,\bm{S})\right)\right]^2 du \right\} \\
		&\stackrel{\text{(ii)}}{\lesssim} \norm{\hat{\mu}(t,\bm{S}) - \bar{\mu}(t,\bm{S})}_{L_2}^2 + h^2\norm{\hat{\beta}(t,\bm{S}) - \bar{\beta}(t,\bm{S})}_{L_2}^2\\
		&= O_P\left(\Upsilon_{1,n}^2 + h^2\Upsilon_{3,n}^2\right) = o_P(1),
	\end{align*}
	where (i) uses the change of variable $u=\frac{t_1-t}{h}$ and (ii) leverages the boundedness of $p_{T|\bm{S}}$ under Assumption~\ref{assump:reg_diff}, the positivity condition (Assumption~\ref{assump:positivity}) on $\bar{p}_{T|\bm{S}}$, the boundedness condition on $K$ under Assumption~\ref{assump:reg_kernel}, as well as $\norm{\hat{\mu}(t,\bm{S}) - \bar{\mu}(t,\bm{S})}_{L_2}=O_P\left(\Upsilon_{1,n}\right)$ and $h\norm{\hat{\beta}(t,\bm{S}) - \bar{\beta}(t,\bm{S})}_{L_2}=O_P\left(h\cdot\Upsilon_{3,n}\right)$ with $\Upsilon_{1,n}, h\cdot\Upsilon_{3,n}\to 0$ as $n\to \infty$. In addition, if $\sup_{t\in \mathcal{T}}\norm{\hat{\mu}(t,\bm{S}) - \bar{\mu}(t,\bm{S})}_{L_2}=o_P(1)$ and $\sup_{t\in \mathcal{T}}\norm{\hat{\beta}(t,\bm{S}) - \bar{\beta}(t,\bm{S})}_{L_2}=o_P\left(\frac{1}{h}\right)$, then the above pointwise rate of convergence can be strengthened to the uniform one as:
	$$\sup_{t\in \mathcal{T}}\left|\mathbb{G}_n\left\{\frac{\left(\frac{T-t}{h}\right) K\left(\frac{T-t}{h}\right)}{\sqrt{h}\cdot \kappa_2\cdot \bar{p}_{T|\bm{S}}(T|\bm{S})}\left[\bar{\mu}(t,\bm{S}) - \hat{\mu}(t,\bm{S}) + (T-t)\left(\bar{\beta}(t,\bm{S}) - \hat{\beta}(t,\bm{S})\right)\right]\right\}\right|=o_P(1).$$
	
	\subsubsection{Analysis of \textbf{Term X} for $\hat{\theta}_{\mathrm{DR}}(t)$}
	\label{app:theta_pos_Term_X}
	
	We first calculate that
	\begin{align*}
		&\mathbb{E}\left|\sqrt{nh^3}\cdot\textbf{Term X}\right|\\
		&\mathbb{E}\left|\sqrt{\frac{n}{h}} \cdot \frac{\left(\frac{T-t}{h}\right) K\left(\frac{T-t}{h}\right)}{\kappa_2}\cdot \frac{\left[\bar{p}_{T|\bm{S}}(T|\bm{S}) - \hat{p}_{T|\bm{S}}(T|\bm{S})\right]}{\hat{p}_{T|\bm{S}}(T|\bm{S})\cdot \bar{p}_{T|\bm{S}}(T|\bm{S})}\cdot \left[\bar{\mu}(t,\bm{S}) - \hat{\mu}(t,\bm{S}) + (T-t)\left(\bar{\beta}(t,\bm{S}) - \hat{\beta}(t,\bm{S})\right)\right] \right| \\
		&\stackrel{\text{(i)}}{\leq} \sqrt{nh} \cdot \sqrt{\mathbb{E}\left\{\frac{K\left(\frac{T-t}{h}\right) \left[\bar{p}_{T|\bm{S}}(T|\bm{S}) - \hat{p}_{T|\bm{S}}(T|\bm{S})\right]^2 }{h\cdot \hat{p}_{T|\bm{S}}^2(T|\bm{S})\cdot \bar{p}_{T|\bm{S}}^2(T|\bm{S})}\right\}}\\
		&\quad \times \sqrt{\mathbb{E}\left\{\frac{\left(\frac{T-t}{h}\right)^2 K\left(\frac{T-t}{h}\right)}{h\cdot \kappa_2^2} \cdot \left[\bar{\mu}(t,\bm{S}) - \hat{\mu}(t,\bm{S}) + (T-t)\left(\bar{\beta}(t,\bm{S}) - \hat{\beta}(t,\bm{S})\right)\right]^2\right\}} \\
		&= \sqrt{nh} \cdot \sqrt{\mathbb{E}\left\{\int_{\mathbb{R}}\frac{K\left(u\right) \left[\bar{p}_{T|\bm{S}}(t+uh|\bm{S}) - \hat{p}_{T|\bm{S}}(t+uh|\bm{S})\right]^2 }{\hat{p}_{T|\bm{S}}^2(t+uh|\bm{S})\cdot \bar{p}_{T|\bm{S}}^2(t+uh|\bm{S})}\cdot p_{T|\bm{S}}(t+uh|\bm{S})\, du\right\}}\\
		&\quad \times \sqrt{\mathbb{E}\left\{\int_{\mathbb{R}} \frac{u^2 K\left(u\right)}{\kappa_2^2} \cdot \left[\bar{\mu}(t,\bm{S}) - \hat{\mu}(t,\bm{S}) + hu\left(\bar{\beta}(t,\bm{S}) - \hat{\beta}(t,\bm{S})\right)\right]^2 p_{T|\bm{S}}(t+uh|\bm{S})\, du\right\}} \\
		&\lesssim \sqrt{nh} \sup_{|u-t|\leq h} \norm{\hat{p}_{T|\bm{S}}(u|\bm{S}) - \bar{p}_{T|\bm{S}}(u|\bm{S})}_{L_2} \left[\norm{\hat{\mu}(t,\bm{S}) - \bar{\mu}(t,\bm{S})}_{L_2} + h \norm{\hat{\beta}(t,\bm{S}) - \bar{\beta}(t,\bm{S})}_{L_2}\right]\\
		&\stackrel{\text{(ii)}}{=} o_P(1),
	\end{align*}
	where (i) uses Cauchy-Schwarz inequality and (ii) leverages our assumption (c) on the doubly robust rate of convergence in the theorem statement. As a result, by Markov's inequality, we obtain that
	\begin{align*}
		&\sqrt{nh^3}\cdot \textbf{Term X} \\
		&= \sqrt{\frac{n}{h}}\cdot \mathbb{P}_n\left\{\frac{\left(\frac{T-t}{h}\right) K\left(\frac{T-t}{h}\right)}{\kappa_2}\cdot \frac{\left[\bar{p}_{T|\bm{S}}(T|\bm{S}) - \hat{p}_{T|\bm{S}}(T|\bm{S})\right]}{\hat{p}_{T|\bm{S}}(T|\bm{S})\cdot \bar{p}_{T|\bm{S}}(T|\bm{S})}\cdot \left[\bar{\mu}(t,\bm{S}) - \hat{\mu}(t,\bm{S}) + (T-t)\left(\bar{\beta}(t,\bm{S}) - \hat{\beta}(t,\bm{S})\right)\right]\right\}\\
		&=o_P(1).
	\end{align*}
	
	\subsubsection{Analysis of \textbf{Term XI} for $\hat{\theta}_{\mathrm{DR}}(t)$}
	\label{app:theta_pos_Term_XI}
	
	By direct calculations with some change of variables, we have that
	\begin{align*}
		&\textbf{Term XI} \\
		&=\mathbb{E}\left\{\left[1- \frac{\left(\frac{T-t}{h}\right)^2 K\left(\frac{T-t}{h}\right)}{h \cdot \kappa_2\cdot \bar{p}_{T|\bm{S}}(T|\bm{S})}\right] \left[\hat{\beta}(t,\bm{S}) - \bar{\beta}(t,\bm{S})\right] \right\} + \mathbb{E}\left\{\frac{\left(\frac{T-t}{h}\right)K\left(\frac{T-t}{h}\right)}{h^2\cdot \kappa_2 \cdot \bar{p}_{T|\bm{S}}(T|\bm{S})} \left[\bar{\mu}(t,\bm{S}) - \hat{\mu}(t,\bm{S})\right]\right\} \\
		&\quad + \mathbb{E}\left\{\frac{\left(\frac{T-t}{h}\right) K\left(\frac{T-t}{h}\right)}{h^2 \cdot \kappa_2}\left[\frac{1}{\hat{p}_{T|\bm{S}}(T|\bm{S})} - \frac{1}{\bar{p}_{T|\bm{S}}(T|\bm{S})}\right]\left[Y -\bar{\mu}(t,\bm{S}) - (T-t)\cdot \bar{\beta}(t,\bm{S})\right]\right\} \\
		&= \underbrace{\mathbb{E}\left\{\left[1- \int_{\mathbb{R}} \frac{u^2 \cdot K(u)\cdot p_{T|\bm{S}}(t+uh|\bm{S})}{\kappa_2\cdot \bar{p}_{T|\bm{S}}(t+uh|\bm{S})} \,du\right]\left[\hat{\beta}(t,\bm{S}) - \bar{\beta}(t,\bm{S})\right] \right\}}_{\textbf{Term XIa}} \\
		&\quad + \underbrace{\mathbb{E}\left\{\int_{\mathbb{R}} \frac{u\cdot K\left(u\right)\cdot p_{T|\bm{S}}(t+uh|\bm{S})}{h\cdot \kappa_2 \cdot \bar{p}_{T|\bm{S}}(t+uh|\bm{S})} \left[\bar{\mu}(t,\bm{S}) - \hat{\mu}(t,\bm{S})\right]\, du\right\}}_{\textbf{Term XIb}} \\
		&\quad + \underbrace{\mathbb{E}\left\{\int_{\mathbb{R}} \frac{u\cdot K(u)\left[\bar{p}_{T|\bm{S}}(t+uh|\bm{S}) - \hat{p}_{T|\bm{S}}(t+uh|\bm{S})\right]}{h\cdot \kappa_2\cdot \bar{p}_{T|\bm{S}}(t+uh|\bm{S}) \cdot \hat{p}_{T|\bm{S}}(t+uh|\bm{S})} \left[\mu(t+uh,\bm{S}) -\bar{\mu}(t,\bm{S}) - hu\cdot \bar{\beta}(t,\bm{S})\right] p_{T|\bm{S}}(t+uh|\bm{S})\, du\right\}}_{\textbf{Term XIc}}.
	\end{align*}
	On one hand, when $\bar{p}_{T|\bm{S}} = p_{T|\bm{S}}$, we know from Assumption~\ref{assump:reg_kernel} that $\textbf{Term XIa} = \textbf{Term XIb} = 0$ and
	\begin{align*}
		\textbf{Term XIc} &= \mathbb{E}\left\{\int_{\mathbb{R}} \frac{u\cdot K(u)\left[p_{T|\bm{S}}(t+uh|\bm{S}) - \hat{p}_{T|\bm{S}}(t+uh|\bm{S})\right]}{h\cdot \kappa_2\cdot \hat{p}_{T|\bm{S}}(t+uh|\bm{S})} \left[\mu(t+uh,\bm{S}) -\bar{\mu}(t,\bm{S}) - hu\cdot \bar{\beta}(t,\bm{S})\right] \, du\right\} \\
		&\lesssim \frac{1}{h}\cdot \sup_{|u-t|\leq h} \norm{\hat{p}_{T|\bm{S}}(u|\bm{S}) - p_{T|\bm{S}}(u|\bm{S})}_{L_2} \\
		&=o_P\left(\sqrt{\frac{1}{nh^3}}\right)
	\end{align*}
	by the boundedness of $\mu,\bar{\mu}$ under Assumption~\ref{assump:reg_diff}, the positivity condition (Assumption~\ref{assump:positivity}), and our assumption (c) on the doubly robust rate of convergence in the theorem statement. Specifically, since $\norm{\hat{\mu}(t,\bm{S}) - \bar{\mu}(t,\bm{S})}_{L_2} + h \norm{\hat{\beta}(t,\bm{S}) - \bar{\beta}(t,\bm{S})}_{L_2}= O_P(1)$ when $\bar{\mu}\neq \mu$ and $\bar{\beta}\neq \beta$, our assumption (c) ensures that $\sup_{|u-t|\leq h} \norm{\hat{p}_{T|\bm{S}}(u|\bm{S}) - p_{T|\bm{S}}(u|\bm{S})}_{L_2} =o_P\left(\sqrt{\frac{1}{nh}}\right)$.
	
	On the other hand, when $\bar{\mu}=\mu$ and $\bar{\beta}=\beta$, we know from Assumption~\ref{assump:positivity} on $\bar{p}_{T|\bm{S}}$ and the boundedness of $p_{T|\bm{S}}$ by Assumption~\ref{assump:den_diff} that
	\begin{align*}
		\textbf{Term XIa} &= \mathbb{E}\left\{\left[1- \int_{\mathbb{R}} \frac{u^2 \cdot K(u)\cdot p_{T|\bm{S}}(t+uh|\bm{S})}{\kappa_2\cdot \bar{p}_{T|\bm{S}}(t+uh|\bm{S})} \,du\right]\left[\hat{\beta}(t,\bm{S}) - \bar{\beta}(t,\bm{S})\right] \right\}\\
		&\lesssim \norm{\hat{\beta}(t,\bm{S}) - \beta(t,\bm{S})}_{L_2}\\
		&=o_P\left(\sqrt{\frac{1}{nh^3}}\right),
	\end{align*}
	where we argue from our assumption (c) on the doubly robust rate of convergence in the theorem statement that $\norm{\hat{\beta}(t,\bm{S}) - \beta(t,\bm{S})}_{L_2}=o_P\left(\sqrt{\frac{1}{nh^3}}\right)$ if $\bar{p}_{T|\bm{S}} \neq p_{T|\bm{S}}$ and $\sup_{|u-t|\leq h} \norm{\hat{p}_{T|\bm{S}}(u|\bm{S}) - p_{T|\bm{S}}(u|\bm{S})}_{L_2}=O_P(1)$. In addition, we also have that
	\begin{align*}
		\textbf{Term XIb} &= \mathbb{E}\left\{\int_{\mathbb{R}} \frac{u\cdot K\left(u\right)\cdot p_{T|\bm{S}}(t+uh|\bm{S})}{h\cdot \kappa_2 \cdot \bar{p}_{T|\bm{S}}(t+uh|\bm{S})} \left[\bar{\mu}(t,\bm{S}) - \hat{\mu}(t,\bm{S})\right]\, du\right\}\\
		&\lesssim \frac{1}{h} \norm{\hat{\mu}(t,\bm{S}) - \mu(t,\bm{S})}_{L_2}\\
		&=o_P\left(\sqrt{\frac{1}{nh^3}}\right),
	\end{align*}
	where we again argue from our assumption (c) on the doubly robust rate of convergence in the theorem statement that $\norm{\hat{\mu}(t,\bm{S}) - \mu(t,\bm{S})}_{L_2}=o_P\left(\sqrt{\frac{1}{nh}}\right)$ if $\bar{p}_{T|\bm{S}} \neq p_{T|\bm{S}}$ and $\sup_{|u-t|\leq h} \norm{\hat{p}_{T|\bm{S}}(u|\bm{S}) - p_{T|\bm{S}}(u|\bm{S})}_{L_2}=O_P(1)$. Finally, we also derive that
	\begin{align*}
		&\textbf{Term XIc} \\
		&= \mathbb{E}\left\{\int_{\mathbb{R}} \frac{u\cdot K(u)\left[\bar{p}_{T|\bm{S}}(t+uh|\bm{S}) - \hat{p}_{T|\bm{S}}(t+uh|\bm{S})\right]}{h\cdot \kappa_2\cdot \bar{p}_{T|\bm{S}}(t+uh|\bm{S}) \cdot \hat{p}_{T|\bm{S}}(t+uh|\bm{S})} \left[\mu(t+uh,\bm{S}) -\mu(t,\bm{S}) - hu\cdot \beta(t,\bm{S})\right] p_{T|\bm{S}}(t+uh|\bm{S})\, du\right\}\\
		&\stackrel{\text{(i)}}{=} \mathbb{E}\Bigg\{\sup_{|u-t|\leq h}\left|\hat{p}_{T|\bm{S}}(u|\bm{S}) - p_{T|\bm{S}}(u|\bm{S}) \right|\\
		&\quad \times \int_{\mathbb{R}} \frac{u K(u) \left[\frac{u^2h^2}{2}\cdot \frac{\partial^2}{\partial t^2}\mu(t,\bm{S}) + O(h^3)\right] \left[p_{T|\bm{S}}(t|\bm{S}) + uh\cdot \frac{\partial}{\partial t} p_{T|\bm{S}}(t|\bm{S}) + O(h^2) \right]}{h\cdot \kappa_2\left[\bar{p}^2_{T|\bm{S}}(t|\bm{S}) + 2uh\cdot \bar{p}_{T|\bm{S}}(t|\bm{S}) \cdot \frac{\partial}{\partial t} \bar{p}^2_{T|\bm{S}}(t|\bm{S}) + O(h^2)\right]\left[1+O_P\left(\Upsilon_{1,n}\right)\right]}\, du\Bigg\}\\
		& = O_P\left(h^2 \sup_{|u-t|\leq h} \norm{\hat{p}_{T|\bm{S}}(u|\bm{S}) - p_{T|\bm{S}}(u|\bm{S})}_{L_2} \right) \\
		&= O_P\left(h^2\cdot \Upsilon_{2,n}\right) \\
		&\stackrel{\text{(ii)}}{=} o_P\left(\sqrt{\frac{1}{nh^3}}\right),
	\end{align*}
	where (i) applies Taylor's expansion and mean-value theorem for integrals as well as uses the fact that the difference between $\bar{p}_{T|\bm{S}}$ and $\hat{p}_{T|\bm{S}}$ is small when $\sup_{|u-t|\leq h} \norm{\hat{p}_{T|\bm{S}}(u|\bm{S}) - \bar{p}_{T|\bm{S}}(u|\bm{S})}_{L_2}=O_P\left(\Upsilon_{2,n}\right)$, while (ii) leverages the arguments that $\sqrt{nh^3}\cdot h^2 = \sqrt{nh^7} \to \sqrt{c_3}\in [0,\infty)$ and $\Upsilon_{2,n}\to 0$ as $n\to \infty$.
	
	\subsubsection{Asymptotic Normality of $\hat{\theta}_{\mathrm{DR}}(t)$}
	\label{app:theta_pos_asym_norm}
	
	For the asymptotic normality of $\hat{\theta}_{\mathrm{DR}}(t)$, it follows from the Lyapunov central limit theorem. Specifically, we already show in \autoref{app:theta_pos_Term_VI} and subsequent subsections that
	\begin{align*}
		\sqrt{nh^3}\left[\hat{\theta}_{\mathrm{DR}}(t) - \theta(t)\right] &= \frac{1}{\sqrt{n}} \sum_{i=1}^n \left\{\phi_{h,t}\left(Y_i,T_i,\bm{S}_i; \bar{\mu},\bar{\beta}, \bar{p}_{T|\bm{S}}\right) + \sqrt{h^3}\left[\bar{\beta}(t,\bm{S}_i) -  \mathbb{E}\left[\beta(t,\bm{S})\right] \right]\right\} +o_P(1)\\
		&= \frac{1}{\sqrt{n}} \sum_{i=1}^n \phi_{h,t}\left(Y_i,T_i,\bm{S}_i; \bar{\mu},\bar{\beta}, \bar{p}_{T|\bm{S}}\right) +o_P(1)
	\end{align*}
	with 
	$$\phi_{h,t}\left(Y,T,\bm{S}; \bar{\mu},\bar{\beta}, \bar{p}_{T|\bm{S}}\right) = \frac{\left(\frac{T-t}{h}\right) K\left(\frac{T-t}{h}\right)}{\sqrt{h}\cdot \kappa_2\cdot \bar{p}_{T|\bm{S}}(T|\bm{S})}\cdot \left[Y - \bar{\mu}(t,\bm{S}) - (T-t)\cdot \bar{\beta}(t,\bm{S})\right]$$ 
	and $V_{\theta}(t) = \mathbb{E}\left[\phi_{h,t}^2\left(Y,T,\bm{S}; \bar{\mu},\bar{\beta}, \bar{p}_{T|\bm{S}}\right)\right] = O(1)$ by our calculation in \textbf{Term VI}. Then, 
	$$\sum_{i=1}^n \mathrm{Var}\left[\frac{1}{\sqrt{n}} \cdot \phi_{h,t}\left(Y_i,T_i,\bm{S}_i; \bar{\mu},\bar{\beta}, \bar{p}_{T|\bm{S}}\right) \right] = O(1)$$ 
	and 
	\begin{align*}
		&\sum_{i=1}^n \mathbb{E}\left|\frac{1}{\sqrt{n}} \cdot \phi_{h,t}\left(Y_i,T_i,\bm{S}_i; \bar{\mu}, \bar{\beta}, \bar{p}_{T|\bm{S}}\right)\right|^{2+c_1} \\
		&= \mathbb{E}\left|\frac{\left(\frac{T-t}{h}\right)^{2+c_1}K^{2+c_1}\left(\frac{T-t}{h}\right) \cdot \left[Y-\bar{\mu}(t,\bm{S}) - (T-t)\cdot \bar{\beta}(t,\bm{S})\right]^{2+c_1}}{n^{\frac{c_1}{2}} h^{1+\frac{c_1}{2}}\cdot \kappa_2^{2+c_1} \cdot \bar{p}_{T|\bm{S}}^{2+c_1}(T|\bm{S})} \right| \\
		&\lesssim \mathbb{E}\left\{\int_{\mathbb{R}} \frac{u^{2+c_1}K^{2+c_1}(u) \cdot \left[\left[ \mu(t+uh,\bm{S})-\bar{\mu}(t,\bm{S}) -hu\cdot \bar{\beta}(t,\bm{S})\right]^{2+c_1} + \mathbb{E}|Y|^{2+c_1}\right]}{\sqrt{(nh)^{c_1}}\cdot \bar{p}_{T|\bm{S}}^{2+c_1}(t+uh|\bm{S})} \cdot p_{T|\bm{S}}(t+uh|\bm{S})\, du \right\}\\
		&=O\left(\sqrt{\frac{1}{(nh)^{c_1}}}\right) =o(1)
	\end{align*}
	by the boundedness of $\mu,\bar{\mu},p_{T|\bm{S}}$, the positivity condition on $\bar{p}_{T|\bm{S}}$, $\mathbb{E}|Y|^{2+c_1}<\infty$ by Assumption~\ref{assump:reg_diff}(c), and the requirement that $nh^3\to\infty$ as $n\to \infty$. Hence, the Lyapunov condition holds, and we have that
	$$\sqrt{nh^3}\left[\hat{\theta}_{\mathrm{DR}}(t) - \theta(t) - h^2 B_{\theta}(t)\right] \stackrel{d}{\to} \mathcal{N}\left(0,V_{\theta}(t)\right)$$
	after subtracting the dominating bias term $h^2 B_{\theta}(t)$ of $\phi_{h,t}\left(Y,T,\bm{S}; \bar{\mu}, \bar{\beta}, \bar{p}_{T|\bm{S}}\right)$ that we have computed in \textbf{Term VI}. The proof of \autoref{thm:theta_pos} is thus completed.
\end{proof}

\subsection{Asymptotic Bias of A Naive AIPW Estimator of $\theta(t)$}
\label{app:theta_pos_AIPW}

Naively, one may combine the RA estimator \eqref{theta_RA} with the IPW estimator \eqref{theta_IPW} to derive an augmented IPW (AIPW) estimator of $\theta(t)$ with the following (or other similar) form as:
\begin{equation}
	\label{theta_AIPW}
	\hat{\theta}_{\mathrm{AIPW}}(t) = \frac{1}{nh}\sum_{i=1}^n \left\{ \frac{K\left(\frac{T_i-t}{h}\right) }{\hat{p}_{T|\bm{S}}(T_i|\bm{S}_i)} \left[ \frac{Y_i}{\kappa_2} \left(\frac{T_i-t}{h^2}\right)-  \hat{\beta}(t,\bm{S}_i)
	\right]+ h\cdot \hat{\beta}(t,\bm{S}_i) \right\}.
\end{equation}
However, this naive AIPW estimator is not doubly robust as demonstrated by Proposition~\ref{prop:theta_pos_AIPW} below---it is only robust to the misspecification of the limiting quantity of $\hat{\beta}(t,\bm{s})$. In other words, $\hat{\theta}_{\mathrm{AIPW}}(t)$ will be asymptotically unbiased only when the estimated conditional density $\hat{p}_{T|\bm{S}}$ converges to the true conditional density $p_{T|\bm{S}}$ in a certain rate.

Analogous to our calculations in \autoref{app:theta_DR_pos}, we can decompose $\hat{\theta}_{\mathrm{AIPW}}(t) - \theta(t)$ under Assumption~\ref{assump:id_cond} as:
\begin{align*}
	&\hat{\theta}_{\mathrm{AIPW}}(t) - \theta(t) \\
	&= \frac{1}{nh}\sum_{i=1}^n \left\{\frac{K\left(\frac{T_i-t}{h}\right)}{\hat{p}_{T|\bm{S}}(T_i|\bm{S}_i)} \left[\frac{Y_i}{\kappa_2}\left(\frac{T_i-t}{h^2}\right) - \hat{\beta}(t,\bm{S}_i)\right]+ h\cdot \hat{\beta}(t,\bm{S}_i) \right\} - \mathbb{E}\left[\frac{\partial}{\partial t}\mu(t,\bm{S})\right] \\
	&= \mathbb{P}_n \tilde{\Phi}_{h,t}\left(Y,T,\bm{S}; \bar{\beta}, \bar{p}_{T|\bm{S}}\right) - \mathbb{E}\left[\frac{\partial}{\partial t}\mu(t,\bm{S})\right] + \mathbb{P}_n \left[\Phi_{h,t}\left(Y,T,\bm{S}; \hat{\beta}, \hat{p}_{T|\bm{S}}\right) - \Phi_{h,t}\left(Y,T,\bm{S}; \bar{\beta}, \bar{p}_{T|\bm{S}}\right) \right] \\
	&= \underbrace{\mathbb{P}_n \tilde{\Phi}_{h,t}\left(Y,T,\bm{S}; \bar{\beta}, \bar{p}_{T|\bm{S}}\right) - \mathbb{E}\left[\beta(t,\bm{S})\right]}_{\textbf{Dominating Term}} + \left(\mathbb{P}_n-\P\right)\left[\hat{\beta}(t,\bm{S}) - \bar{\beta}(t,\bm{S}) \right] \\
	&\quad + \left(\mathbb{P}_n-\P\right)\left\{\frac{K\left(\frac{T-t}{h}\right)}{h}\left[\frac{1}{\hat{p}_{T|\bm{S}}(T|\bm{S})} - \frac{1}{\bar{p}_{T|\bm{S}}(T|\bm{S})} \right] \left[\frac{Y}{\kappa_2}\left(\frac{T-t}{h^2}\right) - \bar{\beta}(t,\bm{S})\right]\right\}\\
	&\quad + \left(\mathbb{P}_n-\P\right)\left\{\frac{K\left(\frac{T-t}{h}\right)}{h\cdot \bar{p}_{T|\bm{S}}(T|\bm{S})} \left[\bar{\beta}(t,\bm{S}) - \hat{\beta}(t,\bm{S})\right]\right\} \\
	&\quad + \mathbb{P}_n\left\{\frac{K\left(\frac{T-t}{h}\right)}{h}\left[\frac{1}{\hat{p}_{T|\bm{S}}(T|\bm{S})} - \frac{1}{\bar{p}_{T|\bm{S}}(T|\bm{S})} \right] \left[\bar{\beta}(t,\bm{S}) - \hat{\beta}(t,\bm{S})\right]\right\} \\
	&\quad + \P\left\{\left[1- \frac{K\left(\frac{T-t}{h}\right)}{h \cdot \bar{p}(T|\bm{S})}\right] \left[\hat{\beta}(t,\bm{S}) - \bar{\beta}(t,\bm{S})\right] +\frac{K\left(\frac{T-t}{h}\right)}{h}\left[\frac{1}{\hat{p}_{T|\bm{S}}(T|\bm{S})} - \frac{1}{\bar{p}_{T|\bm{S}}(T|\bm{S})}\right]\left[\frac{Y}{\kappa_2}\left(\frac{T-t}{h^2}\right) - \bar{\beta}(t,\bm{S})\right]\right\},
\end{align*}
where $\tilde{\Phi}_{h,t}\left(Y,T,\bm{S}; \beta, p_{T|\bm{S}}\right)= \frac{K\left(\frac{T-t}{h}\right)}{h\cdot p_{T|\bm{S}}(T|\bm{S})} \left[\frac{Y}{\kappa_2}\left(\frac{T-t}{h^2}\right) - \beta(t,\bm{S})\right]+ \beta(t,\bm{S})$. Thus, in order to study the asymptotically dominating bias of $\hat{\theta}_{\mathrm{AIPW}}(t)$ in \eqref{theta_AIPW}, it suffices to compute $\mathbb{E}\left[\tilde{\Phi}_{h,t}\left(Y,T,\bm{S}; \bar{\beta}, \bar{p}_{T|\bm{S}}\right) - \beta(t,\bm{S}) \right]$ as in Proposition~\ref{prop:theta_pos_AIPW} below.

\begin{proposition}[Asymptotically dominating bias of $\hat{\theta}_{\mathrm{AIPW}}(t)$]
	\label{prop:theta_pos_AIPW}
	Suppose that Assumptions~\ref{assump:id_cond}, \ref{assump:positivity}, \ref{assump:reg_diff}, \ref{assump:den_diff}, and \ref{assump:reg_kernel} hold as well as $\hat{\beta}, \hat{p}_{T|\bm{S}}$ are constructed on a data sample independent of $\{(Y_i,T_i,\bm{S}_i)\}_{i=1}^n$. For any fixed $t\in \mathcal{T}$, we let $\bar{\mu}(t,\bm{s})$, $\bar{\beta}(t,\bm{s})$, and $\bar{p}_{T|\bm{S}}(t|\bm{s})$ be fixed bounded functions to which $\hat{\mu}(t,\bm{s})$, $\hat{\beta}(t,\bm{s})$ and $\hat{p}_{T|\bm{S}}(t|\bm{s})$ converge. Assume also that $\bar{p}_{T|\bm{S}}$ satisfies Assumptions~\ref{assump:den_diff} and \ref{assump:positivity}. Then, the asymptotically dominating bias of $\hat{\theta}_{\mathrm{AIPW}}(t)$ is given by
	\begin{align*}
		&\mathbb{E}\left[\tilde{\Phi}_{h,t}\left(Y,T,\bm{S}; \bar{\beta}, \bar{p}_{T|\bm{S}}\right) - \beta(t,\bm{S}) \right] \\
		&= \mathbb{E}_{\bm{S}}\left\{\frac{\mu(t,\bm{S})\cdot p_{T|\bm{S}}(t|\bm{S}) }{\bar{p}_{T|\bm{S}}(t|\bm{S})} \cdot \frac{\partial}{\partial t}\log \left[\frac{p_{T|\bm{S}}(t|\bm{S})}{\bar{p}_{T|\bm{S}}(t|\bm{S})}\right]\right\} + \mathbb{E}_{\bm{S}}\left\{\left[\beta(t,\bm{S}) - \bar{\beta}(t,\bm{S})\right]\left[\frac{p_{T|\bm{S}}(t|\bm{S})}{\bar{p}_{T|\bm{S}}(t|\bm{S})} - 1\right]\right\} + O(h^2)
	\end{align*}
	when $h\to 0$ and $n\to \infty$.
\end{proposition}

\begin{proof}[Proof of Proposition~\ref{prop:theta_pos_AIPW}]
	By direct calculations under Assumption~\ref{assump:id_cond}, we derive that
	\begin{align*}
		&\mathbb{E}\left[\tilde{\Phi}_{h,t}\left(Y,T,\bm{S}; \bar{\beta}, \bar{p}_{T|\bm{S}}\right) - \beta(t,\bm{S}) \right] \\
		&= \mathbb{E}\left\{\int_{\mathcal{T}} \frac{K\left(\frac{t_1-t}{h}\right)}{h\cdot \bar{p}_{T|\bm{S}}(t_1|\bm{S})}\cdot \left[\frac{\mu(t_1,\bm{S})}{\kappa_2}\left(\frac{t_1-t}{h^2}\right) - \bar{\beta}(t,\bm{S})\right] p_{T|\bm{S}}(t_1|\bm{S}) \, dt_1\right\} + \mathbb{E}\left[\bar{\beta}(t,\bm{S}) - \beta(t,\bm{S})\right] \\
		&\stackrel{\text{(i)}}{=} \mathbb{E}\left\{\int_{\mathbb{R}} \frac{K\left(u\right)}{\bar{p}_{T|\bm{S}}(t+uh|\bm{S})}\cdot \left[\frac{u\cdot \mu(t+uh,\bm{S})}{h\cdot \kappa_2} - \bar{\beta}(t,\bm{S})\right] p_{T|\bm{S}}(t+uh|\bm{S}) \, du\right\} + \mathbb{E}\left[\bar{\beta}(t,\bm{S}) - \beta(t,\bm{S})\right]\\
		&\stackrel{\text{(ii)}}{=} \mathbb{E}\Bigg\{\int_{\mathbb{R}} \frac{K\left(u\right)\left[p_{T|\bm{S}}(t|\bm{S}) + uh\cdot \frac{\partial}{\partial t} p_{T|\bm{S}}(t|\bm{S}) + \frac{u^2h^2}{2} \cdot \frac{\partial^2}{\partial t^2} p_{T|\bm{S}}(t|\bm{S}) + \frac{u^3h^3}{6}\cdot \frac{\partial^3}{\partial t^3} p_{T|\bm{S}}(t|\bm{S}) + O(h^4)\right]}{\bar{p}_{T|\bm{S}}(t|\bm{S}) + uh\cdot \frac{\partial}{\partial t} \bar{p}_{T|\bm{S}}(t|\bm{S}) + \frac{u^2h^2}{2} \cdot \frac{\partial^2}{\partial t^2} \bar{p}_{T|\bm{S}}(t|\bm{S}) + \frac{u^3h^3}{6}\cdot \frac{\partial^3}{\partial t^3} \bar{p}_{T|\bm{S}}(t|\bm{S}) + O(h^4)}\\ 
		&\quad \times \left[\frac{u\cdot \mu(t,\bm{S})}{h\cdot \kappa_2} + \frac{u^2 \beta(t,\bm{S})}{\kappa_2} - \bar{\beta}(t,\bm{S}) + \frac{u^3 h}{2\kappa_2}\cdot \frac{\partial^2}{\partial t^2} \mu(t,\bm{S}) + \frac{u^4h^2}{6\kappa_2} \cdot \frac{\partial^3}{\partial t^3}\mu(t,\bm{S}) + O(h^3)\right] \, du\Bigg\} \\
		&\quad + \mathbb{E}\left[\bar{\beta}(t,\bm{S}) - \beta(t,\bm{S})\right]\\
		&\stackrel{\text{(iii)}}{=} \mathbb{E}\Bigg\{\int_{\mathbb{R}} K\left(u\right)\left[p_{T|\bm{S}}(t|\bm{S}) + uh\cdot \frac{\partial}{\partial t} p_{T|\bm{S}}(t|\bm{S}) + \frac{u^2h^2}{2} \cdot \frac{\partial^2}{\partial t^2} p_{T|\bm{S}}(t|\bm{S}) + \frac{u^3h^3}{6} \cdot \frac{\partial^3}{\partial t^3} p_{T|\bm{S}}(t|\bm{S}) + O(h^4)\right]\\
		&\quad \quad \times\Bigg[\frac{1}{\bar{p}_{T|\bm{S}}(t|\bm{S})} - \frac{uh\cdot \frac{\partial}{\partial t}\bar{p}_{T|\bm{S}}(t|\bm{S})}{\bar{p}_{T|\bm{S}}^2(t|\bm{S})} - \frac{u^2h^2 \cdot \frac{\partial^2}{\partial t^2}\bar{p}_{T|\bm{S}}(t|\bm{S})}{2\bar{p}_{T|\bm{S}}^2(t|\bm{S})} + \frac{u^2h^2 \left[\frac{\partial}{\partial t}\bar{p}_{T|\bm{S}}(t|\bm{S})\right]^2}{\bar{p}_{T|\bm{S}}^3(t|\bm{S})} - \frac{u^3h^3 \cdot \frac{\partial^3}{\partial t^3} \bar{p}_{T|\bm{S}}(t|\bm{S})}{6\bar{p}_{T|\bm{S}}^2(t|\bm{S})} \\
		&\quad \quad \quad\quad + \frac{u^3h^3\left[\frac{\partial}{\partial t} \bar{p}_{T|\bm{S}}(t|\bm{S})\right]\left[\frac{\partial^2}{\partial t^2} \bar{p}_{T|\bm{S}}(t|\bm{S})\right]}{\bar{p}_{T|\bm{S}}^3(t|\bm{S})} + O(h^4)\Bigg]\\  
		&\quad\quad \times \left[\frac{u\cdot \mu(t,\bm{S})}{h\cdot \kappa_2} + \frac{u^2 \beta(t,\bm{S})}{\kappa_2} - \bar{\beta}(t,\bm{S}) + \frac{u^3 h}{2\kappa_2}\cdot \frac{\partial^2}{\partial t^2} \mu(t,\bm{S}) + \frac{u^4h^2}{6\kappa_2} \cdot \frac{\partial^3}{\partial t^3}\mu(t,\bm{S}) + O(h^3)\right] \, du\Bigg\} \\
		&\quad + \mathbb{E}\left[\bar{\beta}(t,\bm{S}) - \beta(t,\bm{S})\right]\\
		&= \mathbb{E}\left[\frac{\mu(t,\bm{S})\cdot \frac{\partial}{\partial t} p_{T|\bm{S}}(t|\bm{S})}{\bar{p}_{T|\bm{S}}(t|\bm{S})} - \frac{\mu(t,\bm{S})\cdot p_{T|\bm{S}}(t|\bm{S}) \cdot \frac{\partial}{\partial t} \bar{p}_{T|\bm{S}}(t|\bm{S})}{\bar{p}_{T|\bm{S}}^2(t|\bm{S})}\right] + \mathbb{E}\left\{\left[\beta(t,\bm{S}) - \bar{\beta}(t,\bm{S})\right]\left[\frac{p_{T|\bm{S}}(t|\bm{S})}{\bar{p}_{T|\bm{S}}(t|\bm{S})} - 1\right]\right\}\\
		&\quad + h^2 \tilde{B}_{\theta}(t) + O(h^3),
	\end{align*}
	where (i) uses a change of variable $u=\frac{t_1-t}{h}$ while (ii) and (iii) apply Taylor's expansions under Assumptions~\ref{assump:den_diff} and \ref{assump:reg_diff}. Here, the complicated bias term $\tilde{B}_{\theta}(t)$ is given by
	\begin{align*}
		\tilde{B}_{\theta}(t) &= \frac{\kappa_4}{2\kappa_2}\cdot \mathbb{E}_{\bm{S}}\Bigg\{\frac{\mu(t,\bm{S})}{\bar{p}_{T|\bm{S}}(t|\bm{S})}\Bigg[\frac{1}{3}\cdot \frac{\partial^3}{\partial t^3} p_{T|\bm{S}}(t|\bm{S}) - \frac{\partial^2}{\partial t^2}p_{T|\bm{S}}(t|\bm{S}) \cdot \frac{\partial}{\partial t}\log \bar{p}_{T|\bm{S}}(t|\bm{S}) \\
		&\hspace{15mm} + 2\frac{\partial}{\partial t} p_{T|\bm{S}}(t|\bm{S})\left[\frac{\partial}{\partial t} \log \bar{p}_{T|\bm{S}}(t|\bm{S})\right]^2 \\
		& \hspace{15mm}+ \frac{6\frac{\partial}{\partial t} \log \bar{p}_{T|\bm{S}}(t|\bm{S}) \cdot \frac{\partial^2}{\partial t^2} \bar{p}_{T|\bm{S}}(t|\bm{S}) - \frac{\partial^3}{\partial t^3} \bar{p}_{T|\bm{S}}(t|\bm{S}) - 3 \frac{\partial}{\partial t} p_{T|\bm{S}}(t|\bm{S}) \cdot \frac{\partial^2}{\partial t^2} \bar{p}_{T|\bm{S}}(t|\bm{S})}{3 \bar{p}_{T|\bm{S}}(t|\bm{S})}\Bigg] \Bigg\}\\
		&\quad + \frac{\kappa_4}{2\kappa_2} \cdot \mathbb{E}_{\bm{S}}\Bigg\{\frac{\left[\frac{\kappa_4}{\kappa_2}\cdot \beta(t,\bm{S}) - \bar{\beta}(t,\bm{s})\right]}{\bar{p}_{T|\bm{S}}(t|\bm{S})} \Bigg[\frac{\partial^2}{\partial t^2} p_{T|\bm{S}}(t|\bm{S}) - 2\frac{\partial}{\partial t} p_{T|\bm{S}}(t|\bm{S})\cdot \frac{\partial}{\partial t} \log \bar{p}_{T|\bm{S}}(t|\bm{S}) \\
		&\hspace{15mm} + 2p_{T|\bm{S}}(t|\bm{S}) \cdot \left[\frac{\partial}{\partial t} \log \bar{p}_{T|\bm{S}}(t|\bm{S})\right]^2 - \frac{p_{T|\bm{S}}(t|\bm{S})\cdot \frac{\partial^2}{\partial t^2} \bar{p}_{T|\bm{S}}(t|\bm{S})}{\bar{p}_{T|\bm{S}}(t|\bm{S})}\Bigg]\Bigg\}\\
		&\quad + \frac{\kappa_4}{2\kappa_2} \cdot \mathbb{E}_{\bm{S}}\left[\frac{\frac{\partial}{\partial t} p_{T|\bm{S}}(t|\bm{S}) \cdot \frac{\partial^2}{\partial t^2} \mu(t,\bm{S})}{\bar{p}_{T|\bm{S}}(t|\bm{S})} - \frac{p_{T|\bm{S}}(t|\bm{S}) \cdot \frac{\partial}{\partial t} \bar{p}_{T|\bm{S}}(t|\bm{S}) \cdot \frac{\partial^2}{\partial t^2} \mu(t,\bm{S})}{\bar{p}_{T|\bm{S}}^2(t|\bm{S})} + \frac{p_{T|\bm{S}}(t|\bm{S})\cdot \frac{\partial^3}{\partial t^3} \mu(t,\bm{S})}{3\bar{p}_{T|\bm{S}}(t|\bm{S})}\right].
	\end{align*}
	
    When $\bar{p}_{T|\bm{S}} = p_{T|\bm{S}}$, we have that
	$$\mathbb{E}_{\bm{S}}\left\{\frac{\mu(t,\bm{S})\cdot p_{T|\bm{S}}(t|\bm{S}) }{\bar{p}_{T|\bm{S}}(t|\bm{S})} \cdot \frac{\partial}{\partial t}\log \left[\frac{p_{T|\bm{S}}(t|\bm{S})}{\bar{p}_{T|\bm{S}}(t|\bm{S})}\right]\right\} + \mathbb{E}_{\bm{S}}\left\{\left[\beta(t,\bm{S}) - \bar{\beta}(t,\bm{S})\right]\left[\frac{p_{T|\bm{S}}(t|\bm{S})}{\bar{p}_{T|\bm{S}}(t|\bm{S})} - 1\right]\right\} = 0$$
	and $\tilde{B}_{\theta}(t) = \frac{\kappa_4}{6\kappa_2} \cdot \mathbb{E}_{\bm{S}}\left[\frac{\partial^3}{\partial t^3} \mu(t,\bm{S})\right]$. In this case, the dominating bias term is $h^2\tilde{B}_{\theta}(t)$, which tends to 0 as $h\to 0$ and $n\to \infty$. 
	
	However, when $\bar{\beta}=\beta$ (and $\bar{\mu}=\mu$), the dominating bias is equal to 
	$$\mathbb{E}_{\bm{S}}\left\{\frac{\mu(t,\bm{S})\cdot p_{T|\bm{S}}(t|\bm{S}) }{\bar{p}_{T|\bm{S}}(t|\bm{S})} \cdot \frac{\partial}{\partial t}\log \left[\frac{p_{T|\bm{S}}(t|\bm{S})}{\bar{p}_{T|\bm{S}}(t|\bm{S})}\right]\right\} + h^2\tilde{B}_{\theta}(t),$$ 
	which is not necessarily 0. 
	
	This also shows that the naive AIPW estimator \eqref{theta_AIPW} is not doubly robust.
\end{proof}

\section{Proof of \autoref{thm:nonp_eff}}
\label{app:nonp_eff_proof}

\begin{customthm}{2}[Efficient influence function]
Suppose that Assumptions~\ref{assump:positivity}, \ref{assump:reg_diff}, and \ref{assump:reg_kernel} hold for the nonparametric model containing $\P$. For any fixed bandwidth $h>0$ and $t\in \mathcal{T}$, the efficient influence function of $\varpi_{h,t}(\P)$ relative to this model is given by
\begin{align*}
	&\int_{\mathcal{T}} \frac{\mu(t_1,\bm{S})\left(\frac{t_1-t}{h}\right)  K\left(\frac{t_1-t}{h}\right)}{h^2 \cdot \kappa_2} \, dt_1 + \frac{\left[Y - \mu(T,\bm{S})\right] \left(\frac{T-t}{h}\right) K\left(\frac{T-t}{h}\right)}{h^2 \cdot \kappa_2\cdot p_{T|\bm{S}}(T|\bm{S})} - \varpi_{h,t}(\P) \\
	&= \beta(t,\bm{S}) + \frac{\left[Y - \mu(T,\bm{S})\right] \left(\frac{T-t}{h}\right) K\left(\frac{T-t}{h}\right)}{h^2 \cdot \kappa_2\cdot p_{T|\bm{S}}(T|\bm{S})} - \varpi_{h,t}(\P) + h^2 \cdot \tilde{B}_{\theta}(t),
\end{align*}
where $\tilde{B}_{\theta}(t) = \frac{\kappa_4}{6\kappa_2} \cdot \frac{\partial^3}{\partial t^3} \mu(t,\bm{s}) + o(h)$. Under the setup of \autoref{thm:theta_pos}, the resulting estimator of $\theta(t)$ can be written as:
$$\hat{\theta}_{\mathrm{DR,2}}(t) = \frac{1}{nh}\sum_{i=1}^n \left\{ \frac{\left(\frac{T_i-t}{h}\right)K\left(\frac{T_i-t}{h}\right) }{h\cdot \kappa_2\cdot \hat{p}_{T|\bm{S}}(T_i|\bm{S}_i)} \left[Y_i - \hat{\mu}(T_i,\bm{S}_i)
	\right]+ h\cdot \hat{\beta}(t,\bm{S}_i) \right\},$$
which has the same doubly robust properties and asymptotic variance as our proposed one $\hat{\theta}_{\mathrm{DR}}(t)$.
\end{customthm}

\begin{proof}[Proof of \autoref{thm:nonp_eff}]
We start with a parametric submodel \citep{VDV1998,kennedy2024semiparametric} with a family of density functions defined as:
$$p_{\epsilon}(y,t,\bm{s}) = p(y,t,\bm{s}) \left[1 + \epsilon \cdot g(y,t,\bm{s})\right],$$
where $g:\mathcal{Y}\times \mathcal{T}\times \mathcal{S} \to \mathbb{R}$ is a mean-zero function $\int g(y,t,\bm{s})\cdot p(y,t,\bm{s}) \,dydtd\bm{s} = 0$ such that $p_{\epsilon}(y,t,\bm{s}) \geq 0$ for any $\epsilon\in \mathbb{R}$ close to 0. We will also suppress the subscripts for density functions, such as denoting $p_{T|\bm{S}}(t|\bm{s})$ by $p(t|\bm{s})$. The efficient influence function for $\varpi_{h,t}(\P)$ is the mean-zero function that coincides with the canonical gradient $\lim\limits_{\epsilon \to 0} \frac{\varpi_{h,t}(\P_{\epsilon}) - \varpi_{h,t}(\P)}{\epsilon}$. Under this parametric submodel, we know that
$$\varpi_{h,t}(\P_{\epsilon}) = \int_{\mathcal{S}} \int_{\mathcal{T}} \int_{\mathcal{Y}} \frac{y\left(\frac{t_1-t}{h}\right) K\left(\frac{t_1-t}{h}\right)}{h^2\cdot \kappa_2\cdot  p_{\epsilon}(t_1|\bm{s})}\cdot p_{\epsilon}(y,t_1,\bm{s})\, dy dt_1 d\bm{s} \quad \text{ with } \quad p_{\epsilon}(t_1|\bm{s}) = \frac{p_{\epsilon}(t_1,\bm{s})}{p_{\epsilon}(\bm{s})}.$$
Then, we have that
\begin{align}
\label{diff3}
\begin{split}
	&\varpi_{h,t}(\P_{\epsilon}) - \varpi_{h,t}(\P) \\
	&= \int_{\mathcal{S}} \int_{\mathcal{T}} \int_{\mathcal{Y}} \frac{y\left(\frac{t_1-t}{h}\right) K\left(\frac{t_1-t}{h}\right)}{h^2\cdot \kappa_2\cdot  p_{\epsilon}(t_1|\bm{s})}\cdot p_{\epsilon}(y,t_1,\bm{s})\, dy dt_1 d\bm{s} - \int_{\mathcal{S}} \int_{\mathcal{T}} \int_{\mathcal{Y}} \frac{y\left(\frac{t_1-t}{h}\right) K\left(\frac{t_1-t}{h}\right)}{h^2\cdot \kappa_2\cdot  p(t_1|\bm{s})}\cdot p(y,t_1,\bm{s})\, dy dt_1 d\bm{s} \\
	&= \underbrace{\int_{\mathcal{S}} \int_{\mathcal{T}} \int_{\mathcal{Y}} \frac{y\left(\frac{t_1-t}{h}\right) K\left(\frac{t_1-t}{h}\right)}{h^2\cdot \kappa_2}\cdot p_{\epsilon}(y,t_1,\bm{s}) \left[\frac{1}{p_{\epsilon}(t_1|\bm{s})} - \frac{1}{p(t_1|\bm{s})}\right] dy dt_1 d\bm{s}}_{\textbf{Term I}} \\
	&\quad + \underbrace{\int_{\mathcal{S}} \int_{\mathcal{T}} \int_{\mathcal{Y}} \frac{y\left(\frac{t_1-t}{h}\right) K\left(\frac{t_1-t}{h}\right)}{h^2\cdot \kappa_2\cdot  p(t_1|\bm{s})}\left[p_{\epsilon}(y,t_1,\bm{s}) - p(y,t_1,\bm{s})\right] dy dt_1 d\bm{s}}_{\textbf{Term II}}.
\end{split}
\end{align}
Now, by Taylor's expansion, we compute that
\begin{align*}
&\frac{1}{p_{\epsilon}(t_1|\bm{s})} - \frac{1}{p(t_1|\bm{s})} \\
&= \frac{p_{\epsilon}(\bm{s})}{p_{\epsilon}(t_1,\bm{s})} - \frac{1}{p(t_1|\bm{s})} \\
&= \frac{p(\bm{s}) + \epsilon \int_{\mathcal{T}} \int_{\mathcal{Y}} p(y_2,t_2,\bm{s}) \cdot g(y_2,t_2,\bm{s}) dy_2 dt_2}{p(t_1,\bm{s}) + \epsilon \int_{\mathcal{Y}} p(y_2,t_1,\bm{s}) \cdot g(y_2,t_1,\bm{s})\, dy_2} - \frac{1}{p(t_1|\bm{s})}\\
&= \frac{p(\bm{s})}{p(t_1,\bm{s})} - \frac{1}{p(t_1|\bm{s})} + \epsilon\cdot \frac{\int_{\mathcal{T}} \int_{\mathcal{Y}} p(y_2,\bm{s},t_2) \cdot g(y_2,\bm{s},t_2) dy_2 dt_2}{p(t_1,\bm{s})}\\
&\quad - \epsilon\cdot \frac{p(\bm{s})}{p(t_1,\bm{s})^2} \int_{\mathcal{Y}} p(y_2,\bm{s},t_1)\cdot g(y_2,\bm{s},t_1)\, dy_2 + o(\epsilon)\\
&= \epsilon\left[\frac{\int_{\mathcal{T}} \int_{\mathcal{Y}} p(y_2,t_2,\bm{s}) \cdot g(y_2,t_2,\bm{s}) dy_2 dt_2}{p(t_1,\bm{s})} - \frac{p(\bm{s})}{p(t_1,\bm{s})^2} \int_{\mathcal{Y}} p(y_2,t_1,\bm{s})\cdot g(y_2,t_1,\bm{s})\, dy_2 \right] + o(\epsilon).
\end{align*}
Thus, we know that
\begin{align*}
	\textbf{Term I} &= \int_{\mathcal{S}} \int_{\mathcal{T}} \int_{\mathcal{Y}} \frac{y\left(\frac{t_1-t}{h}\right) K\left(\frac{t_1-t}{h}\right)}{h^2\cdot \kappa_2}\cdot p(y,t_1,\bm{s}) \left[1+ O(\epsilon)\right] \left[\frac{1}{p_{\epsilon}(t_1|\bm{s})} - \frac{1}{p(t_1|\bm{s})}\right] dy dt_1 d\bm{s} \\
	&= \epsilon \int_{\mathcal{S}} \int_{\mathcal{T}} \int_{\mathcal{Y}} \frac{y\left(\frac{t_1-t}{h}\right) K\left(\frac{t_1-t}{h}\right)\cdot p(y,t_1,\bm{s})}{h^2\cdot \kappa_2 \cdot p(t_1,\bm{s})} \left[\int_{\mathcal{T}} \int_{\mathcal{Y}} p(y_2,t_2,\bm{s}) \cdot g(y_2,t_2,\bm{s}) \,dy_2 dt_2 \right] dy dt_1 d\bm{s} \\
	&\quad - \epsilon \int_{\mathcal{S}} \int_{\mathcal{T}} \int_{\mathcal{Y}} \frac{y\left(\frac{t_1-t}{h}\right) K\left(\frac{t_1-t}{h}\right)\cdot p(y,t_1,\bm{s})}{h^2\cdot \kappa_2 \cdot p(t_1,\bm{s})\cdot p(t_1|\bm{s})} \left[\int_{\mathcal{Y}} p(y_2,t_1,\bm{s}) \cdot g(y_2,t_1,\bm{s}) \,dy_2 \right] dy dt_1 d\bm{s} + o(\epsilon) \\
	&= \epsilon \int_{\mathcal{S}} \int_{\mathcal{T}} \frac{\mu(t_1,\bm{s})\left(\frac{t_1-t}{h}\right) K\left(\frac{t_1-t}{h}\right)}{h^2\cdot \kappa_2} \left[\int_{\mathcal{T}} \int_{\mathcal{Y}} p(y_2,t_2,\bm{s}) \cdot g(y_2,t_2,\bm{s}) \, dt_2 dy_2 \right] dt_1 d\bm{s} \\
	&\quad - \epsilon \int_{\mathcal{S}} \int_{\mathcal{T}} \int_{\mathcal{Y}} \frac{\mu(t_1,\bm{s})\left(\frac{t_1-t}{h}\right) K\left(\frac{t_1-t}{h}\right)}{h^2\cdot \kappa_2 \cdot p(t_1|\bm{s})}\cdot p(y_2,t_1,\bm{s}) \cdot g(y_2,t_1,\bm{s}) \,dy_2 dt_1 d\bm{s} + o(\epsilon) \\
	&= \epsilon\left\langle \int_{\mathcal{T}} \frac{\mu(t_1,\bm{S})\left(\frac{t_1-t}{h}\right) K\left(\frac{t_1-t}{h}\right)}{h^2\cdot \kappa_2} dt_1 - \frac{\mu(T,\bm{S})\left(\frac{T-t}{h}\right) K\left(\frac{T-t}{h}\right)}{h^2\cdot \kappa_2 \cdot p(T|\bm{S})},\, g(Y,\bm{S},T) \right\rangle_{L_2(\P)} + o(\epsilon),
\end{align*}
where $\big\langle f(Y,T,\bm{S}),\, g(Y,\bm{S},T) \big\rangle_{L_2(\P)} =  \int_{\mathcal{S}} \int_{\mathcal{T}} \int_{\mathcal{Y}} f(y,t,\bm{s})\cdot g(y,t,\bm{s})\cdot p(y,t,\bm{s})\, dy dt d\bm{s}$.
Furthermore, 
\begin{align*}
	\textbf{Term II} &= \epsilon \int_{\mathcal{S}} \int_{\mathcal{T}} \int_{\mathcal{Y}} \frac{y\left(\frac{t_1-t}{h}\right) K\left(\frac{t_1-t}{h}\right)}{h^2\cdot \kappa_2\cdot  p(t_1|\bm{s})} p(y,t_1,\bm{s})\cdot g(y,t_1,\bm{s})\, dy dt_1 d\bm{s} \\
	&= \epsilon\left\langle \frac{Y\left(\frac{T-t}{h}\right) K\left(\frac{T-t}{h}\right)}{h^2\cdot \kappa_2\cdot  p(T|\bm{S})},\, g(Y,\bm{S},T) \right\rangle_{L_2(\P)}.
\end{align*}
Thus, we know that
\begin{align*}
	&\lim_{\epsilon\to 0}\frac{\varpi_{h,t}(\P_{\epsilon}) - \varpi_{h,t}(\P)}{\epsilon} \\
	&= \left\langle \int_{\mathcal{T}} \frac{\mu(t_1,\bm{S})\left(\frac{t_1-t}{h}\right)  K\left(\frac{t_1-t}{h}\right)}{h^2 \cdot \kappa_2} \, dt_1 + \frac{\left[Y - \mu(T,\bm{S})\right] \left(\frac{T-t}{h}\right) K\left(\frac{T-t}{h}\right)}{h^2 \cdot \kappa_2\cdot p(T|\bm{S})},\, g(Y,\bm{S},T) \right\rangle_{L_2(\P)}.
\end{align*}
In addition, $\left\langle C_0, g(Y,\bm{S},T) \right\rangle_{L_2(\P)} =0$ for any constant $C_0$ due to the constraint $\int_{\mathcal{S}} \int_{\mathcal{T}} \int_{\mathcal{Y}} g(y,t,\bm{s}) \cdot p(y,t,\bm{s})\, dy dt d\bm{s}=0$. Hence, 
\begin{align*}
&\lim_{\epsilon\to 0}\frac{\varpi_{h,t}(\P_{\epsilon}) - \varpi_{h,t}(\P)}{\epsilon} \\
&= \left\langle \int_{\mathcal{T}} \frac{\mu(t_1,\bm{S})\left(\frac{t_1-t}{h}\right)  K\left(\frac{t_1-t}{h}\right)}{h^2 \cdot \kappa_2} \, dt_1 + \frac{\left[Y - \mu(T,\bm{S})\right] \left(\frac{T-t}{h}\right) K\left(\frac{T-t}{h}\right)}{h^2 \cdot \kappa_2\cdot p(T|\bm{S})} + C_0,\, g(Y,\bm{S},T) \right\rangle_{L_2(\P)},
\end{align*}
and the efficient influence function is
$$L^*(Y,\bm{S},T) = \int_{\mathcal{T}} \frac{\mu(t_1,\bm{S})\left(\frac{t_1-t}{h}\right)  K\left(\frac{t_1-t}{h}\right)}{h^2 \cdot \kappa_2} \, dt_1 + \frac{\left[Y - \mu(T,\bm{S})\right] \left(\frac{T-t}{h}\right) K\left(\frac{T-t}{h}\right)}{h^2 \cdot \kappa_2\cdot p(T|\bm{S})} + C_0.$$
To determine $C_0$, we use the fact that $\mathbb{E}\left[L^*(Y,\bm{S},T) \right] =0$, which implies that 
\begin{align*}
	C_0 &= \mathbb{E}\left\{-\frac{\left[Y - \mu(T,\bm{S})\right] \left(\frac{T-t}{h}\right) K\left(\frac{T-t}{h}\right)}{h^2 \cdot \kappa_2\cdot p(T|\bm{S})} - \int_{\mathcal{T}} \frac{\mu(t_1,\bm{S})\left(\frac{t_1-t}{h}\right)  K\left(\frac{t_1-t}{h}\right)}{h^2 \cdot \kappa_2} \, dt_1 \right\}\\
	&= 0 - \int_{\mathcal{S}} \int_{\mathcal{T}} \frac{\mu(t_1,\bm{s})\left(\frac{t_1-t}{h}\right) K\left(\frac{t_1-t}{h}\right)}{h^2 \cdot \kappa_2 \cdot p(t_1|\bm{s})} \cdot p(t_1,\bm{s}) \, dt_1 d\bm{s}\\
	&= - \int_{\mathcal{S}} \int_{\mathcal{T}} \int_{\mathcal{Y}} \frac{y \left(\frac{t_1-t}{h}\right)  K\left(\frac{t_1-t}{h}\right)}{h^2\cdot \kappa_2\cdot  p(t_1|\bm{s})} \cdot p(y,t_1,\bm{s}) \,dy dt_1 d\bm{s}\\
	&= - \E\left[\frac{Y \left(\frac{T-t}{h}\right) K\left(\frac{T-t}{h}\right)}{h^2 \cdot \kappa_2\cdot p(T|\bm{S})} \right] \\
	&= -\varpi_{h,t}(\P).
\end{align*}
Therefore, the final efficient influence function for $\varpi_{h,t}(\P) = \E\left[\frac{Y \left(\frac{T-t}{h}\right) K\left(\frac{T-t}{h}\right)}{h^2 \cdot \kappa_2 \cdot p(T|\bm{S})} \right]$ for any fixed $h>0$ and $t\in \mathcal{T}$ is 
\begin{align*}
L^*(Y,\bm{S},T) &= \int_{\mathcal{T}} \frac{\mu(t_1,\bm{S})\left(\frac{t_1-t}{h}\right)  K\left(\frac{t_1-t}{h}\right)}{h^2 \cdot \kappa_2} \, dt_1 + \frac{\left[Y - \mu(T,\bm{S})\right] \left(\frac{T-t}{h}\right) K\left(\frac{T-t}{h}\right)}{h^2 \cdot \kappa_2\cdot p(T|\bm{S})} - \varpi_{h,t}(\P) \\
&= \int_{\mathbb{R}} \frac{\mu(t+uh,\bm{S}) \cdot u  \cdot K\left(u\right)}{h \cdot \kappa_2} \, du + \frac{\left[Y - \mu(T,\bm{S})\right] \left(\frac{T-t}{h}\right) K\left(\frac{T-t}{h}\right)}{h^2 \cdot \kappa_2\cdot p(T|\bm{S})} - \varpi_{h,t}(\P) \\
&= \int_{\mathbb{R}} \frac{u \cdot K\left(u\right)}{h \cdot \kappa_2} \left[\mu(t,\bm{s}) + uh\cdot \frac{\partial}{\partial t} \mu(t,\bm{s}) + \frac{u^2h^2}{2} \cdot \frac{\partial^2}{\partial t^2} \mu(t,\bm{s}) + \frac{u^3h^3}{6} \cdot \frac{\partial^3}{\partial t^3} \mu(t,\bm{s}) + o(h^3)\right] du \\
&\quad + \frac{\left[Y - \mu(T,\bm{S})\right] \left(\frac{T-t}{h}\right) K\left(\frac{T-t}{h}\right)}{h^2 \cdot \kappa_2\cdot p(T|\bm{S})} - \varpi_{h,t}(\P) \\
&= \frac{\partial}{\partial t}\mu(t,\bm{s}) + \frac{\left[Y - \mu(T,\bm{S})\right] \left(\frac{T-t}{h}\right) K\left(\frac{T-t}{h}\right)}{h^2 \cdot \kappa_2\cdot p(T|\bm{S})} - \varpi_{h,t}(\P) + \frac{h^2 \kappa_4}{6\kappa_2} \cdot \frac{\partial^3}{\partial t^3} \mu(t,\bm{s}) + o(h^2)\\
&=\beta(t,\bm{s}) + \frac{\left[Y - \mu(T,\bm{S})\right] \left(\frac{T-t}{h}\right) K\left(\frac{T-t}{h}\right)}{h^2 \cdot \kappa_2\cdot p(T|\bm{S})} - \varpi_{h,t}(\P) + h^2 \cdot \tilde{B}_{\theta}(t),
\end{align*}
where $\tilde{B}_{\theta}(t) = \frac{\kappa_4}{6\kappa_2} \cdot \frac{\partial^3}{\partial t^3} \mu(t,\bm{s}) + o(h)$. This gives rise to another doubly robust estimator of $\theta(t)$ as:
$$\hat{\theta}_{\mathrm{DR,2}}(t) = \frac{1}{nh}\sum_{i=1}^n \left\{ \frac{\left(\frac{T_i-t}{h}\right)K\left(\frac{T_i-t}{h}\right) }{h\cdot \kappa_2\cdot \hat{p}_{T|\bm{S}}(T_i|\bm{S}_i)} \left[Y_i - \hat{\mu}(T_i,\bm{S}_i)
\right]+ h\cdot \hat{\beta}(t,\bm{S}_i) \right\}.$$
One can follow our arguments in the proof of \autoref{thm:theta_pos} to derive the same asymptotic properties of $\hat{\theta}_{\mathrm{DR,2}}(t)$. In particular, we will demonstrate that $\hat{\theta}_{\mathrm{DR}}(t)$ in \eqref{theta_DR} and $\hat{\theta}_{\mathrm{DR,2}}(t)$ above have the same asymptotic variance.

By our arguments in the proof of \autoref{thm:theta_pos}, we know that
\begin{align*}
	\sqrt{nh^3} \left[\hat{\theta}_{\mathrm{DR,2}}(t) -\theta(t)\right] = \underbrace{\frac{1}{\sqrt{n}} \sum_{i=1}^n \frac{\left(\frac{T_i-t}{h}\right) K\left(\frac{T_i-t}{h}\right)}{\sqrt{h} \cdot \kappa_2\cdot \bar{p}_{T|\bm{S}}(T_i|\bm{S}_i)} \cdot \left[Y_i - \bar{\mu}(T_i,\bm{S}_i)\right]}_{\textbf{Term III}} + o_P(1)
\end{align*}
and
\begin{align*}
	\sqrt{nh^3} \left[\hat{\theta}_{\mathrm{DR}}(t) -\theta(t)\right] = \underbrace{\frac{1}{\sqrt{n}} \sum_{i=1}^n \frac{\left(\frac{T_i-t}{h}\right) K\left(\frac{T_i-t}{h}\right)}{\sqrt{h} \cdot \kappa_2\cdot \bar{p}_{T|\bm{S}}(T_i|\bm{S}_i)} \cdot \left[Y_i - \bar{\mu}(t,\bm{S}_i) - (T_i-t)\cdot \bar{\beta}(t,\bm{S}_i)\right]}_{\textbf{Term IV}} + o_P(1),
\end{align*}
where $\bar{\mu}(t,\bm{s})$, $\bar{\beta}(t,\bm{s})$, and $\bar{p}_{T|\bm{S}}(t|\bm{s})$ are fixed bounded functions to which $\hat{\mu}(t,\bm{s})$, $\hat{\beta}(t,\bm{s})$ and $\hat{p}_{T|\bm{S}}(t|\bm{s})$ converge. Thus, the asymptotic variances of $\hat{\theta}_{\mathrm{DR,2}}(t)$ and $\hat{\theta}_{\mathrm{DR}}(t)$ are given by $\mathrm{Var}\left(\textbf{Term III}\right)$ and $\mathrm{Var}\left(\textbf{Term IV}\right)$, respectively. We compute their difference as follows:
\begin{align*}
	&\mathrm{Var}\left(\textbf{Term III}\right) - \mathrm{Var}\left(\textbf{Term IV}\right) \\
	&=\mathrm{Var}\left\{\frac{\left(\frac{T_i-t}{h}\right) K\left(\frac{T_i-t}{h}\right)}{\sqrt{h} \cdot \kappa_2\cdot \bar{p}_{T|\bm{S}}(T_i|\bm{S}_i)} \cdot \left[Y_i - \bar{\mu}(T_i,\bm{S}_i)\right] \right\} \\
	&\quad - \mathrm{Var}\left\{\frac{\left(\frac{T_i-t}{h}\right) K\left(\frac{T_i-t}{h}\right)}{\sqrt{h} \cdot \kappa_2\cdot \bar{p}_{T|\bm{S}}(T_i|\bm{S}_i)} \cdot \left[Y_i - \bar{\mu}(t,\bm{S}_i) - (T_i-t)\cdot \bar{\beta}(t,\bm{S}_i)\right] \right\}\\
	&= \mathbb{E}\left\{\int_{\mathcal{T}} \frac{\left(\frac{t_1-t}{h}\right)^2 K^2\left(\frac{t_1-t}{h}\right)}{h \cdot \kappa_2^2\cdot \bar{p}_{T|\bm{S}}^2(t_1|\bm{S})} \cdot \left[\mu(t_1,\bm{S}) - \bar{\mu}(t_1,\bm{S})\right]^2 p_{T|\bm{S}}(t_1|\bm{S})\, dt_1 \right\} \\
	&\quad - \left[\mathbb{E}\left\{\int_{\mathcal{T}} \frac{\left(\frac{t_1-t}{h}\right) K\left(\frac{t_1-t}{h}\right)}{\sqrt{h} \cdot \kappa_2\cdot \bar{p}_{T|\bm{S}}(t_1|\bm{S})} \cdot \left[\mu(t_1,\bm{S}) - \bar{\mu}(t_1,\bm{S})\right] p_{T|\bm{S}}(t_1|\bm{S})\, dt_1 \right\}  \right]^2 \\
	&\quad - \mathbb{E}\left\{\int_{\mathcal{T}} \frac{\left(\frac{t_1-t}{h}\right)^2 K^2\left(\frac{t_1-t}{h}\right)}{h \cdot \kappa_2^2\cdot \bar{p}_{T|\bm{S}}^2(t_1|\bm{S})} \cdot \left[\mu(t_1,\bm{S}) - \bar{\mu}(t,\bm{S}) - (t_1-t)\cdot \bar{\beta}(t,\bm{S})\right]^2 p_{T|\bm{S}}(t_1|\bm{S})\, dt_1 \right\} \\
	&\quad + \left[\mathbb{E}\left\{\int_{\mathcal{T}} \frac{\left(\frac{t_1-t}{h}\right) K\left(\frac{t_1-t}{h}\right)}{\sqrt{h} \cdot \kappa_2\cdot \bar{p}_{T|\bm{S}}(t_1|\bm{S})} \cdot \left[\mu(t_1,\bm{S}) - \bar{\mu}(t,\bm{S}) - (t_1-t)\cdot \bar{\beta}(t,\bm{S})\right] p_{T|\bm{S}}(t_1|\bm{S})\, dt_1 \right\}  \right]^2 \\
	&= \mathbb{E}\left\{\int_{\mathbb{R}} \frac{u^2\cdot K^2(u)}{\kappa_2^2\cdot \bar{p}_{T|\bm{S}}^2(t+uh|\bm{S})} \cdot \left[\mu(t+uh,\bm{S}) - \bar{\mu}(t+uh,\bm{S})\right]^2 p_{T|\bm{S}}(t+uh|\bm{S})\, du \right\} \\
	&\quad - \left[\mathbb{E}\left\{\int_{\mathbb{R}} \frac{\sqrt{h}\cdot u \cdot K(u)}{\kappa_2\cdot \bar{p}_{T|\bm{S}}(t+uh|\bm{S})} \cdot \left[\mu(t+uh,\bm{S}) - \bar{\mu}(t+uh,\bm{S})\right] p_{T|\bm{S}}(t+uh|\bm{S})\, du \right\} \right]^2 \\
	&\quad - \mathbb{E}\left\{\int_{\mathbb{R}} \frac{u^2\cdot K^2\left(u\right)}{\kappa_2^2\cdot \bar{p}_{T|\bm{S}}^2(t+uh|\bm{S})} \cdot \left[\mu(t+uh,\bm{S}) - \bar{\mu}(t,\bm{S}) - uh\cdot \bar{\beta}(t,\bm{S})\right]^2 p_{T|\bm{S}}(t+uh|\bm{S})\, du \right\} \\
	&\quad + \left[\mathbb{E}\left\{\int_{\mathbb{R}} \frac{\sqrt{h} \cdot u\cdot K\left(u\right)}{\kappa_2\cdot \bar{p}_{T|\bm{S}}(t+uh|\bm{S})} \cdot \left[\mu(t+uh,\bm{S}) - \bar{\mu}(t,\bm{S}) - uh\cdot \bar{\beta}(t,\bm{S})\right] p_{T|\bm{S}}(t+uh|\bm{S})\, du \right\}  \right]^2 \\
	&\stackrel{\text{(i)}}{=} \mathbb{E}\Bigg\{\int_{\mathbb{R}} \frac{u^2\cdot K^2(u)}{\kappa_2^2\cdot \bar{p}_{T|\bm{S}}^2(t+uh|\bm{S})} \left[\mu(t+uh,\bm{S}) - \bar{\mu}(t,\bm{S}) - uh\cdot \bar{\beta}(t,\bm{S}) - \frac{u^2h^2}{2}\cdot \frac{\partial^2}{\partial t^2} \bar{\mu}(t,\bm{S}) + O(h^3)\right]^2 \\
	&\quad \quad \times p_{T|\bm{S}}(t+uh|\bm{S})\, du \Bigg\} \\
	&\quad - \mathbb{E}\left\{\int_{\mathbb{R}} \frac{u^2\cdot K^2\left(u\right)}{\kappa_2^2\cdot \bar{p}_{T|\bm{S}}^2(t+uh|\bm{S})}  \left[\mu(t+uh,\bm{S}) - \bar{\mu}(t,\bm{S}) - uh\cdot \bar{\beta}(t,\bm{S})\right]^2 p_{T|\bm{S}}(t+uh|\bm{S})\, du \right\} + O(h^3)\\
	&= -\mathbb{E}\left\{\int_{\mathbb{R}} \frac{h^2\cdot u^4\cdot K^2\left(u\right)}{\kappa_2^2\cdot \bar{p}_{T|\bm{S}}^2(t+uh|\bm{S})}  \left[\mu(t+uh,\bm{S}) - \bar{\mu}(t,\bm{S}) - uh\cdot \bar{\beta}(t,\bm{S})\right] \cdot \frac{\partial^2}{\partial t^2} \bar{\mu}(t,\bm{S}) \cdot p_{T|\bm{S}}(t+uh|\bm{S})\, du \right\} \\
	&\quad + O(h^3)\\
	&= O(h^2),
\end{align*}
where (i) uses the second-order kernel property with $\int_{\mathbb{R}} u\cdot K(u)\, du=0$ to argue that
\begin{align*}
	&\mathbb{E}\left\{\int_{\mathbb{R}} \frac{\sqrt{h} \cdot u\cdot K\left(u\right)}{\kappa_2\cdot \bar{p}_{T|\bm{S}}(t+uh|\bm{S})} \cdot \left[\mu(t+uh,\bm{S}) - \bar{\mu}(t,\bm{S}) - uh\cdot \bar{\beta}(t,\bm{S})\right] p_{T|\bm{S}}(t+uh|\bm{S})\, du \right\}\\
	&\asymp \mathbb{E}\left\{\int_{\mathbb{R}} \frac{\sqrt{h}\cdot u \cdot K(u)}{\kappa_2\cdot \bar{p}_{T|\bm{S}}(t+uh|\bm{S})} \cdot \left[\mu(t+uh,\bm{S}) - \bar{\mu}(t+uh,\bm{S})\right] p_{T|\bm{S}}(t+uh|\bm{S})\, du \right\} \\
	&= O\left(\sqrt{h^3}\right).
\end{align*}
As a result, the asymptotic variances of $\hat{\theta}_{\mathrm{DR,2}}(t)$ and $\hat{\theta}_{\mathrm{DR}}(t)$ are indeed the same when $n\to \infty$ and $h\to 0$.

Notice that if the regression model is correctly specified, \emph{i.e.}, $\bar{\mu}=\mu, \bar{\beta}=\beta$, then the differences between the asymptotic variances of $\hat{\theta}_{\mathrm{DR,2}}(t)$ and $\hat{\theta}_{\mathrm{DR}}(t)$ will be of smaller order as $O(h^4)$.
\end{proof}

\section{Proofs of Propositions~\ref{prop:IPW_bias_m}, \ref{prop:theta_IPW_new}, and \ref{prop:theta_IPW_new2}}
\label{app:IPW_bias_proof}

\subsection{Proof of Proposition~\ref{prop:IPW_bias_m}}
\label{app:IPW_bias_proof_m}

\begin{customprop}{3}[Inconsistency of IPW estimators]
Suppose that Assumptions~\ref{assump:id_cond}(a-c), \ref{assump:reg_diff}, \ref{assump:den_diff}(c), and \ref{assump:reg_kernel}(a-b) hold under the additive confounding model \eqref{add_conf_model}. Assume also that when the bandwidth $h$ is small, the Lebesgue measure of the symmetric difference set satisfies 
$$\left|\mathcal{S}(t+uh)\triangle \mathcal{S}(t)\right| = \left|\left[\mathcal{S}(t+uh)\setminus \mathcal{S}(t)\right] \cup \left[\mathcal{S}(t)\setminus \mathcal{S}(t+uh)\right]\right|=o(1)$$
for any $t\in \mathcal{T}$ and $u\in \mathbb{R}$. Then, when $h$ is small, the expectation of $\tilde{m}_{\mathrm{IPW}}(t)$ in \eqref{IPW_oracle} is given by
$$\E\left[\tilde{m}_{\mathrm{IPW}}(t)\right] = \bar{m}(t)\cdot \rho(t) + \omega(t) + o(1),$$
where $\rho(t) = \P\left(\bm{S}\in \mathcal{S}(t)\right)$ and $\omega(t) = \E\left[\eta(\bm{S}) \mathbbm{1}_{\{\bm{S}\in \mathcal{S}(t)\}}\right]$. If, in addition, there exists a constant $A_h>0$ depending on $h$ such that
\begin{equation*}
	\int_{\mathbb{R}}\E\left\{\left[\bar{m}(t) +\eta(\bm{S})\right]\left[\mathbbm{1}_{\{\bm{S}\in \mathcal{S}(t+uh)\setminus \mathcal{S}(t)\}} - \mathbbm{1}_{\{\bm{S}\in \mathcal{S}(t)\setminus \mathcal{S}(t+uh)\}}\right] \right\} u\cdot K(u)\, du=O(A_h)
\end{equation*}
for any $t\in \mathcal{T}$ and $u\in \mathbb{R}$ when $h$ is small, then the expectation of $\tilde{\theta}_{\mathrm{IPW}}(t)$ in \eqref{IPW_forms} is given by
$$\E\left[\tilde{\theta}_{\mathrm{IPW}}(t)\right]=\bar{m}'(t)\cdot \rho(t) + O\left(\frac{A_h}{h}\right).$$
\end{customprop}

\begin{proof}[Proof of Proposition~\ref{prop:IPW_bias_m}]
	Notice that the conditional density support $\mathcal{S}(t)$ depends on $t$ when the positivity condition is violated. Under the additive confounding model \eqref{add_conf_model}, we have that
	\begin{align*}
		\E\left[\tilde{m}_{\mathrm{IPW}}(t)\right] &= \E\left[\frac{1}{h}\frac{K\left(\frac{T_i-t}{h}\right)}{ p_{T|\bm{S}}(T_i|\bm{S}_i)}\cdot Y_i\right]\\
		&=\E\left\{\frac{1}{h}\frac{K\left(\frac{T_i-t}{h}\right)}{ p_{T|\bm{S}}(T_i|\bm{S}_i)}\left[\bar{m}(T_i)+\eta(\bm{S}_i)\right]\right\}\\
		&= \int_{\mathcal{T}} \int_{\mathcal{S}(\tilde{t})}\frac{1}{h}\frac{K(\frac{\tilde{t}-t}{h})}{ p_{T|\bm{S}}(\tilde{t}|\bm{s})}  \left[\bar{m}(\tilde{t})+\eta(\bm{s})\right] p(\tilde{t},\bm{s})\,d\bm{s}d\tilde{t}\\
		& =\int_{\mathcal{T}} \int_{\mathcal{S}(\tilde{t})}\frac{1}{h}K\left(\frac{\tilde{t}-t}{h}\right) \left[\bar{m}(\tilde{t})+\eta(\bm{s})\right] p_S(\bm{s})\,d\bm{s} d\tilde{t}\\ 
		&\stackrel{\text{(i)}}{=} \int_{\mathbb{R}} \int_{\mathcal{S}(t+uh)} K\left(u\right) \left[\bar{m}(t+uh)+\eta(\bm{s})\right] p_S(\bm{s})\,d\bm{s} du\\ 
		& \stackrel{\text{(ii)}}{=} \int_{\mathbb{R}} \int_{\mathcal{S}(t+uh)} K(u) \left[\bar{m}(t)+\eta(\bm{s}) + uh\cdot \bar{m}'(t) + O(h^2)\right] p_S(\bm{s})\,d\bm{s} du \\
		&= \int_{\mathbb{R}} \int_{\mathcal{S}(t)} K(u) \left[\bar{m}(t)+\eta(\bm{s}) + uh\cdot \bar{m}'(t) + O(h^2)\right] p_S(\bm{s})\,d\bm{s} du \\
		&\quad + \int_{\mathbb{R}} \int_{\mathcal{S}(t+uh)\setminus \mathcal{S}(t)} K(u) \left[\bar{m}(t)+\eta(\bm{s}) + uh\cdot \bar{m}'(t) + O(h^2)\right] p_S(\bm{s})\,d\bm{s} du\\
		& \quad - \int_{\mathbb{R}} \int_{\mathcal{S}(t) \setminus \mathcal{S}(t+uh)} K(u) \left[\bar{m}(t)+\eta(\bm{s}) + uh\cdot \bar{m}'(t) + O(h^2)\right] p_S(\bm{s})\,d\bm{s} du\\
		&\stackrel{\text{(iii)}}{=} \int_{\mathcal{S}(t)} \left[\bar{m}(t)+\eta(\bm{s})\right] \cdot p_S(\bm{s}) d\bm{s} + O(h^2) +o(1)\\
		&= \bar{m}(t) \cdot \P(\bm{S}\in\mathcal{S}(t)) + \E\left[\eta(\bm{S})\cdot  \mathbbm{1}_{\{\bm{S}\in \mathcal{S}(t)\}}\right] +o(1) \\
		&= \bar{m}(t)\cdot \rho(t) + \omega(t) + o(1),
	\end{align*}
	where (i) follows from a change of variable $u=\frac{\tilde{t}-t}{h}$, (ii) is due to Taylor's expansion under Assumption~\ref{assump:reg_diff}, and (iii) relies on our assumption on the Lebesgue measure $\left|\mathcal{S}(t+uh)\triangle \mathcal{S}(t)\right| = o(1)$ for any $t\in \mathcal{T}$ and $u\in \mathbb{R}$.
	
	Similarly, we can also derive that
	\begin{align*}
		&\mathbb{E}\left[\tilde{\theta}_{\mathrm{IPW}}(t) \right]\\ &=\E\left[\frac{\left[\bar{m}(T_i) +\eta(\bm{S}_i)\right] \left(\frac{T_i-t}{h}\right) K\left(\frac{T_i-t}{h}\right)}{h^2\cdot \kappa_2\cdot p_{T|\bm{S}}(T_i|\bm{S}_i)}\right]\\
		&= \int_{\mathcal{T}} \int_{\mathcal{S}(t_1)} \frac{\left[\bar{m}(t_1) + \eta(\bm{s}_1)\right]\left(\frac{t_1-t}{h}\right) K\left(\frac{t_1-t}{h}\right)}{h^2\cdot \kappa_2} \cdot p_S(\bm{s}_1) \, d\bm{s}_1 dt_1\\
		&\stackrel{\text{(iv)}}{=} \int_{\mathbb{R}} \int_{\mathcal{S}(t+uh)} \frac{\left[\bar{m}(t+uh) + \eta(\bm{s}_1)\right]u\cdot K(u)}{h\cdot \kappa_2} \cdot p_S(\bm{s}_1) \, d\bm{s}_1 du\\
		&=\int_{\mathbb{R}} \int_{\mathcal{S}(t)} \frac{\left[\bar{m}(t+uh) + \eta(\bm{s}_1)\right]u\cdot K(u)}{h\cdot \kappa_2} \cdot p_S(\bm{s}_1) \, d\bm{s}_1 du\\
		&\quad + \int_{\mathbb{R}}\left[\int_{\mathcal{S}(t+uh) \setminus \mathcal{S}(t)} - \int_{\mathcal{S}(t) \setminus \mathcal{S}(t+uh)}\right] \frac{\left[\bar{m}(t+uh) + \eta(\bm{s}_1)\right]u\cdot K(u)}{h\cdot \kappa_2} \cdot p_S(\bm{s}_1) \, d\bm{s}_1 du\\
		&\stackrel{\text{(v)}}{=}\int_{\mathbb{R}} \int_{\mathcal{S}(t)} \frac{\left[\bar{m}(t) + \eta(\bm{s}_1) + uh\cdot \bar{m}'(t) + O(h^2)\right]u\cdot K(u)}{h\cdot \kappa_2} \cdot p_S(\bm{s}_1) \, d\bm{s}_1 du\\
		&\quad + \int_{\mathbb{R}}\left[\int_{\mathcal{S}(t+uh) \setminus \mathcal{S}(t)} - \int_{\mathcal{S}(t) \setminus \mathcal{S}(t+uh)}\right] \frac{\left[\bar{m}(t) + \eta(\bm{s}_1) + uh\cdot \bar{m}'(t) + O(h^2)\right]u\cdot K(u)}{h\cdot \kappa_2} \cdot p_S(\bm{s}_1) \, d\bm{s}_1 du\\
		&\stackrel{\text{(vi)}}{=} \int_{\mathcal{S}(t)} \bar{m}'(t)\cdot p_S(\bm{s}_1)\, d\bm{s}_1 + O(h) \\
		&\quad + \frac{1}{h\cdot \kappa_2} \int_{\mathbb{R}}\E\left\{\left[\bar{m}(t) +\eta(\bm{S})\right]\left[\mathbbm{1}_{\{\bm{S}\in \mathcal{S}(t+uh)\setminus \mathcal{S}(t)\}} - \mathbbm{1}_{\{\bm{S}\in \mathcal{S}(t)\setminus \mathcal{S}(t+uh)\}}\right] \right\} u\cdot K(u)\, du\\
		&\quad + \frac{1}{\kappa_2}\int_{\mathbb{R}}\E\left\{\bar{m}'(t)\left[\mathbbm{1}_{\{\bm{S}\in \mathcal{S}(t+uh)\setminus \mathcal{S}(t)\}} - \mathbbm{1}_{\{\bm{S}\in \mathcal{S}(t)\setminus \mathcal{S}(t+uh)\}}\right] \right\} u^2 K(u)\, du\\
		&\stackrel{\text{(vii)}}{=} \bar{m}'(t)\cdot \rho(t) + O\left(\frac{A_h}{h}\right),
	\end{align*}
    where (iv) follows from a change of variable $u=\frac{\tilde{t}-t}{h}$, (v) is due to Taylor's expansion under Assumption~\ref{assump:reg_diff}, (vi) leverages the property of the second-order kernel $K$, and (vii) relies on our assumptions on the Lebesgue measure $\left|\mathcal{S}(t+uh)\triangle \mathcal{S}(t)\right| = o(1)$ and \eqref{quadratic_cond} for any $t\in \mathcal{T}$ and $u\in \mathbb{R}$. The results follow.
\end{proof}

\subsection{Proof of Proposition~\ref{prop:theta_IPW_new}}
\label{app:IPW_new_proof}

\begin{customprop}{4}
Suppose that Assumptions~\ref{assump:id_cond}(a-c), \ref{assump:reg_diff}, \ref{assump:den_diff}(c), and \ref{assump:reg_kernel}(a-b) hold under the additive confounding model \eqref{add_conf_model}. Then, when the bandwidth $h$ is small, the expectation of the modified IPW quantity \eqref{IPW_quantity} is given by
\begin{align*}
	&\mathbb{E}\left[\tilde{\Xi}_t(Y,T,\bm{S})\right] = \bar{m}'(t) + O(h^2) \\
	&\quad\quad + \int_{\mathbb{R}}\E\left\{\left[\bar{m}(t+uh) +\eta(\bm{S})\right]\left[\mathbbm{1}_{\{\bm{S}\in \mathcal{S}(t+uh)\setminus \mathcal{S}(t)\}} - \mathbbm{1}_{\{\bm{S}\in \mathcal{S}(t)\setminus \mathcal{S}(t+uh)\}}\right] \Big| T=t\right\} u\cdot K(u)\, du.
\end{align*}
\end{customprop}

\begin{proof}[Proof of Proposition~\ref{prop:theta_IPW_new}]
Recall that the conditional density support $\mathcal{S}(t)$ depends on $t$ when the positivity condition fails to hold. Under the additive confounding model \eqref{add_conf_model}, we have that
\begin{align*}
	&\mathbb{E}\left[\tilde{\Xi}_t(Y,T,\bm{S})\right] \\
	&= \int_{\mathcal{T}} \int_{\mathcal{S}(t_1)} \frac{\left[\bar{m}(t_1) + \eta(\bm{s}_1)\right] \left(\frac{t_1-t}{h}\right) K\left(\frac{t_1-t}{h}\right) p_{\bm{S}|T}(\bm{s}_1|t)}{h^2 \cdot \kappa_2} \, d\bm{s}_1 dt_1\\
	&= \int_{\mathcal{T}} \int_{\mathcal{S}(t)} \frac{\left[\bar{m}(t_1) + \eta(\bm{s}_1)\right] \left(\frac{t_1-t}{h}\right) K\left(\frac{t_1-t}{h}\right) p_{\bm{S}|T}(\bm{s}_1|t)}{h^2 \cdot \kappa_2} \, d\bm{s}_1 dt_1 \\
	&\quad + \left\{\left[\int_{\mathcal{T}} \int_{\mathcal{S}(t_1) \setminus \mathcal{S}(t)} - \int_{\mathcal{T}} \int_{\mathcal{S}(t) \setminus \mathcal{S}(t_1)} \right] \frac{\left[\bar{m}(t_1) + \eta(\bm{s}_1)\right] \left(\frac{t_1-t}{h}\right) K\left(\frac{t_1-t}{h}\right) p_{\bm{S}|T}(\bm{s}_1|t)}{h^2 \cdot \kappa_2} \, d\bm{s}_1 dt_1 \right\}\\
	&\stackrel{\text{(i)}}{=} \int_{\mathbb{R}} \int_{\mathcal{S}(t)} \frac{\left[\bar{m}(t+uh) + \eta(\bm{s}_1)\right] u \cdot K(u) \cdot p_{\bm{S}|T}(\bm{s}_1|t)}{\kappa_2 \cdot h} \,d\bm{s}_1 du\\
	&\quad + \left\{\int_{\mathbb{R}} \left[\int_{\mathcal{S}(t+uh) \setminus \mathcal{S}(t)} - \int_{\mathcal{S}(t) \setminus \mathcal{S}(t+uh)} \right] \frac{\left[\bar{m}(t+uh) + \eta(\bm{s}_1)\right] u\cdot K(u) p_{\bm{S}|T}(\bm{s}_1|t)}{h \cdot \kappa_2} \, d\bm{s}_1 du \right\}\\
	&\stackrel{\text{(ii)}}{=} \int_{\mathbb{R}} \int_{\mathcal{S}(t)} \frac{\left[\bar{m}(t) +\eta(\bm{s}_1) + uh \cdot \bar{m}'(t) + \frac{u^2h^2}{2} \cdot \bar{m}''(t) + O\left(h^3\right)\right] u \cdot K(u) \cdot p(\bm{s}_1|t)}{\kappa_2 \cdot h} \,d\bm{s}_1 du\\
	&\quad + \left\{\int_{\mathbb{R}}\left[\int_{\mathcal{S}(t+uh) \setminus \mathcal{S}(t)} - \int_{\mathcal{S}(t) \setminus \mathcal{S}(t+uh)} \right] \frac{\left[\bar{m}(t+uh) + \eta(\bm{s}_1)\right] u\cdot K(u) p_{\bm{S}|T}(\bm{s}_1|t)}{h \cdot \kappa_2} \, d\bm{s}_1 du \right\}\\
	&\stackrel{\text{(iii)}}{=} \bar{m}'(t) + O(h^2)\\
	&\quad + \left\{\int_{\mathbb{R}}\left[ \int_{\mathcal{S}(t+uh) \setminus \mathcal{S}(t)} - \int_{\mathcal{S}(t) \setminus \mathcal{S}(t+uh)} \right] \frac{\left[\bar{m}(t+uh) + \eta(\bm{s}_1)\right] u\cdot K(u) p_{\bm{S}|T}(\bm{s}_1|t)}{h \cdot \kappa_2} \, d\bm{s}_1 du \right\}\\
	&= \bar{m}'(t) + O(h^2) \\
	&\quad + \int_{\mathbb{R}}\E\left\{\left[\bar{m}(t+uh) +\eta(\bm{S})\right]\left[\mathbbm{1}_{\{\bm{S}\in \mathcal{S}(t+uh)\setminus \mathcal{S}(t)\}} - \mathbbm{1}_{\{\bm{S}\in \mathcal{S}(t)\setminus \mathcal{S}(t+uh)\}}\right] \Big| T=t\right\} u\cdot K(u)\, du,
\end{align*}
where (i) follows from a change of variable $u=\frac{t_1-t}{h}$,  (ii) is due to Taylor's expansion, and (iii) utilizes the fact that $K$ is a second-order kernel function by Assumption~\ref{assump:reg_kernel}(b). The result thus follows.\\

As stated in Remark~\ref{remark:IPW_modified}, we should evaluate the conditional density $p_{\bm{S}|T}$ at the (query) point $(t,\bm{S})$ instead of the sample point $(T,\bm{S})$ in the modified IPW quantity \eqref{IPW_quantity}. To see this, we consider the alternative modified IPW quantity 
$$\tilde{\Xi}_{t,2}(Y,T,\bm{S}) = \frac{Y\left(\frac{T-t}{h}\right) K\left(\frac{T-t}{h}\right) p_{\bm{S}|T}(\bm{S}|T)}{h^2\cdot \kappa_2 \cdot p(T,\bm{S})}$$
and compute its expectation as:
\begin{align*}
	&\mathbb{E}\left[\tilde{\Xi}_{t,2}(Y,T,\bm{S}) \right] \\
	&= \int_{\mathcal{T}} \int_{\mathcal{S}(t_1)} \frac{\left[\bar{m}(t_1) + \eta(\bm{s}_1)\right]\left(\frac{t_1-t}{h}\right) K\left(\frac{t_1-t}{h}\right) \cdot p_{\bm{S}|T}(\bm{s}_1|t_1)}{h^2\cdot \kappa_2} \,d\bm{s}_1 dt_1 \\
	&\stackrel{\text{(iv)}}{=} \int_{\mathbb{R}} \int_{\mathcal{S}(t+uh)} \frac{\left[\bar{m}(t+uh) + \eta(\bm{s}_1)\right] u \cdot K\left(u\right) \cdot p_{\bm{S}|T}(\bm{s}_1|t+uh)}{h\cdot \kappa_2} \,d\bm{s}_1 du \\
	&\stackrel{\text{(v)}}{=} \int_{\mathbb{R}} \int_{\mathcal{S}(t+uh)} \frac{\left[\bar{m}(t) + \eta(\bm{s}_1) + uh\cdot \bar{m}'(t) + \frac{u^2h^2}{2}\cdot \bar{m}''(t) + O(h^3)\right] u \cdot K\left(u\right) \cdot p_{\bm{S}|T}(\bm{s}_1|t+uh)}{h\cdot \kappa_2} \,d\bm{s}_1 du \\
	&\stackrel{\text{(vi)}}{=} \bar{m}'(t) + O(h^2) + \int_{\mathbb{R}} \int_{\mathcal{S}(t+uh)} \frac{ \eta(\bm{s}_1) \cdot u \cdot K\left(u\right) \cdot p_{\bm{S}|T}(\bm{s}_1|t+uh)}{h\cdot \kappa_2}\, d\bm{s}_1 du \\
	&= \bar{m}'(t) + O(h^2) + \int_{\mathcal{T}} \frac{\left(\frac{t_1-t}{h}\right) K\left(\frac{t_1-t}{h}\right)}{h^2\cdot \kappa_2} \cdot \mathbb{E}\left[\eta(\bm{S})|T = t_1\right]\, dt_1,
\end{align*}
where (i) follows from a change of variable $u=\frac{t_1-t}{h}$,  (ii) is due to Taylor's expansion, and (iii) utilizes the fact that $K$ is a second-order kernel function by Assumption~\ref{assump:reg_kernel}(b). Hence, it is unclear how we can eliminate the additional bias term $\int_{\mathcal{T}} \frac{\left(\frac{t_1-t}{h}\right) K\left(\frac{t_1-t}{h}\right)}{h^2\cdot \kappa_2} \cdot \mathbb{E}\left[\eta(\bm{S})|T = t_1\right]\, dt_1$ from $\mathbb{E}\left[\tilde{\Xi}_{t,2}(Y,T,\bm{S}) \right]$ unless $\E\left[\eta(\bm{S})|T=t\right]=0$, which is not true in general.
\end{proof}

\subsection{Proof of Proposition~\ref{prop:theta_IPW_new2}}
\label{app:IPW_new_proof2}

\begin{customprop}{5}
Suppose that Assumptions~\ref{assump:id_cond}(a-c), \ref{assump:reg_diff}, \ref{assump:den_diff}(c), \ref{assump:reg_kernel}(a-b), and \ref{assump:cond_support_smooth} hold under the additive confounding model \eqref{add_conf_model}. Then, when the bandwidth $h>0$ is small, the expectation of the modified IPW quantity \eqref{IPW_quantity2} is given by
$$\mathbb{E}\left[\tilde{\Xi}_{t,\zeta}(Y,T,\bm{S}) \right] = \bar{m}'(t) + \frac{h^2\kappa_4}{6\kappa_2}\cdot \bar{m}^{(3)}(t) + O\left(h^3\right).$$
\end{customprop}

\begin{proof}[Proof of Proposition~\ref{prop:theta_IPW_new2}]
	Since the kernel function $K$ has a compact support under Assumption~\ref{assump:reg_kernel}, we can assume, without loss of generality, that it is supported on $[-1,1]$. Then, when $h < \frac{\zeta}{A_0}$ (or, equivalently $A_0|uh| < \zeta$ for any $u\in [-1,1]$), we have that
	$$\mathcal{S}(t) \ominus \zeta \subset \mathcal{S}(t) \ominus \left(A_0|uh|\right) \subset S(t+uh) \quad \text{ and } \quad \mathcal{L}_{\zeta}(t) \subset \mathcal{L}_{A_0|\delta|}(t) \subset \mathcal{S}(t+\delta)$$
	by Assumption~\ref{assump:cond_support_smooth}. Then, under model \eqref{add_conf_model} and the support shrinking approach for $p_{\zeta}(\bm{s}|t)$, the expectation of $\tilde{\Xi}_{t,\zeta}(Y,T,\bm{S})$ is given by
	\begin{align*}
		&\mathbb{E}\left[\tilde{\Xi}_{t,\zeta}(Y,T,\bm{S})\right] \\
		&= \int_{\mathcal{T}} \int_{\mathcal{S}(t_1)} \frac{\left[\bar{m}(t_1) + \eta(\bm{s}_1)\right] \left(\frac{t_1-t}{h}\right)K\left(\frac{t_1-t}{h}\right) \cdot p_{\zeta}(\bm{s}_1|t)}{\kappa_2 \cdot h} \,d\bm{s}_1 du \\
		&= \int_{\mathbb{R}} \int_{\mathcal{S}(t+uh)} \frac{\left[\bar{m}(t+uh) + \eta(\bm{s}_1)\right]\cdot u \cdot K(u) \cdot p_{\zeta}(\bm{s}_1|t)}{\kappa_2 \cdot h} \,d\bm{s}_1 du \\
		&\stackrel{\text{(i)}}{=} \int_{\mathbb{R}} \int_{\mathcal{S}(t) \ominus \zeta} \frac{\left[\bar{m}(t+uh) + \eta(\bm{s}_1)\right] \cdot u \cdot K(u) \cdot p_{\zeta}(\bm{s}_1|t)}{\kappa_2 \cdot h} \,d\bm{s}_1 du\\
		&= \int_{\mathbb{R}} \int_{\mathcal{S}(t) \ominus \zeta} \frac{\left[\bar{m}(t) + \eta(\bm{s}_1) + uh \cdot \bar{m}'(t) + \frac{u^2h^2}{2} \cdot \bar{m}''(t) + \frac{u^3h^3}{6} \cdot \bar{m}^{(3)}(t) + O\left(h^4\right)\right] u\cdot K(u) \cdot p_{\zeta}(\bm{s}_1|t)}{\kappa_2 \cdot h} \,d\bm{s}_1 du\\
		&\stackrel{\text{(ii)}}{=} \underbrace{\int_{\mathbb{R}} \int_{\mathcal{S}(t) \ominus \zeta} \frac{\left[\bar{m}(t) + \eta(\bm{s}_1) \right]\cdot u \cdot K(u) \cdot p_{\zeta}(\bm{s}_1|t)}{\kappa_2 \cdot h} \,d\bm{s}_1 du}_{=0} + \int_{\mathbb{R}} \int_{\mathcal{S}(t) \ominus \zeta} \frac{\bar{m}'(t)\cdot u^2 K(u) \cdot p_{\zeta}(\bm{s}_1|t)}{\kappa_2} \,d\bm{s}_1 du\\
		&\quad + \underbrace{\int_{\mathbb{R}} \int_{\mathcal{S}(t) \ominus \zeta} \frac{ h\cdot \bar{m}''(t) \cdot u^3 K(u) \cdot p_{\zeta}(\bm{s}_1|t)}{2\kappa_2} \,d\bm{s}_1 du}_{=0} + \int_{\mathbb{R}} \int_{\mathcal{S}(t) \ominus \zeta} \frac{ h^2\cdot \bar{m}^{(3)}(t) \cdot u^4 K(u) \cdot p_{\zeta}(\bm{s}_1|t)}{6\kappa_2} \,d\bm{s}_1 du \\
		&\quad + O\left(h^3\right)\\
		&= \bar{m}'(t) + \frac{h^2\kappa_4}{6\kappa_2}\cdot \bar{m}^{(3)}(t) + O\left(h^3\right)\\
		&= \mathbb{E}\left[\frac{\partial}{\partial t}\mu(t,\bm{S}) \Big| T=t\right] + \frac{h^2\kappa_4}{6\kappa_2}\cdot \bar{m}^{(3)}(t) + O\left(h^3\right),
	\end{align*}
	where (i) uses the definition of the $\zeta$-interior conditional density \eqref{interior_cond_density} and (ii) follows from the fact that $K$ is a second-order kernel function under Assumption~\ref{assump:reg_kernel}(b). The result under the level set approach for $p_{\zeta}(\bm{s}|t)$ follows from almost identical arguments.
\end{proof}

\section{Proof of \autoref{thm:theta_nopos}}
\label{app:theta_nopos_proof}

\begin{customthm}{6}[Consistency of estimating $\theta(t)$ without positivity]
	Suppose that Assumptions~\ref{assump:id_cond}(a-c), \ref{assump:reg_diff}, \ref{assump:den_diff}, \ref{assump:reg_kernel}, and \ref{assump:cond_support_smooth} hold under the additive confounding model \eqref{add_conf_model}, and the support $\mathcal{S} \subset \mathbb{R}^d$ of the marginal density $p_{\bm{S}}$ is compact. In addition, $\hat{\mu},\hat{\beta}, \hat{p}_{\zeta}, \hat{p}$ are constructed on a data sample independent of $\{(Y_i,T_i,\bm{S}_i)\}_{i=1}^n$. For any fixed $t\in \mathcal{T}$, we let $\bar{\mu}(t,\bm{s})$, $\bar{\beta}(t,\bm{s})$, $\bar{p}_{\zeta}(\bm{s}|t)$, and $\bar{p}(t,\bm{s})$ be fixed bounded functions to which $\hat{\mu}(t,\bm{s})$, $\hat{\beta}(t,\bm{s})$, $\hat{p}_{\zeta}(\bm{s}|t)$, and $\hat{p}(t,\bm{s})$ converge under the rates of convergence as:
	\begin{align*}
		&\norm{\hat{\beta}(t,\bm{S}) - \bar{\beta}(t,\bm{S})}_{L_2} = O_P\left(\Upsilon_{3,n}\right), \quad \sup_{\bm{s}\in \mathcal{S}}\left|\hat{F}_{\bm{S}|T}(\bm{s}|t) - F_{\bm{S}|T}(\bm{s}|t)\right| = O_P\left(\Upsilon_{4,n}\right), \\
		& \norm{\hat{p}_{\zeta}(\bm{S}|t) - \bar{p}_{\zeta}(\bm{S}|t)}_{L_2}=O_P\left(\Upsilon_{5,n}\right), \quad \text{ and } \quad \sup_{|u-t|\leq h}\norm{\hat{p}(u,\bm{S}) - \bar{p}(u,\bm{S})}_{L_2}=O_P\left(\Upsilon_{6,n}\right),
	\end{align*}
	where $\Upsilon_{3,n},\Upsilon_{4,n},\Upsilon_{5,n}, \Upsilon_{6,n}\to 0$ as $n\to \infty$. Then, as $h\to 0$ and $nh^3\to \infty$, we have that
	\begin{align*}
		&\hat{\theta}_{\mathrm{C,RA}}(t) - \theta(t) = O_P\left(\Upsilon_{3,n}+\Upsilon_{4,n}+\norm{\bar{\beta}(t,\bm{S}) - \beta(t,\bm{S})}_{L_2}\right), \\ 
		&\hat{\theta}_{\mathrm{C,IPW}}(t) - \theta(t) = O(h^2) + O_P\left(\sqrt{\frac{1}{nh^3}} + \Upsilon_{5,n} + \Upsilon_{6,n} + \sup_{|u-t|\leq h} \norm{\bar{p}(u,\bm{S}) - p(u,\bm{S})}_{L_2}\right). 
	\end{align*}
	If, in addition, we assume that 
	\begin{enumerate}[label=(\alph*)]
		\item $\bar{p},\bar{p}_{\zeta}$ satisfy Assumptions~\ref{assump:den_diff} and \ref{assump:cond_support_smooth} as well as $\sqrt{nh^3}\cdot \Upsilon_{5,n}=o(1)$;
		\item either (i) ``$\,\bar{\mu}=\mu$ and $\bar{\beta}=\beta$'' or (ii) ``$\,\bar{p} = p$'';
		\item {\small $\sqrt{nh} \left[\norm{\hat{p}_{\zeta}(\bm{S}|t) - \bar{p}_{\zeta}(\bm{S}|t)}_{L_2} + \sup\limits_{|u-t|\leq h} \norm{\hat{p}(u,\bm{S}) - p(u,\bm{S})}_{L_2} \right]\left[\norm{\hat{\mu}(t,\bm{S}) - \mu(t,\bm{S})}_{L_2} + h \norm{\hat{\beta}(t,\bm{S}) - \beta(t,\bm{S})}_{L_2}\right] = o_P(1)$},
	\end{enumerate}
	then 
	\begin{align*}
		&\sqrt{nh^3}\left[\hat{\theta}_{\mathrm{C,DR}}(t) - \theta(t)\right] \\
		&= \frac{1}{\sqrt{n}} \sum_{i=1}^n \left\{\phi_{C,h,t}\left(Y_i,T_i,\bm{S}_i;\bar{\mu}, \bar{\beta}, \bar{p},\bar{p}_{\zeta}\right) + \sqrt{h^3}\left[\int \bar{\beta}(t,\bm{s})\cdot \bar{p}_{\zeta}(\bm{s}|t)\, d\bm{s} - \theta(t)\right] \right\} + o_P(1)
	\end{align*}
	when $nh^7\to c_3$ for some finite number $c_3\geq 0$, where $$\phi_{C,h,t}\left(Y,T,\bm{S}; \bar{\mu},\bar{\beta}, \bar{p},\bar{p}_{\zeta}\right) = \frac{\left(\frac{T-t}{h}\right) K\left(\frac{T-t}{h}\right) \cdot \bar{p}_{\zeta}(\bm{S}|t)}{\sqrt{h}\cdot \kappa_2\cdot \bar{p}(T,\bm{S})}\cdot \left[Y - \bar{\mu}(t,\bm{S}) - (T-t)\cdot \bar{\beta}(t,\bm{S})\right].$$
	Furthermore, 
	$$\sqrt{nh^3}\left[\hat{\theta}_{\mathrm{C,DR}}(t) - \theta(t) - h^2 B_{C,\theta}(t)\right] \stackrel{d}{\to} \mathcal{N}\left(0,V_{C,\theta}(t)\right)$$
	with $V_{C,\theta}(t) = \mathbb{E}\left[\phi_{C,h,t}^2\left(Y,T,\bm{S};\bar{\mu}, \bar{\beta}, \bar{p},\bar{p}_{\zeta}\right)\right]$ and 
	\begin{align*}
		B_{C,\theta}(t) &= \begin{cases}
			\frac{\kappa_4}{6\kappa_2} \int \left\{\frac{3\frac{\partial}{\partial t} p(t,\bm{s}) \cdot \bar{m}''(t) + p(t,\bm{s})\left[\bar{m}^{(3)}(t) - 3\frac{\partial}{\partial t} \log\bar{p}(t,\bm{s}) \cdot \bar{m}''(t) \right]}{\bar{p}(t,\bm{s})} \right\} \bar{p}_{\zeta}(\bm{s}|t)\, d\bm{s}& \text{ when } \bar{\mu}=\mu \text{ and } \bar{\beta}=\beta,\\
			\frac{\kappa_4}{6\kappa_2} \cdot \bar{m}^{(3)}(t) & \text{ when } \bar{p} = p.
		\end{cases}
	\end{align*}
\end{customthm}

\begin{proof}[Proof of \autoref{thm:theta_nopos}]
	We derive the rates of convergence of $\hat{\theta}_{\mathrm{C,RA}}(t)$ given by \eqref{theta_RA_corrected} and $\hat{\theta}_{\mathrm{IPW}}(t)$ given by \eqref{theta_IPW_bnd_corrected} in \autoref{app:theta_RA_nopos} and \autoref{app:theta_IPW_nopos}, respectively. We also prove the asymptotic linearity, double robustness, and asymptotic normality of $\hat{\theta}_{\mathrm{C,DR}}(t)$ given by \eqref{theta_DR_bnd_corrected} in \autoref{app:theta_DR_nopos}.
	
	\subsection{Rate of Convergence of $\hat{\theta}_{\mathrm{C,RA}}(t)$}
	\label{app:theta_RA_nopos}
	
	Firstly, we derive the rate of convergence for $\hat{\theta}_{\mathrm{C,RA}}(t)$ in \eqref{theta_RA_corrected}. By Proposition~\ref{prop:id_additive}, we know that $\theta(t)=\bar{m}'(t) = \mathbb{E}\left[\frac{\partial}{\partial t}\mu(T,\bm{S}) \Big| T=t\right] = \mathbb{E}\left[\beta(T,\bm{S}) \big| T=t\right]$ and 
	\begin{align*}
		&\hat{\theta}_{\mathrm{C,RA}}(t) - \theta(t) \\
		&= \int \hat{\beta}(t,\bm{s}) \, d\hat{F}_{\bm{S}|T}(\bm{s}|t) - \int \beta(t,\bm{s})\, dF_{\bm{S}|T}(\bm{s}|t)\\
		&= \underbrace{\int \left[\hat{\beta}(t,\bm{s}) - \bar{\beta}(t,\bm{s}) \right] d\hat{F}_{\bm{S}|T}(\bm{s}|t)}_{\textbf{Term I}} + \underbrace{\int \bar{\beta}(t,\bm{s})\, d\left[\hat{F}_{\bm{S}|T}(\bm{s}|t) - F_{\bm{S}|T}(\bm{s}|t)\right]}_{\textbf{Term II}} + \underbrace{\int \left[\bar{\beta}(t,\bm{s}) - \beta(t,\bm{s}) \right] d F_{\bm{S}|T}(\bm{s}|t)}_{\textbf{Term III}}.
	\end{align*}
	
	$\bullet$ {\bf Term I:} By Markov’s inequality (and H\"older's inequality), we know that
	$$\left|\hat{\beta}(t,\bm{S}_1) - \bar{\beta}(t,\bm{S}_1)\right| = O_P\left(\mathbb{E}\left|\hat{\beta}(t,\bm{S}_1) - \bar{\beta}(t,\bm{S}_1)\right|\right) =O_P\left(\norm{\hat{\beta}(t,\bm{S}) - \bar{\beta}(t,\bm{S})}_{L_2}\right)=O_P\left(\Upsilon_{3,n}\right)$$
	for any random vector $\bm{S}_1$ supported on $\mathcal{S}\subset \mathbb{R}^d$. Thus, 
	$$\textbf{Term I} \leq \int \left|\hat{\beta}(t,\bm{s}) - \bar{\beta}(t,\bm{s}) \right| d\hat{F}_{\bm{S}|T}(\bm{s}|t) = O_P\left(\norm{\hat{\beta}(t,\bm{S}) - \bar{\beta}(t,\bm{S})}_{L_2}\right)=O_P\left(\Upsilon_{3,n}\right).$$
	
	$\bullet$ {\bf Term II:} By the compactness of $\mathcal{S}$ and the fact that $\mathcal{S}(t) \subset \mathcal{S}$, we know that the Lebesgue measure $|\mathcal{S}(t)|$ satisfies $|\mathcal{S}(t)|\leq |\mathcal{S}|<\infty$ for any $t\in \mathcal{T}$ and thus,
	\begin{align*}
		\textbf{Term II} &\leq \sup_{\bm{s}\in \mathcal{S}}\left|\bar{\beta}(t,\bm{s})\right| \cdot \norm{\hat{F}_{\bm{S}|T}(\cdot|t) - F_{\bm{S}|T}(\cdot |t)}_{\mathrm{TV}} \\
		&\leq \sup_{\bm{s}\in \mathcal{S}}\left|\bar{\beta}(t,\bm{s})\right| \cdot \sup_{\bm{s}\in \mathcal{S}}\left|\hat{F}_{\bm{S}|T}(\bm{s}|t) - F_{\bm{S}|T}(\bm{s}|t)\right| \cdot |\mathcal{S}|\\
		&=O_P(\Upsilon_{4,n})
	\end{align*}
    under Assumption~\ref{assump:reg_diff} and the condition that $\sup\limits_{\bm{s}\in \mathcal{S}}\left|\hat{F}_{\bm{S}|T}(\bm{s}|t) - F_{\bm{S}|T}(\bm{s}|t)\right| = O_P\left(\Upsilon_{4,n}\right)$, where $\norm{\hat{F}_{\bm{S}|T}(\cdot|t) - F_{\bm{S}|T}(\cdot |t)}_{\mathrm{TV}}$ is the total variation distance between the probability measures associated with $\hat{F}_{\bm{S}|T}(\cdot|t)$ and $F_{\bm{S}|T}(\cdot |t)$. Notice that $\hat{F}_{\bm{S}|T}(\cdot|t)$ can be constructed on the same data sample $\left\{(Y_i,T_i,\bm{S}_i)\right\}_{i=1}^n$.\\
	
	$\bullet$ {\bf Term III:} Similar to the argument for \textbf{Term I}, we have that
	$$\textbf{Term III} \leq \int \left|\bar{\beta}(t,\bm{s}) - \beta(t,\bm{s}) \right| d\hat{F}_{\bm{S}|T}(\bm{s}|t) = O_P\left(\norm{\bar{\beta}(t,\bm{S}) - \beta(t,\bm{S})}_{L_2}\right).$$
	
	In summary, we conclude that
	$$\hat{\theta}_{\mathrm{C,RA}}(t) - \theta(t) = O_P\left(\Upsilon_{3,n}+\Upsilon_{4,n}+\norm{\bar{\beta}(t,\bm{S}) - \beta(t,\bm{S})}_{L_2}\right).$$

	\subsection{Rate of Convergence of $\hat{\theta}_{\mathrm{C,IPW}}(t)$}
	\label{app:theta_IPW_nopos}
	
	Secondly, we derive the rate of convergence for $\hat{\theta}_{\mathrm{C,IPW}}(t)$ in \eqref{theta_IPW_bnd_corrected}. Recall from \eqref{IPW_quantity2} that
	\begin{align*}
		\hat{\theta}_{\mathrm{C,IPW}}(t) - \theta(t) 
		&= \mathbb{P}_n\left[\tilde{\Xi}_{t,\zeta}(Y,T,\bm{S}) \right] - \theta(t) + \hat{\theta}_{\mathrm{C,IPW}}(t) - \mathbb{P}_n\left[\tilde{\Xi}_{t,\zeta}(Y,T,\bm{S}) \right]\\
		&= \underbrace{\frac{1}{nh}\sum_{i=1}^n \frac{Y_i\left(\frac{T_i-t}{h}\right) K\left(\frac{T_i-t}{h}\right) \bar{p}_{\zeta}(\bm{S}_i|t)}{h\cdot \kappa_2 \cdot p(T_i,\bm{S}_i)} - \mathbb{E}\left[\beta(T,\bm{S}) \big|T=t\right]}_{\textbf{Term IV}} \\
		&\quad + \underbrace{\frac{1}{nh}\sum_{i=1}^n \frac{Y_i\left(\frac{T_i-t}{h}\right) K\left(\frac{T_i-t}{h}\right)}{h\cdot \kappa_2} \left[\frac{\hat{p}_{\zeta}(\bm{S}_i|t)}{\hat{p}(T_i,\bm{S}_i)} - \frac{\bar{p}_{\zeta}(\bm{S}_i|t)}{p(T_i,\bm{S}_i)}\right]}_{\textbf{Term V}}.
	\end{align*}
	We shall handle \textbf{Term IV} and \textbf{Term V} in \autoref{app:theta_nopos_term_IV} and \autoref{app:theta_nopos_term_V}, respectively.
	
	\subsubsection{Rate of Convergence of \textbf{Term IV} for $\hat{\theta}_{\mathrm{C,IPW}}(t)$}
	\label{app:theta_nopos_term_IV}
	
	We already computed in Proposition~\ref{prop:theta_IPW_new2} that 
	$$\mathbb{E}\left[\tilde{\Xi}_{t,\zeta}(Y,T,\bm{S}) \right] - \theta(t) = \frac{h^2\kappa_4}{6\kappa_2}\cdot \bar{m}^{(3)}(t) + O\left(h^3\right)$$
	under model \eqref{add_conf_model}. In particular, the above equality holds true even when we replace the true $\zeta$-interior conditional density $p_{\zeta}$ with the limiting $\zeta$-interior conditional density $\bar{p}_{\zeta}$ because $\bar{p}_{\zeta}$ also satisfies Assumptions~\ref{assump:den_diff} and \ref{assump:cond_support_smooth}. Additionally, we calculate the variance of $\mathbb{P}_n\left[\tilde{\Xi}_{t,\zeta}(Y,T,\bm{S})\right]$ under $\bar{p}_{\zeta}$ as:
	\begin{align*}
		&\mathrm{Var}\left\{\mathbb{P}_n\left[\tilde{\Xi}_{t,\zeta}(Y,T,\bm{S})\right]\right\} \\
		&=\frac{1}{nh^4\kappa_2^2}\cdot\mathrm{Var}\left[\frac{Y_i\left(\frac{T_i-t}{h}\right)K\left(\frac{T_i-t}{h}\right) \cdot \bar{p}_{\zeta}(\bm{S}_i|t)}{p(T_i,\bm{S}_i)}\right] \\
		&= \frac{1}{nh^4\kappa_2^2}\cdot\mathbb{E}\left[\frac{Y_i^2\left(\frac{T_i-t}{h}\right)^2 K^2\left(\frac{T_i-t}{h}\right) \cdot \bar{p}_{\zeta}^2(\bm{S}_i|t)}{p^2(T_i,\bm{S}_i)}\right] - \frac{1}{nh^4\kappa_2^2} \left\{\mathbb{E}\left[\frac{Y_i\left(\frac{T_i-t}{h}\right)K\left(\frac{T_i-t}{h}\right) \cdot \bar{p}_{\zeta}(\bm{S}_i|t)}{p(T_i,\bm{S}_i)}\right]\right\}^2\\
		&\stackrel{\text{(i)}}{=} \frac{1}{nh^4\kappa_2^2} \int_{\mathcal{S}\times \mathcal{T}} \frac{\left(\frac{t_1-t}{h}\right)^2K^2\left(\frac{t_1-t}{h}\right) \cdot \bar{p}_{\zeta}^2(\bm{s}_1|t) }{p(t_1,\bm{s}_1)} \cdot \left[\mu(t_1,\bm{s}_1)^2 +\sigma^2\right] \, dt_1 d\bm{s}_1 - \frac{\left\{\mathbb{E}\left[\beta(T,\bm{S})\big| T=t\right]\right\}^2}{n} \\
		&\quad + O\left(\frac{h^2}{n}\right)\\
		&\stackrel{\text{(ii)}}{=} \frac{1}{nh^3\kappa_2^2} \int_{\mathcal{S}} \int_{\mathbb{R}} \frac{u^2K^2(u)\cdot \bar{p}_{\zeta}^2(\bm{s}_1|t)}{p(t+uh,\bm{s}_1)} \cdot \left[\mu(t+uh,\bm{s}_1)^2 +\sigma^2\right] \, du d\bm{s}_1 + O\left(\frac{1}{n}\right) \\
		&\stackrel{\text{(iii)}}{=} \frac{1}{nh^3\kappa_2^2} \int_{\mathcal{S}} \int_{\mathbb{R}} \frac{u^2K^2(u) \cdot \bar{p}_{\zeta}^2(\bm{s}_1|t)}{p(t,\bm{s}_1) + uh \cdot \frac{\partial}{\partial t}p(t',\bm{s}_1)} \left[\mu(t,\bm{s}_1)^2 + 2uh\cdot \mu(t'',\bm{s}_1)\cdot \frac{\partial}{\partial t} \mu(t'',\bm{s}_1) + \sigma^2\right]\, du d\bm{s}_1 \\
		&\quad + O\left(\frac{1}{n}\right)\\
		&\stackrel{\text{(iv)}}{=} \frac{1}{nh^3\kappa_2^2} \int_{\mathcal{S}} \int_{\mathbb{R}} \frac{u^2 K^2(u) \cdot \bar{p}_{\zeta}^2(\bm{s}_1|t)}{p(t,\bm{s}_1)} \left[\mu(t,\bm{s}_1)^2 + \sigma^2\right] \, du d\bm{s}_1 + O\left(\frac{1}{n}\right)\\
		&\stackrel{\text{(v)}}{=} O\left(\frac{1}{nh^3}\right)
	\end{align*}
    with $\mu(t,\bm{s})=\bar{m}(t)+\eta(\bm{s})$ under model \eqref{add_conf_model}, where (i) utilizes our result in Proposition~\ref{prop:theta_IPW_new2} for $\mathbb{E}\left[\tilde{\Xi}_{t,\zeta}(Y,T,\bm{S}) \right]$, (ii) uses a change of variable $u = \frac{t_1-t}{h}$ and the boundedness of $\beta(t,\bm{s})$, (iii) applies the Taylor's expansion under Assumptions~\ref{assump:reg_diff} and \ref{assump:den_diff} with $t',t''$ being two points between $t$ and $t+uh$, (iv) absorbs the higher order terms to $O\left(\frac{1}{n}\right)$, and (iv) utilizes the properties of $K$ under Assumption~\ref{assump:reg_kernel} and the positivity of $\frac{\bar{p}_{\zeta}^2(\bm{s}|t)}{p(t,\bm{s})}$ within the support $\mathcal{J}$ of $p(t,\bm{s})$. Now, by Chebyshev's inequality and our above calculations, we obtain that
	\begin{align*}
		\mathbb{P}_n\left[\tilde{\Xi}_{t,\zeta}(Y,T,\bm{S}) \right] - \theta(t) &= \mathbb{P}_n\left[\tilde{\Xi}_{t,\zeta}(Y,T,\bm{S}) \right] - \mathbb{E}\left[\tilde{\Xi}_{t,\zeta}(Y,T,\bm{S})\right] + \mathbb{E}\left[\tilde{\Xi}_{t,\zeta}(Y,T,\bm{S})\right] - \theta(t)\\
		&= O_P\left(\sqrt{\mathrm{Var}\left\{\mathbb{P}_n\left[\tilde{\Xi}_{t,\zeta}(Y,T,\bm{S}) \right]\right\}}\right) + O(h^2)\\
		&= O_P\left(\sqrt{\frac{1}{nh^3}}\right) + O(h^2)
	\end{align*}
	as $h\to 0$ and $nh^3\to \infty$. As a side note, under the VC-type condition on $K$ (Assumption~\ref{assump:reg_kernel}(c)), we can apply Theorem 2 in \cite{einmahl2005uniform} to strengthen the above pointwise rate of convergence to the uniform one as:
	$$\sup_{t\in \mathcal{T}}\left|\mathbb{P}_n\left[\tilde{\Xi}_{t,\zeta}(Y,T,\bm{S}) \right] - \theta(t) \right| = O_P\left(\sqrt{\frac{|\log h|}{nh^3}}\right) + O(h^2).$$
	
	\subsubsection{Rate of Convergence of \textbf{Term V} for $\hat{\theta}_{\mathrm{C,IPW}}(t)$}
	\label{app:theta_nopos_term_V}
	
	By direct calculations, we have that
	\begin{align*}
		&\textbf{Term V} \\
		&= \frac{1}{nh}\sum_{i=1}^n \frac{Y_i\left(\frac{T_i-t}{h}\right) K\left(\frac{T_i-t}{h}\right) \cdot \bar{p}_{\zeta}(\bm{S}_i|t)}{h\cdot \kappa_2 \cdot p(T_i,\bm{S}_i)} \left[\frac{\hat{p}_{\zeta}(\bm{S}_i|t) \cdot p(T_i,\bm{S}_i) - \bar{p}_{\zeta}(\bm{S}_i|t) \cdot \hat{p}(T_i,\bm{S}_i)}{\hat{p}(T_i,\bm{S}_i)\cdot \bar{p}_{\zeta}(\bm{S}_i|t)} \right]\\
		&= \frac{1}{nh}\sum_{i=1}^n \frac{Y_i\left(\frac{T_i-t}{h}\right) K\left(\frac{T_i-t}{h}\right) \cdot \bar{p}_{\zeta}(\bm{S}_i|t)}{h\cdot \kappa_2 \cdot p(T_i,\bm{S}_i)} \left\{\frac{\left[\hat{p}_{\zeta}(\bm{S}_i|t) - \bar{p}_{\zeta}(\bm{S}_i|t)\right] p(T_i,\bm{S}_i) - \bar{p}_{\zeta}(\bm{S}_i|t)\left[ \hat{p}(T_i,\bm{S}_i) - p(T_i,\bm{S}_i)\right]}{\hat{p}(T_i,\bm{S}_i)\cdot \bar{p}_{\zeta}(\bm{S}_i|t)} \right\}\\
		&\stackrel{\text{(i)}}{=} \left\{\mathbb{E}\left[\beta(T,\bm{S})\big| T=t\right] + O(h^2) +O_P\left(\sqrt{\frac{1}{nh^3}}\right)\right\}\\
		&\quad \times \Bigg[\frac{O_P\left(\norm{\hat{p}_{\zeta}(\bm{S}|t) - \bar{p}_{\zeta}(\bm{S}|t)}_{L_2} \right)}{\inf\limits_{(t,\bm{s})\in \mathcal{T}\times \mathcal{S}}p(t,\bm{s}) - O_P\left(\sup\limits_{|u-t|\leq h} \left[\norm{\hat{p}(u,\bm{S}) - \bar{p}(u,\bm{S})}_{L_2} + \norm{\bar{p}(u,\bm{S}) - p(u,\bm{S})}_{L_2} \right] \right)} \\
		&\quad \quad + \frac{O_P\left(\sup\limits_{|u-t|\leq h} \left[\norm{\hat{p}(u,\bm{S}) - \bar{p}(u,\bm{S})}_{L_2} + \norm{\bar{p}(u,\bm{S}) - p(u,\bm{S})}_{L_2} \right] \right)}{\inf\limits_{(t,\bm{s})\in \mathcal{T}\times \mathcal{S}}p(t,\bm{s}) - O_P\left(\sup\limits_{|u-t|\leq h} \left[\norm{\hat{p}(u,\bm{S}) - \bar{p}(u,\bm{S})}_{L_2} + \norm{\bar{p}(u,\bm{S}) - p(u,\bm{S})}_{L_2} \right] \right)}\Bigg]\\
		&= \left[O(1+h^2) + O_P\left(\sqrt{\frac{1}{nh^3}}\right)\right] \cdot O_P\left(\Upsilon_{5,n} + \Upsilon_{6,n} + \sup_{|u-t|\leq h} \norm{\bar{p}(u,\bm{S}) - p(u,\bm{S})}_{L_2}\right) \\
		&= O_P\left(\Upsilon_{5,n} + \Upsilon_{6,n} + \sup_{|u-t|\leq h} \norm{\bar{p}(u,\bm{S}) - p(u,\bm{S})}_{L_2}\right)
	\end{align*}
	as $h\to 0$ and $nh^3\to \infty$, where (i) utilizes our results for \textbf{Term IV} and Markov’s inequality.
	
	Combining our results for \textbf{Term IV} and \textbf{Term V}, we conclude that
	\begin{align*}
		&\hat{\theta}_{\mathrm{C,IPW}}(t) - \theta(t) = O(h^2) + O_P\left(\sqrt{\frac{1}{nh^3}} + \Upsilon_{5,n} + \Upsilon_{6,n} + \sup_{|u-t|\leq h} \norm{\bar{p}(u,\bm{S}) - p(u,\bm{S})}_{L_2}\right).
	\end{align*}
	
	\subsection{Asymptotic Properties of $\hat{\theta}_{\mathrm{C,DR}}(t)$}
	\label{app:theta_DR_nopos}
	
	Finally, using some similar arguments to \autoref{app:theta_DR_pos}, we shall establish the asymptotic properties of $\hat{\theta}_{\mathrm{C,DR}}(t)$ in \eqref{theta_DR_bnd_corrected}. By Proposition~\ref{prop:id_additive}, we have that 
	$$\theta(t) = \mathbb{E}\left[\frac{\partial}{\partial t}\mu(T,\bm{S}) \big| T=t\right] = \int \beta(t,\bm{s}) \, dF_{\bm{S}|T}(\bm{s}|t) = \int \beta(t,\bm{s}) \cdot p_{\zeta}(\bm{s}|t)\, d\bm{s} = \int \beta(t,\bm{s})\cdot \bar{p}_{\zeta}(\bm{s}|t)\, d\bm{s}$$
	where $\mu(t,\bm{s})=\bar{m}(t) + \eta(\bm{s})$. 
	Therefore, 
	\begin{align*}
		&\hat{\theta}_{\mathrm{C,DR}}(t) - \theta(t) \\
		&= \frac{1}{n}\sum_{i=1}^n \frac{\left(\frac{T_i-t}{h}\right)K\left(\frac{T_i-t}{h}\right)\cdot \hat{p}_{\zeta}(\bm{S}_i|t)}{h^2\cdot \kappa_2\cdot \hat{p}(T_i,\bm{S}_i)} \left[Y_i - \hat{\mu}(t,\bm{S}_i) - (T_i-t)\cdot \hat{\beta}(t,\bm{S}_i)\right]+ \int \hat{\beta}(t,\bm{s})\cdot \hat{p}_{\zeta}(\bm{s}|t)\, d\bm{s} \\
		&\quad - \int \beta(t,\bm{s})\cdot \bar{p}_{\zeta}(\bm{s}|t)\, d\bm{s} \\
		&= \mathbb{P}_n \left[\frac{1}{\sqrt{h^3}}\cdot \phi_{C,h,t}\left(Y,T,\bm{S}; \bar{\mu},\bar{\beta}, \bar{p}, \bar{p}_{\zeta}\right)\right] + \int \bar{\beta}(t,\bm{s})\cdot \bar{p}_{\zeta}(\bm{s}|t)\, d\bm{s} - \int \beta(t,\bm{s})\cdot \bar{p}_{\zeta}(\bm{s}|t)\, d\bm{s}\\
		&\quad + \int \hat{\beta}(t,\bm{s})\cdot \hat{p}_{\zeta}(\bm{s}|t)\, d\bm{s} - \int \bar{\beta}(t,\bm{s})\cdot \bar{p}_{\zeta}(\bm{s}|t)\, d\bm{s} \\
		&\quad + \mathbb{P}_n \left[\frac{1}{\sqrt{h^3}}\cdot \phi_{C,h,t}\left(Y,T,\bm{S}; \hat{\mu}, \hat{\beta}, \hat{p}, \hat{p}_{\zeta}\right) - \frac{1}{\sqrt{h^3}}\cdot \phi_{C,h,t}\left(Y,T,\bm{S}; \bar{\mu},\bar{\beta}, \bar{p}, \bar{p}_{\zeta}\right) \right] \\
		&= \underbrace{\mathbb{P}_n \left[\frac{1}{\sqrt{h^3}}\cdot \phi_{C,h,t}\left(Y,T,\bm{S}; \bar{\mu},\bar{\beta}, \bar{p}, \bar{p}_{\zeta}\right)\right] + \int \bar{\beta}(t,\bm{s})\cdot \bar{p}_{\zeta}(\bm{s}|t)\, d\bm{s} - \int \beta(t,\bm{s})\cdot \bar{p}_{\zeta}(\bm{s}|t)\, d\bm{s}}_{\textbf{Term VI}} \\
		&\quad + \underbrace{\int \hat{\beta}(t,\bm{s})\left[\hat{p}_{\zeta}(\bm{s}|t) - \bar{p}_{\zeta}(\bm{s}|t)\right]\, d\bm{s}}_{\textbf{Term VII}} \\
		&\quad + \underbrace{\left(\mathbb{P}_n-\P\right)\left\{\frac{\left(\frac{T-t}{h}\right) K\left(\frac{T-t}{h}\right)}{h^2\kappa_2}\left[\frac{\hat{p}_{\zeta}(\bm{S}|t)}{\hat{p}(T,\bm{S})} - \frac{\bar{p}_{\zeta}(\bm{S}|t)}{\bar{p}(T,\bm{S})} \right] \left[Y - \bar{\mu}(t,\bm{S}) - (T-t)\cdot \bar{\beta}(t,\bm{S})\right]\right\}}_{\textbf{Term VIII}}\\
		&\quad + \underbrace{\left(\mathbb{P}_n-\P\right)\left\{\frac{\left(\frac{T-t}{h}\right)K\left(\frac{T-t}{h}\right)\cdot \bar{p}_{\zeta}(\bm{S}|t)}{h^2\kappa_2\cdot \bar{p}(T,\bm{S})} \left[\bar{\mu}(t,\bm{S}) - \hat{\mu}(t,\bm{S}) + (T-t)\left[\bar{\beta}(t,\bm{S}) - \hat{\beta}(t,\bm{S})\right] \right]\right\}}_{\textbf{Term IX}} \\
		&\quad + \underbrace{\mathbb{P}_n\left\{\frac{\left(\frac{T-t}{h}\right) K\left(\frac{T-t}{h}\right)}{h^2\kappa_2}\left[\frac{\hat{p}_{\zeta}(\bm{S}|t)}{\hat{p}(T,\bm{S})} - \frac{\bar{p}_{\zeta}(\bm{S}|t)}{\bar{p}(T,\bm{S})} \right] \left[\bar{\mu}(t,\bm{S}) - \hat{\mu}(t,\bm{S}) + (T-t)\left[\bar{\beta}(t,\bm{S}) - \hat{\beta}(t,\bm{S})\right]\right]\right\}}_{\textbf{Term X}} \\
		&\quad + \underbrace{\P\left\{\frac{\left(\frac{T-t}{h}\right)^2 K\left(\frac{T-t}{h}\right) \cdot \bar{p}_{\zeta}(\bm{S}|t)}{h \cdot \kappa_2\cdot \bar{p}(T,\bm{S})} \left[\bar{\beta}(t,\bm{S}) - \hat{\beta}(t,\bm{S})\right]\right\} + \int \left[\hat{\beta}(t,\bm{s}) - \bar{\beta}(t,\bm{s})\right]\bar{p}_{\zeta}(\bm{s}|t)\, d\bm{s}}_{\textbf{Term XIa}} \\ 
		&\quad + \underbrace{\P\left\{\frac{\left(\frac{T-t}{h}\right) K\left(\frac{T-t}{h}\right)\cdot \bar{p}_{\zeta}(\bm{S}|t)}{h^2 \cdot \kappa_2\cdot \bar{p}(T,\bm{S})}\left[\bar{\mu}(t,\bm{S}) - \hat{\mu}(t,\bm{S})\right]\right\}}_{\textbf{Term XIb}}\\
		&\quad + \underbrace{\P\left\{\frac{\left(\frac{T-t}{h}\right) K\left(\frac{T-t}{h}\right)}{h^2\kappa_2}\left[\frac{\hat{p}_{\zeta}(\bm{S}|t)}{\hat{p}(T,\bm{S})} - \frac{\bar{p}_{\zeta}(\bm{S}|t)}{\bar{p}(T,\bm{S})} \right] \left[Y - \bar{\mu}(t,\bm{S}) - (T-t)\cdot \bar{\beta}(t,\bm{S})\right]\right\}}_{\textbf{Term XIc}},
	\end{align*}
	where $\phi_{C,h,t}\left(Y,T,\bm{S}; \bar{\mu},\bar{\beta}, \bar{p},\bar{p}_{\zeta}\right) = \frac{\left(\frac{T-t}{h}\right) K\left(\frac{T-t}{h}\right) \cdot \bar{p}_{\zeta}(\bm{S}|t)}{\sqrt{h}\cdot \kappa_2\cdot \bar{p}(T,\bm{S})}\cdot \left[Y - \bar{\mu}(t,\bm{S}) - (T-t)\cdot \bar{\beta}(t,\bm{S})\right]$. It remains to show that the dominating \textbf{Term VI} is of order $O(h^2)+O_P\left(\sqrt{\frac{1}{nh^3}}\right)$ in \autoref{app:theta_nopos_Term_VI} and the remainder terms are of order $o_P\left(\sqrt{\frac{1}{nh^3}}\right)$ for any fixed $t\in \mathcal{T}$ in \autoref{app:theta_nopos_Term_VII}, \autoref{app:theta_nopos_Term_VIII}, and \autoref{app:theta_nopos_Term_X}, and \autoref{app:theta_nopos_Term_XI}. We shall also derive the asymptotic normality of $\hat{\theta}_{\mathrm{C,DR}}(t)$ in \autoref{app:theta_nopos_asym_norm}.
	
	\subsubsection{Analysis of \textbf{Term VI} for $\hat{\theta}_{\mathrm{C,DR}}(t)$}
	\label{app:theta_nopos_Term_VI}
	
	We analyze the variance and bias of \textbf{Term VI} separately as follows. By direct calculations, we have that
	\begin{align*}
		\mathrm{Var}\left[\textbf{Term VI}\right] &= \mathrm{Var}\left\{\mathbb{P}_n\left[\frac{1}{\sqrt{h^3}} \cdot \phi_{C,h,t}\left(Y,T,\bm{S}; \bar{\mu},\bar{\beta}, \bar{p},\bar{p}_{\zeta}\right)\right]\right\} \\
		&= \frac{1}{nh^4} \cdot \mathrm{Var}\left\{\frac{\left(\frac{T-t}{h}\right) K\left(\frac{T-t}{h}\right) \cdot \bar{p}_{\zeta}(\bm{S}|t)}{\kappa_2\cdot \bar{p}(T,\bm{S})}\cdot \left[Y - \bar{\mu}(t,\bm{S}) - (T-t)\cdot \bar{\beta}(t,\bm{S})\right]\right\}\\
		&= \frac{1}{nh^4} \cdot \mathbb{E}\left\{\frac{\left(\frac{T-t}{h}\right)^2 K^2\left(\frac{T-t}{h}\right) \cdot \bar{p}_{\zeta}^2(\bm{S}|t)}{\kappa_2^2\cdot \bar{p}^2(T,\bm{S})}\cdot \left[Y - \bar{\mu}(t,\bm{S}) - (T-t)\cdot \bar{\beta}(t,\bm{S})\right]^2 \right\} \\
		&\quad - \frac{1}{nh^4}\left\{\mathbb{E}\left[\frac{\left(\frac{T-t}{h}\right) K\left(\frac{T-t}{h}\right) \cdot \bar{p}_{\zeta}(\bm{S}|t)}{\kappa_2\cdot \bar{p}(T,\bm{S})}\cdot \left[Y - \bar{\mu}(t,\bm{S}) - (T-t)\cdot \bar{\beta}(t,\bm{S})\right]\right]\right\}^2\\
		&\stackrel{\text{(i)}}{\lesssim} \frac{1}{nh^3}\int_{\mathbb{R}} \int_{\mathcal{S}(t+uh)} \frac{u^2 K^2\left(u\right) \cdot \bar{p}_{\zeta}^2(\bm{s}_1|t)\cdot p(t+uh,\bm{s}_1)}{\kappa_2^2\cdot \bar{p}^2(t+uh,\bm{s}_1)}\\
		&\quad \quad \times \left\{\left[\mu(t+uh,\bm{s}_1) - \bar{\mu}(t,\bm{s}_1) - hu\cdot \bar{\beta}(t,\bm{s}_1)\right]^2 + \sigma^2\right\} d\bm{s}_1 du\\
		&\stackrel{\text{(ii)}}{=} \frac{1}{nh^3}\int_{\mathbb{R}} \int_{\mathcal{S}(t+uh)} \frac{u^2 K^2\left(u\right) \cdot \bar{p}_{\zeta}^2(\bm{s}_1|t)\cdot \left[p(t,\bm{s}_1) + O(h)\right]}{\kappa_2^2\cdot \left[\bar{p}^2(t,\bm{s}_1) + O(h^2)\right]}\\
		&\quad \quad \times \left\{\left[\mu(t,\bm{s}_1) +O(h) - \bar{\mu}(t,\bm{s}_1) - hu\cdot \bar{\beta}(t,\bm{s}_1)\right]^2 + \sigma^2\right\} d\bm{s}_1 du\\
		&\stackrel{\text{(iii)}}{=} O\left(\frac{1}{nh^3}\right),
	\end{align*}
	where (i) uses a change of variable and only keeps the dominating first term, (ii) leverages Taylor's expansions, and (iii) utilizes the upper boundedness of $\mu,\bar{\mu},\bar{\beta},\bar{p}_{\zeta}$ under Assumption~\ref{assump:reg_diff} as well as the fact that $\bar{p}$ is lower bounded away from 0 around the support $\mathcal{J}$. Now, by Chebyshev’s inequality, we conclude that
	\begin{align*}
		\left(\mathbb{P}_n -\P\right)\left[\frac{1}{\sqrt{h^3}}\cdot \phi_{C,h,t}\left(Y,T,\bm{S}; \bar{\mu},\bar{\beta}, \bar{p},\bar{p}_{\zeta}\right)\right] &= O_P\left(\sqrt{\mathrm{Var}\left\{\mathbb{P}_n\left[\frac{1}{\sqrt{h^3}} \cdot \phi_{C,h,t}\left(Y,T,\bm{S}; \bar{\mu},\bar{\beta}, \bar{p},\bar{p}_{\zeta}\right)\right]\right\}}\right) \\
		&= O_P\left(\sqrt{\frac{1}{nh^3}}\right).
	\end{align*}
	In addition, we calculate the bias term as:
	\begin{align*}
		&\mathrm{Bias}\left[\textbf{Term VI}\right] \\
		&=\mathbb{E}\left[\frac{1}{\sqrt{h^3}} \cdot \phi_{C,h,t}\left(Y,T,\bm{S}; \bar{\mu},\bar{\beta}, \bar{p},\bar{p}_{\zeta}\right)\right] + \int \bar{\beta}(t,\bm{s})\cdot \bar{p}_{\zeta}(\bm{s}|t)\, d\bm{s} - \int \beta(t,\bm{s})\cdot \bar{p}_{\zeta}(\bm{s}|t)\, d\bm{s} \\
		&= \mathbb{E}\left\{\frac{\left(\frac{T-t}{h}\right) K\left(\frac{T-t}{h}\right) \cdot \bar{p}_{\zeta}(\bm{S}|t)}{h^2\cdot \kappa_2\cdot \bar{p}(T,\bm{S})}\cdot \left[Y - \bar{\mu}(t,\bm{S}) - (T-t)\cdot \bar{\beta}(t,\bm{S})\right]\right\} \\
		&\quad + \int \bar{\beta}(t,\bm{s})\cdot \bar{p}_{\zeta}(\bm{s}|t)\, d\bm{s} - \int \beta(t,\bm{s})\cdot \bar{p}_{\zeta}(\bm{s}|t)\, d\bm{s} \\
		&= \int_{\mathcal{T}} \int_{\mathcal{S}(t_1)} \frac{\left(\frac{t_1-t}{h}\right) K\left(\frac{t_1-t}{h}\right) \cdot \bar{p}_{\zeta}(\bm{s}_1|t)}{h^2\kappa_2\cdot \bar{p}(t_1,\bm{s}_1)}\cdot \left[\mu(t_1,\bm{s}_1) - \bar{\mu}(t,\bm{s}_1) - (t_1-t)\cdot \bar{\beta}(t,\bm{s}_1)\right] p(t_1,\bm{s}_1) \, d\bm{s}_1 dt_1 \\
		&\quad + \int \bar{\beta}(t,\bm{s})\cdot \bar{p}_{\zeta}(\bm{s}|t)\, d\bm{s} - \int \beta(t,\bm{s})\cdot \bar{p}_{\zeta}(\bm{s}|t)\, d\bm{s} \\
		&\stackrel{\text{(i)}}{=} \int_{\mathbb{R}} \int_{\mathcal{S}(t+uh)} \frac{u K\left(u\right) \cdot \bar{p}_{\zeta}(\bm{s}_1|t)}{h\cdot \kappa_2\cdot \bar{p}(t+uh,\bm{s}_1)}\cdot \left[\mu(t+uh,\bm{s}_1) - \bar{\mu}(t,\bm{s}_1) - hu\cdot \bar{\beta}(t,\bm{s}_1)\right] p(t+uh,\bm{s}_1) \, d\bm{s}_1 du \\
		&\quad + \int \bar{\beta}(t,\bm{s})\cdot \bar{p}_{\zeta}(\bm{s}|t)\, d\bm{s} - \int \beta(t,\bm{s})\cdot \bar{p}_{\zeta}(\bm{s}|t)\, d\bm{s} \\
		&\stackrel{\text{(ii)}}{=} \int_{\mathbb{R}} \int_{\bar{\mathcal{S}}(t)\ominus \zeta} \frac{u\cdot K\left(u\right)\cdot \bar{p}_{\zeta}(\bm{s}_1|t)\left[p(t,\bm{s}_1) + uh\cdot \frac{\partial}{\partial t} p(t,\bm{s}_1) + \frac{u^2h^2}{2} \cdot \frac{\partial^2}{\partial t^2} p(t,\bm{s}_1) + \frac{u^3h^3}{6}\cdot \frac{\partial^3}{\partial t^3} p(t,\bm{s}_1) + O(h^4)\right]}{h\cdot \kappa_2\left[\bar{p}(t,\bm{s}_1) + uh\cdot \frac{\partial}{\partial t} \bar{p}(t,\bm{s}_1) + \frac{u^2h^2}{2} \cdot \frac{\partial^2}{\partial t^2} \bar{p}(t,\bm{s}_1) + \frac{u^3h^3}{6}\cdot \frac{\partial^3}{\partial t^3} \bar{p}(t,\bm{s}_1) + O(h^4)\right]}\\ 
		&\quad \times \left[\left(\mu(t,\bm{s}_1) - \bar{\mu}(t,\bm{s}_1)\right) + hu\left(\beta(t,\bm{s}_1) - \bar{\beta}(t,\bm{s}_1)\right) + \frac{u^2 h^2}{2}\cdot \frac{\partial^2}{\partial t^2} \mu(t,\bm{s}_1) + \frac{u^3h^3}{6} \cdot \frac{\partial^3}{\partial t^3}\mu(t,\bm{s}_1) + O(h^4)\right] \, d\bm{s}_1du \\
		&\quad + \int \bar{\beta}(t,\bm{s})\cdot \bar{p}_{\zeta}(\bm{s}|t)\, d\bm{s} - \int \beta(t,\bm{s})\cdot \bar{p}_{\zeta}(\bm{s}|t)\, d\bm{s} \\
		&= \int_{\bar{\mathcal{S}}(t)\ominus \zeta} \int_{\mathbb{R}} \frac{u\cdot K\left(u\right)\cdot \bar{p}_{\zeta}(\bm{s}_1|t)}{h\cdot \kappa_2}\left[p(t,\bm{s}_1) + uh\cdot \frac{\partial}{\partial t} p(t,\bm{s}_1) + \frac{u^2h^2}{2} \cdot \frac{\partial^2}{\partial t^2} p(t,\bm{s}_1) + \frac{u^3h^3}{6}\cdot \frac{\partial^3}{\partial t^3} p(t,\bm{s}_1) + O(h^4)\right]\\
		&\quad \quad \times\Bigg[\frac{1}{\bar{p}(t,\bm{s}_1)} - \frac{uh\cdot \frac{\partial}{\partial t}\bar{p}(t,\bm{s}_1)}{\bar{p}^2(t,\bm{s}_1)} - \frac{u^2h^2 \cdot \frac{\partial^2}{\partial t^2}\bar{p}(t,\bm{s}_1)}{2\bar{p}^2(t,\bm{s}_1)} + \frac{u^2h^2 \left[\frac{\partial}{\partial t}\bar{p}(t,\bm{s}_1)\right]^2}{\bar{p}^3(t,\bm{s}_1)} - \frac{u^3h^3 \cdot \frac{\partial^3}{\partial t^3} \bar{p}(t,\bm{s}_1)}{6\bar{p}^2(t,\bm{s}_1)} \\
		&\quad \quad \quad\quad + \frac{u^3h^3\left[\frac{\partial}{\partial t} \bar{p}(t,\bm{s}_1)\right]\left[\frac{\partial^2}{\partial t^2} \bar{p}(t,\bm{s}_1)\right]}{\bar{p}^3(t,\bm{s}_1)} + O(h^4)\Bigg]\\  
		&\quad\quad \times \left[\left(\mu(t,\bm{s}_1) - \bar{\mu}(t,\bm{s}_1)\right) + hu\left(\beta(t,\bm{s}_1) - \bar{\beta}(t,\bm{s}_1)\right) + \frac{u^2 h^2}{2}\cdot \frac{\partial^2}{\partial t^2} \mu(t,\bm{s}_1) + \frac{u^3h^3}{6} \cdot \frac{\partial^3}{\partial t^3}\mu(t,\bm{s}_1) + O(h^4)\right] \, d\bm{s}_1du \\
		&\quad + \int \bar{\beta}(t,\bm{s})\cdot \bar{p}_{\zeta}(\bm{s}|t)\, d\bm{s} - \int \beta(t,\bm{s})\cdot \bar{p}_{\zeta}(\bm{s}|t)\, d\bm{s} \\
		&= \int \left[\frac{\frac{\partial}{\partial t} p(t,\bm{s})}{\bar{p}(t,\bm{s})} - \frac{p(t,\bm{s}) \cdot \frac{\partial}{\partial t} \bar{p}(t,\bm{s})}{\bar{p}^2(t,\bm{s})}\right]\left[\mu(t,\bm{s}) - \bar{\mu}(t,\bm{s})\right] \cdot \bar{p}_{\zeta}(\bm{s}|t)\, d\bm{s}\\
		&\quad + \int \left[\beta(t,\bm{s}) - \bar{\beta}(t,\bm{s})\right]\left[\frac{p(t,\bm{s})}{\bar{p}(t,\bm{s})} - 1\right] \bar{p}_{\zeta}(\bm{s}|t)\, d\bm{s} + h^2 \cdot B_{C,\theta}(t) + O(h^3),
	\end{align*}
	where (i) uses a change of variable $u=\frac{t_1-t}{h}$ and (ii) applies Taylor's expansions. Here, the complicated bias term $B_{C,\theta}(t)$ is given by
	\begin{align*}
		B_{C,\theta}(t) &= \frac{\kappa_4}{2\kappa_2} \int\frac{\left[\mu(t,\bm{s}) - \bar{\mu}(t,\bm{s})\right]}{\bar{p}(t,\bm{s})}\\
		&\quad \hspace{10mm} \times \Bigg[\frac{1}{3}\cdot \frac{\partial^3}{\partial t^3} p(t,\bm{s}) - \frac{\partial^2}{\partial t^2} p(t,\bm{s}) \cdot \frac{\partial}{\partial t} \log \bar{p}(t,\bm{s}) + 2\frac{\partial}{\partial t} p(t,\bm{s}) \cdot \left[\frac{\partial}{\partial t} \log \bar{p}(t,\bm{s})\right]^2\\
		&\quad\hspace{15mm} + \frac{6\frac{\partial}{\partial t} \log \bar{p}(t,\bm{s}) \cdot \frac{\partial^2}{\partial t^2} \bar{p}(t,\bm{s}) - \frac{\partial^3}{\partial t^3} \bar{p}(t,\bm{s}) - 3 \frac{\partial}{\partial t} p(t,\bm{s}) \cdot \frac{\partial^2}{\partial t^2} \bar{p}(t,\bm{s})}{3 \bar{p}(t,\bm{s})}\Bigg] \bar{p}_{\zeta}(\bm{s}|t)\, d\bm{s} \\
		&\quad + \frac{\kappa_4}{2\kappa_2} \int \frac{\left[\beta(t,\bm{s}) - \bar{\beta}(t,\bm{s})\right]}{\bar{p}(t,\bm{s})} \Bigg[\frac{\partial^2}{\partial t^2} p(t,\bm{s}) - 2\frac{\partial}{\partial t} p(t,\bm{s})\cdot \frac{\partial}{\partial t} \log \bar{p}(t,\bm{s}) \\
		&\hspace{15mm} + 2p(t,\bm{s}) \cdot \left[\frac{\partial}{\partial t} \log \bar{p}(t,\bm{s})\right]^2 - \frac{p(t,\bm{s})\cdot \frac{\partial^2}{\partial t^2} \bar{p}(t,\bm{s})}{\bar{p}(t,\bm{s})}\Bigg] \bar{p}_{\zeta}(\bm{s}|t)\, d\bm{s}\\
		&\quad + \frac{\kappa_4}{2\kappa_2} \int\left[\frac{\frac{\partial}{\partial t} p(t,\bm{s}) \cdot \frac{\partial^2}{\partial t^2} \mu(t,\bm{s})}{\bar{p}(t,\bm{s})} - \frac{p(t,\bm{s}) \cdot \frac{\partial}{\partial t} \bar{p}(t,\bm{s}) \cdot \frac{\partial^2}{\partial t^2} \mu(t,\bm{s})}{\bar{p}^2(t,\bm{s})} + \frac{p(t,\bm{s})\cdot \frac{\partial^3}{\partial t^3} \mu(t,\bm{s})}{3\bar{p}(t,\bm{s})}\right] \bar{p}_{\zeta}(\bm{s}|t)\, d\bm{s}.
	\end{align*}
	Under the condition that either $\bar{\mu}=\mu$ and $\bar{\beta}=\beta$ or $\bar{p} = p$, we have that
	\begin{align*}
		&\int \left[\frac{\frac{\partial}{\partial t} p(t,\bm{s})}{\bar{p}(t,\bm{s})} - \frac{p(t,\bm{s}) \cdot \frac{\partial}{\partial t} \bar{p}(t,\bm{s})}{\bar{p}^2(t,\bm{s})}\right]\left[\mu(t,\bm{s}) - \bar{\mu}(t,\bm{s})\right] \cdot \bar{p}_{\zeta}(\bm{s}|t)\, d\bm{s}\\
		&\quad + \int \left[\beta(t,\bm{s}) - \bar{\beta}(t,\bm{s})\right]\left[\frac{p(t,\bm{s})}{\bar{p}(t,\bm{s})} - 1\right] \bar{p}_{\zeta}(\bm{s}|t)\, d\bm{s} \\
		& = 0
	\end{align*}
	and 
	\begin{align*}
		B_{C,\theta}(t) &= 
		\begin{cases}
			\frac{\kappa_4}{2\kappa_2} \int \left[\frac{\frac{\partial}{\partial t} p(t,\bm{s}) \cdot \frac{\partial^2}{\partial t^2} \mu(t,\bm{s})}{\bar{p}(t,\bm{s})} - \frac{p(t,\bm{s}) \cdot \frac{\partial}{\partial t} \bar{p}(t,\bm{s}) \cdot \frac{\partial^2}{\partial t^2} \mu(t,\bm{s})}{\bar{p}^2(t,\bm{s})} + \frac{p(t,\bm{s})\cdot \frac{\partial^3}{\partial t^3} \mu(t,\bm{s})}{3\bar{p}(t,\bm{s})} \right] \bar{p}_{\zeta}(\bm{s}|t)\, d\bm{s} \; \text{ when } \bar{\mu}=\mu \text{ and } \bar{\beta}=\beta,\\
			\frac{\kappa_4}{6\kappa_2} \int \left[\frac{\partial^3}{\partial t^3} \mu(t,\bm{s}) \right] \bar{p}_{\zeta}(\bm{s}|t)\, d\bm{s} \quad\quad \text{ when } \bar{p} = p,
		\end{cases}\\
		&= \begin{cases}
			\frac{\kappa_4}{6\kappa_2} \int \left\{\frac{3\frac{\partial}{\partial t} p(t,\bm{s}) \cdot \bar{m}''(t) + p(t,\bm{s})\left[\bar{m}^{(3)}(t) - 3\frac{\partial}{\partial t} \log\bar{p}(t,\bm{s}) \cdot \bar{m}''(t) \right]}{\bar{p}(t,\bm{s})} \right\} \bar{p}_{\zeta}(\bm{s}|t)\, d\bm{s} &\text{ when } \bar{\mu}=\mu \text{ and } \bar{\beta}=\beta,\\
			\frac{\kappa_4}{6\kappa_2} \cdot \bar{m}^{(3)}(t) & \text{ when } \bar{p} = p.
		\end{cases}
	\end{align*}
	As a result, as $h\to 0$ and $nh^3\to \infty$, we know that
	\begin{align*}
		\textbf{Term VI} &= \mathbb{P}_n\left[\frac{1}{\sqrt{h^3}} \cdot \phi_{C,h,t}\left(Y,T,\bm{S}; \bar{\mu},\bar{\beta}, \bar{p},\bar{p}_{\zeta}\right) \right] + \int \bar{\beta}(t,\bm{s})\cdot \bar{p}_{\zeta}(\bm{s}|t)\, d\bm{s} - \int \beta(t,\bm{s})\cdot \bar{p}_{\zeta}(\bm{s}|t)\, d\bm{s} \\
		&= h^2 B_{C,\theta}(t) + o(h^2) + O_P\left(\sqrt{\frac{1}{nh^3}}\right)\\
		&= O(h^2) + O_P\left(\sqrt{\frac{1}{nh^3}}\right).
	\end{align*}
	As a side note, under some VC-type condition on the kernel function $K$ \citep{einmahl2005uniform}, we can strengthen the above pointwise rate of convergence to the following uniform one as:
	$$\sup_{t\in \mathcal{T}} \left|\textbf{Term VI}\right| = O(h^2) + O_P\left(\sqrt{\frac{|\log h|}{nh^3}}\right);$$
	see Theorem 4 in \cite{einmahl2005uniform} for details.
	
	\subsubsection{Analysis of \textbf{Term VII} for $\hat{\theta}_{\mathrm{C,DR}}(t)$}
	\label{app:theta_nopos_Term_VII}
	
	Notice that
	\begin{align*}
		\textbf{Term VII} &= \int \hat{\beta}(t,\bm{s})\left[\hat{p}_{\zeta}(\bm{s}|t) - \bar{p}_{\zeta}(\bm{s}|t)\right]\, d\bm{s}\\
		&\leq \int \left|\bar{\beta}(t,\bm{s}) + \hat{\beta}(t,\bm{s}) - \bar{\beta}(t,\bm{s})\right| \left|\hat{p}_{\zeta}(\bm{s}|t) - \bar{p}_{\zeta}(\bm{s}|t)\right|\, d\bm{s}\\
		&\stackrel{\text{(i)}}{\lesssim} \int \left[\left|\bar{\beta}(t,\bm{s})\right| + \left|\hat{\beta}(t,\bm{s}) - \bar{\beta}(t,\bm{s})\right|\right] \left|\hat{p}_{\zeta}(\bm{s}|t) - \bar{p}_{\zeta}(\bm{s}|t)\right| p_S(\bm{s})\, d\bm{s} \\
		&\stackrel{\text{(ii)}}{=} O_P\left(\norm{\hat{p}_{\zeta}(\bm{S}|t) - \bar{p}_{\zeta}(\bm{S}|t)}_{L_2}\right)\\
		&= O_P\left(\Upsilon_{5,n}\right),
	\end{align*}
	where (i) uses the fact that the marginal density $p_S$ is lower bounded away from 0 within the union set $\left(\mathcal{S}(t)\ominus\zeta\right) \cup \left(\hat{\mathcal{S}}(t)\ominus\zeta\right)$ and (ii) leverages the boundedness of $\bar{\beta}$ under Assumption~\ref{assump:reg_diff} as well as $\norm{\hat{\beta}(t,\bm{S}) - \bar{\beta}(t,\bm{S})}_{L_2} = O_P\left(\Upsilon_{3,n}\right)$ with $\Upsilon_{3,n} \to 0$ as $n\to \infty$. Hence, $\sqrt{nh^3}\cdot \textbf{Term VII} =o_P(1)$ when $\sqrt{nh^3} \cdot \Upsilon_{5,n} =o(1)$.
	
	\subsubsection{Analyses of \textbf{Term VIII} and \textbf{Term IX} for $\hat{\theta}_{\mathrm{C,DR}}(t)$}
	\label{app:theta_nopos_Term_VIII}
	
	By Markov's inequality, we know that
	\begin{align*}
		\sqrt{nh^3}\cdot \textbf{Term VIII} &=  \mathbb{G}_n\left\{\frac{\left(\frac{T-t}{h}\right) K\left(\frac{T-t}{h}\right)}{\sqrt{h}\cdot \kappa_2}\left[\frac{\hat{p}_{\zeta}(\bm{S}|t)}{\hat{p}(T,\bm{S})} - \frac{\bar{p}_{\zeta}(\bm{S}|t)}{\bar{p}(T,\bm{S})}\right]\left[Y -\bar{\mu}(t,\bm{S}) - (T-t)\cdot \bar{\beta}(t,\bm{S})\right] \right\}\\
		&= O_P\left(\Upsilon_{5,n} + \Upsilon_{6,n}\right) = o_P(1)
	\end{align*}
	because 
	\begin{align*}
		&\mathbb{E}\left\{\frac{\left(\frac{T-t}{h}\right)^2K^2\left(\frac{T-t}{h}\right)}{h\cdot \kappa_2^2} \left[\frac{\hat{p}_{\zeta}(\bm{S}|t)}{\hat{p}(T,\bm{S})} - \frac{\bar{p}_{\zeta}(\bm{S}|t)}{\bar{p}(T,\bm{S})}\right]^2 \left[Y -\bar{\mu}(t,\bm{S}) - (T-t)\cdot \bar{\beta}(t,\bm{S})\right]^2 \right\} \\
		&\lesssim  \mathbb{E}\Bigg\{\frac{\left(\frac{T-t}{h}\right)^2K^2\left(\frac{T-t}{h}\right)}{h\cdot \kappa_2^2}\cdot \frac{\left[\hat{p}_{\zeta}(\bm{S}|t) - \bar{p}_{\zeta}(\bm{S}|t)\right]^2 \bar{p}^2(T,\bm{S}) + \left[\hat{p}(T,\bm{S}) - \bar{p}(T,\bm{S})\right] \bar{p}_{\zeta}^2(\bm{S}|t)}{\hat{p}^2(T,\bm{S})\cdot \bar{p}^2(T,\bm{S})}\\
		&\quad \times \left[\left(\mu(T,\bm{S}) - \bar{\mu}(t,\bm{S}) - (T-t)\cdot \bar{\beta}(t,\bm{S})\right)^2 + \sigma^2\right] \Bigg\}\\
		&\stackrel{\text{(i)}}{=} \int_{\mathbb{R}} \int_{\mathcal{S}(t+uh)} \frac{u^2 K^2(u)}{\kappa_2^2 } \cdot \frac{\left[\hat{p}_{\zeta}(\bm{s}_1|t) - \bar{p}_{\zeta}(\bm{s}_1|t)\right]^2 \bar{p}^2(t+uh,\bm{s}_1) + \left[\hat{p}(t+uh,\bm{s}_1) - \bar{p}(t+uh,\bm{s}_1)\right] \bar{p}_{\zeta}^2(\bm{s}_1|t)}{\hat{p}^2(t+uh,\bm{s}_1)\cdot \bar{p}^2(t+uh,\bm{s}_1)}\\
		&\quad\quad \times \left[\left(\mu(t+uh,\bm{S}) -\bar{\mu}(t,\bm{S}) -hu\cdot \bar{\beta}(t,\bm{S}) \right)^2 + \sigma^2\right]p(t+uh,\bm{s}_1)\, d\bm{s}_1 du\\
		&\stackrel{\text{(ii)}}{\lesssim} \norm{\hat{p}_{\zeta}(\bm{S}|t) - \bar{p}_{\zeta}(\bm{S}|t)}_{L_2}^2 + \sup_{|u-t|\leq h} \norm{\hat{p}(u,\bm{S}) - \bar{p}(u,\bm{S})}_{L_2}^2\\
		&= O_P\left(\Upsilon_{5,n}^2 + \Upsilon_{6,n}^2\right) = o_P(1),
	\end{align*}
	where (i) uses the change of variable $u=\frac{T-t}{h}$ in the integration as well as (ii) leverages the upper boundedness of $\mu,\bar{\mu},\bar{\beta}$ under Assumption~\ref{assump:reg_diff} and the lower boundedness on $\frac{\bar{p}_{\zeta}}{\bar{p}},\hat{p}$ away from 0 around the support $\mathcal{J}$ by definition. As a side note, under the VC-type condition on the kernel function $K$ \citep{einmahl2005uniform} and 
	$$\sup_{t\in \mathcal{T}}\left[\norm{\hat{p}_{\zeta}(\bm{S}|t) - \bar{p}_{\zeta}(\bm{S}|t)}_{L_2} + \sup_{|u-t|\leq h} \norm{\hat{p}(u,\bm{S}) - \bar{p}(u,\bm{S})}_{L_2}\right]=o_P(1),$$ 
	we can strengthen the above pointwise rate of convergence to the following uniform result as:
	$$\sup_{t\in \mathcal{T}}\left|\mathbb{G}_n\left\{\frac{\left(\frac{T-t}{h}\right) K\left(\frac{T-t}{h}\right)}{\sqrt{h}\cdot \kappa_2}\left[\frac{\hat{p}_{\zeta}(\bm{S}|t)}{\hat{p}(T,\bm{S})} - \frac{\bar{p}_{\zeta}(\bm{S}|t)}{\bar{p}(T,\bm{S})}\right]\left[Y -\bar{\mu}(t,\bm{S}) - (T-t)\cdot \bar{\beta}(t,\bm{S})\right] \right\}\right|=o_P(1).$$
	Similarly, by Markov's inequality, we have that
	\begin{align*}
		\sqrt{nh^3} \cdot \textbf{Term IX} &= \mathbb{G}_n\left\{\frac{\left(\frac{T-t}{h}\right) K\left(\frac{T-t}{h}\right)\cdot \bar{p}_{\zeta}(\bm{S}|t)}{\sqrt{h}\cdot \kappa_2\cdot \bar{p}(T,\bm{S})}\left[\bar{\mu}(t,\bm{S}) - \hat{\mu}(t,\bm{S}) + (T-t)\left(\bar{\beta}(t,\bm{S}) - \hat{\beta}(t,\bm{S})\right)\right]\right\} \\
		&= O_P\left(\max\left\{\Upsilon_{1,n}, h\cdot \Upsilon_{3,n}\right\}\right) = o_P(1)
	\end{align*}
	because
	\begin{align*}
		&\mathbb{E}\left\{\frac{\left(\frac{T-t}{h}\right)^2K^2\left(\frac{T-t}{h}\right)\cdot \bar{p}_{\zeta}^2(\bm{S}|t)}{h\cdot \kappa_2^2\cdot  \bar{p}^2(T,\bm{S})}\left[\bar{\mu}(t,\bm{S}) - \hat{\mu}(t,\bm{S}) + (T-t)\left(\bar{\beta}(t,\bm{S}) - \hat{\beta}(t,\bm{S})\right)\right]^2\right\}\\ 
		&=\int_{\mathcal{T}} \int_{\mathcal{S}(t)}\frac{\left(\frac{t_1-t}{h}\right)^2 K^2\left(\frac{t_1-t}{h}\right) \cdot \bar{p}_{\zeta}^2(\bm{s}_1|t) \cdot p(t_1,\bm{s}_1)}{h\cdot \kappa_2^2\cdot  \bar{p}^2(t_1,\bm{s}_1)} \left[\bar{\mu}(t,\bm{s}_1) - \hat{\mu}(t,\bm{s}_1) + (t_1-t)\left(\bar{\beta}(t,\bm{s}_1) - \hat{\beta}(t,\bm{s}_1)\right)\right]^2\, d\bm{s}_1 dt_1 \\
		&\stackrel{\text{(i)}}{=} \int_{\mathbb{R}} \int_{\mathcal{S}(t+uh)}\frac{u^2 K^2(u) \cdot \bar{p}_{\zeta}^2(\bm{s}_1|t) \cdot p(t+uh,\bm{s}_1)}{h\cdot \kappa_2^2\cdot  \bar{p}^2(t+uh,\bm{s}_1)} \left[\bar{\mu}(t,\bm{s}_1) - \hat{\mu}(t,\bm{s}_1) + hu\cdot \left(\bar{\beta}(t,\bm{s}_1) - \hat{\beta}(t,\bm{s}_1)\right)\right]^2\, d\bm{s}_1 du \\
		&\stackrel{\text{(ii)}}{\lesssim} \norm{\hat{\mu}(t,\bm{S}) - \bar{\mu}(t,\bm{S})}_{L_2}^2 + h^2\norm{\hat{\beta}(t,\bm{S}) - \bar{\beta}(t,\bm{S})}_{L_2}^2\\
		&= O_P\left(\Upsilon_{1,n}^2 + h^2\Upsilon_{3,n}^2\right) = o_P(1),
	\end{align*}
	where (i) uses the change of variable $u=\frac{t_1-t}{h}$ and (ii) leverages the boundedness of $p,p_{\zeta}$ under Assumption~\ref{assump:reg_diff}, the lower bound on $\frac{\bar{p}_{\zeta}}{\bar{p}}$ away from 0 around the support $\mathcal{J}$ by definition, the boundedness condition on $K$ under Assumption~\ref{assump:reg_kernel}, as well as $\norm{\hat{\mu}(t,\bm{S}) - \bar{\mu}(t,\bm{S})}_{L_2}=O_P\left(\Upsilon_{1,n}\right)$ and $h\norm{\hat{\beta}(t,\bm{S}) - \bar{\beta}(t,\bm{S})}_{L_2}=O_P\left(h\cdot\Upsilon_{3,n}\right)$ with $\Upsilon_{1,n}, h\cdot\Upsilon_{3,n}\to 0$ as $n\to \infty$. In addition, if $\sup_{t\in \mathcal{T}}\norm{\hat{\mu}(t,\bm{S}) - \bar{\mu}(t,\bm{S})}_{L_2}=o_P(1)$ and $\sup_{t\in \mathcal{T}}\norm{\hat{\beta}(t,\bm{S}) - \bar{\beta}(t,\bm{S})}_{L_2}=o_P\left(\frac{1}{h}\right)$, then the above pointwise rate of convergence can be strengthened to the uniform one as:
	$$\sup_{t\in \mathcal{T}}\left|\mathbb{G}_n\left\{\frac{\left(\frac{T-t}{h}\right) K\left(\frac{T-t}{h}\right)\cdot \bar{p}_{\zeta}(\bm{S}|t)}{\sqrt{h}\cdot \kappa_2\cdot \bar{p}(T|\bm{S})}\left[\bar{\mu}(t,\bm{S}) - \hat{\mu}(t,\bm{S}) + (T-t)\left(\bar{\beta}(t,\bm{S}) - \hat{\beta}(t,\bm{S})\right)\right]\right\}\right|=o_P(1).$$
	
	\subsubsection{Analysis of \textbf{Term X} for $\hat{\theta}_{\mathrm{C,DR}}(t)$}
	\label{app:theta_nopos_Term_X}
	
	We first calculate that
	{\small\begin{align*}
		&\mathbb{E}\left|\sqrt{nh^3}\cdot\textbf{Term X}\right|\\
		&\mathbb{E}\left|\sqrt{\frac{n}{h}} \cdot \frac{\left(\frac{T-t}{h}\right) K\left(\frac{T-t}{h}\right)}{\kappa_2} \left[\frac{\hat{p}_{\zeta}(\bm{S}|t)}{\hat{p}(T,\bm{S})} - \frac{\bar{p}_{\zeta}(\bm{S}|t)}{\bar{p}(T,\bm{S})}\right] \left[\bar{\mu}(t,\bm{S}) - \hat{\mu}(t,\bm{S}) + (T-t)\left(\bar{\beta}(t,\bm{S}) - \hat{\beta}(t,\bm{S})\right)\right] \right| \\
		&\stackrel{\text{(i)}}{\lesssim} \sqrt{nh} \cdot \sqrt{\mathbb{E}\left\{\frac{K\left(\frac{T-t}{h}\right)}{h} \left[ \frac{\left[\bar{p}_{\zeta}(\bm{S}|t) - \hat{p}_{\zeta}(\bm{S}|t)\right]^2}{\hat{p}^2(T,\bm{S})} + \frac{\left[\hat{p}(T,\bm{S}) - \bar{p}(T,\bm{S})\right]^2 \bar{p}_{\zeta}^2(\bm{S}|t)}{\bar{p}^2(T,\bm{S})} \right]\right\}}\\
		&\quad \times \sqrt{\mathbb{E}\left\{\frac{\left(\frac{T-t}{h}\right)^2 K\left(\frac{T-t}{h}\right)}{h\cdot \kappa_2^2} \cdot \left[\bar{\mu}(t,\bm{S}) - \hat{\mu}(t,\bm{S}) + (T-t)\left(\bar{\beta}(t,\bm{S}) - \hat{\beta}(t,\bm{S})\right)\right]^2\right\}} \\
		&= \sqrt{nh} \cdot \sqrt{\int_{\mathbb{R}}\int_{\mathcal{S}(t+uh)} K(u) \left[\frac{\left[\bar{p}_{\zeta}(\bm{s}_1|t) - \hat{p}_{\zeta}(\bm{s}_1|t)\right]^2}{\hat{p}^2(t+uh,\bm{s}_1)} + \frac{\left[\hat{p}(t+uh,\bm{s}_1) - \bar{p}(t+uh,\bm{s}_1)\right]^2 \bar{p}_{\zeta}^2(\bm{s}_1|t)}{\bar{p}^2(t+uh,\bm{s}_1)} \right] p(t+uh,\bm{s}_1)\,d\bm{s}_1 du}\\
		&\quad \times \sqrt{\mathbb{E}\left\{\int_{\mathbb{R}} \frac{u^2 K\left(u\right)}{\kappa_2^2} \cdot \left[\bar{\mu}(t,\bm{S}) - \hat{\mu}(t,\bm{S}) + hu\left(\bar{\beta}(t,\bm{S}) - \hat{\beta}(t,\bm{S})\right)\right]^2 p_{T|\bm{S}}(t+uh|\bm{S})\, du\right\}} \\
		&\lesssim \sqrt{nh} \left[\norm{\hat{p}_{\zeta}(\bm{S}|t) - \bar{p}_{\zeta}(\bm{S}|t)}_{L_2} + \sup_{|u-t|\leq h} \norm{\hat{p}(u,\bm{S}) - \bar{p}(u,\bm{S})}_{L_2}\right] \left[\norm{\hat{\mu}(t,\bm{S}) - \bar{\mu}(t,\bm{S})}_{L_2} + h \norm{\hat{\beta}(t,\bm{S}) - \bar{\beta}(t,\bm{S})}_{L_2}\right]\\
		&\stackrel{\text{(ii)}}{=} o_P(1),
	\end{align*}}
	where (i) uses Cauchy-Schwarz inequality and (ii) leverages our assumption (c) on the doubly robust rate of convergence in the theorem statement. As a result, by Markov's inequality, we obtain that
	\begin{align*}
		&\sqrt{nh^3}\cdot \textbf{Term X} \\
		&= \sqrt{\frac{n}{h}}\cdot \mathbb{P}_n\left\{\frac{\left(\frac{T-t}{h}\right) K\left(\frac{T-t}{h}\right)}{\kappa_2} \left[\frac{\hat{p}_{\zeta}(\bm{S}|t)}{\hat{p}(T,\bm{S})} - \frac{\bar{p}_{\zeta}(\bm{S}|t)}{\bar{p}(T,\bm{S})}\right] \left[\bar{\mu}(t,\bm{S}) - \hat{\mu}(t,\bm{S}) + (T-t)\left(\bar{\beta}(t,\bm{S}) - \hat{\beta}(t,\bm{S})\right)\right]\right\}\\
		&=o_P(1).
	\end{align*}
	
	\subsubsection{Analysis of \textbf{Term XI} for $\hat{\theta}_{\mathrm{C,DR}}(t)$}
	\label{app:theta_nopos_Term_XI}
	
	By direct calculations under model \eqref{add_conf_model} with some change of variables, we have that
	{\small\begin{align*}
		&\textbf{Term XI} \\
		&= \mathbb{E}\left\{\frac{\left(\frac{T-t}{h}\right)^2 K\left(\frac{T-t}{h}\right) \cdot \bar{p}_{\zeta}(\bm{S}|t)}{h \cdot \kappa_2\cdot \bar{p}(T,\bm{S})} \left[\bar{\beta}(t,\bm{S}) - \hat{\beta}(t,\bm{S})\right]\right\} + \int \left[\hat{\beta}(t,\bm{s}) - \bar{\beta}(t,\bm{s})\right] \bar{p}_{\zeta}(\bm{s}|t)\, d\bm{s} \\ 
		&\quad + \mathbb{E}\left\{\frac{\left(\frac{T-t}{h}\right) K\left(\frac{T-t}{h}\right)\cdot \bar{p}_{\zeta}(\bm{S}|t)}{h^2 \cdot \kappa_2\cdot \bar{p}(T,\bm{S})}\left[\bar{\mu}(t,\bm{S}) - \hat{\mu}(t,\bm{S})\right]\right\}\\
		&\quad + \mathbb{E}\left\{\frac{\left(\frac{T-t}{h}\right) K\left(\frac{T-t}{h}\right)}{h^2\kappa_2}\left[\frac{\hat{p}_{\zeta}(\bm{S}|t)}{\hat{p}(T,\bm{S})} - \frac{\bar{p}_{\zeta}(\bm{S}|t)}{\bar{p}(T,\bm{S})} \right] \left[Y - \bar{\mu}(t,\bm{S}) - (T-t)\cdot \bar{\beta}(t,\bm{S})\right]\right\}\\
		&= \underbrace{\int_{\mathbb{R}} \int_{\mathcal{S}(t+uh)} \frac{u^2 K\left(u\right) \cdot \bar{p}_{\zeta}(\bm{s}_1|t)\cdot p(t+uh,\bm{s}_1)}{\kappa_2\cdot \bar{p}(t+uh,\bm{s}_1)} \left[\bar{\beta}(t,\bm{s}_1) - \hat{\beta}(t,\bm{s}_1)\right] d\bm{s}_1 du + \int \left[\hat{\beta}(t,\bm{s}) - \bar{\beta}(t,\bm{s})\right]\bar{p}_{\zeta}(\bm{s}|t)\, d\bm{s}}_{\textbf{Term XIa}} \\
		& \quad + \underbrace{\int_{\mathbb{R}} \int_{\mathcal{S}(t+uh)} \frac{u K\left(u\right)\cdot \bar{p}_{\zeta}(\bm{s}_1|t) \cdot p(t+uh,\bm{s}_1)}{h \cdot \kappa_2\cdot \bar{p}(t+uh,\bm{s}_1)}\left[\bar{\mu}(t,\bm{s}_1) - \hat{\mu}(t,\bm{s}_1)\right] d\bm{s}_1 du}_{\textbf{Term XIb}}\\
		&\quad + \underbrace{\int_{\mathbb{R}} \int_{\mathcal{S}(t+uh)} \frac{u K\left(u\right) \cdot p(t+uh,\bm{s}_1)}{h\cdot \kappa_2}\left[\frac{\hat{p}_{\zeta}(\bm{s}_1|t)}{\hat{p}(t+uh,\bm{s}_1)} - \frac{\bar{p}_{\zeta}(\bm{s}_1|t)}{\bar{p}(t+uh,\bm{s}_1)} \right] \left[\mu(t+uh,\bm{s}_1) - \bar{\mu}(t,\bm{s}_1) - hu\cdot \bar{\beta}(t,\bm{s}_1)\right] d\bm{s}_1 du}_{\textbf{Term XIc}}.
	\end{align*}}
	
	On one hand, when $\bar{p}(t,\bm{s})=p(t,\bm{s})$, we know that $\textbf{Term XIb}=0$, and under Assumption~\ref{assump:cond_support_smooth} and model \eqref{add_conf_model},
	\begin{align*}
		\textbf{Term XIa} &= \int_{\mathbb{R}} \int_{\mathcal{S}(t) \ominus \zeta} \frac{u^2 K\left(u\right) \cdot \bar{p}_{\zeta}(\bm{s}_1|t)}{\kappa_2} \left[\bar{\beta}(t,\bm{s}_1) - \hat{\beta}(t,\bm{s}_1)\right] d\bm{s}_1 du + \int \left[\hat{\beta}(t,\bm{s}) - \bar{\beta}(t,\bm{s})\right]\bar{p}_{\zeta}(\bm{s}|t)\, d\bm{s}\\
		&= \int_{\mathcal{S}(t) \ominus \zeta} p_{\zeta}(\bm{s}_1|t) \left[\bar{\beta}(t,\bm{s}_1) - \hat{\beta}(t,\bm{s}_1)\right] d\bm{s}_1 + \int \left[\hat{\beta}(t,\bm{s}) - \bar{\beta}(t,\bm{s})\right]\bar{p}_{\zeta}(\bm{s}|t)\, d\bm{s}\\
		&=0.
	\end{align*}
	In addition, we also have that
	\begin{align*}
		&\textbf{Term XIc} \\
		&= \int_{\mathbb{R}} \int_{\mathcal{S}(t+uh)} \frac{u K\left(u\right) \cdot \left[p(t+uh,\bm{s}_1) - \hat{p}(t+uh,\bm{s}_1)\right]\cdot \bar{p}_{\zeta}(\bm{s}_1|t)}{h\cdot \kappa_2\cdot \hat{p}(t+uh,\bm{s}_1)} \left[\mu(t+uh,\bm{s}_1) - \bar{\mu}(t,\bm{s}_1) - hu\cdot \bar{\beta}(t,\bm{s}_1)\right] d\bm{s}_1 du \\
		&\quad + \int_{\mathbb{R}} \int_{\mathcal{S}(t+uh)} \frac{u K\left(u\right) \cdot \left[\hat{p}_{\zeta}(\bm{s}_1|t) - \bar{p}_{\zeta}(\bm{s}_1|t)\right]\cdot p(t+uh,\bm{s}_1)}{h\cdot \kappa_2\cdot \hat{p}(t+uh,\bm{s}_1)} \left[\mu(t+uh,\bm{s}_1) - \bar{\mu}(t,\bm{s}_1) - hu\cdot \bar{\beta}(t,\bm{s}_1)\right] d\bm{s}_1 du \\
		&\lesssim \frac{1}{h}\left[\norm{\hat{p}_{\zeta}(\bm{S}|t) - p_{\zeta}(\bm{S}|t)}_{L_2} + \sup_{|u-t|\leq h} \norm{\hat{p}(u,\bm{S}) - p(u,\bm{S})}_{L_2}\right] \\
		&=o_P\left(\sqrt{\frac{1}{nh^3}}\right)
	\end{align*}
	by the upper boundedness of $\mu,\bar{\mu}$ under Assumption~\ref{assump:reg_diff}, the lower boundedness of $\hat{p}$ away from 0 around the support $\mathcal{J}$, and our assumption (c) on the doubly robust rate of convergence in the theorem statement. Specifically, since $\norm{\hat{\mu}(t,\bm{S}) - \bar{\mu}(t,\bm{S})}_{L_2} + h \norm{\hat{\beta}(t,\bm{S}) - \bar{\beta}(t,\bm{S})}_{L_2}= O_P(1)$ when $\bar{\mu}\neq \mu$ and $\bar{\beta}\neq \beta$, our assumption (c) ensures that $\norm{\hat{p}_{\zeta}(\bm{S}|t) - \bar{p}_{\zeta}(\bm{S}|t)}_{L_2} + \sup_{|u-t|\leq h} \norm{\hat{p}(u,\bm{S}) - p(u,\bm{S})}_{L_2} =o_P\left(\sqrt{\frac{1}{nh}}\right)$.
	
	On the other hand, when $\bar{\mu}=\mu$ and $\bar{\beta}=\beta$, we know from Assumptions~\ref{assump:den_diff} and \ref{assump:cond_support_smooth} on $\bar{p},\bar{p}_{\zeta}$ that $\frac{\bar{p}_{\zeta}}{\bar{p}}$ is bounded away from 0 within the support and thus,
	\begin{align*}
		&\textbf{Term XIa} \\
		&= \int_{\mathbb{R}} \int_{\bar{\mathcal{S}}(t)\ominus \zeta} \frac{u^2 K\left(u\right) \cdot \bar{p}_{\zeta}(\bm{s}_1|t)\cdot p(t+uh,\bm{s}_1)}{\kappa_2\cdot \bar{p}(t+uh,\bm{s}_1)} \left[\beta(t,\bm{s}_1) - \hat{\beta}(t,\bm{s}_1)\right] d\bm{s}_1 du + \int \left[\hat{\beta}(t,\bm{s}) - \beta(t,\bm{s})\right]\bar{p}_{\zeta}(\bm{s}|t)\, d\bm{s}\\
		&\stackrel{\text{(i)}}{\lesssim} \norm{\hat{\beta}(t,\bm{S}) - \beta(t,\bm{S})}_{L_2}\\
		&\stackrel{\text{(ii)}}{=}o_P\left(\sqrt{\frac{1}{nh^3}}\right),
	\end{align*}
	where (i) uses the facts that $\bar{p}_{\zeta}$ is upper bounded while the marginal density $p_S$ is lower bounded away from 0 within the $\zeta$-interior conditional support $\mathcal{S}(t)\ominus \zeta$ and (ii) utilizes our assumption (c) on the doubly robust rate of convergence in the theorem statement to argue that $\norm{\hat{\beta}(t,\bm{S}) - \beta(t,\bm{S})}_{L_2}=o_P\left(\sqrt{\frac{1}{nh^3}}\right)$ if $\bar{p}\neq p, \bar{p}_{\zeta}\neq p_{\zeta}$ and $\norm{\hat{p}_{\zeta}(\bm{S}|t) - p_{\zeta}(\bm{S}|t)}_{L_2} + \sup_{|u-t|\leq h} \norm{\hat{p}(u,\bm{S}) - p(u,\bm{S})}_{L_2}=O_P(1)$. In addition, we also have that
	\begin{align*}
		\textbf{Term XIb} &= \int_{\mathbb{R}} \int_{\mathcal{S}(t+uh)} \frac{u K\left(u\right)\cdot \bar{p}_{\zeta}(\bm{s}_1|t) \cdot p(t+uh,\bm{s}_1)}{h \cdot \kappa_2\cdot \bar{p}(t+uh,\bm{s}_1)}\left[\mu(t,\bm{s}_1) - \hat{\mu}(t,\bm{s}_1)\right] d\bm{s}_1 du \\
		&\lesssim \frac{1}{h} \norm{\hat{\mu}(t,\bm{S}) - \mu(t,\bm{S})}_{L_2}\\
		&=o_P\left(\sqrt{\frac{1}{nh^3}}\right),
	\end{align*}
	where we again argue from our assumption (c) on the doubly robust rate of convergence in the theorem statement that $\norm{\hat{\mu}(t,\bm{S}) - \mu(t,\bm{S})}_{L_2}=o_P\left(\sqrt{\frac{1}{nh}}\right)$ if $\bar{p}\neq p, \bar{p}_{\zeta}\neq p_{\zeta}$ and $\norm{\hat{p}_{\zeta}(\bm{S}|t) - p_{\zeta}(\bm{S}|t)}_{L_2} + \sup_{|u-t|\leq h} \norm{\hat{p}(u,\bm{S}) - p(u,\bm{S})}_{L_2}=O_P(1)$. Finally, we also derive that
	\begin{align*}
		&\textbf{Term XIc} \\
		&= \int_{\mathbb{R}} \int_{\mathcal{S}(t+uh)} \frac{u K\left(u\right) \cdot p(t+uh,\bm{s}_1)}{h\cdot \kappa_2}\left[\frac{\hat{p}_{\zeta}(\bm{s}_1|t)}{\hat{p}(t+uh,\bm{s}_1)} - \frac{\bar{p}_{\zeta}(\bm{s}_1|t)}{\bar{p}(t+uh,\bm{s}_1)} \right] \left[\mu(t+uh,\bm{s}_1) - \mu(t,\bm{s}_1) - hu\cdot \beta(t,\bm{s}_1)\right] d\bm{s}_1 du\\
		&= \int_{\mathbb{R}} \int_{\mathcal{S}(t+uh)} \frac{u K\left(u\right) \cdot p(t+uh,\bm{s}_1)\left[\hat{p}_{\zeta}(\bm{s}_1|t) - \bar{p}_{\zeta}(\bm{s}_1|t) \right]}{h\cdot \kappa_2\cdot \hat{p}(t+uh,\bm{s}_1)} \left[\mu(t+uh,\bm{s}_1) - \mu(t,\bm{s}_1) - hu\cdot \beta(t,\bm{s}_1)\right] d\bm{s}_1 du\\
		&\quad + \int_{\mathbb{R}} \int_{\bar{\mathcal{S}}(t)\ominus \zeta} \frac{u K\left(u\right) \cdot p(t+uh,\bm{s}_1)\cdot \bar{p}_{\zeta}(\bm{s}_1|t)\left[\bar{p}(t+uh,\bm{s}_1) - \hat{p}(t+uh,\bm{s}_1) \right]}{h\cdot \kappa_2\cdot \hat{p}(t+uh,\bm{s}_1)\cdot \bar{p}(t+uh,\bm{s}_1)}\\
		&\quad \quad \times \left[\mu(t+uh,\bm{s}_1) - \mu(t,\bm{s}_1) - hu\cdot \beta(t,\bm{s}_1)\right] d\bm{s}_1 du\\
		&\stackrel{\text{(i)}}{=} \int_{\mathbb{R}} \int_{\mathcal{S}(t+uh)} \frac{u K\left(u\right) \left[\hat{p}_{\zeta}(\bm{s}_1|t) - \bar{p}_{\zeta}(\bm{s}_1|t) \right] \left[\frac{u^2h^2}{2}\cdot \frac{\partial^2}{\partial t^2}\mu(t,\bm{s}_1) + O(h^3)\right] \left[p(t,\bm{s}_1) + 2uh\cdot \frac{\partial}{\partial t} p(t,\bm{s}_1) + O(h^2)\right]}{h\cdot \kappa_2\left[\bar{p}(t,\bm{s}_1) + 2uh\cdot \frac{\partial}{\partial t} \bar{p}(t,\bm{s}_1) + O(h^2)\right]\left[1+O_P\left(\Upsilon_{6,n}\right)\right]} \, d\bm{s}_1 du\\
		&\quad + \int_{\bar{\mathcal{S}}(t)\ominus \zeta} \sup_{|u-t|\leq h}\left|\bar{p}(u,\bm{s}_1) - \hat{p}(u,\bm{s}_1) \right|\\
		&\quad\quad \times \left|\int_{\mathbb{R}} \frac{u\cdot  K\left(u\right) \cdot \bar{p}_{\zeta}(\bm{s}_1|t)\left[\frac{u^2h^2}{2}\cdot \frac{\partial^2}{\partial t^2}\mu(t,\bm{s}_1) + O(h^3)\right] \left[p(t,\bm{s}_1) + 2uh\cdot \frac{\partial}{\partial t} p(t,\bm{s}_1) + O(h^2)\right]}{h\cdot \kappa_2\left[\bar{p}^2(t,\bm{s}_1) + 2uh\cdot \bar{p}(t,\bm{s}_1) \cdot \frac{\partial}{\partial t} \bar{p}(t,\bm{s}_1) + O(h^2)\right]\left[1+O_P\left(\Upsilon_{6,n}^2\right)\right]} \, du\right| d\bm{s}_1\\
		&= O_P\left(h^2 \norm{\hat{p}_{\zeta}(\bm{S}|t) - \bar{p}_{\zeta}(\bm{S}|t)}_{L_2}\right) + O_P\left(h^2\sup_{|u-t|\leq h}\norm{\bar{p}(u,\bm{S}) - \hat{p}(u,\bm{S})}_{L_2} \right)\\
		&= O_P\left(h^2\left[\Upsilon_{5,n}+\Upsilon_{6,n}\right]\right) \\
		&\stackrel{\text{(ii)}}{=} o_P\left(\sqrt{\frac{1}{nh^3}}\right),
	\end{align*}
	where (i) applies Taylor's expansion and mean-value theorem for integrals as well as uses the fact that the difference between $\bar{p}$ and $\hat{p}$ is small when $\sup_{|u-t|\leq h} \norm{\hat{p}(u,\bm{S}) - \bar{p}(u,\bm{S})}_{L_2}=O_P\left(\Upsilon_{6,n}\right)$, while (ii) leverages the arguments that $\sqrt{nh^3}\cdot h^2 = \sqrt{nh^7} \to \sqrt{c_3}\in [0,\infty)$ and $\Upsilon_{5,n},\Upsilon_{6,n}\to 0$ as $n\to \infty$.
	
	\subsubsection{Asymptotic Normality of $\hat{\theta}_{\mathrm{C,DR}}(t)$}
	\label{app:theta_nopos_asym_norm}
	
	For the asymptotic normality of $\hat{\theta}_{\mathrm{C,DR}}(t)$, it follows from the Lyapunov central limit theorem. Specifically, we already show in \autoref{app:theta_nopos_Term_VI} and subsequent subsections that
	\begin{align*}
		&\sqrt{nh^3}\left[\hat{\theta}_{\mathrm{C,DR}}(t) - \theta(t)\right] \\
		&= \frac{1}{\sqrt{n}} \sum_{i=1}^n \left\{\phi_{C,h,t}\left(Y_i,T_i,\bm{S}_i;\bar{\mu}, \bar{\beta}, \bar{p},\bar{p}_{\zeta}\right) + \sqrt{h^3}\left[\int \bar{\beta}(t,\bm{s})\cdot \bar{p}_{\zeta}(\bm{s}|t)\, d\bm{s} - \theta(t)\right] \right\} +o_P(1) \\
		&= \frac{1}{\sqrt{n}} \sum_{i=1}^n \phi_{C,h,t}\left(Y_i,T_i,\bm{S}_i;\bar{\mu}, \bar{\beta}, \bar{p},\bar{p}_{\zeta}\right)  +o_P(1)
	\end{align*}
	with 
	$$\phi_{C,h,t}\left(Y,T,\bm{S}; \bar{\mu},\bar{\beta}, \bar{p},\bar{p}_{\zeta}\right) = \frac{\left(\frac{T-t}{h}\right) K\left(\frac{T-t}{h}\right) \cdot \bar{p}_{\zeta}(\bm{S}|t)}{\sqrt{h}\cdot \kappa_2\cdot \bar{p}(T,\bm{S})}\cdot \left[Y - \bar{\mu}(t,\bm{S}) - (T-t)\cdot \bar{\beta}(t,\bm{S})\right]$$
	and $V_{C,\theta}(t) = \mathbb{E}\left[\phi_{C,h,t}^2\left(Y,T,\bm{S}; \bar{\mu},\bar{\beta}, \bar{p},\bar{p}_{\zeta}\right)\right] = O(1)$ by our calculation in \textbf{Term VI}. Then, 
	$$\sum_{i=1}^n \mathrm{Var}\left[\frac{1}{\sqrt{n}} \cdot \phi_{C,h,t}\left(Y_i,T_i,\bm{S}_i; \bar{\mu},\bar{\beta}, \bar{p},\bar{p}_{\zeta}\right) \right] = O(1)$$ 
	and 
	\begin{align*}
		&\sum_{i=1}^n \mathbb{E}\left|\frac{1}{\sqrt{n}} \cdot \phi_{C,h,t}\left(Y_i,T_i,\bm{S}_i; \bar{\mu},\bar{\beta}, \bar{p},\bar{p}_{\zeta}\right)\right|^{2+c_1} \\
		&= \mathbb{E}\left|\frac{\left(\frac{T-t}{h}\right)^{2+c_1}K^{2+c_1}\left(\frac{T-t}{h}\right) \cdot \bar{p}_{\zeta}^{2+c_1}(\bm{S}|t)\cdot \left[Y-\bar{\mu}(t,\bm{S}) - (T-t)\cdot \bar{\beta}(t,\bm{S})\right]^{2+c_1}}{n^{\frac{c_1}{2}} h^{1+\frac{c_1}{2}}\cdot \kappa_2^{2+c_1} \cdot \bar{p}^{2+c_1}(T,\bm{S})} \right| \\
		&\lesssim \int_{\mathbb{R}} \int_{\bar{\mathcal{S}}(t)\ominus \zeta} \frac{u^{2+c_1}K^{2+c_1}(u) \cdot \bar{p}_{\zeta}^{2+c_1}(\bm{s}_1|t)\left[\left[ \mu(t+uh,\bm{s}_1)-\bar{\mu}(t,\bm{s}_1) -hu\cdot \bar{\beta}(t,\bm{s}_1)\right]^{2+c_1} + \mathbb{E}|\epsilon|^{2+c_1}\right]}{\sqrt{(nh)^{c_1}}\bar{p}^{2+c_1}(t+uh,\bm{s}_1)} \\ 
		&\quad \times p(t+uh,\bm{s}_1)\, d\bm{s}_1 du \\
		&=O\left(\sqrt{\frac{1}{(nh)^{c_1}}}\right) =o(1)
	\end{align*}
	by the upper boundedness of $\mu,\bar{\mu},p$ under Assumptions~\ref{assump:reg_diff} and \ref{assump:den_diff}, the lower boundedness of $\bar{p}$ away from 0 around the support $\mathcal{J}$, the assumption that $\mathbb{E}|\epsilon|^{2+c_1}<\infty$, and the requirement that $nh^3\to\infty$ as $n\to \infty$. Hence, the Lyapunov condition holds, and we have that
	$$\sqrt{nh^3}\left[\hat{\theta}_{\mathrm{C,DR}}(t) - \theta(t) - h^2 B_{C,\theta}(t)\right] \stackrel{d}{\to} \mathcal{N}\left(0,V_{C,\theta}(t)\right)$$
	after subtracting the dominating bias term $h^2 B_{C,\theta}(t)$ of $\phi_{C,h,t}\left(Y,T,\bm{S}; \bar{\mu}, \bar{\beta}, \bar{p}, \bar{p}_{\zeta}\right)$ that we have computed in \textbf{Term VI}. The proof is thus completed.
\end{proof}

\section{Asymptotic Theory of Estimating $m(t)$ Without Positivity}
\label{app:m_nopos_proof}

Under the additive confounding model \eqref{add_conf_model}, Proposition~\ref{prop:id_additive} implies that we can define the integral RA estimator \eqref{m_RA_corrected} as well as the integral IPW and DR estimators of $m(t)$ based on \eqref{theta_DR_bnd_corrected} and \eqref{theta_IPW_bnd_corrected} as:
$$\hat{m}_{\mathrm{C,IPW}}(t) = \frac{1}{n}\sum_{i=1}^n \left[Y_i + \int_{\tilde{t}=T_i}^{\tilde{t}=t} \hat{\theta}_{\mathrm{C,IPW}}(\tilde{t})\, d\tilde{t} \right] \quad \text{ and } \quad \hat{m}_{\mathrm{C,DR}}(t) = \frac{1}{n}\sum_{i=1}^n \left[Y_i + \int_{\tilde{t}=T_i}^{\tilde{t}=t} \hat{\theta}_{\mathrm{C,DR}}(\tilde{t})\, d\tilde{t} \right].$$ 
We prove in the following corollary that these integral estimators are consistent to $m(t)$ without assuming the positivity condition.

\begin{corollary}[Consistency of estimating $m(t)$ without positivity]
\label{cor:m_nopos}
	Suppose that Assumptions~\ref{assump:id_cond}, \ref{assump:reg_diff}, \ref{assump:den_diff}, \ref{assump:reg_kernel}, and \ref{assump:cond_support_smooth} are valid under the additive confounding model \eqref{add_conf_model}. In addition, $\hat{\mu},\hat{\beta}, \hat{p}_{\zeta}, \hat{p}$ are constructed on a data sample independent of $\{(Y_i,T_i,\bm{S}_i)\}_{i=1}^n$. For any fixed $t\in \mathcal{T}$, we let $\bar{\mu}(t,\bm{s})$, $\bar{\beta}(t,\bm{s})$, $\bar{p}_{\zeta}(\bm{s}|t)$, and $\bar{p}(t,\bm{s})$ be fixed bounded functions to which $\hat{\mu}(t,\bm{s})$, $\hat{\beta}(t,\bm{s})$, $\hat{p}_{\zeta}(\bm{s}|t)$, and $\hat{p}(t,\bm{s})$ converge under the rates of convergence as:
	\begin{align*}
		&\sup_{t\in \mathcal{T}}\norm{\hat{\beta}(t,\bm{S}) - \bar{\beta}(t,\bm{S})}_{L_2} = O_P\left(\bar{\Upsilon}_{3,n}\right), \quad \sup_{t\in \mathcal{T}}\sup_{\bm{s}\in \mathcal{S}}\left|\hat{F}_{\bm{S}|T}(\bm{s}|t) - F_{\bm{S}|T}(\bm{s}|t)\right| = O_P\left(\bar{\Upsilon}_{4,n}\right), \\
		& \sup_{t\in \mathcal{T}}\norm{\hat{p}_{\zeta}(\bm{S}|t) - \bar{p}_{\zeta}(\bm{S}|t)}_{L_2}=O_P\left(\bar{\Upsilon}_{5,n}\right), \quad \text{ and } \quad \sup_{u\in \mathcal{T} \oplus h}\norm{\hat{p}(u,\bm{S}) - \bar{p}(u,\bm{S})}_{L_2}=O_P\left(\bar{\Upsilon}_{6,n}\right),
	\end{align*}
	where $\mathcal{T} \oplus h=\left\{u\in \mathbb{R}: \inf_{t\in \mathcal{T}} |u-t| \leq h\right\}$ and $\bar{\Upsilon}_{3,n}, \bar{\Upsilon}_{4,n}, \bar{\Upsilon}_{5,n}, \bar{\Upsilon}_{6,n}\to 0$ as $n\to \infty$. Then, as $h\to 0$ and $nh^3\to \infty$, we have that
	\begin{align*}
		&\sup_{t\in \mathcal{T}}\left|\hat{m}_{\mathrm{C,RA}}(t) - m(t)\right| = O_P\left(\bar{\Upsilon}_{3,n}+\bar{\Upsilon}_{4,n}+ \sup_{t\in \mathcal{T}}\norm{\bar{\beta}(t,\bm{S}) - \beta(t,\bm{S})}_{L_2} + \frac{1}{\sqrt{n}}\right), \\ 
		&\sup_{t\in \mathcal{T}}\left|\hat{m}_{\mathrm{C,IPW}}(t) - m(t)\right| = O(h^2) + O_P\left(\sqrt{\frac{|\log h|}{nh^3}} + \bar{\Upsilon}_{5,n} + \bar{\Upsilon}_{6,n} + \sup_{u\in \mathcal{T} \oplus h} \norm{\bar{p}(u,\bm{S}) - p(u,\bm{S})}_{L_2}\right). 
	\end{align*}
	If, in addition, we assume that 
	\begin{enumerate}[label=(\alph*)]
		\item $\bar{p},\bar{p}_{\zeta}$ satisfy Assumptions~\ref{assump:den_diff} and \ref{assump:cond_support_smooth} as well as $\sqrt{nh^3}\cdot \bar{\Upsilon}_{5,n}=o(1)$;
		\item either $\bar{\mu}=\mu$ and $\bar{\beta}=\beta$ or $\bar{p} = p$;
		\item {\small $\sup\limits_{t\in \mathcal{T}}\left[\norm{\hat{p}_{\zeta}(\bm{S}|t) - \bar{p}_{\zeta}(\bm{S}|t)}_{L_2} + \sup\limits_{|u-t|\leq h} \norm{\hat{p}(u,\bm{S}) - p(u,\bm{S})}_{L_2} \right]\left[\norm{\hat{\mu}(t,\bm{S}) - \mu(t,\bm{S})}_{L_2} + h \norm{\hat{\beta}(t,\bm{S}) - \beta(t,\bm{S})}_{L_2}\right] = o_P\left(\frac{1}{\sqrt{nh}}\right)$},
	\end{enumerate}
	then 
	\begin{align*}
		&\sqrt{nh^3}\left[\hat{m}_{\mathrm{C,DR}}(t) - m(t)\right] \\
		&= \frac{1}{\sqrt{n}} \sum_{i=1}^n \left\{ \mathbb{E}_{T_1}\left[\int_{T_1}^t  \left\{\phi_{C,h,\tilde{t}}\left(Y_i,T_i,\bm{S}_i; \bar{\mu},\bar{\beta}, \bar{p},\bar{p}_{\zeta}\right) + \sqrt{h^3}\int \left[\bar{\beta}(\tilde{t},\bm{s}) - \beta(\tilde{t},\bm{s})\right] \bar{p}_{\zeta}(\bm{s}|\tilde{t})\, d\bm{s} \right\} d\tilde{t}\right] \right\} + o_P(1)
	\end{align*}
	when $nh^7\to c_3$ for some finite number $c_3\geq 0$, where $\phi_{C,h,t}\left(Y,T,\bm{S}; \bar{\mu},\bar{\beta}, \bar{p},\bar{p}_{\zeta}\right)$ is defined in \autoref{thm:theta_nopos}. 
	Furthermore, 
	$$\sqrt{nh^3}\left\{\hat{m}_{\mathrm{C,DR}}(t) - m(t) - h^2\cdot  \mathbb{E}_{T_1}\left[\int_{T_1}^t B_{C,\theta}(\tilde{t})\, d\tilde{t} \right]\right\} \stackrel{d}{\to} \mathcal{N}\left(0,V_{C,m}(t)\right)$$
	with $V_{C,m}(t) = \mathbb{E}\left[\left\{\mathbb{E}_{T_1}\left[\int_{T_1}^t \phi_{C,h,\tilde{t}}\left(Y,T,\bm{S};\bar{\mu}, \bar{\beta}, \bar{p},\bar{p}_{\zeta}\right) d\tilde{t}\right] \right\}^2 \right]$ with $B_{C,\theta}(t)$ defined in \autoref{thm:theta_nopos}.
\end{corollary}

\begin{proof}[Proof of Corollary~\ref{cor:m_nopos}]
	Recall from \eqref{m_RA_corrected} that our integral estimator of $m(t)$ is defined as:
	$$\hat{m}_C(t) = \frac{1}{n}\sum_{i=1}^n \left[Y_i + \int_{\tilde{t}=T_i}^{\tilde{t}=t} \hat{\theta}_C(\tilde{t})\, d\tilde{t} \right],$$
	where $\hat{\theta}_C(t)$ can be either RA \eqref{theta_RA_corrected}, IPW \eqref{theta_IPW_bnd_corrected}, or DR \eqref{theta_DR_bnd_corrected} estimators. By \eqref{id_m} in Proposition~\ref{prop:id_additive}, we have that
	\begin{align}
		\label{m_decomp}
		\hat{m}_C(t) - m(t) &= \frac{1}{n}\sum_{i=1}^n Y_i - \mathbb{E}(Y) + \frac{1}{n} \sum_{i=1}^n \int_{T_i}^t \theta(\tilde{t})\, d\tilde{t} - \mathbb{E}\left[\int_T^t \theta(\tilde{t})\, d\tilde{t}\right] + \frac{1}{n}\sum_{i=1}^n \int_{T_i}^t \left[\hat{\theta}_C(\tilde{t}) - \theta(\tilde{t})\right]d\tilde{t}.
	\end{align}
	Under Assumption~\ref{assump:reg_diff} and the condition $\mathbb{E}|\epsilon^{2+c_1}| < \infty$ for some constant $c_1>0$, it is valid that $\mathrm{Var}(Y) <\infty$ and $\mathrm{Var}\left[\int_T^t \theta(\tilde{t})\, d\tilde{t}\right] < \infty$. Thus, by Chebyshev's inequality, we know that
	$$\frac{1}{n}\sum_{i=1}^n Y_i - \mathbb{E}(Y) = O_P\left(\frac{1}{\sqrt{n}}\right) \quad \text{ and } \quad \frac{1}{n} \sum_{i=1}^n \int_{T_i}^t \theta(\tilde{t})\, d\tilde{t} - \mathbb{E}\left[\int_T^t \theta(\tilde{t})\, d\tilde{t}\right]= O_P\left(\frac{1}{\sqrt{n}}\right).$$
	Furthermore, under Assumption~\ref{assump:reg_diff}, $\left|\int_T^{t_1} \theta(\tilde{t})\, d\tilde{t} - \int_T^{t_2} \theta(\tilde{t})\, d\tilde{t} \right| \leq \sup_{t\in \mathcal{T}}|\theta(t)| \cdot |t_1-t_2|$. Together with the compactness of $\mathcal{T}$ and Example 19.7 in \cite{VDV1998}, we also deduce that
	$$\sup_{t\in \mathcal{T}}\left|\frac{1}{n} \sum_{i=1}^n \int_{T_i}^t \theta(\tilde{t})\, d\tilde{t} - \mathbb{E}\left[\int_T^t \theta(\tilde{t})\, d\tilde{t}\right] \right|= O_P\left(\frac{1}{\sqrt{n}}\right).$$
	Therefore, plugging the above rates of convergence back into \eqref{m_decomp}, we conclude that
	\begin{align}
		\label{m_decomp2}
		\hat{m}_C(t) - m(t) &= \frac{1}{n}\sum_{i=1}^n \int_{T_i}^t \left[\hat{\theta}_C(\tilde{t}) - \theta(\tilde{t})\right]d\tilde{t} + O_P\left(\frac{1}{\sqrt{n}}\right).
	\end{align}
	Now, we derive the uniform rates of convergence for $\hat{m}_{\mathrm{C,RA}}(t)$ and $\hat{m}_{\mathrm{C,IPW}}(t)$ when $\hat{\theta}_C(t)$ are given by $\hat{\theta}_{\mathrm{C,RA}}(t)$ in \eqref{theta_RA_corrected} and $\hat{\theta}_{\mathrm{C,IPW}}(t)$ in \eqref{theta_IPW_bnd_corrected} respectively in \autoref{app:m_RA_nopos} and \autoref{app:m_IPW_nopos}. We also prove the asymptotic linearity, double robustness, and asymptotic normality of $\hat{m}_{\mathrm{C,DR}}(t)$ when $\hat{\theta}_C(t)$ is given by $\hat{\theta}_{\mathrm{C,DR}}(t)$ in \eqref{theta_DR_bnd_corrected} in \autoref{app:m_DR_nopos}.
	
	\subsection{Uniform Rate of Convergence of $\hat{m}_{\mathrm{C,RA}}(t)$}
	\label{app:m_RA_nopos}
	
	Firstly, we derive the rate of convergence for $\hat{m}_{\mathrm{C,RA}}(t)$ when the derivative estimator is given by $\hat{\theta}_{\mathrm{C,RA}}(t)$ in \eqref{theta_RA_corrected}. By \eqref{m_decomp2}, we have that
	\begin{align*}
		\sup_{t\in \mathcal{T}}\left|\hat{m}_{\mathrm{C,RA}}(t) - m(t)\right| &\leq \sup_{t\in \mathcal{T}}\left|\frac{1}{n}\sum_{i=1}^n \int_{T_i}^t \left[\hat{\theta}_{\mathrm{C,RA}}(\tilde{t}) - \theta(\tilde{t})\right]d\tilde{t}\right| + O_P\left(\frac{1}{\sqrt{n}}\right) \\
		&\leq \sup_{t\in \mathcal{T}} \left[\left|\hat{\theta}_{\mathrm{C,RA}}(t) - \theta(t) \right|\cdot \frac{1}{n}\sum_{i=1}^n |t-T_i| \right] + O_P\left(\frac{1}{\sqrt{n}}\right).
	\end{align*}
	The compactness of $\mathcal{T}$ by Assumption~\ref{assump:reg_diff} and Markov's inequality imply that $\sup_{t\in \mathcal{T}} \frac{1}{n}\sum_{i=1}^n |t-T_i| = O_P(1)$. Therefore, by \autoref{thm:theta_nopos}, we conclude that
	\begin{align*}
		\sup_{t\in \mathcal{T}}\left|\hat{m}_{\mathrm{C,RA}}(t) - m(t)\right| &\leq \sup_{t\in \mathcal{T}}\left|\hat{\theta}_{\mathrm{C,RA}}(t) - \theta(t) \right| \cdot \sup_{t\in \mathcal{T}} \left[\frac{1}{n}\sum_{i=1}^n |t-T_i| \right] + O_P\left(\frac{1}{\sqrt{n}}\right)\\
		&= O_P\left(\bar{\Upsilon}_{3,n}+\bar{\Upsilon}_{4,n}+ \sup_{t\in \mathcal{T}}\norm{\bar{\beta}(t,\bm{S}) - \beta(t,\bm{S})}_{L_2} + \frac{1}{\sqrt{n}}\right).
	\end{align*}

	\subsection{Uniform Rate of Convergence of $\hat{m}_{\mathrm{C,IPW}}(t)$}
	\label{app:m_IPW_nopos}
	
	Secondly, we derive the rate of convergence for $\hat{m}_{\mathrm{C,IPW}}(t)$ when the derivative estimator is given by $\hat{\theta}_{\mathrm{C,IPW}}(t)$ in \eqref{theta_IPW_bnd_corrected}. Similar to our arguments in \autoref{app:m_RA_nopos}, we derive from our results in the proof of \autoref{thm:theta_nopos} (\autoref{app:theta_IPW_nopos}) that 
	\begin{align*}
		&\sup_{t\in \mathcal{T}}\left|\hat{m}_{\mathrm{C,IPW}}(t) - m(t)\right|\\ &\leq \sup_{t\in \mathcal{T}}\left|\hat{\theta}_{\mathrm{C,IPW}}(t) - \theta(t) \right| \cdot \sup_{t\in \mathcal{T}} \left[\frac{1}{n}\sum_{i=1}^n |t-T_i| \right] + O_P\left(\frac{1}{\sqrt{n}}\right)\\
		&= O(h^2) + O_P\left(\sqrt{\frac{|\log h|}{nh^3}} + \bar{\Upsilon}_{5,n} + \bar{\Upsilon}_{6,n} + \sup_{u\in \mathcal{T} \oplus h} \norm{\bar{p}(u,\bm{S}) - p(u,\bm{S})}_{L_2}\right).
	\end{align*}

	\subsection{Asymptotic Properties of $\hat{m}_{\mathrm{C,DR}}(t)$}
	\label{app:m_DR_nopos}
	
	Finally, using the asymptotic properties of $\hat{\theta}_{\mathrm{C,DR}}(t)$ in \autoref{thm:theta_nopos}, we shall establish the asymptotic properties of $\hat{m}_{\mathrm{C,DR}}(t)$ when the derivative function is given by $\hat{\theta}_{\mathrm{C,DR}}(t)$ in \eqref{theta_DR_bnd_corrected}. Recall from the proof of \autoref{thm:theta_nopos} (specifically, \autoref{app:theta_DR_nopos}) that 
	\begin{align*}
		&\hat{\theta}_{\mathrm{C,DR}}(t) - \theta(t) \\
		&= \mathbb{P}_n\left[\frac{1}{\sqrt{h^3}}\cdot \phi_{C,h,t}\left(Y_i,T_i,\bm{S}_i;\bar{\mu}, \bar{\beta}, \bar{p},\bar{p}_{\zeta}\right)\right] +  \int \left[\bar{\beta}(t,\bm{s})- \beta(t,\bm{s})\right] \bar{p}_{\zeta}(\bm{s}|t)\, d\bm{s} + A_{n,h}(t),
	\end{align*}
	where $\phi_{C,h,t}\left(Y,T,\bm{S}; \bar{\mu},\bar{\beta}, \bar{p},\bar{p}_{\zeta}\right) = \frac{\left(\frac{T-t}{h}\right) K\left(\frac{T-t}{h}\right) \cdot \bar{p}_{\zeta}(\bm{S}|t)}{\sqrt{h}\cdot \kappa_2\cdot \bar{p}(T,\bm{S})}\cdot \left[Y - \bar{\mu}(t,\bm{S}) - (T-t)\cdot \bar{\beta}(t,\bm{S})\right]$ and $A_{n,h}(t)=o_P\left(\sqrt{\frac{1}{nh^3}}\right)$ consists of \textbf{Term VII--XI} for $\hat{\theta}_{\mathrm{C,DR}}(t)$ in \autoref{app:theta_DR_nopos}. Then, by \eqref{m_decomp2} and model \eqref{add_conf_model}, we know that
	\begin{align*}
		&\hat{m}_{\mathrm{C,DR}}(t) - m(t) \\
		&= \frac{1}{n}\sum_{i_1=1}^n \int_{T_{i_1}}^t \left\{\frac{1}{n}\sum_{i_2=1}^n \frac{\left(\frac{T_{i_2}-\tilde{t}}{h}\right) K\left(\frac{T_{i_2}-\tilde{t}}{h}\right) \cdot \bar{p}_{\zeta}(\bm{S}_{i_2}|\tilde{t})}{h^2\cdot \kappa_2\cdot \bar{p}(T_{i_2},\bm{S}_{i_2})}\cdot \left[Y_{i_2} - \bar{\mu}(\tilde{t},\bm{S}_{i_2}) - (T_{i_2}-\tilde{t})\cdot \bar{\beta}(\tilde{t},\bm{S}_{i_2})\right]\right\} d\tilde{t} \\
		&\quad + \frac{1}{n}\sum_{i=1}^n \int_{T_i}^t \int \left[\bar{\beta}(\tilde{t},\bm{s}) - \beta(\tilde{t},\bm{s})\right] \bar{p}_{\zeta}(\bm{s}|\tilde{t})\, d\bm{s} d\tilde{t} + \frac{1}{n}\sum_{i=1}^n \int_{T_i}^t A_{n,h}(\tilde{t})\, d\tilde{t} + O_P\left(\frac{1}{\sqrt{n}}\right)\\
		&= \underbrace{\frac{1}{n^2}\sum_{i_1=1}^n \sum_{i_2=1}^n \int_{T_{i_1}}^t \frac{1}{\sqrt{h^3}}\cdot \phi_{C,h,\tilde{t}}\left(Y_{i_2},T_{i_2},\bm{S}_{i_2}; \bar{\mu},\bar{\beta}, \bar{p},\bar{p}_{\zeta}\right) d\tilde{t} + \mathbb{E}\left\{\int_T^t \int \left[\bar{\beta}(\tilde{t},\bm{s}) - \beta(\tilde{t},\bm{s})\right] \bar{p}_{\zeta}(\bm{s}|\tilde{t})\, d\bm{s} d\tilde{t}\right\}}_{\textbf{Term I}}\\
		&\quad + \underbrace{\frac{1}{n}\sum_{i=1}^n \int_{T_i}^t \int \left[\bar{\beta}(\tilde{t},\bm{s}) - \beta(\tilde{t},\bm{s})\right] \bar{p}_{\zeta}(\bm{s}|\tilde{t})\, d\bm{s} d\tilde{t} - \mathbb{E}\left\{\int_T^t \int \left[\bar{\beta}(\tilde{t},\bm{s}) - \beta(\tilde{t},\bm{s})\right] \bar{p}_{\zeta}(\bm{s}|\tilde{t})\, d\bm{s} d\tilde{t}\right\}}_{\textbf{Term II}}\\
		&\quad + \underbrace{\frac{1}{n}\sum_{i=1}^n \int_{T_i}^t A_{n,h}(\tilde{t})\, d\tilde{t}}_{\textbf{Term III}} + O_P\left(\frac{1}{\sqrt{n}}\right).
	\end{align*}
	As for \textbf{Term II}, we know from Assumptions~\ref{assump:reg_diff} and \ref{assump:den_diff} that $\mathbb{E}\left\{\left[\int_T^t \int \left[\bar{\beta}(\tilde{t},\bm{s}) - \beta(\tilde{t},\bm{s})\right] \bar{p}_{\zeta}(\bm{s}|\tilde{t})\, d\bm{s} d\tilde{t}\right]^2\right\} < \infty$ for any $t\in \mathcal{T}$. By the central limit theorem, 
	\begin{align*}
		\textbf{Term II} &= \frac{1}{n}\sum_{i=1}^n \int_{T_i}^t \int \left[\bar{\beta}(\tilde{t},\bm{s}) - \beta(\tilde{t},\bm{s})\right] \bar{p}_{\zeta}(\bm{s}|\tilde{t})\, d\bm{s} d\tilde{t} - \mathbb{E}\left\{\int_T^t \int \left[\bar{\beta}(\tilde{t},\bm{s}) - \beta(\tilde{t},\bm{s})\right] \bar{p}_{\zeta}(\bm{s}|\tilde{t})\, d\bm{s} d\tilde{t}\right\} \\
		&= O_P\left(\frac{1}{\sqrt{n}}\right).
	\end{align*}
	Thus, it remains to derive the asymptotic linearity of $\hat{m}_{\mathrm{C,DR}}(t) - m(t)$ from \textbf{Term I} in \autoref{app:m_nopos_Term_I} and argue that $\textbf{Term III} = o_P\left(\sqrt{\frac{1}{nh^3}}\right)$ in \autoref{app:m_nopos_Term_III}.
	
	\subsubsection{Analysis of \textbf{Term I} for $\hat{m}_{\mathrm{C,DR}}(t)$}
	\label{app:m_nopos_Term_I}
	
	Notice that the first term in \textbf{Term I} takes a form of V-statistics with a symmetric ``kernel'' defined as: 
	\begin{align*}
		\Lambda_t(\bm{U}_{i_1}, \bm{U}_{i_2}) &= \frac{1}{2}\int_{T_{i_1}}^t \frac{1}{\sqrt{h^3}}\cdot \phi_{C,h,\tilde{t}}\left(Y_{i_2},T_{i_2},\bm{S}_{i_2}; \bar{\mu},\bar{\beta}, \bar{p},\bar{p}_{\zeta}\right) d\tilde{t} \\
		&\quad + \frac{1}{2} \int_{T_{i_2}}^t \frac{1}{\sqrt{h^3}}\cdot \phi_{C,h,\tilde{t}}\left(Y_{i_1},T_{i_1},\bm{S}_{i_1}; \bar{\mu},\bar{\beta}, \bar{p},\bar{p}_{\zeta}\right) d\tilde{t},
	\end{align*}
	where $\bm{U}_i\equiv (Y_i,T_i,\bm{S}_i)$ for $i=1,...,n$.
	By Pascal's rule, we know that
	\begin{align*}
		\textbf{Term I} &= \mathbb{P}_n^2 \Lambda_t + \mathbb{E}\left\{\int_T^t \int \left[\bar{\beta}(\tilde{t},\bm{s}) - \beta(\tilde{t},\bm{s})\right] \bar{p}_{\zeta}(\bm{s}|\tilde{t})\, d\bm{s} d\tilde{t}\right\}\\
		&= 2\left(\mathbb{P}_n - \P\right)\P \Lambda_t + \left(\mathbb{P}_n - \P\right)^2 \Lambda_t + \P^2 \Lambda_t + \mathbb{E}\left\{\int_T^t \int \left[\bar{\beta}(\tilde{t},\bm{s}) - \beta(\tilde{t},\bm{s})\right] \bar{p}_{\zeta}(\bm{s}|\tilde{t})\, d\bm{s} d\tilde{t}\right\}\\
		&= \underbrace{\mathbb{E}_{T_1}\left\{\int_{T_1}^t \left(\mathbb{P}_n-\P\right)\left[\frac{1}{\sqrt{h^3}}\cdot \phi_{C,h,\tilde{t}}\left(Y,T,\bm{S}; \bar{\mu},\bar{\beta}, \bar{p},\bar{p}_{\zeta}\right) \right] d\tilde{t}\right\}}_{\textbf{Term Ia}} \\
		&\quad + \underbrace{\left(\mathbb{P}_n-\P\right)\left\{\int_T^t \mathbb{E}\left[\frac{1}{\sqrt{h^3}}\cdot \phi_{C,h,\tilde{t}}\left(Y,T,\bm{S}; \bar{\mu},\bar{\beta}, \bar{p},\bar{p}_{\zeta}\right) \right] d\tilde{t}\right\}}_{\textbf{Term Ib}} \\
		&\quad + \underbrace{\left(\mathbb{P}_n-\P\right)^2 \left\{\int_{T_{i_1}}^t \frac{1}{\sqrt{h^3}}\cdot \phi_{C,h,\tilde{t}}\left(Y_{i_2},T_{i_2},\bm{S}_{i_2}; \bar{\mu},\bar{\beta}, \bar{p},\bar{p}_{\zeta}\right) d\tilde{t}\right\}}_{\textbf{Term Ic}} \\
		&\quad + \underbrace{\mathbb{E}_{T_1}\left\{\int_{T_1}^t \left(\mathbb{E}\left[\frac{1}{\sqrt{h^3}}\cdot \phi_{C,h,\tilde{t}}\left(Y,T,\bm{S}; \bar{\mu},\bar{\beta}, \bar{p},\bar{p}_{\zeta}\right) \right] + \int \left[\bar{\beta}(\tilde{t},\bm{s}) - \beta(\tilde{t},\bm{s})\right] \bar{p}_{\zeta}(\bm{s}|\tilde{t})\, d\bm{s}\right) d\tilde{t}\right\}}_{\textbf{Term Id}},
	\end{align*}
	where we use the shorthand notation $\P \Lambda_t$ referring to the function $\bm{U}_{i_1} \mapsto \int \Lambda_t(\bm{U}_{i_1}, \bm{u}_{i_2}) \, d\P(\bm{u}_{i_2})$ and $\P^2\Lambda_t = \int \int \Lambda_t(\bm{u}_{i_1}, \bm{u}_{i_2}) \, d\P(\bm{u}_{i_1})d\P(\bm{u}_{i_2})$. 
	
	We shall show that the dominating terms \textbf{Term Ia} and \textbf{Term Id} are of orders $O_P\left(\sqrt{\frac{1}{nh^3}}\right)$ and $O(h^2)$ respectively, and the remainder terms \textbf{Term Ib} and \textbf{Term Ic} are of order $o_P\left(\sqrt{\frac{1}{nh^3}}\right)$ as follows.
	
	$\bullet$ \textbf{Term Ia:} By our calculations in \autoref{app:theta_nopos_Term_VI}, we know that
	\begin{align*}
		&\mathrm{Var}\left[\textbf{Term Ia}\right] \\
		&= \mathrm{Var}\left[\left(\mathbb{P}_n-\P\right)\left\{\mathbb{E}_{T_1}\left[\int_{T_1}^t \frac{\left(\frac{T-\tilde{t}}{h}\right) K\left(\frac{T-\tilde{t}}{h}\right) \cdot \bar{p}_{\zeta}(\bm{S}|\tilde{t})}{h^2\cdot \kappa_2\cdot \bar{p}(T,\bm{S})}\cdot \left[Y - \bar{\mu}(\tilde{t},\bm{S}) - (T-\tilde{t})\cdot \bar{\beta}(\tilde{t},\bm{S})\right] d\tilde{t} \right] \right\} \right]\\
		&=\frac{1}{n} \cdot \mathrm{Var}\left\{\mathbb{E}_{T_1}\left[\int_{T_1}^t \frac{\left(\frac{T-t}{h}\right) K\left(\frac{T-\tilde{t}}{h}\right) \cdot \bar{p}_{\zeta}(\bm{S}|\tilde{t})}{h^2\cdot \kappa_2\cdot \bar{p}(T,\bm{S})}\cdot \left[Y - \bar{\mu}(\tilde{t},\bm{S}) - (T-\tilde{t})\cdot \bar{\beta}(\tilde{t},\bm{S})\right] d\tilde{t} \right] \right\}\\
		&= \frac{1}{n} \cdot \mathbb{E}\left[\left\{\mathbb{E}_{T_1}\left[\int_{T_1}^t \frac{\left(\frac{T-t}{h}\right) K\left(\frac{T-\tilde{t}}{h}\right) \cdot \bar{p}_{\zeta}(\bm{S}|\tilde{t})}{h^2\cdot \kappa_2\cdot \bar{p}(T,\bm{S})}\cdot \left[Y - \bar{\mu}(\tilde{t},\bm{S}) - (T-\tilde{t})\cdot \bar{\beta}(\tilde{t},\bm{S})\right] d\tilde{t} \right] \right\}^2 \right] - \frac{1}{n} \cdot \P^2\Lambda_t \\
		&\stackrel{\text{(i)}}{\leq} \frac{1}{nh^4} \cdot \mathbb{E}_{T_1}\left[\left|t-T_1\right| \int_{T_1}^t \mathbb{E}\left\{\frac{\left(\frac{T-t}{h}\right)^2 K^2\left(\frac{T-\tilde{t}}{h}\right) \cdot \bar{p}_{\zeta}^2(\bm{S}|\tilde{t})}{\kappa_2^2\cdot \bar{p}^2(T,\bm{S})}\cdot \left[Y - \bar{\mu}(\tilde{t},\bm{S}) - (T-\tilde{t})\cdot \bar{\beta}(\tilde{t},\bm{S})\right]^2 \right\} d\tilde{t} \right] - \frac{1}{n} \cdot \P^2\Lambda_t \\
		&\stackrel{\text{(ii)}}{\lesssim} \frac{1}{nh^3} \int_{\mathcal{T}}  \int_{\mathbb{R}} \int_{\mathcal{S}(t+uh)} |t-t_1| \cdot p_T(t_1)\cdot \frac{u^2 K^2\left(u\right) \cdot \bar{p}_{\zeta}^2(\bm{s}_1|t)\cdot p(t+uh,\bm{s}_1)}{\kappa_2^2\cdot \bar{p}^2(t+uh,\bm{s}_1)} \\
		&\quad \times \left\{\left[\mu(t+uh,\bm{s}_1) - \bar{\mu}(t,\bm{s}_1) - hu\cdot \bar{\beta}(t,\bm{s}_1)\right]^2 + \sigma^2\right\} d\bm{s}_1 du dt_1\\
		&\stackrel{\text{(iii)}}{=} O_P\left(\frac{1}{nh^3}\right),
	\end{align*}
	where (i) follows from Jensen's inequality on the squared function and Cauchy-Schwarz inequality as $\left[\int_{T_1}^t g(\tilde{t}) \,d\tilde{t} \right]^2 \leq \left[\int_{T_1}^t g(\tilde{t}) \,d\tilde{t} \right] \cdot |t-T_1|$ for the function $g(\tilde{t}) = \mathbb{E}\left[\phi_{C,h,\tilde{t}}\left(Y,T,\bm{S}; \bar{\mu},\bar{\beta}, \bar{p},\bar{p}_{\zeta}\right) \right]$, (ii) uses a change of variable and only keeps the dominating first term, and (iii) leverages our arguments in \autoref{app:theta_nopos_Term_VI}. Moreover, $\mathrm{Var}\left[\sqrt{nh^3}\cdot \textbf{Term Ia}\right]$ is strictly positive as long as $\mathrm{Var}(\epsilon|T,\bm{S})>\sigma^2>0$. Then, by Chebyshev's inequality, we obtain that
	\begin{align*}
		\textbf{Term Ia} = O_P\left(\sqrt{\mathrm{Var}\left[\textbf{Term Ia}\right]}\right) = O_P\left(\sqrt{\frac{1}{nh^3}}\right).
	\end{align*}
	
	$\bullet$ \textbf{Term Id:} By our calculation of the bias term in \autoref{app:theta_nopos_Term_VI}, we know that
	\begin{align*}
		\textbf{Term Id} &= h^2\cdot\mathbb{E}_{T_1}\left[\int_{T_1}^t B_{C,\theta}(\tilde{t})\, d\tilde{t} \right] + O(h^3) = O(h^2),
	\end{align*}
	where $B_{C,\theta}(t)$ is defined in \autoref{thm:theta_nopos} as:
	\begin{align*}
		B_{C,\theta}(t) &= \begin{cases}
			\frac{\kappa_4}{6\kappa_2} \int \left\{\frac{3\frac{\partial}{\partial t} p(t,\bm{s}) \cdot \bar{m}''(t) + p(t,\bm{s})\left[\bar{m}^{(3)}(t) - 3\frac{\partial}{\partial t} \log\bar{p}(t,\bm{s}) \cdot \bar{m}''(t) \right]}{\bar{p}(t,\bm{s})} \right\} \bar{p}_{\zeta}(\bm{s}|t)\, d\bm{s}& \text{ when } \bar{\mu}=\mu \text{ and } \bar{\beta}=\beta,\\
			\frac{\kappa_4}{6\kappa_2} \cdot \bar{m}^{(3)}(t) & \text{ when } \bar{p} = p.
		\end{cases}
	\end{align*}
	
	$\bullet$ \textbf{Term Ib:} By Chebyshev's inequality, we have that
	\begin{align*}
		\textbf{Term Ib} & = \left(\mathbb{P}_n-\P\right)\left\{\int_T^t \mathbb{E}\left[\frac{1}{\sqrt{h^3}}\cdot \phi_{C,h,\tilde{t}}\left(Y,T,\bm{S}; \bar{\mu},\bar{\beta}, \bar{p},\bar{p}_{\zeta}\right) \right] d\tilde{t}\right\}\\
		&= O_P\left(\sqrt{\frac{\mathrm{Var}\left\{\int_T^t \mathbb{E}\left[\frac{1}{\sqrt{h^3}}\cdot \phi_{C,h,\tilde{t}}\left(Y,T,\bm{S}; \bar{\mu},\bar{\beta}, \bar{p},\bar{p}_{\zeta}\right) \right] d\tilde{t}\right\}}{n}}\right) \\
		&\stackrel{\text{(i)}}{=}O_P\left(\sqrt{\frac{1}{nh^2}}\right)  =o_P\left(\sqrt{\frac{1}{nh^3}}\right).
	\end{align*}
	Here, the equality (i) above follows from the calculation that
	\begin{align*}
		& \mathrm{Var}\left\{\int_T^t \mathbb{E}\left[\frac{1}{\sqrt{h^3}}\cdot \phi_{C,h,\tilde{t}}\left(Y,T,\bm{S}; \bar{\mu},\bar{\beta}, \bar{p},\bar{p}_{\zeta}\right) \right] d\tilde{t}\right\}\\
		&= \int_{\mathcal{T}} \left\{\int_{t_2}^t \int_{\mathcal{T}} \int_{\mathcal{S}(t_1)} \frac{\left(\frac{t_1-\tilde{t}}{h}\right) K\left(\frac{t_1-\tilde{t}}{h}\right) \cdot \bar{p}_{\zeta}(\bm{s}_1|\tilde{t})}{h^2 \cdot \kappa_2\cdot \bar{p}(t_1,\bm{s}_1)}\cdot \left[\mu(t_1,\bm{s}_1) - \bar{\mu}(\tilde{t},\bm{s}_1) - (t_1-\tilde{t})\cdot \bar{\beta}(\tilde{t},\bm{s}_1)\right] d\bm{s}_1 dt_1 d\tilde{t}\right\}^2 p_T(t_2)\, dt_2\\
		&\stackrel{\text{(ii)}}{=} \int_{\mathcal{T}} \left\{\int_{t_2}^t \int_{\mathbb{R}} \int_{\mathcal{S}(\tilde{t}+uh)} \frac{u\cdot K\left(u\right) \cdot \bar{p}_{\zeta}(\bm{s}_1|\tilde{t})}{h \cdot \kappa_2\cdot \bar{p}(\tilde{t}+uh,\bm{s}_1)}\cdot \left[\mu(\tilde{t}+uh,\bm{s}_1) - \bar{\mu}(\tilde{t},\bm{s}_1) - uh\cdot \bar{\beta}(\tilde{t},\bm{s}_1)\right] d\bm{s}_1 du d\tilde{t}\right\}^2 p_T(t_2)\, dt_2\\
		&\stackrel{\text{(iii)}}{=} O\left(\frac{1}{h^2}\right),
	\end{align*}
	where (ii) utilizes a change of variable and (iii) utilizes the upper boundedness of $\mu,\bar{\mu}, \bar{\beta},\bar{p}_{\zeta},p_T$ under Assumptions~\ref{assump:reg_diff} and \ref{assump:den_diff} as well as the fact that $\bar{p}$ is lower bounded away from 0 around the support $\mathcal{J}$.\\
	
	$\bullet$ \textbf{Term Ic:} Recall from Assumption~\ref{assump:reg_kernel}(c) that $\mathcal{K} = \left\{t'\mapsto \left(\frac{t'-t}{h}\right)^{k_1} K\left(\frac{t'-t}{h}\right): t\in \mathcal{T}, h>0, k_1=0,1\right\}$ is a bounded VC-type class of measurable functions on $\mathbb{R}$. Under Assumption~\ref{assump:reg_diff} and the condition that $\mathbb{E}|\epsilon^{2+c_1}| <\infty$, we deduce by Theorem 4 in \cite{einmahl2005uniform} that with probability 1,
	\begin{align*}
		&\sup_{t\in \mathcal{T}}\left|\mathbb{G}_n\left[\phi_{C,h,t}\left(Y,T,\bm{S}; \bar{\mu},\bar{\beta}, \bar{p},\bar{p}_{\zeta}\right) \right] \right| \\
		&= \sup_{\tilde{t}\in \mathcal{T}}\left|\mathbb{G}_n\left\{\frac{\left(\frac{T-t}{h}\right) K\left(\frac{T-t}{h}\right) \cdot \bar{p}_{\zeta}(\bm{S}|t)}{\sqrt{h}\cdot \kappa_2\cdot \bar{p}(T,\bm{S})}\cdot \left[Y - \bar{\mu}(t,\bm{S}) - (T-t)\cdot \bar{\beta}(t,\bm{S})\right] \right\} \right| \\
		&= O\left(\sqrt{|\log h|}\right)
	\end{align*}
	when $\frac{|\log h|}{\log\log n} \to \infty$. Thus, by Chebyshev's inequality, 
	\begin{align*}
		\sqrt{nh^3}\cdot \textbf{Term Ic} &= \left(\mathbb{P}_n - \P\right)\left\{\int_{T_{i_1}}^t \mathbb{G}_n \left[\phi_{C,h,\tilde{t}}\left(Y,T,\bm{S}; \bar{\mu},\bar{\beta}, \bar{p},\bar{p}_{\zeta}\right) \right] d\tilde{t}\right\} \\
		&= O_P\left(\sqrt{\frac{1}{n}\cdot \mathbb{E}\left\{\left[\int_{T_{i_1}}^t \mathbb{G}_n \left[\phi_{C,h,\tilde{t}}\left(Y,T,\bm{S}; \bar{\mu},\bar{\beta}, \bar{p},\bar{p}_{\zeta}\right) \right] d\tilde{t} \right]^2\right\}} \right)\\
		&= O_P\left(\sqrt{\frac{|\log h|}{n}} \right) = o_P\left(\sqrt{\frac{1}{nh^3}}\right).
	\end{align*}
	Here, the last equality follows from the calculation that
	\begin{align*}
		\mathbb{E}\left\{\left[\int_{T_{i_1}}^t \mathbb{G}_n \left[\phi_{C,h,\tilde{t}}\left(Y,T,\bm{S}; \bar{\mu},\bar{\beta}, \bar{p},\bar{p}_{\zeta}\right) \right] d\tilde{t} \right]^2\right\} &\stackrel{\text{(i)}}{\leq} \mathbb{E}\left[|t-T_{i_1}|^2 \cdot \sup_{\tilde{t}\in \mathcal{T}}\left|\mathbb{G}_n \left[\phi_{C,h,\tilde{t}}\left(Y,T,\bm{S}; \bar{\mu},\bar{\beta}, \bar{p},\bar{p}_{\zeta}\right) \right] \right|^2 \right]\\
		&= O\left(|\log h|\right),
	\end{align*}
	where (i) applies the mean-value theorem for integrals.\\
	
	As a summary for this subsection, we conclude that
	\begin{align*}
		&\sqrt{nh^3} \cdot \textbf{Term I} \\
		&= \frac{1}{\sqrt{n}} \sum_{i=1}^n \left\{ \mathbb{E}_{T_1}\left[\int_{T_1}^t  \left\{\phi_{C,h,\tilde{t}}\left(Y_i,T_i,\bm{S}_i; \bar{\mu},\bar{\beta}, \bar{p},\bar{p}_{\zeta}\right) + \sqrt{h^3}\int \left[\bar{\beta}(\tilde{t},\bm{s}) - \beta(\tilde{t},\bm{s})\right] \bar{p}_{\zeta}(\bm{s}|\tilde{t})\, d\bm{s} \right\} d\tilde{t}\right] \right\} + o_P(1).
	\end{align*}

	\subsubsection{Analysis of \textbf{Term III} for $\hat{m}_{\mathrm{C,DR}}(t)$}
	\label{app:m_nopos_Term_III}
	
	Recall from those \textbf{Term VII--XI} for $\hat{\theta}_{\mathrm{C,DR}}(t)$ in \autoref{app:theta_DR_nopos} that \textbf{Term III} for $\hat{m}_{\mathrm{C,DR}}(t)$ here is given by
	\begin{align*}
		&\textbf{Term III} \\
		&= \frac{1}{n}\sum_{i=1}^n \int_{T_i}^t A_{n,h}(\tilde{t})\, d\tilde{t}\\
		&= \underbrace{\frac{1}{n}\sum_{i=1}^n \int_{T_i}^t\int \hat{\beta}(\tilde{t},\bm{s}) \left[\hat{p}_{\zeta}(\bm{s}|\tilde{t}) - \bar{p}_{\zeta}(\bm{s}|\tilde{t})\right]d\bm{s} d\tilde{t}}_{\textbf{Term IIIa}} \\
		&\quad + \underbrace{\frac{1}{n}\sum_{i=1}^n \int_{T_i}^t \left(\mathbb{P}_n-\P\right)\left\{\frac{\left(\frac{T-\tilde{t}}{h}\right) K\left(\frac{T-\tilde{t}}{h}\right)}{h^2\kappa_2}\left[\frac{\hat{p}_{\zeta}(\bm{S}|\tilde{t})}{\hat{p}(T,\bm{S})} - \frac{\bar{p}_{\zeta}(\bm{S}|\tilde{t})}{\bar{p}(T,\bm{S})} \right] \left[Y - \bar{\mu}(\tilde{t},\bm{S}) - (T-\tilde{t})\cdot \bar{\beta}(\tilde{t},\bm{S})\right]\right\} d\tilde{t}}_{\textbf{Term IIIb}}\\
		&\quad + \underbrace{\frac{1}{n}\sum_{i=1}^n \int_{T_i}^t \left(\mathbb{P}_n-\P\right)\left\{\frac{\left(\frac{T-\tilde{t}}{h}\right)K\left(\frac{T-\tilde{t}}{h}\right)\cdot \bar{p}_{\zeta}(\bm{S}|\tilde{t})}{h^2\kappa_2\cdot \bar{p}(T,\bm{S})} \left[\bar{\mu}(\tilde{t},\bm{S}) - \hat{\mu}(\tilde{t},\bm{S}) + (T-\tilde{t})\left[\bar{\beta}(\tilde{t},\bm{S}) - \hat{\beta}(\tilde{t},\bm{S})\right] \right]\right\} d\tilde{t}}_{\textbf{Term IIIc}} \\
		&\quad + \underbrace{\frac{1}{n}\sum_{i=1}^n \int_{T_i}^t \mathbb{P}_n\left\{\frac{\left(\frac{T-\tilde{t}}{h}\right) K\left(\frac{T-\tilde{t}}{h}\right)}{h^2\kappa_2}\left[\frac{\hat{p}_{\zeta}(\bm{S}|\tilde{t})}{\hat{p}(T,\bm{S})} - \frac{\bar{p}_{\zeta}(\bm{S}|\tilde{t})}{\bar{p}(T,\bm{S})} \right] \left[\bar{\mu}(\tilde{t},\bm{S}) - \hat{\mu}(\tilde{t},\bm{S}) + (T-\tilde{t})\left[\bar{\beta}(\tilde{t},\bm{S}) - \hat{\beta}(\tilde{t},\bm{S})\right]\right]\right\} d\tilde{t}}_{\textbf{Term IIId}} \\
		&\quad + \underbrace{\frac{1}{n}\sum_{i=1}^n \int_{T_i}^t \left(\P\left\{\frac{\left(\frac{T-\tilde{t}}{h}\right)^2 K\left(\frac{T-\tilde{t}}{h}\right) \cdot \bar{p}_{\zeta}(\bm{S}|\tilde{t})}{h \cdot \kappa_2\cdot \bar{p}(T,\bm{S})} \left[\bar{\beta}(\tilde{t},\bm{S}) - \hat{\beta}(\tilde{t},\bm{S})\right]\right\} + \int \left[\hat{\beta}(\tilde{t},\bm{s}) - \bar{\beta}(\tilde{t},\bm{s})\right]\bar{p}_{\zeta}(\bm{s}|\tilde{t})\, d\bm{s} \right) d\tilde{t}}_{\textbf{Term IIIe}} \\ 
		&\quad + \underbrace{\frac{1}{n}\sum_{i=1}^n \int_{T_i}^t \P\left\{\frac{\left(\frac{T-\tilde{t}}{h}\right) K\left(\frac{T-\tilde{t}}{h}\right)\cdot \bar{p}_{\zeta}(\bm{S}|\tilde{t})}{h^2 \cdot \kappa_2\cdot \bar{p}(T,\bm{S})}\left[\bar{\mu}(\tilde{t},\bm{S}) - \hat{\mu}(\tilde{t},\bm{S})\right]\right\} d\tilde{t}}_{\textbf{Term IIIf}}\\
		&\quad + \underbrace{\frac{1}{n}\sum_{i=1}^n \int_{T_i}^t \P\left\{\frac{\left(\frac{T-\tilde{t}}{h}\right) K\left(\frac{T-\tilde{t}}{h}\right)}{h^2\kappa_2}\left[\frac{\hat{p}_{\zeta}(\bm{S}|\tilde{t})}{\hat{p}(T,\bm{S})} - \frac{\bar{p}_{\zeta}(\bm{S}|\tilde{t})}{\bar{p}(T,\bm{S})} \right] \left[Y - \bar{\mu}(\tilde{t},\bm{S}) - (T-\tilde{t})\cdot \bar{\beta}(\tilde{t},\bm{S})\right]\right\} d\tilde{t}}_{\textbf{Term IIIg}}.
	\end{align*}
	We shall argue that all these terms above are of order $o_P\left(\sqrt{\frac{1}{nh^3}}\right)$ respectively as follows.
	
	$\bullet$ \textbf{Term IIIa:} By direct calculations, we know that
	\begin{align*}
		\sqrt{nh^3} \cdot \textbf{Term IIIa} &= \sqrt{\frac{h^3}{n}} \sum_{i=1}^n \int_{T_i}^t \int \hat{\beta}(\tilde{t},\bm{s})\left[\hat{p}_{\zeta}(\bm{s}|\tilde{t}) -\bar{p}_{\zeta}(\bm{s}|\tilde{t})\right] d\bm{s} d\tilde{t}\\
		&\leq \left(\sqrt{\frac{h^3}{n}} \sum_{i=1}^n |t-T_i|\right) \sup_{\tilde{t}\in \mathcal{T}} \left|\int \hat{\beta}(\tilde{t},\bm{s})\left[\hat{p}_{\zeta}(\bm{s}|\tilde{t}) -\bar{p}_{\zeta}(\bm{s}|\tilde{t})\right] d\bm{s} \right|\\
		&=O_P\left(\sqrt{nh^3} \cdot \bar{\Upsilon}_{5,n}\right) = o_P(1),
	\end{align*}
	where in the last equality, we use the compactness of the marginal support $\mathcal{T} \subset \mathbb{R}$ to argue that $\frac{1}{n} \sum_{i=1}^n |t-T_i| = O_P(1)$ for any $t\in \mathcal{T}$ and utilize our derivations in \autoref{app:theta_nopos_Term_VII} to obtain that $\sup_{\tilde{t}\in \mathcal{T}} \left|\int \hat{\beta}(\tilde{t},\bm{s})\left[\hat{p}_{\zeta}(\bm{s}|\tilde{t}) -\bar{p}_{\zeta}(\bm{s}|\tilde{t})\right] d\bm{s} \right| = O_P\left(\norm{\hat{p}_{\zeta}(\bm{S}|t) - \bar{p}_{\zeta}(\bm{S}|t)}_{L_2}\right) = O_P\left(\bar{\Upsilon}_{5,n}\right)$.\\
	
	$\bullet$ \textbf{Term IIIb:} Notice that \textbf{Term IIIb} can be written in a form of V-statistics. Specifically, 
	\begin{align*}
		&\sqrt{nh^3}\cdot \textbf{Term IIIb} \\
		&= \frac{1}{\sqrt{n}}\sum_{i=1}^n \int_{T_i}^t \left(\mathbb{P}_n-\P\right)\left\{\frac{\left(\frac{T-\tilde{t}}{h}\right) K\left(\frac{T-\tilde{t}}{h}\right)}{\sqrt{h}\cdot \kappa_2}\left[\frac{\hat{p}_{\zeta}(\bm{S}|\tilde{t})}{\hat{p}(T,\bm{S})} - \frac{\bar{p}_{\zeta}(\bm{S}|\tilde{t})}{\bar{p}(T,\bm{S})} \right] \left[Y - \bar{\mu}(\tilde{t},\bm{S}) - (T-\tilde{t})\cdot \bar{\beta}(\tilde{t},\bm{S})\right]\right\} d\tilde{t}\\
		&:= \frac{1}{n^{\frac{3}{2}}} \sum_{i_1=1}^n \sum_{i_2=1}^n \int_{T_{i_1}}^t \left\{\bm{Z}_{i_2}(\tilde{t}) - \mathbb{E}\left[\bm{Z}_{i_2}(\tilde{t}) \right]\right\} d\tilde{t}
	\end{align*}
	with $\bm{Z}_i(\tilde{t}) = \frac{\left(\frac{T_i-\tilde{t}}{h}\right) K\left(\frac{T_i-\tilde{t}}{h}\right)}{\sqrt{h}\cdot \kappa_2}\left[\frac{\hat{p}_{\zeta}(\bm{S}_i|\tilde{t})}{\hat{p}(T_i,\bm{S})} - \frac{\bar{p}_{\zeta}(\bm{S}_i|\tilde{t})}{\bar{p}(T_i,\bm{S}_i)} \right] \left[Y_i - \bar{\mu}(\tilde{t},\bm{S}_i) - (T_i-\tilde{t})\cdot \bar{\beta}(\tilde{t},\bm{S}_i)\right]$. Note that the random variables $\int_{T_i}^t \left\{\bm{Z}_j(\tilde{t}) - \mathbb{E}\left[\bm{Z}_j(\tilde{t}) \right]\right\} d\tilde{t}$ and $\int_{T_k}^t \left\{\bm{Z}_{\ell}(\tilde{t}) - \mathbb{E}\left[\bm{Z}_{\ell}(\tilde{t}) \right]\right\} d\tilde{t}$ are dependent and have a nonzero covariance when any of $i,j,k,\ell$ coincides. Thus, the variance of 
	$$\sum_{i_1=1}^n \sum_{i_2=1}^n \int_{T_{i_1}}^t \left\{\bm{Z}_{i_2}(\tilde{t}) - \mathbb{E}\left[\bm{Z}_{i_2}(\tilde{t}) \right]\right\} d\tilde{t}$$ 
	involves a sum of $O(n^3)$ nonzero terms; see, \emph{e.g.}, Section 6.1 in \cite{lehmann1999elements} for detailed explanations. By Chebyshev's inequality,
	\begin{align*}
		\sqrt{nh^3}\cdot \textbf{Term IIIb} &= \frac{1}{n^{\frac{3}{2}}} \sum_{i_1=1}^n \sum_{i_2=1}^n \int_{T_{i_1}}^t \left\{\bm{Z}_{i_2}(\tilde{t}) - \mathbb{E}\left[\bm{Z}_{i_2}(\tilde{t}) \right]\right\} \\
		&= O_P\left(\mathrm{Var}\left[\int_{T_{i_1}}^t \left\{\bm{Z}_{i_2}(\tilde{t}) - \mathbb{E}\left[\bm{Z}_{i_2}(\tilde{t}) \right]\right\} \right] \right)\\
		&= O_P\left(\sqrt{\mathbb{E}\left[\left(\int_{T_{i_1}}^t \left\{\bm{Z}_{i_2}(\tilde{t}) - \mathbb{E}\left[\bm{Z}_{i_2}(\tilde{t}) \right]\right\} \right)^2 \right]}\right)\\
		&=o_P(1).
	\end{align*}
	Here, the last equality follows from the calculations that
	\begin{align*}
		&\mathbb{E}\left[\left(\int_{T_{i_1}}^t \left\{\bm{Z}_{i_2}(\tilde{t}) - \mathbb{E}\left[\bm{Z}_{i_2}(\tilde{t}) \right]\right\} \right)^2 \right] \\ &\stackrel{\text{(i)}}{\leq} \mathbb{E}\left[(t-T_{i_1})^2\left|\bm{Z}_{i_2}(t') - \mathbb{E}\left[\bm{Z}_{i_2}(t')\right]\right|^2\right]\\
		&= \sup_{t'\in \mathcal{T}} \mathbb{E}\left[(t-T_{i_1})^2\left|\bm{Z}_{i_2}(t')\right|^2\right]\\
		&\stackrel{\text{(ii)}}{\lesssim} \sup_{t'\in \mathcal{T}}\mathbb{E}\left\{\frac{\left(\frac{T_i-t'}{h}\right)^2 K^2\left(\frac{T_i-t'}{h}\right)}{h\cdot \kappa_2^2}\left[\frac{\hat{p}_{\zeta}(\bm{S}_i|t')}{\hat{p}(T_i,\bm{S})} - \frac{\bar{p}_{\zeta}(\bm{S}_i|t')}{\bar{p}(T_i,\bm{S}_i)} \right]^2 \left[Y_i - \bar{\mu}(t',\bm{S}_i) - (T_i-t')\cdot \bar{\beta}(t',\bm{S}_i)\right]^2\right\}\\
		&\stackrel{\text{(iii)}}{\lesssim} \sup_{t\in \mathcal{T}}\norm{\hat{p}_{\zeta}(\bm{S}|t) - \bar{p}_{\zeta}(\bm{S}|t)}_{L_2}^2 + \sup_{u\in \mathcal{T} \oplus h} \norm{\hat{p}(u,\bm{S}) - \bar{p}(u,\bm{S})}_{L_2}^2\\
		&= O_P\left(\bar{\Upsilon}_{5,n}^2 + \bar{\Upsilon}_{6,n}^2\right) = o_P(1),
	\end{align*}
	where (i) applies the mean-value theorem for integrals with $t'$ lying between $t, T_{i_1}\in \mathcal{T}$, (ii) uses the compactness of $\mathcal{T} \subset \mathbb{R}$, and (iii) utilizes our derivations in \autoref{app:theta_nopos_Term_VIII}. \\
	
	$\bullet$ \textbf{Term IIIc:} Analogous to our arguments for \textbf{Term IIIb}, we write \textbf{Term IIIc} in terms of V-statistics and deduce that
	\begin{align*}
		&\sqrt{nh^3}\cdot \textbf{Term IIIc} \\
		&= \frac{1}{\sqrt{n}}\sum_{i=1}^n \int_{T_i}^t \left(\mathbb{P}_n-\P\right)\left\{\frac{\left(\frac{T-\tilde{t}}{h}\right)K\left(\frac{T-\tilde{t}}{h}\right)\cdot \bar{p}_{\zeta}(\bm{S}|\tilde{t})}{\sqrt{h}\cdot \kappa_2\cdot \bar{p}(T,\bm{S})} \left[\bar{\mu}(\tilde{t},\bm{S}) - \hat{\mu}(\tilde{t},\bm{S}) + (T-\tilde{t})\left[\bar{\beta}(\tilde{t},\bm{S}) - \hat{\beta}(\tilde{t},\bm{S})\right] \right]\right\} d\tilde{t}\\
		&:= \frac{1}{n^{\frac{3}{2}}} \sum_{i_1=1}^n \sum_{i_2=1}^n \int_{T_{i_1}}^t \left\{\bar{\bm{Z}}_{i_2}(\tilde{t}) - \mathbb{E}\left[\bar{\bm{Z}}_{i_2}(\tilde{t}) \right]\right\} d\tilde{t}\\
		&= O_P\left(\sqrt{\mathbb{E}\left[\left(\int_{T_{i_1}}^t \left\{\bar{\bm{Z}}_{i_2}(\tilde{t}) - \mathbb{E}\left[\bar{\bm{Z}}_{i_2}(\tilde{t}) \right]\right\} \right)^2 \right]}\right)\\
		&=o_P(1).
	\end{align*}
	Here, $\bar{\bm{Z}}_i(\tilde{t}) = \frac{\left(\frac{T_i-\tilde{t}}{h}\right)K\left(\frac{T_i-\tilde{t}}{h}\right)\cdot \bar{p}_{\zeta}(\bm{S}_i|\tilde{t})}{\sqrt{h}\cdot \kappa_2\cdot \bar{p}(T_i,\bm{S}_i)} \left[\bar{\mu}(\tilde{t},\bm{S}_i) - \hat{\mu}(\tilde{t},\bm{S}_i) + (T_i-\tilde{t})\left[\bar{\beta}(\tilde{t},\bm{S}_i) - \hat{\beta}(\tilde{t},\bm{S}_i)\right] \right]$, and the last equality above follows from some similar calculations as:
	\begin{align*}
		&\mathbb{E}\left[\left(\int_{T_{i_1}}^t \left\{\bar{\bm{Z}}_{i_2}(\tilde{t}) - \mathbb{E}\left[\bar{\bm{Z}}_{i_2}(\tilde{t}) \right]\right\} \right)^2 \right]\\
		&\leq \mathbb{E}\left[(t-T_{i_1})^2\left|\bar{\bm{Z}}_{i_2}(t') - \mathbb{E}\left[\bar{\bm{Z}}_{i_2}(t')\right]\right|^2\right]\\
		&= \sup_{t'\in \mathcal{T}} \mathbb{E}\left[(t-T_{i_1})^2\left|\bar{\bm{Z}}_{i_2}(t')\right|^2\right]\\
		&\lesssim \sup_{t'\in \mathcal{T}} \mathbb{E}\left\{\frac{\left(\frac{T_i-\tilde{t}}{h}\right)^2 K^2\left(\frac{T_i-\tilde{t}}{h}\right)\cdot \bar{p}_{\zeta}^2(\bm{S}_i|\tilde{t})}{h\cdot \kappa_2^2\cdot \bar{p}^2(T_i,\bm{S}_i)} \left[\bar{\mu}(\tilde{t},\bm{S}_i) - \hat{\mu}(\tilde{t},\bm{S}_i) + (T_i-\tilde{t})\left[\bar{\beta}(\tilde{t},\bm{S}_i) - \hat{\beta}(\tilde{t},\bm{S}_i)\right] \right]^2 \right\}\\
		&\stackrel{\text{(iv)}}{\lesssim} \sup_{t\in \mathcal{T}}\left[\norm{\hat{\mu}(t,\bm{S}) - \bar{\mu}(t,\bm{S})}_{L_2}^2 + h^2\norm{\hat{\beta}(t,\bm{S}) - \bar{\beta}(t,\bm{S})}_{L_2}^2 \right] \\
		&= O_P\left(\bar{\Upsilon}_{1,n}^2 + h^2\bar{\Upsilon}_{3,n}^2\right) = o_P(1),
	\end{align*}
	where (iv) again leverages our derivations in \autoref{app:theta_nopos_Term_VIII}.\\
	
	$\bullet$ \textbf{Term IIId:} Similar to our arguments for \textbf{Term IIIb}, we also write \textbf{Term IIId} in terms of V-statistics and utilize Markov's inequality to deduce that
	\begin{align*}
		&\sqrt{nh^3}\cdot \textbf{Term IIId} \\
		&= \frac{1}{\sqrt{n}} \sum_{i=1}^n \int_{T_i}^t \mathbb{P}_n\left\{\frac{\left(\frac{T-\tilde{t}}{h}\right) K\left(\frac{T-\tilde{t}}{h}\right)}{\sqrt{h}\cdot \kappa_2}\left[\frac{\hat{p}_{\zeta}(\bm{S}|\tilde{t})}{\hat{p}(T,\bm{S})} - \frac{\bar{p}_{\zeta}(\bm{S}|\tilde{t})}{\bar{p}(T,\bm{S})} \right] \left[\bar{\mu}(\tilde{t},\bm{S}) - \hat{\mu}(\tilde{t},\bm{S}) + (T-\tilde{t})\left[\bar{\beta}(\tilde{t},\bm{S}) - \hat{\beta}(\tilde{t},\bm{S})\right]\right]\right\} d\tilde{t} \\
		&:= \frac{1}{n^{\frac{3}{2}}} \sum_{i_1=1}^n \sum_{i_2=1}^n \int_{T_{i_1}}^t \bm{V}_{i_2}(\tilde{t}) \, d\tilde{t} \\
		&=O_P\left(\sqrt{n}\cdot \mathbb{E}\left|\int_{T_{i_1}}^t \bm{V}_{i_2}(\tilde{t}) \, d\tilde{t}\right|\right)\\
		&=o_P(1).
	\end{align*}
	Here, $\bm{V}_i(\tilde{t}) = \frac{\left(\frac{T_i-\tilde{t}}{h}\right) K\left(\frac{T_i-\tilde{t}}{h}\right)}{\sqrt{h}\cdot \kappa_2}\left[\frac{\hat{p}_{\zeta}(\bm{S}_i|\tilde{t})}{\hat{p}(T_i,\bm{S}_i)} - \frac{\bar{p}_{\zeta}(\bm{S}_i|\tilde{t})}{\bar{p}(T_i,\bm{S}_i)} \right] \left[\bar{\mu}(\tilde{t},\bm{S}_i) - \hat{\mu}(\tilde{t},\bm{S}_i) + (T_i-\tilde{t})\left[\bar{\beta}(\tilde{t},\bm{S}_i) - \hat{\beta}(\tilde{t},\bm{S}_i)\right]\right]$, and the last equality above follows from the following calculation as:
	\begin{align*}
		&\sqrt{n}\cdot \mathbb{E}\left|\int_{T_{i_1}}^t \bm{V}_{i_2}(\tilde{t}) \, d\tilde{t}\right|\\
		&\stackrel{\text{(i)}}{\leq} \sqrt{n}\cdot \mathbb{E}\left[|t-T_{i_1}|\cdot \left|\bm{V}_{i_2}(t') \right| \right] \\
		&\stackrel{\text{(ii)}}{\lesssim} \sqrt{n}\cdot \sup_{\tilde{t}\in \mathcal{T}} \mathbb{E}\left|\frac{\left(\frac{T-\tilde{t}}{h}\right) K\left(\frac{T-\tilde{t}}{h}\right)}{\sqrt{h}\cdot \kappa_2}\left[\frac{\hat{p}_{\zeta}(\bm{S}|\tilde{t})}{\hat{p}(T,\bm{S})} - \frac{\bar{p}_{\zeta}(\bm{S}|\tilde{t})}{\bar{p}(T,\bm{S})} \right] \left[\bar{\mu}(\tilde{t},\bm{S}) - \hat{\mu}(\tilde{t},\bm{S}) + (T-\tilde{t})\left[\bar{\beta}(\tilde{t},\bm{S}) - \hat{\beta}(\tilde{t},\bm{S})\right]\right]\right|\\
		&\stackrel{\text{(iii)}}{\lesssim} \sqrt{nh} \cdot \sup\limits_{t\in \mathcal{T}}\Bigg\{\left[\norm{\hat{p}_{\zeta}(\bm{S}|t) - p_{\zeta}(\bm{S}|t)}_{L_2} + \sup\limits_{|u-t|\leq h} \norm{\hat{p}(u,\bm{S}) - p(u,\bm{S})}_{L_2} \right]\\
		&\quad \times \left[\norm{\hat{\mu}(t,\bm{S}) - \mu(t,\bm{S})}_{L_2} + h \norm{\hat{\beta}(t,\bm{S}) - \beta(t,\bm{S})}_{L_2}\right] \Bigg\}\\
		&\stackrel{\text{(iv)}}{=} o_P(1),
	\end{align*}
	where applies the mean-value theorem for integrals with $t'$ lying between $t, T_{i_1}\in \mathcal{T}$, (ii) uses the compactness of $\mathcal{T} \subset \mathbb{R}$, (iii) utilizes our derivations in \autoref{app:theta_nopos_Term_X}, and (iv) leverages our assumption (c) on the doubly robust rate of convergence in the corollary statement. \\
	
	$\bullet$ \textbf{Term IIIe, Term IIIf, and Term IIIg:} On one hand, when $\bar{p}(t,\bm{s}) = p(t,\bm{s})$ for all $(t,\bm{s})\in \mathcal{T}\times \mathcal{S}$, we know from our calculations in \autoref{app:theta_nopos_Term_XI} that $\textbf{Term IIIe} = \textbf{Term IIIf}=0$ and 
	\begin{align*}
		&\textbf{Term IIIg} \\
		&\stackrel{\text{(i)}}{\leq} \left(\frac{1}{n}\sum_{i=1}^n |T_i-t| \right) \cdot \sup_{\tilde{t}\in \mathcal{T}}\left|\P\left\{\frac{\left(\frac{T-\tilde{t}}{h}\right) K\left(\frac{T-\tilde{t}}{h}\right)}{h^2\kappa_2}\left[\frac{\hat{p}_{\zeta}(\bm{S}|\tilde{t})}{\hat{p}(T,\bm{S})} - \frac{\bar{p}_{\zeta}(\bm{S}|\tilde{t})}{p(T,\bm{S})} \right] \left[Y - \bar{\mu}(\tilde{t},\bm{S}) - (T-\tilde{t})\cdot \bar{\beta}(\tilde{t},\bm{S})\right]\right\} \right|\\
		&\stackrel{\text{(ii)}}{\lesssim} \left(\frac{1}{n}\sum_{i=1}^n |T_i-t| \right) \cdot \frac{1}{h} \left[\sup_{t\in \mathcal{T}} \norm{\hat{p}_{\zeta}(\bm{S}|t) - \bar{p}_{\zeta}(\bm{S}|t)}_{L_2} + \sup_{u\in \mathcal{T} \oplus h} \norm{\hat{p}(u,\bm{S}) - p(u,\bm{S})}_{L_2}\right]\\
		&\stackrel{\text{(iii)}}{=} O_P\left(\frac{1}{h} \left(\bar{\Upsilon}_{5,n} + \bar{\Upsilon}_{6,n}\right)\right)\\
		&=o_P\left(\sqrt{\frac{1}{nh^3}}\right),
	\end{align*}
	where (i) utilizes the mean-value theorem for integrals, (ii) follows from our derivations in \autoref{app:theta_nopos_Term_XI}, and (iii) applies Markov's inequality to $\frac{1}{n}\sum_{i=1}^n |T_i-t|$ as well as leverages our assumption (c) on the doubly robust rate of convergence in the corollary statement. Specifically, since $\norm{\hat{\mu}(t,\bm{S}) - \bar{\mu}(t,\bm{S})}_{L_2} + h \norm{\hat{\beta}(t,\bm{S}) - \bar{\beta}(t,\bm{S})}_{L_2}= O_P(1)$ for any $t\in \mathcal{T}$ when $\bar{\mu}\neq \mu$ and $\bar{\beta}\neq \beta$, our assumption (c) ensures that $\sup_{t\in \mathcal{T}}\norm{\hat{p}_{\zeta}(\bm{S}|t) - \bar{p}_{\zeta}(\bm{S}|t)}_{L_2} + \sup_{u\in \mathcal{T} \oplus h} \norm{\hat{p}(u,\bm{S}) - p(u,\bm{S})}_{L_2} =o_P\left(\sqrt{\frac{1}{nh}}\right)$.
	
	On the other hand, when $\bar{\mu}=\mu$ and $\bar{\beta}=\beta$, we again know from our calculations in \autoref{app:theta_nopos_Term_XI} that 
	\begin{align*}
		&\textbf{Term IIIe} \\
		&\leq \left(\frac{1}{n}\sum_{i=1}^n |T_i-t| \right) \\
		&\quad \times \sup_{\tilde{t}\in \mathcal{T}}\left|\P\left\{\frac{\left(\frac{T-\tilde{t}}{h}\right)^2 K\left(\frac{T-\tilde{t}}{h}\right) \cdot \bar{p}_{\zeta}(\bm{S}|\tilde{t})}{h \cdot \kappa_2\cdot \bar{p}(T,\bm{S})} \left[\beta(\tilde{t},\bm{S}) - \hat{\beta}(\tilde{t},\bm{S})\right]\right\} + \int \left[\hat{\beta}(\tilde{t},\bm{s}) - \beta(\tilde{t},\bm{s})\right]\bar{p}_{\zeta}(\bm{s}|\tilde{t})\, d\bm{s} \right|\\
		&= O_P\left(\sup_{t\in \mathcal{T}}\norm{\hat{\beta}(t,\bm{S}) - \beta(t,\bm{S})}_{L_2} \right)\\
		&=o_P\left(\sqrt{\frac{1}{nh^3}}\right),
	\end{align*}
	\begin{align*}
		\textbf{Term IIIf} &\leq \left(\frac{1}{n}\sum_{i=1}^n |T_i-t| \right) \cdot \sup_{\tilde{t}\in \mathcal{T}}\left|\P\left\{\frac{\left(\frac{T-\tilde{t}}{h}\right) K\left(\frac{T-\tilde{t}}{h}\right)\cdot \bar{p}_{\zeta}(\bm{S}|\tilde{t})}{h^2 \cdot \kappa_2\cdot \bar{p}(T,\bm{S})}\left[\mu(\tilde{t},\bm{S}) - \hat{\mu}(\tilde{t},\bm{S})\right]\right\} \right|\\
		&= O_P\left(\frac{1}{h}\cdot \sup_{t\in \mathcal{T}}\norm{\hat{\mu}(t,\bm{S}) - \mu(t,\bm{S})}_{L_2} \right)\\
		&=o_P\left(\sqrt{\frac{1}{nh^3}}\right),
	\end{align*}
	and
	\begin{align*}
		&\textbf{Term IIIg} \\
		&\leq \left(\frac{1}{n}\sum_{i=1}^n |T_i-t| \right) \cdot \sup_{\tilde{t}\in \mathcal{T}}\left|\P\left\{\frac{\left(\frac{T-\tilde{t}}{h}\right) K\left(\frac{T-\tilde{t}}{h}\right)}{h^2\kappa_2}\left[\frac{\hat{p}_{\zeta}(\bm{S}|\tilde{t})}{\hat{p}(T,\bm{S})} - \frac{\bar{p}_{\zeta}(\bm{S}|\tilde{t})}{\bar{p}(T,\bm{S})} \right] \left[Y - \bar{\mu}(\tilde{t},\bm{S}) - (T-\tilde{t})\cdot \bar{\beta}(\tilde{t},\bm{S})\right]\right\} \right|\\
		&=O_P\left(h^2\left[\sup_{t\in \mathcal{T}} \norm{\hat{p}_{\zeta}(\bm{S}|t) - p_{\zeta}(\bm{S}|t)}_{L_2} + \sup_{u\in \mathcal{T} \oplus h} \norm{\hat{p}(u,\bm{S}) - p(u,\bm{S})}_{L_2} \right]\right)\\
		&=o_P\left(\sqrt{\frac{1}{nh^3}}\right).
	\end{align*}
	
	As a summary for this subsection, we conclude that $\sqrt{nh^3}\cdot \textbf{Term III} = o_P(1)$.

	\subsubsection{Asymptotic Normality of $\hat{m}_{\mathrm{C,DR}}(t)$}
	\label{app:m_nopos_asym_norm}
	
	For the asymptotic normality of $\hat{m}_{\mathrm{C,DR}}(t)$, it follows from the Lyapunov central limit theorem. Specifically, we already show in \autoref{app:m_nopos_Term_I} and \autoref{app:m_nopos_Term_III} that
	\begin{align*}
		&\sqrt{nh^3}\left[\hat{m}_{\mathrm{C,DR}}(t) - m(t)\right] \\
		&= \frac{1}{\sqrt{n}} \sum_{i=1}^n \left\{ \mathbb{E}_{T_1}\left[\int_{T_1}^t  \left\{\phi_{C,h,\tilde{t}}\left(Y_i,T_i,\bm{S}_i; \bar{\mu},\bar{\beta}, \bar{p},\bar{p}_{\zeta}\right) + \sqrt{h^3} \int \left[\bar{\beta}(\tilde{t},\bm{s}) - \beta(\tilde{t},\bm{s})\right] \bar{p}_{\zeta}(\bm{s}|\tilde{t})\, d\bm{s} \right\} d\tilde{t}\right] \right\} + o_P(1) \\
		&= \frac{1}{\sqrt{n}} \sum_{i=1}^n \mathbb{E}_{T_1}\left[\int_{T_1}^t \phi_{C,h,\tilde{t}}\left(Y_i,T_i,\bm{S}_i;\bar{\mu}, \bar{\beta}, \bar{p},\bar{p}_{\zeta}\right) d\tilde{t}\right]  +o_P(1)
	\end{align*}
	with 
	$$\phi_{C,h,t}\left(Y,T,\bm{S}; \bar{\mu},\bar{\beta}, \bar{p},\bar{p}_{\zeta}\right) = \frac{\left(\frac{T-t}{h}\right) K\left(\frac{T-t}{h}\right) \cdot \bar{p}_{\zeta}(\bm{S}|t)}{\sqrt{h}\cdot \kappa_2\cdot \bar{p}(T,\bm{S})}\cdot \left[Y - \bar{\mu}(t,\bm{S}) - (T-t)\cdot \bar{\beta}(t,\bm{S})\right]$$
	and $V_{C,m}(t) = \mathbb{E}\left[\left\{\mathbb{E}_{T_1}\left[\int_{T_1}^t \phi_{C,h,\tilde{t}}\left(Y,T,\bm{S};\bar{\mu}, \bar{\beta}, \bar{p},\bar{p}_{\zeta}\right) d\tilde{t}\right] \right\}^2 \right] = O(1)$ by our calculation of \textbf{Term Ia} in \autoref{app:m_nopos_Term_I}. Then, 
	$$\sum_{i=1}^n \mathrm{Var}\left\{\frac{1}{\sqrt{n}} \cdot \mathbb{E}_{T_1}\left[\int_{T_1}^t \phi_{C,h,\tilde{t}}\left(Y_i,T_i,\bm{S}_i;\bar{\mu}, \bar{\beta}, \bar{p},\bar{p}_{\zeta}\right) d\tilde{t}\right] \right\} = O(1)$$ 
	and 
	\begin{align*}
		&\sum_{i=1}^n \mathbb{E}\left|\frac{1}{\sqrt{n}} \cdot \mathbb{E}_{T_1}\left[\int_{T_1}^t \phi_{C,h,\tilde{t}}\left(Y,T,\bm{S};\bar{\mu}, \bar{\beta}, \bar{p},\bar{p}_{\zeta}\right) d\tilde{t}\right] \right|^{2+c_4} \\
		&\leq \mathbb{E}\left[\frac{1}{\sqrt{n^{c_4}h^{c_4+2}}}\cdot \mathbb{E}_{T_1}\left|\int_{T_1}^t \frac{\left(\frac{T-\tilde{t}}{h}\right) K\left(\frac{T-\tilde{t}}{h}\right) \cdot \bar{p}_{\zeta}(\bm{S}|\tilde{t})}{\kappa_2\cdot \bar{p}(T,\bm{S})}\cdot \left[Y - \bar{\mu}(\tilde{t},\bm{S}) - (T-\tilde{t})\cdot \bar{\beta}(\tilde{t},\bm{S})\right] d\tilde{t} \right|^{2+c_4} \right]\\
		&= \mathbb{E}\Bigg\{\frac{1}{\sqrt{n^{c_4}h^{c_4+2}}} \mathbb{E}_{T_1}\Bigg[|t-T_1|^{2+c_4} \\
		&\quad \times \sup_{\tilde{t}\in \mathcal{T}} \frac{\left(\frac{T-\tilde{t}}{h}\right)^{2+c_4} K^{2+c_4}\left(\frac{T-\tilde{t}}{h}\right) \cdot \bar{p}_{\zeta}^{2+c_4}(\bm{S}|\tilde{t})}{\kappa_2^{2+c_4}\cdot \bar{p}^{2+c_4}(T,\bm{S})}\cdot \left[\mu(T,\bm{S}) + \epsilon - \bar{\mu}(\tilde{t},\bm{S}) - (T-\tilde{t})\cdot \bar{\beta}(\tilde{t},\bm{S})\right]^{2+c_4} d\tilde{t} \Bigg] \Bigg\}\\
		&=O\left(\sqrt{\frac{1}{n^{c_4}h^{2+c_4}}}\right) =o(1)
	\end{align*}
	by the upper boundedness of $\mu,\bar{\mu},p$ under Assumptions~\ref{assump:reg_diff} and \ref{assump:den_diff}, the upper boundedness of the kernel function under Assumption~\ref{assump:reg_kernel}(c), the lower boundedness of $\bar{p}$ away from 0 around the support $\mathcal{J}$, the assumption that $\mathbb{E}|\epsilon|^{2+c_1}<\infty$ for some constant $c_1\geq 1$, and the requirement that $nh^3\to\infty$ as $n\to \infty$. Hence, the Lyapunov condition holds, and we have that
	$$\sqrt{nh^3}\left\{\hat{m}_{\mathrm{C,DR}}(t) - m(t) - h^2\cdot  \mathbb{E}_{T_1}\left[\int_{T_1}^t B_{C,\theta}(\tilde{t})\, d\tilde{t} \right]\right\} \stackrel{d}{\to} \mathcal{N}\left(0,V_{C,m}(t)\right)$$
	after subtracting the dominating bias term $h^2\cdot  \mathbb{E}_{T_1}\left[\int_{T_1}^t B_{C,\theta}(\tilde{t})\, d\tilde{t} \right]$ that we have computed in \textbf{Term Id} in \autoref{app:m_nopos_Term_I}. The proof is thus completed.
\end{proof}

\end{document}